\documentclass{lmcs}
\pdfoutput=1

\usepackage{lastpage}
\lmcsdoi{16}{4}{6}
\lmcsheading{}{\pageref{LastPage}}{}{}%
{Mar.~27,~2019}{Oct.~23,~2020}{}

\usepackage[english]{babel}	% option clash with LIPIcs
\usepackage[utf8]{inputenc}	% already preselectionned by LIPIcs

\usepackage{booktabs}   %% For formal tables:
\usepackage{amsmath}
\usepackage{amsthm}
\usepackage{amssymb}
\usepackage{stmaryrd}

\usepackage{cmll}
\usepackage{mathrsfs}

\usepackage{array}
\usepackage[all]{xy}

\usepackage{color}

\usepackage{etex}

\usepackage{bm}

\DeclareOldFontCommand{\sf}{\normalfont\sffamily}{\mathsf}

\usepackage[%
  void
]{optional}

\usepackage{comment}

\usepackage{ifthen}

\usepackage{macros}

\ifdef\note
  {\renewcommand\note[1]{\opt{draftnotes}{{\sf \color{orange} #1}}}}
  {\newcommand\note[1]{\opt{draftnotes}{{\sf \color{orange} #1}}}}
\newcommand\cnote[1]{\opt{draftnotes}{\begin{center}\note{#1}\end{center}}}
\newcommand{\fntext}{}
\newcommand{\fn}{\footnote{\fntext}}

\renewcommand\vec[1]{\bm{#1}}

\newcommand\CR{{\color{blue} CR}}

\newcommand\TODO{{\color{red} TODO}}

\opt{fullproofs}{\newcommand\flagfullproofs{}}
\ifdef{\flagfullproofs}{
  \newenvironment{fullproof}{\begin{proof}}{\end{proof}}
}{\excludecomment{fullproof}}

\ifdef{\flagfullproofs}
  {\newcommand\full[1]{#1}}
  {\newcommand\full[1]{}}
\ifdef{\flagfullproofs}
  {\newcommand\short[1]{}}
  {\newcommand\short[1]{#1}}

\theoremstyle{plain}
\newtheorem{subclm}[thm]{Claim}

\newenvironment{subproof}
  {\begin{proof}}
  {\end{proof}}

\begin{document}
%%%%%%%%%%%%%%%%%%%%%%%%%%%%%%%%%%%%%%%%%%%%%%%%%%%%%%%%%%%%%%%%%%%%%%%%%%%

\title{A Functional (Monadic) Second-Order Theory of Infinite Trees}
%\title{A Complete Axiomatization of $\MSO$ on Infinite Trees}
%Positional determinacy and uniform strategies in a
%functional second-order theory of infinite trees}

\author[A.~Das]{Anupam Das}	%required
\address{University of Birmingham}	%required
\email{anupam.das@di.ku.dk}  %optional
%\thanks{}	%optional

\author[C.~Riba]{Colin Riba}	%optional
\address{LIP -- ENS de Lyon}	%optional
\email{colin.riba@ens-lyon.fr}  %optional
%\thanks{}	%optional

\date{\today}

%%%%%%%%%%%%%%%%%%%%%%%%%%%%%%%%%%%%%%%%%%%%%%%%%%%%%%%%%%%%%%%%%%%%%%%%%%%
\begin{abstract}
%%%%%%%%%%%%%%%%%%%%%%%%%%%%%%%%%%%%%%%%%%%%%%%%%%%%%%%%%%%%%%%%%%%%%%%%%%%
\noindent
This paper presents a complete axiomatization of Monadic Second-Order
Logic ($\MSO$) over infinite trees.
$\MSO$ on infinite trees is a rich system, 
and its decidability (``Rabin's Tree Theorem'')
is one of the most powerful known results concerning the decidability of logics.

By a complete axiomatization we mean a 
complete deduction system
with a polynomial-time recognizable set of axioms.
By naive enumeration of formal derivations, this formally gives
a proof of Rabin's Tree Theorem.
%This formally gives, by naive enumeration of formal derivations,
%a new proof of Rabin's Tree Theorem.
The deduction system consists of the usual rules
for second-order logic seen as two-sorted first-order logic,
together with the natural adaptation 
to infinite trees of the axioms of $\MSO$ on $\omega$-words.
In addition, it contains an axiom scheme expressing the (positional)
determinacy of certain parity games.

The main difficulty resides in the limited expressive power of the language of $\MSO$.
We actually devise an extension of $\MSO$,
called \emph{Functional (Monadic) Second-Order Logic} ($\FSO$),
which allows us to uniformly manipulate (hereditarily) finite sets
and corresponding labeled trees,
and whose language allows for higher
abstraction than that of $\MSO$.
\end{abstract}

%%% Local Variables:
%%% mode: latex
%%% TeX-master: "main.tex"
%%% End:

\maketitle

%%%%%%%%%%%%%%%%%%%%%%%%%%%%%%%%%%%%%%%%%%%%%%%%%%%%%%%%%%%%%%%%%%%%%%%%%%%
\section{Introduction}
%%%%%%%%%%%%%%%%%%%%%%%%%%%%%%%%%%%%%%%%%%%%%%%%%%%%%%%%%%%%%%%%%%%%%%%%%%%

\noindent
This paper presents a complete axiomatization of Monadic Second-Order
Logic ($\MSO$) over infinite trees.
$\MSO$ on infinite trees is a rich system which
contains non trivial mathematical theories (see~\eg~\cite{rabin69tams,bgg97cdp})
and subsumes many logics, 
including modal logics (see \eg~\cite{brv02modal})
and logics for verification (see \eg~\cite{vw08chapter}).
Rabin's Tree Theorem~\cite{rabin69tams}, the decidability of $\MSO$ on infinite trees,
is one of the most powerful known results concerning the decidability of logics
(see \eg~\cite{bgg97cdp}).

The original decidability proof of~\cite{rabin69tams} relied on 
an effective translation
of formulae to finite state automata running on infinite trees.
Since then, there has been considerable work
on Rabin's Tree Theorem, culminating in streamlined decidability proofs,
as presented \eg\@ in ~\cite{thomas97handbook,gtw02alig,pp04book}.
Most current approaches to $\MSO$ on infinite trees
(with the notable exception of~\cite{blumensath13tcs})
are based on translations of formulae to automata.

By a `complete axiomatization' we mean a
complete deduction system
with a \emph{polynomial-time} recognizable set of axioms and rules.
This condition on axiom/rule recognizability is typical in proof theory, where it is known as the \emph{Cook-Reckhow criterion} \cite{CR79}.
The point is that proofs should be `easily checkable', 
%and we also rule out trivial axiomatizations that are less insightful,
%which rules out trivial axiomatizations %that are less insightful,
%such as the set of all true formulae.
which rules out axiomatizations based on enumerations of all true formulae.
In this way, a complete axiomatization not only constitutes an alternative demonstration of Rabin's Tree theorem itself, by naive enumeration of formal derivations, but also yields a meaningful notion of `proof certificate' for theorems.

Our deduction system consists of the usual rules
for second-order logic seen as two-sorted first-order logic
(see \eg~\cite{riba12ifip}), together with the natural adaptation 
to infinite trees of the axioms of $\MSO$ on $\omega$-words~\cite{siefkes70lnm}.
In addition, it contains an axiom scheme expressing the (positional)
determinacy of certain parity games.

We continue a line of work begun by B\"uchi and Siefkes,
who gave axiomatizations of $\MSO$ on various classes of linear orders
(see \eg~\cite{siefkes70lnm,bs73lnm}),
as well as an axiomatization of \emph{Weak} $\MSO$ ($\WMSO$) over infinite 
trees~\cite{siefkes78ijm}
($\WMSO$ is $\MSO$ with set quantifications restricted to \emph{finite} sets).
These works essentially rely on formalizations of automata in the logic.
A major result in the axiomatic treatment of logics over infinite structures
is Walukiewicz's proof of completeness of Kozen's axiomatization
of the modal $\mu$-calculus~\cite{walukiewicz00ic}
(see also~\cite{al17lics} for an alternative recent proof of this result).
Another trend relies on model-theoretic techniques.
For instance~\cite{gf10fossacs,gc12lmcs}
give complete axiomatizations of
$\MSO$ and the modal $\mu$-calculus over finite trees;
a reworking of the completeness of $\MSO$ on $\omega$-words~\cite{siefkes70lnm}
is proposed in~\cite{riba12ifip};
and~\cite{sv10apal} gives a model-theoretic completeness proof 
for a fragment of the modal $\mu$-calculus.
An attractive feature of model-theoretic completeness proofs 
for the aforementioned logics
is that they
allow elegant reformulations of algebraic approaches to these logics.
Unfortunately, in the case of $\MSO$ over infinite trees,
the only known algebraic approach~\cite{blumensath13tcs}
seems too complex to be easily formalized.
%Model theory allows elegant reformulations of algebraic approaches.
%However, for the case of infinite trees the algebraic approach of~\cite{blumensath13tcs}
%appears to be much more complicated.
We therefore directly formalize a translation
of formulae to automata in the axiomatic theory.

%More recent works 
%include complete axiomatizations of
%$\MSO$ and the modal $\mu$-calculus over finite trees~\cite{gf10fossacs,gc12lmcs}.
%These works (as well as the reworking~\cite{riba12ifip}
%of Siefkes's completeness proof of $\MSO$ on $\omega$-words~\cite{siefkes70lnm}),
%rely on model-theoretic techniques, which allow elegant reformulations
%of algebraic approaches.
%For the case of infinite trees the algebraic approach of~\cite{blumensath13tcs}
%appears to be much more complicated.
%We therefore directly formalize a translation
%of formulae to automata in the axiomatic theory.

Mirroring usual automata based decidability proofs
(see \eg~\cite{thomas97handbook,gtw02alig,pp04book}),
our method for proving completeness proceeds in two steps.
We first formalize
a translation of $\MSO$-formulae to tree automata
(using the positional determinacy of parity games to prove the complementation lemma),
so that each closed formula is provably equivalent to an automaton over the
singleton alphabet.
The second (and much shorter) step is a variant of the Büchi-Landweber Theorem~\cite{bl69tams}
which states that $\MSO$ decides winning for (definable) games of finite graphs,
and which is obtained thanks to the completeness of $\MSO$ over $\omega$-words.
%The second (and much shorter) step relies on the completeness of
%$\MSO$ on $\omega$-words to obtain a variant of the Büchi-Landweber Theorem~\cite{bl69tams},
%stating that $\MSO$ decides winning for (definable) games of finite graphs.

The main expositional difficulty resides in the limited expressive power of the language of $\MSO$.
To ameliorate this we actually devise an extension of $\MSO$,
called \emph{Functional (Monadic) Second-Order Logic} ($\FSO$),
allowing uniform manipulation of (hereditarily) finite sets
%and corresponding labelings of the full $\Dir$-ary tree.
and corresponding labeled infinite trees.
We intuitively see $\FSO$ as providing a language for higher
abstraction than that of $\MSO$,
allowing a uniform formalization of automata and games
which would have been difficult to write down in $\MSO$.
However, since $\FSO$ is interpretable in $\MSO$ (as we show), its language
has the same intrinsic limitations as the language
of $\MSO$. In particular it suffers from the inexpressibility of choice over tree
positions~\cite{gs83choice,cl07csl},
%A consequence %of this limitation
and so predicates such as length comparison
of tree positions are not expressible in $\FSO$.
This implies that only positional strategies 
(\wrt\@ our specific notion of acceptance games),
are expressible in $\FSO$
and moreover that usually unproblematic reasoning on infinite plays
can become cumbersome in this setting.
%and moreover that some reasonings on infinite plays
%which can be easily performed in usual mathematical practice
%can become cumbersome in our setting.

There are several ways to translate $\MSO$ to tree automata.
%(see \eg~\cite{gtw02alig,thomas97handbook,walukiewicz02tcs}).
We choose to translate formulae to alternating parity automata,
following~\cite{walukiewicz02tcs}.
The two non-trivial steps in the translation %of formulae to alternating automata
are negation and (existential) quantification.
Negation requires the complementation of automata,
relying on the determinacy of acceptance games,
while
%On the other hand,
existential quantifiers
require us
to simulate an alternating automaton by an equivalent
non-deterministic one
(this is the \emph{Simulation Theorem}~\cite{ej91focs,ms95tcs}),
thence obtaining an automaton computing the appropriate projection.

%the \emph{Simulation Theorem}~\cite{ej91focs,ms95tcs},
%stating that each alternating automaton is equivalent to a
%non-deterministic one,
%on which 
%%(which are required to handle existential quantifications).
%%to simulate an alternating automaton by an equivalent
%%non-deterministic one
%%(this is the \emph{Simulation Theorem}~\cite{ej91focs,ms95tcs})
%%whence we can obtain an automaton computing the appropriate projection.

%%%%%%%%%%%%%%%%%%%%%%%%%%%%%%%%%%%%%%%%%%%%%%%%%%%%%%%%%%%%%%%%%%%%%%%%%%%
\renewcommand\fntext{The approach of~\cite{ms95tcs}
to the Simulation Theorem actually contains a proof of McNaughton Theorem,
but we do not see how to easily formalize it in our context.}
%but seems not to be realistically formalizable in our context.}
%%%%%%%%%%%%%%%%%%%%%%%%%%%%%%%%%%%%%%%%%%%%%%%%%%%%%%%%%%%%%%%%%%%%%%%%%%%
%Besides, 
As usual with translations of $\MSO$ to tree automata,
%(see \eg~\cite{thomas97handbook,ej91focs,ms95tcs}),
we rely on McNaughton's Theorem~\cite{mcnaughton66ic}
(see also \eg~\cite{thomas90handbook,pp04book}),
stating that non-deterministic Büchi automata on $\omega$-words
are effectively equivalent to deterministic \emph{parity} (or Muller, Rabin, Streett) automata
on $\omega$-words.
In translations of $\MSO$ to alternating tree automata,
McNaughton's Theorem is usually invoked for the Simulation Theorem.\fn\@
In our context, the relevant instances of McNaughton's Theorem
are imported into $\FSO$ via the completeness of $\MSO$ on
$\omega$-words~\cite{siefkes70lnm}.

%Since we translate $\MSO$-formulae to alternating tree automata,
%we invoke McNaughton's Theorem in our formalized proof
%of the Simulation Theorem.
%%stating that each alternating automaton is equivalent to a
%%non-deterministic one (which are required to handle existential quantifications).
%In our context, relevant instances of McNaughton's Theorem
%are imported in $\FSO$ via the completeness of $\MSO$ on $\omega$-words.

It is well-known that the $\MSO$ theory of $k$-ary trees
can be embedded in that of the binary tree~\cite{rabin69tams}.
However, it does not seem that such an embedding
yields an axiomatization of $k$-ary trees from an axiomatization of the binary tree.
Therefore, in this work, we axiomatize the $\MSO$ theory of the full infinite $\Dir$-ary
tree for an arbitrary non-empty finite set $\Dir$.
%We denote by $\MSOD$ and $\FSOD$ the corresponding formal systems.

This paper is a corrected version of~\cite{dr15lics},
which contains a flaw in the positional determinacy argument (Thm.\@ VI.15).
In the present paper, we augment the systems $\FSO$ and $\MSO$
with an axiom expressing the positional determinacy of parity games,
thereby obtaining complete axiomatizations.
We do not know yet whether the theory $\MSO$ of~\cite{dr15lics} is complete,
but let us mention that the axiomatization of $\WMSO$
over infinite trees given in~\cite{siefkes78ijm} augments the natural analogue 
for trees of Peano's arithmetic with an axiom of induction over finite trees.

%%%%%%%%%%%%%%%%%%%%%%%%%%%%%%%%%%%%%%%%%%%%%%%%%%%%%%%%%%%%%%%%%%%%%%%%%%%
\subsection*{Outline.}
%%%%%%%%%%%%%%%%%%%%%%%%%%%%%%%%%%%%%%%%%%%%%%%%%%%%%%%%%%%%%%%%%%%%%%%%%%%
The paper is organized as follows.
We present the basic formal theory for $\MSO$ in~\S\ref{sec:mso} 
%by presenting 
%with
%a basic formal theory for $\MSO$.
Our theory $\FSO$ is then presented in~\S\ref{sec:fso}
and we sketch its mutual interpretability with $\MSO$.
\S\ref{sec:pos} and~\S\ref{sec:games} discuss a formalization
of two-players infinite games in $\FSO$,
and, in particular, we give a formulation of the axiom $(\PosDet)$
of positional determinacy of parity games.
This provides us with the required tools to formalize in~\S\ref{sec:aut}
(alternating) tree automata, acceptance games
and basic operations on them
(including complementation in $\FSO + (\PosDet)$).
%This formalization of games is then applied in~\S\ref{sec:aut} to define a representation
%of (alternating) tree automata and acceptance games in $\FSO$.
\S\ref{sec:msow} is an interlude discussing
a complete theory of $\MSO$ over $\omega$-words
within the infinite paths of $\FSO$.
Building on~\S\ref{sec:aut} and~\S\ref{sec:msow},
we then give our completeness argument for 
$\FSO + (\PosDet)$ 
and $\MSO + (\PosDet)$ 
in~\S\ref{sec:compl}.
\full{Finally,~\S\ref{sec:sim} contains a proof of the Simulation Theorem in $\FSO$,
and the mutual interpretations of $\FSO$ and $\MSO$
are proved correct
in Appendix~\ref{sec:app:cons}.}%
\short{Finally,~\S\ref{sec:sim} contains a proof of the Simulation Theorem in $\FSO$.
The full version~\cite{dr20full} gives details of proofs omitted here.}

%We then turn to our completeness argument:
%\S\ref{sec:msow} presents a complete theory of $\MSO$ over $\omega$-words
%within the infinite paths of $\FSO$.
%Together with our formalization of automata, this gives our completeness
%result in~\S\ref{sec:compl}.
%Our completeness argument then occupies~\S\ref{sec:msow} and~\S\ref{sec:compl}.

%%% Local Variables:
%%% mode: latex
%%% TeX-master: "main.tex"
%%% End:

%%%%%%%%%%%%%%%%%%%%%%%%%%%%%%%%%%%%%%%%%%%%%%%%%%%%%%%%%%%%%%%%%%%%%%%%%%%
\section{Preliminaries: $\MSO$ on Infinite Trees as a Second-Order Logic}
\label{sec:mso}
%%%%%%%%%%%%%%%%%%%%%%%%%%%%%%%%%%%%%%%%%%%%%%%%%%%%%%%%%%%%%%%%%%%%%%%%%%%

\noindent
%This preliminary Section presents 
We present here
a basic formal theory
of \emph{Monadic Second-Order Logic} ($\MSO$) over infinite trees.
This theory can be seen as an analogue for trees of Peano's axioms
for second order arithmetic.
In order to obtain a complete theory, $\MSO$
will be augmented with 
%the Axiom $\MI{(\PosDet)}$ of positional determinacy of parity games 
an axiom of positional determinacy of parity games 
(see~\S\ref{sec:cons},~\S\ref{sec:posdet} and~\S\ref{sec:compl}).

We are going to define the theory $\MSOD$ of the infinite full $\Dir$-ary tree $\univ$,
for $\Dir$ a finite non-empty set.
Both the language and the axioms of $\MSOD$ will depend on $\Dir$.
The language of $\MSOD$ is the usual language of two-sorted first-order
logic, with one sort for \emph{Individuals}
and one sort for \emph{(Monadic)} \emph{Predicates}.
The axioms of $\MSOD$ are the expected axioms on the relational structure
of the full $\Dir$-ary tree, together with induction and comprehension.
The theory $\MSOD$ is essentially that of~\cite{siefkes78ijm},
but with second-order quantifications intended to range over arbitrary
subsets of $\univ$ (instead of just finite ones), and without the axiom
of induction over finite trees.

We fix for the rest of this Section
a finite non-empty set $\Dir$ of \emph{tree directions}.

%%%%%%%%%%%%%%%%%%%%%%%%%%%%%%%%%%%%%%%%%%%%%%%%%%%%%%%%%%%%%%%%%%%%%%%%%%%
%\paragraph{Language.}
\subsection{The Language of $\MSOD$}
\label{sec:mso:lgge}
%%%%%%%%%%%%%%%%%%%%%%%%%%%%%%%%%%%%%%%%%%%%%%%%%%%%%%%%%%%%%%%%%%%%%%%%%%%

The language of $\MSOD$ has two sorts:
\begin{itemize}
\item
The sort of \emph{Individuals}, 
%denoted $t,u,\etc$,
%and 
intended to range over tree positions $p \in \univ$.
We have infinitely many Individual variables $x,y,z$ etc.
We also have one constant symbol $\Root$ (for the root of $\univ$),
and one unary function symbol $\Succ_d$
for each $d \in \Dir$ (for the successor function $p\mapsto p.d$).
Individual terms, written $t,u, \etc.$ are given by:
\begin{equation*}
\tag{for $d \in \Dir$}
t \quad\bnf\quad
x \gs \Root \gs \Succ_{d}(t)
%\qquad\qquad\text{(for $d \in \Dir$)}
\end{equation*}

\item
The sort of \emph{(Monadic) Predicates}, 
with variables $X,Y,Z,\etc$,
intended to range over sets of tree positions $A \in \Po(\univ)$.
There are no other term formers for this sort.
\end{itemize}

\noindent
Formulae of $\MSOD$ are given by the following grammar:
\[
\varphi,\psi \in \F_\Dir \quad\bnf\quad
  X(t)
\gs t \Eq u
\gs t \Lt u
\gs (\varphi \lor \psi)
\gs \lnot \varphi
\gs (\exists x)\varphi
\gs (\exists X)\varphi
\]
where $t$ and $u$ are Individual terms.
We use the usual derived formulae:
\[
\begin{array}{r !{~~\deq~~} l !{\qquad} r !{~~\deq~~} l}
  (\forall x)\varphi
& \lnot (\exists x)(\lnot \varphi)
& \varphi \land \psi
& \lnot(\lnot \varphi \lor \lnot \psi)
\\
  (\forall X)\varphi
& \lnot (\exists X)(\lnot \varphi)
& \varphi \limp \psi
& \lnot \varphi \lor \psi
\\
  \True
& (\forall x)(x \Eq x)
& \False
& \lnot \True
\\
  (t \Leq u)
& (t \Lt u) \lor (t\Eq u)
\end{array}
\]

\noindent
We employ usual writing conventions for formulae,
for instance omitting internal and external brackets when appropriate.

%%%%%%%%%%%%%%%%%%%%%%%%%%%%%%%%%%%%%%%%%%%%%%%%%%%%%%%%%%%%%%%%%%%%%%%%%%%
\subsection{The Deduction System of $\MSOD$}
\label{sec:mso:th}
%%%%%%%%%%%%%%%%%%%%%%%%%%%%%%%%%%%%%%%%%%%%%%%%%%%%%%%%%%%%%%%%%%%%%%%%%%%
Deduction for $\MSOD$ is defined by the system presented
in Figure~\ref{fig:ded:prop} and Figure~\ref{fig:ded:pred}
(where $\Phi$ stands for
a multiset of %(not nec.\@ distinct)
formulae),
together with the following axioms.
\begin{itemize}
\item \emph{Equality on Individuals:} %(for all formula $\varphi$):
\begin{equation}
\tag{for each $\varphi$}
(\forall x) (x \Eq x)
\qquad\text{and}\qquad
(\forall x) (\forall y)
\big( x \Eq y ~~\longlimp~~ \varphi[x/z] ~~\longlimp~~ \varphi[y/z] \big)
\end{equation}

\item The \emph{Tree Axioms} of Figure~\ref{fig:ded:tree}.

\item \emph{Comprehension Scheme:}
\begin{equation}
\tag{for each $\varphi$, with $X$ not free in $\varphi$}
(\exists X) (\forall y)\big[X(y) ~~\longliff~~ \varphi\big]
%\qquad\text{(for each $\varphi$, with $X$ not free in $\varphi$)}
\end{equation}

\item \emph{Induction Axiom:}
\[
(\forall X)\left(
X(\Root) ~~\longlimp~~
\mathord{\bigconj}_{d \in \Dir}
(\forall y)
\big[
X(y) ~~\longlimp~~
  X(\Succ_d(y))
\big]
~~\longlimp~~ (\forall y) X(y)
\right)
\]
\end{itemize}

%%%%%%%%%%%%%%%%%%%%%%%%%%%%%%%%%%%%%%%%%%%%%%%%%%%%%%%%%%%%%%%%%%%%%%%%%%%
\begin{figure}[tbp]
%%%%%%%%%%%%%%%%%%%%%%%%%%%%%%%%%%%%%%%%%%%%%%%%%%%%%%%%%%%%%%%%%%%%%%%%%%%
\[
\begin{array}{c !{\qquad} c !{\qquad} c}
  \dfrac{}{\Phi \thesis \varphi \lor \lnot \varphi}
%\qquad\qquad
& \dfrac{}{\Phi, \varphi \thesis \varphi}
%\qquad\qquad
& \dfrac{\Phi \thesis \varphi \qquad \Phi \thesis \lnot \varphi}
  {\Phi \thesis \psi}
\\\\
  \dfrac{\Phi \thesis \varphi}{\Phi \thesis \varphi \lor \psi}
& \dfrac{\Phi \thesis \psi}{\Phi \thesis \varphi \lor \psi}
& \dfrac{\Phi \thesis \varphi \lor \psi
\qquad \Phi, \varphi \thesis \vartheta
\qquad \Phi, \psi \thesis \vartheta}
  {\Phi \thesis \vartheta}
\\[1em]
\end{array}
\]
%\hrule
\caption{Deduction Rules for Propositional Logic.\label{fig:ded:prop}}
\end{figure}

%%%%%%%%%%%%%%%%%%%%%%%%%%%%%%%%%%%%%%%%%%%%%%%%%%%%%%%%%%%%%%%%%%%%%%%%%%%
\begin{figure}[tbp]
%%%%%%%%%%%%%%%%%%%%%%%%%%%%%%%%%%%%%%%%%%%%%%%%%%%%%%%%%%%%%%%%%%%%%%%%%%%
\[
\begin{array}{c}
\dfrac{\Phi \thesis \varphi[t/x]}{\Phi \thesis (\exists x)\varphi}
\qquad\qquad
\dfrac{\Phi \thesis (\exists x) \varphi \qquad \Phi,\varphi \thesis \psi}
  {\Phi \thesis \psi}
~\text{($x$ not free in $\Phi,\psi$)}
\\\\
\dfrac{\Phi \thesis \varphi[Y/X]}{\Phi \thesis (\exists X) \varphi}
\qquad\qquad
\dfrac{\Phi \thesis (\exists X) \varphi \qquad \Phi,\varphi \thesis \psi}
  {\Phi \thesis \psi}
~\text{($X$ not free in $\Phi,\psi$)}
\\[1em]
\end{array}
\]
%\hrule
\caption{Deduction Rules for Predicate Logic.\label{fig:ded:pred}}
\end{figure}

%%%%%%%%%%%%%%%%%%%%%%%%%%%%%%%%%%%%%%%%%%%%%%%%%%%%%%%%%%%%%%%%%%%%%%%%%%%
\begin{rem}
\label{rem:mso:ded}
%%%%%%%%%%%%%%%%%%%%%%%%%%%%%%%%%%%%%%%%%%%%%%%%%%%%%%%%%%%%%%%%%%%%%%%%%%%
As usual, one can derive
$\thesis
(\varphi \limp \psi \limp \vartheta) ~\liff~
((\varphi \land \psi) \limp \vartheta)$
and we have the \emph{Deduction Theorem}:
\[
\Phi,\varphi \thesis \psi
\qquad\text{iff}\qquad
\Phi \thesis \varphi \limp \psi
\]
Indeed, 
if $\Phi,\varphi\thesis \psi$, then one gets $\Phi \thesis \lnot \varphi \lor \psi$
by $\lor$-Elimination on the Excluded Middle $\Phi \thesis \varphi \lor \lnot \varphi$.
Conversely, if $\Phi \thesis \lnot \varphi \lor \psi$,
then one gets $\Phi,\varphi \thesis \psi$ by $\lor$-Elimination.
One similarly obtains the \emph{Modus Ponens} as a derived rule
\[
\dfrac{\Phi \thesis \psi \limp \varphi \qquad \Phi\thesis \psi}
  {\Phi \thesis \varphi}
\]
%%%%%%%%%%%%%%%%%%%%%%%%%%%%%%%%%%%%%%%%%%%%%%%%%%%%%%%%%%%%%%%%%%%%%%%%%%%
\end{rem}
%%%%%%%%%%%%%%%%%%%%%%%%%%%%%%%%%%%%%%%%%%%%%%%%%%%%%%%%%%%%%%%%%%%%%%%%%%%

%%%%%%%%%%%%%%%%%%%%%%%%%%%%%%%%%%%%%%%%%%%%%%%%%%%%%%%%%%%%%%%%%%%%%%%%%%%
\begin{figure}[tbp]
%%%%%%%%%%%%%%%%%%%%%%%%%%%%%%%%%%%%%%%%%%%%%%%%%%%%%%%%%%%%%%%%%%%%%%%%%%%
\[
\begin{array}{c !{\qquad\quad} c}
\lnot (\exists x) \bigdisj_{d \neq d'}
\big( \Succ_{d}(x) \Eq \Succ_{d'}(x) \big)
%\quad\text{if $d \neq d'$}
&
(\forall x) (\forall y) \bigconj_{d \in \Dir}
  \big( \Succ_d(x) \Eq \Succ_d(y) ~~\limp~~ x \Eq y \big)
\\\\
\lnot (\exists x) \big(x \Lt x \big)
&
(\forall x) (\forall y) (\forall z)
\big( x \Lt y ~~\limp~~ y \Lt z ~~\limp~~ x \Lt z \big)
\\\\
(\forall x) \big(\Root \Leq x \big)
&
(\forall x) (\forall y)
\Big(
\big( \bigdisj_{d \in \Dir} x \Lt \Succ_d(y) \big) ~~\liff~~ x \Leq y \Big)
\end{array}
\]
%\hrule
\caption{Tree Axioms of $\MSOD$ and $\FSOD$
(where $(x \Leq y)$ stands for $(x \Lt y \lor x \Eq y)$).\label{fig:ded:tree}}
\end{figure}

%%%%%%%%%%%%%%%%%%%%%%%%%%%%%%%%%%%%%%%%%%%%%%%%%%%%%%%%%%%%%%%%%%%%%%%%%%%
\begin{nota}
%%%%%%%%%%%%%%%%%%%%%%%%%%%%%%%%%%%%%%%%%%%%%%%%%%%%%%%%%%%%%%%%%%%%%%%%%%%
Henceforth, we write $\MSO$ instead of $\MSOD$ when the set of directions
$\Dir$ is clear from the context.
\end{nota}

%%% Local Variables:
%%% mode: latex
%%% TeX-master: "main.tex"
%%% End:

%%%%%%%%%%%%%%%%%%%%%%%%%%%%%%%%%%%%%%%%%%%%%%%%%%%%%%%%%%%%%%%%%%%%%%%%%%%
\section{A Functional Extension of $\MSO$ on Infinite Trees}
\label{sec:fso}
%%%%%%%%%%%%%%%%%%%%%%%%%%%%%%%%%%%%%%%%%%%%%%%%%%%%%%%%%%%%%%%%%%%%%%%%%%%

\noindent
In this Section,
we present (bounded) \emph{Functional (Monadic) Second-Order Logic}
over the full $\Dir$-ary tree ($\FSOD$),
an extension of $\MSOD$ %over the full infinite $\Dir$-ary tree
with (hereditarily) \emph{finite sets} and bounded quantification over them.
As with $\MSOD$ in~\S\ref{sec:mso}, we will simply write $\FSO$
for $\FSOD$ when $\Dir$ is irrelevant or clear from the context.

$\FSOD$ is equipped with a basic axiomatization which will allow us,
in~\S\ref{sec:pos}-\S\ref{sec:aut}, to formalize a basic theory of games
and automata,
and in particular to state an axiom scheme $(\PosDet)$
expressing the positional determinacy of (suitably represented)
parity games (\S\ref{sec:posdet}).
We will then show in~\S\ref{sec:compl}
that $\FSOD + (\PosDet)$ is complete.

%%%%%%%%%%%%%%%%%%%%%%%%%%%%%%%%%%%%%%%%%%%%%%%%%%%%%%%%%%%%%%%%%%%%%%%%%%%
\subsection{Motivations and Overview.}
\label{sec:fso:intro}
%%%%%%%%%%%%%%%%%%%%%%%%%%%%%%%%%%%%%%%%%%%%%%%%%%%%%%%%%%%%%%%%%%%%%%%%%%%
Let us first discuss the motivations and guiding principles
in the design of $\FSOD$.
As usual, within the language of $\MSOD$ presented in~\S\ref{sec:mso},
we can simulate a labeling of $\univ$ over a finite non-empty set $\Sigma$
\[
T : \univ \longto \Sigma
\]
There are different ways to achieve this.
A possibility is, for say $\Sigma = \{\al a_1,\dots,\al a_n\}$,
to code $T : \univ \to \Sigma$ using a tuple of Monadic variables
$X_1,\dots,X_n$ such that
\begin{equation}
\tag{for $i = 1,\dots,n$}
x \in X_i \quad\text{iff}\quad
T(x) = \al a_i
\end{equation}
A more succinct coding could be obtained using $\lceil \log n \rceil$
monadic variables to encode the letter index $i$ of $\al a_i$ in binary. 
%\noindent
However, directly working with such codings would make it
cumbersome to formalize games and automata
as presented in this paper.
We will therefore rather work in the system $\FSOD$,
%Instead, we will work in the system $\FSO$,
which is built around the following principles:
\begin{enumerate}
\item
$\FSOD$ has no primitive notion of \emph{Monadic variables}.
Instead, $\FSOD$ has a primitive notion of \emph{Function variables},
of the form
\begin{equation}
\tag{$\Sigma$ a finite set}
F : \univ \longto \Sigma
\end{equation}
%for $\Sigma$ a non-empty finite set.

\item
In addition, $\FSOD$ allows us to work \emph{uniformly} with arbitrary
finite sets.
In particular, we have an explicit sort for them, including
terms, variables and quantifications.

\item
$\FSOD$ is faithfully interpretable in $\MSOD$.
To this end, all quantifications over finite sets in $\FSOD$-formulae
are required to be bounded.
\end{enumerate}

\noindent
In particular, there is a syntactic translation $\MI{-}$
of $\FSOD$-formulae to $\MSOD$-formulae.
The basic idea of this translation is to interpret finite sets
using propositional logic,
and to interpret Functions
$F : \univ \to \{\al a_1,\dots,\al a_n\}$
as partitions $X_1,\dots,X_n$ of $\univ$.
But while $\FSOD$ handles free variables over finite sets in a uniform way,
the translation $\MI{-}$ only applies to $\FSOD$-formulae
without free variables over finite sets.
This means that for an $\FSOD$-formula $\varphi(k)$ with $k$ a variable over
finite sets, for each finite set $\kappa$ we will have a specific
$\MSOD$-formula $\MI{\varphi(\kappa)}$.

Technically, the finite sets of $\FSOD$ will be the usual
\emph{hereditarily} finite sets.

%%%%%%%%%%%%%%%%%%%%%%%%%%%%%%%%%%%%%%%%%%%%%%%%%%%%%%%%%%%%%%%%%%%%%%%%%%%
\begin{defi}
\label{def:fso:hf}
%%%%%%%%%%%%%%%%%%%%%%%%%%%%%%%%%%%%%%%%%%%%%%%%%%%%%%%%%%%%%%%%%%%%%%%%%%%
Let $V_0 \deq \emptyset$,
and $V_{n+1} \deq \Po(V_n)$ for each $n \in \NN$.
The set $V_\omega$
of \emph{hereditarily finite} sets (HF-sets)
is defined as
%is the set $V_\omega$ defined as
\[
V_\omega \quad\deq\quad \bigcup_{n \in \NN} V_n
\]
\end{defi}

%%%%%%%%%%%%%%%%%%%%%%%%%%%%%%%%%%%%%%%%%%%%%%%%%%%%%%%%%%%%%%%%%%%%%%%%%%%
\begin{rem}
\label{rem:fso:hf}
%%%%%%%%%%%%%%%%%%%%%%%%%%%%%%%%%%%%%%%%%%%%%%%%%%%%%%%%%%%%%%%%%%%%%%%%%%%
In the context of this paper, it is useful to note that,
as is well-known (see \eg~\cite[Exercise 12.9]{jech06set}),
$V_\omega$
is a model of $\ZFCM$
(\ie\@ of $\ZFC$ without the infinity axiom).
\end{rem}

%%%%%%%%%%%%%%%%%%%%%%%%%%%%%%%%%%%%%%%%%%%%%%%%%%%%%%%%%%%%%%%%%%%%%%%%%%%
\begin{conv}
%%%%%%%%%%%%%%%%%%%%%%%%%%%%%%%%%%%%%%%%%%%%%%%%%%%%%%%%%%%%%%%%%%%%%%%%%%%
We will always assume the finite non-empty set $\Dir$
of \emph{tree directions} to be an $\HF$-set.
\end{conv}

\noindent
The language of $\FSOD$ will have the same sort of Individuals as $\MSOD$
and a sort for $\HF$-sets, and
its Function variables will be
%We will also have a sort for Functions variables
of the form $F:\univ \to K$ for $K$ a term over $\HF$-sets ($\HF$-term).
The design of $\FSOD$ is obtained as a compromise between the following two
conflicting desiderata:
\begin{enumerate}
\item To be as flexible as possible to allow an easy formalization of games
and automata.
\item To be as simple as possible to allow an easy translation to $\MSOD$.
\end{enumerate}

\noindent
This leads us to two peculiar design choices.
\begin{enumerate}
\item We have, in addition to the above mentioned sorts, a distinct sort of
\emph{Functions over $\HF$-sets}.
This sort contains only constants (so these functions cannot be quantified over),
whose purpose is to provide Skolem functions
for those $\forall\exists$ (bounded) statements over $\HF$-sets
which are provable in $\ZFCM$.

\item In order to facilitate the translation of $\FSOD$ to $\MSOD$,
Function variables, written $(F:K)$ (``$F$ has codomain $K$''),
cannot occur in $\HF$-terms.
%are not directly allowed in $\HF$-terms.
%The interaction of Functions $(F:K)$ with $\HF$-sets is
%formally reduced to atomic formulae of the form
%Formally, the interaction of Functions $(F:K)$ with $\HF$-sets is
%restricted to atomic formulae of the form
Formally, Functions $(F:K)$ are only allowed in atomic formulae of the form
\[
\tag{for $L$ an $\HF$-term}
F(t) \Eq L
\]
%where $L$ is an $\HF$-term.
\end{enumerate}

\noindent
The axioms of $\FSOD$ will contain the obvious adaptation of the Tree Axioms
and the Induction Axiom of $\MSOD$.
We also have axioms defining the aforementioned Skolem functions.
In addition, the Comprehension Scheme of $\MSOD$ will be replaced
by \emph{Functional Choice Axioms}
allowing us to define Functions $F : \univ \to K$
from $\forall\exists$-statements:
\[
(\forall x)(\exists k \in K)\varphi(x,k)
~~\longlimp~~
(\exists F : \univ \to K)(\forall x)\varphi(x,F(x))
\]

%%%%%%%%%%%%%%%%%%%%%%%%%%%%%%%%%%%%%%%%%%%%%%%%%%%%%%%%%%%%%%%%%%%%%%%%%%%
\begin{rem}
\label{rem:fso:choice}
%%%%%%%%%%%%%%%%%%%%%%%%%%%%%%%%%%%%%%%%%%%%%%%%%%%%%%%%%%%%%%%%%%%%%%%%%%%
%We insist that
\emph{Functional Choice Axioms} as above
actually amount to \emph{Comprehension} in $\MSO$ (\S\ref{sec:mso:th}).
Such axioms do not create choice predicates for Individuals,
which are known to be undefinable in
$\MSO$~\cite{gs83choice,cl07csl},
and moreover to break decidability when added to the language of
$\MSO$~\cite{bg00lics,cl07csl}.
\end{rem}

The rest of this Section is organized as follows.
The system $\FSOD$ is defined in~\S\ref{sec:fso:fso}-\ref{sec:ax},
and its (expected)
interpretation in the standard model of $\Dir$-ary trees
is given in~\S\ref{sec:std}.
Then in~\S\ref{sec:cons} we discuss the interpretation of $\FSOD$ in $\MSOD$
and a straightforward embedding of $\MSOD$ in $\FSOD$.
Finally,~\S\ref{sec:not} presents notation whose purpose is to
allow some flexibility in the manipulation of functions.
The language and axioms of $\FSOD$ are summarized in Figure~\ref{fig:fso},
with references to the relevant parts of the text.

%%%%%%%%%%%%%%%%%%%%%%%%%%%%%%%%%%%%%%%%%%%%%%%%%%%%%%%%%%%%%%%%%%%%%%%%%%%
\subsection{The Language of $\FSOD$}
\label{sec:fso:fso}
%%%%%%%%%%%%%%%%%%%%%%%%%%%%%%%%%%%%%%%%%%%%%%%%%%%%%%%%%%%%%%%%%%%%%%%%%%%
We now formally define the language of $\FSOD$, for $\Dir$ an $\HF$-set.
It consists of the the following sorts:
%The language of $\FSOD$ consists of the following sorts:
\begin{itemize}
\item
The sort of \emph{Hereditarily finite ($\HF$) sets},
with infinitely many $\HF$-variables $k,\ell$ etc.,
and with one constant symbol $\const\kappa$ for each 
%hereditarily finite set $\kappa$
$\kappa \in V_\omega$
(we often simply write $\kappa$ for $\const\kappa$ in formulae, omitting the overset dot).

\item
The same sort of \emph{Individuals} as $\MSOD$ (see \S\ref{sec:mso:lgge}).
%with infinitely many variables $x,y,z$ etc.,
%with one constant symbol $\Root$,
%and with one unary function symbol $\Succ_d$
%for each $d \in \Dir$.

\item
The sort of \emph{Functions}, with infinitely many variables $F,G,H,\etc$.

\item
The sort of \emph{$\HF$-Functions}, with no variable.
For each pair $(n,m) \in \NN \times \NN$, we assume given
a constant symbol $\const g_{n,m}$ of arity $n$.
The interpretation of these constant symbols is discussed in~\S\ref{sec:ax:hf}.
% 
%In the following, we often write $\const g_{n,m}$ for $\const g_{n,m}$ in formulae.
\end{itemize}

\noindent
The language of $\FSOD$ has two kinds of terms.
The \emph{Individual terms} are the same as those of $\MSOD$.
%They are thus given by
%\begin{equation*}
%\tag{for $d \in \Dir$}
%t \quad\bnf\quad
%x \gs \Root \gs \Succ_{d}(t)
%%\qquad\qquad\text{(for $d \in \Dir$)}
%\end{equation*}
%
%%\[
%%t \quad\bnf\quad
%%x \gs \Root \gs \Succ_{d}(t)
%%\qquad\qquad\text{(for $d \in \Dir$)}
%%\]
%
%\noindent
In addition, $\FSOD$ also has \emph{$\HF$-terms}, which are given by
\[
K,L \quad\bnf\quad
    k 
\gs \const\kappa
%\gs F(t)
%\gs g(L_1,\dots,L_n)
\gs \const g_{n,m}(L_1,\dots,L_n)
\]

\noindent
The \emph{formulae} of $\FSOD$ are built as follows:
%where $\GC \in \{\Eq,\In,\Sle\}$:
\[
\begin{array}{r r l}
  \varphi,\psi
& \bnf
& t \Eq u \gs t \Lt u
\\
& \gs
& K \Eq L \gs K \In L \gs K \Sle L
  \gs F(t) \Eq K
\\
& \gs
& \varphi \lor \psi \gs \lnot \varphi
\\
& \gs
& (\exists x)\varphi
  \gs
  (\exists F:K)\varphi
  \gs
  %\exists K \GC L.\varphi
  (\exists k \In L)\varphi \gs (\exists k \Sle L)\varphi
\end{array}
\]

\noindent
An $\FSOD$-formula $\varphi$ is \emph{$\HF$-closed}
if it contains no free $\HF$-variable.

%%%%%%%%%%%%%%%%%%%%%%%%%%%%%%%%%%%%%%%%%%%%%%%%%%%%%%%%%%%%%%%%%%%%%%%%%%%
\begin{nota}
\label{not:fso}
%%%%%%%%%%%%%%%%%%%%%%%%%%%%%%%%%%%%%%%%%%%%%%%%%%%%%%%%%%%%%%%%%%%%%%%%%%%
\hfill
\begin{enumerate}
\item
%Usual derived formulae are defined as expected, where $\GC \in \{\In,\Sle\}$:
Usual derived formulae are defined similarly as with $\MSO$
(where $\GC$ is either $\In$ or $\Sle$):
%$\GC \in \{\In,\Sle\}$):
%where $\mathcal{X}$ stands for a variable of any sort:
\[
\begin{array}{r !{~~\deq~~} l !{\qquad} r !{~~\deq~~} l}
  (\forall x)\varphi
& \lnot (\exists x)(\lnot \varphi)
& \varphi \land \psi
& \lnot(\lnot \varphi \lor \lnot \psi)
\\
  (\forall F:L)\varphi
& \lnot (\exists F:L)(\lnot \varphi)
& \varphi \limp \psi
& \lnot \varphi \lor \psi
\\
  (\forall k \GC L)\varphi
& \lnot (\exists k \GC L)(\lnot \varphi)
& (t \Leq u)
& (t \Lt u) \lor (t\Eq u)
\\
  \True
& (\forall x)(x \Eq x)
& \False
& \lnot \True
\end{array}
\]

\item
%As well as when we quantify over Functions,
In addition to bounded quantification $(\exists F:K)$,
we use the notation $(F : K)$ within formulae
as the defined formula:
%with the following meaning:
\[
(F : K) \quad\deq\quad
(\forall x)(\exists k \In K)(F(x) \Eq k)
\]

\item
For
variables
$\vec K = K_1,\dots,K_n$
and
$\vec L = L_1,\dots,L_n$,
and
$\GC$ either $\Eq$, $\In$ or $\Sle$,
we let
\[
\vec K \GC \vec L
\quad=\quad
  (K_1,\dots,K_n) \GC (L_1,\dots,L_n)
\quad\deq\quad \bigconj_{1 \leq i \leq n} K_i \GC L_i
\]
%For $\GC$ either $\Eq$, $\In$ or $\Sle$,
%%For $\GC \in \{\Eq,\In,\Sle\}$,
%$\vec K = K_1,\dots,K_n$
%and
%$\vec L = L_1,\dots,L_n$,
%we let
%\[
%\vec K \GC \vec L
%\quad=\quad
%  (K_1,\dots,K_n) \GC (L_1,\dots,L_n)
%\quad\deq\quad \bigconj_{1 \leq i \leq n} K_i \GC L_i
%\]
\end{enumerate}
\end{nota}

%%%%%%%%%%%%%%%%%%%%%%%%%%%%%%%%%%%%%%%%%%%%%%%%%%%%%%%%%%%%%%%%%%%%%%%%%%%
\begin{rem}
\label{rem:fso:dir}
%%%%%%%%%%%%%%%%%%%%%%%%%%%%%%%%%%%%%%%%%%%%%%%%%%%%%%%%%%%%%%%%%%%%%%%%%%%
The (hereditarily) finite set $\Dir$ of tree directions is considered
both as a parameter in the definition of $\FSOD$,
via the successor term constructors $\Succ_d$ (for $d \in \Dir$)
and the corresponding axioms (see~\S\ref{sec:ax}),
and as a (constant) $\HF$ set, which can occur as such in $\FSOD$ formulae.
Strictly speaking, we should write $\const \Dir $ rather than $\Dir$ in the latter case, but we usually simply omit the overset dot, as with other HF-sets.
\end{rem}

%%%%%%%%%%%%%%%%%%%%%%%%%%%%%%%%%%%%%%%%%%%%%%%%%%%%%%%%%%%%%%%%%%%%%%%%%%%
\subsection{The Deduction System of $\FSOD$}
\label{sec:ded}
%%%%%%%%%%%%%%%%%%%%%%%%%%%%%%%%%%%%%%%%%%%%%%%%%%%%%%%%%%%%%%%%%%%%%%%%%%%
Deduction for $\FSOD$ is defined by the system presented
on Figure~\ref{fig:ded:prop}
(with $\FSOD$ formulae instead of $\MSOD$ formulae)
and Figure~\ref{fig:ded:bfso},
together with all the axioms of~\S\ref{sec:ax}.
The language and axioms of $\FSOD$ are summarized in Figure~\ref{fig:fso}.

%%%%%%%%%%%%%%%%%%%%%%%%%%%%%%%%%%%%%%%%%%%%%%%%%%%%%%%%%%%%%%%%%%%%%%%%%%%
\begin{figure}[tbp]
%%%%%%%%%%%%%%%%%%%%%%%%%%%%%%%%%%%%%%%%%%%%%%%%%%%%%%%%%%%%%%%%%%%%%%%%%%%
\[
\begin{array}{c}
\dfrac{\Phi \thesis \varphi[t/x]}{\Phi \thesis (\exists x)\varphi}
\qquad\qquad
\dfrac{\Phi \thesis (\exists x)\varphi \qquad \Phi, \varphi \thesis \psi}
  {\Phi \thesis \psi}
~~\text{($x$ not free in $\Phi,\psi$)}
\\\\

\dfrac{\Phi \thesis \varphi[G/F]
  \qquad \Phi \thesis (G:K)}
{\Phi \thesis (\exists F:K)\varphi}
\quad~~~\,
\dfrac{\Phi \thesis (\exists F:K) \varphi \qquad
  \Phi,\, (F:K) ,\, \varphi \thesis \psi}
  {\Phi \thesis \psi}
~\text{($F$ not free in $\Phi,\psi$)}
\\\\

\dfrac{\Phi \thesis \varphi[K/k] \qquad \Phi \thesis K \GC L}
  {\Phi \thesis (\exists k \GC L)\varphi}
~~\text{(for $\GC$ either $\In$ or $\Sle$)}
%(\GC \in \{\In,\Sle\})
\\\\
\dfrac{\Phi \thesis (\exists k \GC K)\varphi
  \qquad \Phi,\, k \GC K,\, \varphi \thesis \psi}
  {\Phi \thesis \psi}
~~\text{(for $\GC$ either $\In$ or $\Sle$, and $k$ not free in $\Phi,\psi$)}
%\text{($\GC \in \{\In,\Sle\}$ and $k$ not free in $\Phi,\psi$)}
\\\\

\dfrac{\Phi \thesis \varphi}
  {\Phi[F(t)/k] \thesis \varphi[F(t)/k]}
~~\text{($\Phi[F(t)/k], \varphi[F(t)/k]$ $\FSO$-formulae)}
\\[1em]
\end{array}
\]
%\hrule
\caption{Deduction Rules for $\FSOD$.\label{fig:ded:bfso}}
\end{figure}

%%%%%%%%%%%%%%%%%%%%%%%%%%%%%%%%%%%%%%%%%%%%%%%%%%%%%%%%%%%%%%%%%%%%%%%%%%%
\begin{figure}[tbp]
%%%%%%%%%%%%%%%%%%%%%%%%%%%%%%%%%%%%%%%%%%%%%%%%%%%%%%%%%%%%%%%%%%%%%%%%%%%
\[
\begin{array}{l l r c}
\toprule
\multicolumn{1}{l}{\textbf{Language}}
\\
\midrule

  \text{Individual Terms}
& t \quad\bnf\quad x \gs \Root \gs \Succ_{d}(t)
& \text{($d \in \Dir$)}
& \text{(\S\ref{sec:mso:lgge})}
\\

  \text{Functions}
& F,G,H,\etc
& \text{(only variables)}
& \text{(\S\ref{sec:fso:fso})}
\\

%  \text{$\HF$-Terms}
%& K,L \quad\bnf\quad k \gs \const\kappa \gs \const g_{n,m}(L_1,\dots,L_n) 
%& 
%& \text{(\S\ref{sec:fso:fso})}
%\\
%\multicolumn{3}{r}{\text{($k$ $\HF$-variable,~~ $\kappa \in V_\omega$,~~
%  $\const g_{n,m}$ $\HF$-Function)}}
%\\\\

  \text{$\HF$-Terms}
& K,L \quad\bnf\quad k 
& \text{($k$ $\HF$-variable)}
& \text{(\S\ref{sec:fso:fso})}
\\
& \gs \const\kappa
& \text{($\kappa \in V_\omega$)}
\\
& \gs \const g_{n,m}(L_1,\dots,L_n) 
& \text{($\const g_{n,m}$ $\HF$-Function)}
\\\\

  \text{Formulae}
& \varphi,\psi \quad\bnf\quad
&
& \text{(\S\ref{sec:fso:fso})}
\\

& \multicolumn{3}{l}{
  \gs t \Eq u \gs t \Lt u \gs F(t) \Eq K}
\\

& \multicolumn{3}{l}{
\gs K \Eq L \gs K \In L \gs K \Sle L}
\\

& \multicolumn{3}{l}{
  \gs \varphi \lor \psi \gs \lnot \varphi
}
\\

& \multicolumn{3}{l}{
  \gs (\exists x)\varphi \gs (\exists F:K)\varphi
  \gs (\exists k \In L)\varphi \gs (\exists k \Sle L)\varphi
}

\\

%  \text{Formulae}
%& \begin{array}{r r l}
%    \varphi,\psi
%  & \bnf
%  & t \Eq u \gs t \Lt u
%  \\
%  & \gs
%  & K \Eq L \gs K \In L \gs K \Sle L
%    \gs F(t) \Eq K
%  \\
%  & \gs
%  & \varphi \lor \psi \gs \lnot \varphi
%  \\
%  & \gs
%  & (\exists x)\varphi
%    \gs
%    (\exists F:K)\varphi
%    \gs
%    %\exists K \GC L.\varphi
%    (\exists k \In L)\varphi \gs (\exists k \Sle L)\varphi
%  \end{array}
%&
%& \text{(\S\ref{sec:fso:fso})}

\\
\toprule
\multicolumn{1}{l}{\textbf{Axioms}}
\\
\midrule

  \multicolumn{2}{l}{\text{\emph{Equality:}}}
&
& \text{(\S\ref{sec:ax:eq})}
\\
  \multicolumn{3}{c}{
  \begin{array}{c !{\qquad} c}
    (\forall x) (x \Eq x)
  & (\forall x) (\forall y)
    \big(
    x \Eq y ~~\longlimp~~ \varphi[x/z] ~~\longlimp~~ \varphi[y/z]
    \big)
  \\\\
    K \Eq K
  & \left(
    K \Eq L ~~\longlimp~~ \varphi[K/m] ~~\longlimp~~ \varphi[L/m]
    \right)
  \end{array}}
\\\\

  \multicolumn{2}{l}{\text{\emph{Induction:}}}
&
& \text{(\S\ref{sec:ax:ind})}
\\
  \multicolumn{3}{c}{
  \varphi(\Root)
  \quad\longlimp\quad
  \mathord{\bigconj}_{d \in \Dir}
  (\forall x)
  \big[
     \varphi(x) ~~\longlimp~~  \varphi(\Succ_d(x))
  \big]
  \quad\longlimp\quad
  (\forall x) \varphi(x)
  }
\\\\

  \multicolumn{2}{l}{\text{\emph{Tree Axioms:}}}
&
& \text{(\S\ref{sec:ax:tree})}
\\
  \multicolumn{3}{c}{
  \begin{array}{c !{\quad} c}
  \lnot (\exists x) \bigdisj_{d \neq d'}
  \big( \Succ_{d}(x) \Eq \Succ_{d'}(x) \big)
  %\quad\text{if $d \neq d'$}
  &
  (\forall x) (\forall y) \bigconj_{d \in \Dir}
    \big( \Succ_d(x) \Eq \Succ_d(y) ~~\limp~~ x \Eq y \big)
  \\\\
  
  \lnot (\exists x) \big(x \Lt x \big)
  &
  (\forall x) (\forall y) (\forall z)
  \big( x \Lt y ~~\limp~~ y \Lt z ~~\limp~~ x \Lt z \big)
  \\\\
  
  (\forall x) \big(\Root \Leq x \big)
  &
  (\forall x) (\forall y)
  \Big(
  \big( \bigdisj_{d \in \Dir} x \Lt \Succ_d(y) \big) ~~\liff~~ x \Leq y \Big)
  \end{array}}
\\\\

  \multicolumn{2}{l}{\text{\emph{Axioms on $\HF$-Sets:}}}
  %\text{(provided\quad
  %$\Sk(\ZFCM)\thesis (\forall k_1,\dots,k_n) (\exists! \ell) \varphi_{n,m}$):}}
&
& \text{(\S\ref{sec:ax:hf})}
\\
& \multicolumn{2}{l}{\varphi_{n,m}[\vec K/\vec k][\const g_{n,m}(\vec K) / \ell]}
\\
& \multicolumn{2}{r}{\text{(provided\quad
$\Sk(\ZFCM)\thesis (\forall k_1,\dots,k_n) (\exists! \ell) \varphi_{n,m}$)}}
\\\\

  \multicolumn{2}{l}{\text{\emph{Functional Choice Axioms:}}}
&
& \text{(\S\ref{sec:ax:choice})}
\\

  \multicolumn{4}{c}{
  \begin{array}{r !{\quad\longlimp\quad} l}
    (\forall k \In K) (\exists \ell \In L) \varphi(k,\ell)
  & \big( \exists f \In L^K\big) (\forall k \In K)
    \varphi(k,f(k))
  \\[1em]
  
    (\forall x) (\exists k \In K) \varphi(x,k)
  & (\exists F : K) (\forall x) (\exists k \In K)
    \left(F(x) \Eq k ~\land~ \varphi(x,k) \right)
  \\[1em]
  
    (\forall k \In K) (\exists F : L)\varphi(k,F)
  & \big( \exists G : L^K \big) (\forall k \In K) \varphi(k,F)[G(k) \sslash F]
  \end{array}}

\\\\
\toprule
\end{array}
\]

\caption{Summary of $\FSOD$.\label{fig:fso}}
\end{figure}

\cnote{\CR:NOTE
\begin{itemize}
\item $\FSOD$ has an explicit substitution rule for Functions
(see also Rem.~\ref{rem:ax:subst}).

\item Moreover,
the Equality Axioms (\S\ref{sec:ax:eq}) and the Axioms on $\HF$-Sets
(\S\ref{sec:ax:hf})
have been closed under substitution of $\HF$-terms.
Note that an other possibility would have been to add a substitution rule
\[
\dfrac{\Phi \thesis \varphi}{\Phi[K/k] \thesis \varphi[K/k]}
\]

\item This is for the expected substitution lemma to hold.
\end{itemize}}

%%% Local Variables:
%%% mode: latex
%%% TeX-master: "main.tex"
%%% End:

%%%%%%%%%%%%%%%%%%%%%%%%%%%%%%%%%%%%%%%%%%%%%%%%%%%%%%%%%%%%%%%%%%%%%%%%%%%
\subsection{Basic Axiomatization}
\label{sec:ax}
%%%%%%%%%%%%%%%%%%%%%%%%%%%%%%%%%%%%%%%%%%%%%%%%%%%%%%%%%%%%%%%%%%%%%%%%%%%
We now present the axioms of $\FSOD$.
The first group
(Equality, Tree Axioms and Induction,~\S\ref{sec:ax:eq}-\S\ref{sec:ax:ind})
corresponds to its counterpart in $\MSOD$.
We then present our specific axioms for $\HF$-Sets in~\S\ref{sec:ax:hf}
and our Functional Choice Axioms in~\S\ref{sec:ax:choice}.

\subsubsection{Equality}
\label{sec:ax:eq}
%%%%%%%%%%%%%%%%%%%%%%%%%%%%%%%%%%%%%%%%%%%%%%%%%%%%%%%%%%%%%%%%%%%%%%%%%%%

The theory $\FSOD$ has
usual equality axioms for individuals and $\HF$-Sets.
\begin{itemize}
\item \emph{Equality on Individuals.}
%(for all formula $\varphi$):
\begin{equation*}
\tag{for all formula $\varphi$}
(\forall x) (x \Eq x)
\qquad\text{and}\qquad
(\forall x) (\forall y)
\big(
x \Eq y ~~\longlimp~~ \varphi[x/z] ~~\longlimp~~ \varphi[y/z]
\big)
%\quad\text{(for all formula $\varphi$)}
\end{equation*}

\item \emph{Equality on $\HF$-Sets}
%(for all formula $\varphi$ and all $\HF$-variables $k,\ell,m$):
(for all formula $\varphi$, all $\HF$-terms $K,L$ and all $\HF$-variable $m$):
\[
K \Eq K
\qquad\text{and}\qquad
\left(
K \Eq L ~~\longlimp~~ \varphi[K/m] ~~\longlimp~~ \varphi[L/m]
\right)
%\quad\text{(for all formula $\varphi$)}
\]
\end{itemize}

%%%%%%%%%%%%%%%%%%%%%%%%%%%%%%%%%%%%%%%%%%%%%%%%%%%%%%%%%%%%%%%%%%%%%%%%%%%
\begin{rem}
\label{rem:ax:subst}
%%%%%%%%%%%%%%%%%%%%%%%%%%%%%%%%%%%%%%%%%%%%%%%%%%%%%%%%%%%%%%%%%%%%%%%%%%%
Note that $\FSOD$ is equipped with an explicit \emph{Substitution Rule}
\begin{equation*}
\tag{$\Phi[F(t)/k], \varphi[F(t)/k]$ $\FSO$-formulae}
\dfrac{\Phi \thesis \varphi}
  {\Phi[F(t)/k] \thesis \varphi[F(t)/k]}
%~~\text{($\Phi[F(t)/k], \varphi[F(t)/k]$ $\FSO$-formulae)}
\end{equation*}

\noindent
Substitution entails the following
(where $\varphi(F(t))$ is an $\FSO$-formula):
\[
(F(t) \Eq K) ~~\longlimp~~ \varphi(K) ~~\longlimp~~ \varphi(F(t))
\]

\noindent
as well as
the derived rule
\[
\dfrac{\Phi \thesis \varphi(F(t)) \qquad \Phi \thesis (F:K)}
  {\Phi \thesis (\exists k \In K)\varphi(k)}
\]

\noindent
The former is a direct consequence of the Substitution rule
together with elimination of equality on $\HF$-Sets.
For the latter, first note that
Remark~\ref{rem:mso:ded} also holds for $\FSOD$.
In particular, one can derive
%always derive
\[
%\dfrac{}{
  %\Phi \thesis 
  (k \In K) ~~\longlimp~~ \varphi(k) ~~\longlimp~~ (\exists k \In K)\varphi(k)
%}
\]
On the other hand, we have
\[
(\exists \ell \In K)(k \Eq \ell)
%\thesis
~~\longlimp~~
(k \In K)
\]

\noindent
We therefore get
\[
%\dfrac{}{
  %\Phi \thesis 
  (\exists \ell \In K)(k \Eq \ell)
  ~~\longlimp~~ \varphi(k) ~~\longlimp~~ (\exists k \In K)\varphi(k)
%}
\]
and the Substitution rule gives
\[
%\dfrac{}{
  %\Phi \thesis
  (\exists \ell \In K)(F(t) \Eq \ell)
  ~~\longlimp~~ \varphi(F(t)) ~~\longlimp~~ (\exists k \In K)\varphi(k)
%}
  \tag*{\qed}
\]
\end{rem}

%%%%%%%%%%%%%%%%%%%%%%%%%%%%%%%%%%%%%%%%%%%%%%%%%%%%%%%%%%%%%%%%%%%%%%%%%%%
\subsubsection{Induction}
\label{sec:ax:ind}
%%%%%%%%%%%%%%%%%%%%%%%%%%%%%%%%%%%%%%%%%%%%%%%%%%%%%%%%%%%%%%%%%%%%%%%%%%%
We have the following \emph{Induction Scheme}:
%(for each formula $\varphi$):
\begin{equation*}
\tag{for each formula $\varphi$}
  \varphi(\Root)
  ~~\longlimp~~
  \mathord{\bigconj}_{d \in \Dir}
  (\forall x)
  \big[
     \varphi(x) ~~\longlimp~~  \varphi(\Succ_d(x))
  \big]
  ~~\longlimp~~
  (\forall x) \varphi(x)
\end{equation*}

%%%%%%%%%%%%%%%%%%%%%%%%%%%%%%%%%%%%%%%%%%%%%%%%%%%%%%%%%%%%%%%%%%%%%%%%%%%
\subsubsection{Tree Axioms}
\label{sec:ax:tree}
%%%%%%%%%%%%%%%%%%%%%%%%%%%%%%%%%%%%%%%%%%%%%%%%%%%%%%%%%%%%%%%%%%%%%%%%%%%
For the tree structure of $\univ$,
we have the same Tree Axioms as $\MSOD$, displayed in Figure~\ref{fig:ded:tree}
(recall that $\FSOD$ has the same Individuals as $\MSOD$).
%we have the axioms of Figure~\ref{fig:ded:tree}.
%(where $(x \Leq y)$ stands for $(x \Lt y \lor x \Eq y)$).
%
%We list here some immediate consequences of the tree axioms.

We now state expected results on the axioms so far introduced.
To this end, let $\FSODZ$ be the system consisting of 
the deduction rules of Figures~\ref{fig:ded:prop} and~\ref{fig:ded:bfso},
together with the Equality Axioms (\S\ref{sec:ax:eq})
the Induction Scheme (\S\ref{sec:ax:ind})
and the Tree Axioms (Figure~\ref{fig:ded:tree}).

%%%%%%%%%%%%%%%%%%%%%%%%%%%%%%%%%%%%%%%%%%%%%%%%%%%%%%%%%%%%%%%%%%%%%%%%%%%
\begin{prop}
\label{prop:ax:tree}
%%%%%%%%%%%%%%%%%%%%%%%%%%%%%%%%%%%%%%%%%%%%%%%%%%%%%%%%%%%%%%%%%%%%%%%%%%%
$\FSODZ$ proves the following.
\begin{enumerate}
\item $(\forall x) (\forall y) (x \Leq y \Leq x ~~\longlimp~~ x \Eq y)$
\item $(\forall x) ( x \Leq \Root ~~\longlimp~~ x \Eq \Root)$
\item $\lnot(\exists x) (x \Lt \Root )$
\item $\lnot (\exists x) (\Succ_d(x) \Eq \Root)$
\item 
\label{eq:ax:tree:rootorsucc}
$(\forall x)
(x \Eq \Root ~~\lor~~ (\exists y) \lor_{d \in \Dir} x \Eq \Succ_d(y))$

\item
$(\forall x)(\forall y)
(x \Lt y ~~\longlimp~~ \lor_{d \in \Dir} \Succ_d(x) \Leq y)$
\end{enumerate}
\end{prop}

%%%%%%%%%%%%%%%%%%%%%%%%%%%%%%%%%%%%%%%%%%%%%%%%%%%%%%%%%%%%%%%%%%%%%%%%%%%
\begin{fullproof}
%%%%%%%%%%%%%%%%%%%%%%%%%%%%%%%%%%%%%%%%%%%%%%%%%%%%%%%%%%%%%%%%%%%%%%%%%%%
\hfill
\begin{enumerate}
\item If $x \Lt y \Lt x$ then by transitivity of $\Lt$ we have $x \Lt x$,
contradicting the irreflexivity of $\Lt$.

\item 
If $x \Leq \Root$, we have $x \Leq \Root \Leq x$, so that $x \Eq \Root$.

\item
If $x \Leq \Root$ then $x \Eq \Root$ so that we cannot have
$(x \Leq \Root) \land \lnot(X \Eq \Root)$.

\item
If $\Succ_d(x) \Eq \Root$, then since $x \Lt \Succ_d(x)$ we have
$x \Lt \Root$, a contradiction.

\item
A direct application of the Induction Scheme.

\item
Assuming given $x$,
we apply the Induction Scheme on the formula
$\varphi(y) \deq (x \Lt y \limp \bigdisj_{d \in \Dir} \Succ_d(x) \Leq y)$.

We trivially get $\varphi(\Root)$ since $\lnot(x \Lt \Root)$.
Assuming now $\varphi(y)$ we show $\varphi(\Succ_{d'}(y))$.
So assume $x \Lt \Succ_{d'}(y)$.
Then the Tree Axioms give $x \Leq y$.
If $x \Lt y$, then we are done thanks to $\varphi(y)$.
Otherwise, we have $x \Eq y$, so that $\Succ_{d'}(x) \Eq \Succ_{d'}(y)$
and we are done.
\qedhere
\end{enumerate}
%%%%%%%%%%%%%%%%%%%%%%%%%%%%%%%%%%%%%%%%%%%%%%%%%%%%%%%%%%%%%%%%%%%%%%%%%%%
\end{fullproof}
%%%%%%%%%%%%%%%%%%%%%%%%%%%%%%%%%%%%%%%%%%%%%%%%%%%%%%%%%%%%%%%%%%%%%%%%%%%

A consequence of Proposition~\ref{prop:ax:tree}
is that the Induction Scheme of $\FSOD$ (\S\ref{sec:ax:ind})
implies the usual scheme of \emph{Well-Founded Induction} \wrt\@ 
the strict prefix order $\Lt$.

%%%%%%%%%%%%%%%%%%%%%%%%%%%%%%%%%%%%%%%%%%%%%%%%%%%%%%%%%%%%%%%%%%%%%%%%%%%
\begin{thm}[Well-Founded Induction]
\label{thm:ax:wfind}
%%%%%%%%%%%%%%%%%%%%%%%%%%%%%%%%%%%%%%%%%%%%%%%%%%%%%%%%%%%%%%%%%%%%%%%%%%%
$\FSODZ$
proves the following form of well-founded induction:
\[
(\forall x) \big[
(\forall y \Lt x)(\varphi(y))
~~\longlimp~~ \varphi(x)
\big]
~~\longlimp~~
(\forall x) \varphi(x)
\]
\end{thm}

%%%%%%%%%%%%%%%%%%%%%%%%%%%%%%%%%%%%%%%%%%%%%%%%%%%%%%%%%%%%%%%%%%%%%%%%%%%
\begin{fullproof}
%%%%%%%%%%%%%%%%%%%%%%%%%%%%%%%%%%%%%%%%%%%%%%%%%%%%%%%%%%%%%%%%%%%%%%%%%%%
Assume
\[
\forall x\left[
\forall y\left(y \Lt x ~~\longlimp~~ \varphi(y) \right)
~~\longlimp~~ \varphi(x)
\right]
\]
We apply induction on the formula
\[
\psi(x) \quad\deq\quad (\forall y \Leq x) \varphi(y)
\]
We have $\psi(\Root)$ since $\forall y \lnot(y \Lt \Root)$.
Assuming $\psi(x)$, we get $\psi(\Succ_d(x))$
as follows.
If $y \Lt \Succ_d(x)$, we have $y \Leq x$, hence $\varphi(y)$
since we assumed $\psi(x)$.
Moreover, $\varphi(\Succ_d(x))$ follows from the fact that
$y \Lt \Succ_d(x)$ implies $y \Leq x$, hence $\varphi(y)$
since we assumed $\psi(x)$.
%%%%%%%%%%%%%%%%%%%%%%%%%%%%%%%%%%%%%%%%%%%%%%%%%%%%%%%%%%%%%%%%%%%%%%%%%%%
\end{fullproof}
%%%%%%%%%%%%%%%%%%%%%%%%%%%%%%%%%%%%%%%%%%%%%%%%%%%%%%%%%%%%%%%%%%%%%%%%%%%

%%%%%%%%%%%%%%%%%%%%%%%%%%%%%%%%%%%%%%%%%%%%%%%%%%%%%%%%%%%%%%%%%%%%%%%%%%%
\begin{rem}
%%%%%%%%%%%%%%%%%%%%%%%%%%%%%%%%%%%%%%%%%%%%%%%%%%%%%%%%%%%%%%%%%%%%%%%%%%%
Both Proposition~\ref{prop:ax:tree}
and Theorem~\ref{thm:ax:wfind} also hold for $\MSOD$.
\end{rem}

%%%%%%%%%%%%%%%%%%%%%%%%%%%%%%%%%%%%%%%%%%%%%%%%%%%%%%%%%%%%%%%%%%%%%%%%%%%
\subsubsection{$\HF$-Sets}
\label{sec:ax:hf}
%%%%%%%%%%%%%%%%%%%%%%%%%%%%%%%%%%%%%%%%%%%%%%%%%%%%%%%%%%%%%%%%%%%%%%%%%%%
We now present our axioms on $\HF$-Sets.
%whose
Their
purpose is to ease formalization in $\FSOD$.
Recall that $\HF$-sets range over $V_\omega$
(Definition~\ref{def:fso:hf}).
The idea of these axioms
is to incorporate in $\FSOD$ as much of the theory of $V_\omega$ as possible,
%while still being $\MSO$-interpretable.
while keeping $\FSOD$ interpretable in $\MSOD$
and with a semi-recursive notion of provability.
The interpretation of $\FSOD$ in $\MSOD$
relies on the fact that in $\FSOD$-formulae,
all quantifications over $\HF$-Sets are bounded (either by $\In$ or $\Sle$),
so that in a closed $\FSOD$-formula, quantifications
over $\HF$-Sets can be interpreted using usual propositional logic.

We will have, as particular cases of our axioms on $\HF$-Sets,
all bounded formulae over $\HF$-Sets which are true in $V_\omega$.
Moreover, \wrt\@ the interpretation of $\FSOD$ in $\MSOD$ (\S\ref{sec:cons})
and in particular \wrt\@ its application to $\MSOD$ over $\omega$-words
(\S\ref{sec:msow}, \S\ref{sec:compl} and~\S\ref{sec:sim}),
it is important to have sufficiently many functions over $V_\omega$
available within \emph{closed} $\HF$-terms.
This is the main purpose of our axioms on $\HF$-Sets.
They %essentially
state that the $\HF$-Functions $\const g_{n,m}$
are Skolem functions for
$\forall\exists!$-statements over $\HF$-Sets.
%the $\forall\exists$-statements over $\HF$-Sets which are true in $V_\omega$.
These axioms are further commented in~\S\ref{sec:remcompl}.

%We rely on the fact that $V_\omega$ is a model of $\ZFCM$ (Remark~\ref{rem:fso:hf}).

%%%%%%%%%%%%%%%%%%%%%%%%%%%%%%%%%%%%%%%%%%%%%%%%%%%%%%%%%%%%%%%%%%%%%%%%%%%
\begin{defi}[$\HF$-Formula]
\label{def:ax:hf:form}
%%%%%%%%%%%%%%%%%%%%%%%%%%%%%%%%%%%%%%%%%%%%%%%%%%%%%%%%%%%%%%%%%%%%%%%%%%%
An \emph{$\HF$-formula} is an $\FSOD$-formula with atoms
of the form $K \Eq L$, $K \In L$ or $K \Sle L$
where $K$ and $L$ are $\HF$-terms.
%containing no Function variable. %and no $\HF$-Function.
%variable of arity $(1,n)$.
\end{defi}

%%%%%%%%%%%%%%%%%%%%%%%%%%%%%%%%%%%%%%%%%%%%%%%%%%%%%%%%%%%%%%%%%%%%%%%%%%%
%\paragraph{Inclusion.}
%%%%%%%%%%%%%%%%%%%%%%%%%%%%%%%%%%%%%%%%%%%%%%%%%%%%%%%%%%%%%%%%%%%%%%%%%%%

Fix a distinguished $\HF$-variable $\ell$,
and an enumeration $k_1,k_2,\dots$ of distinct $\HF$-variables all different from $\ell$.
Furthermore, fix an enumeration $(\varphi_{n,m})_{n,m \in \NN}$
of $\HF$-formulae satisfying the following conditions:
\begin{enumerate}
\item Each formula $\varphi_{n,m}$ has free variables among $k_1,\dots,k_n,\ell$.
\item All $\HF$-Functions occurring in $\varphi_{n,m}$
have the form $\const g_{n',m'}$ with $m' < m$.
\item Each $\HF$-formula $\varphi$ satisfying (1) and (2) occurs infinitely often in 
$(\varphi_{n,m})_{n,m \in \NN}$, in the following sense.
If $\varphi$ has free variables among $k_1,\dots,k_n,\ell$,
then there are infinitely many $m \in \NN$ such that $\varphi$
is $\varphi_{n,m}$.
\end{enumerate}

\noindent
Recall from Remark~\ref{rem:fso:hf} that $V_\omega$
is a model of $\ZFCM$.
The idea of our Axioms on $\HF$-Sets is that 
\begin{equation}
\label{eq:ax:hf}
\FSOD\thesis \varphi_{n,m}[\const g_{n,m}(k_1,\dots,k_n)/\ell]
\qquad\text{whenever}\qquad
\ZFCM\thesis (\forall k_1,\dots,k_n) (\exists! \ell) \varphi_{n,m}
\end{equation}
(where $(k \Sle k')$ is interpreted as $(\forall m\In k)(m \In k')$).
However, recall that $\varphi_{n,m}$ in~\eqref{eq:ax:hf}
may contain $\HF$-Functions $\const g_{n',m'}$ with $m'<m$.
The premise of~\eqref{eq:ax:hf} can thus not be formulated in $\ZFCM$,
but requires a suitable extension of it.
We let $\Sk(\ZFCM)$ consist of $\ZFCM$ augmented with the axioms
\begin{equation}
\tag{for each $n,m \in \NN$}
(\forall k_1,\dots,k_n)
(\exists! \ell)(\varphi_{n,m})
~~\longlimp~~
(\forall k_1,\dots,k_n)
\varphi_{n,m}[\const g_{n,m}(k_1,\dots,k_n)/\ell]
%(\forall k_1,\dots,k_n)
%\Big(
%(\exists! \ell)(\varphi_{n,m})
%~~\limp~~
%\varphi_{n,m}[\const g_{n,m}(k_1,\dots,k_n)/\ell]
%\Big)
\end{equation}

\cnote{\CR:NOTE
\begin{itemize}
\item \cite[Thm.\@ 3.4.4]{dalen04book}
gives the conservativity with axioms of the form
\[
(\forall k_1,\dots,k_n)
\Big(
(\exists \ell)(\varphi_{n,m})
~~\limp~~
\varphi_{n,m}[\const g_{n,m}(k_1,\dots,k_n)/\ell]
\Big)
\]
which entails conservativity in our cases since our axioms are implied
by those of~\cite[Thm.\@ 3.4.4]{dalen04book}.
\end{itemize}}

\noindent
It is well-known that $\Sk(\ZFCM)$ is thus a conservative extension of $\ZFCM$
(see \eg~\cite[\S 3.4]{dalen04book}).
$\FSOD$ has the following axiom scheme for $\HF$-Sets,
which simply consists of~\eqref{eq:ax:hf} formulated for $\Sk(\ZFCM)$
rather than $\ZFCM$:
\begin{itemize}
\item \emph{Axioms on $\HF$-Sets.}
For each $n,m \in \NN$
such that
%$\Sk(\ZFCM)\thesis (\forall k_1,\dots,k_n) (\exists! \ell) \varphi_{n,m}$,
\begin{equation}
\label{eq:ax:hf:zfc}
\Sk(\ZFCM)\thesis\quad (\forall k_1,\dots,k_n) (\exists! \ell) \varphi_{n,m}
\end{equation}

\noindent
and for all $\HF$-terms $\vec K = K_1,\dots,K_n$, we have the axiom
%we have as axiom
\[
\varphi_{n,m}[\vec K/\vec k][\const g_{n,m}(\vec K) / \ell]
\]
\end{itemize}

%%%%%%%%%%%%%%%%%%%%%%%%%%%%%%%%%%%%%%%%%%%%%%%%%%%%%%%%%%%%%%%%%%%%%%%%%%%
\begin{rem}
\label{rem:ax:hf:rec}
%%%%%%%%%%%%%%%%%%%%%%%%%%%%%%%%%%%%%%%%%%%%%%%%%%%%%%%%%%%%%%%%%%%%%%%%%%%
Note that this axiom scheme makes the axiom set of $\FSOD$ not recursive.
But as expected for a proof system,
\emph{provability} in $\FSOD$ remains semi-recursive.
\end{rem}

We fix here once and for all an interpretation
of the $\HF$-Function symbols $\const g_{n,m}$
as functions over $V_\omega$.
%One expects 
The idea is
that if~\eqref{eq:ax:hf:zfc} holds, 
then $\const g_{n,m}$ is interpreted
as a computable function $\hf g_{n,m} : V_\omega^n \to V_\omega$
such that
\begin{equation}
\label{eq:ax:hf:vomega}
V_\omega \models\quad
(\forall k_1,\dots,k_n)\varphi_{n,m}[\hf g_{n,m}(k_1,\dots,k_n)/\ell]
\end{equation}

\noindent
But again recall that $\varphi_{n,m}$ 
may contain $\HF$-Functions $\const g_{n',m'}$ with $m'<m$.
We therefore proceed inductively, as follows.

%%%%%%%%%%%%%%%%%%%%%%%%%%%%%%%%%%%%%%%%%%%%%%%%%%%%%%%%%%%%%%%%%%%%%%%%%%%
\begin{conv}
\label{conv:ax:hf:ac}
%%%%%%%%%%%%%%%%%%%%%%%%%%%%%%%%%%%%%%%%%%%%%%%%%%%%%%%%%%%%%%%%%%%%%%%%%%%
By induction on $m \in \NN$,
we interpret the 
$\HF$-Function symbols $\const g_{n,m}$
%The $\HF$-Function symbols $\const g_{n,m}$ are interpreted
as computable functions
\[
\hf g_{n,m} ~~:~~ V_\omega^n ~~\longto~~ V_\omega
\]
%by induction on $m \in \NN$.

\noindent
Consider the formula $\varphi_{n,m}$, and assume that all $\HF$-Functions
$\const g_{n',m'}$ with $m' < m$ are already interpreted.
If~\eqref{eq:ax:hf:zfc} holds
then by the Countable Axiom of Choice
we interpret $\const g_{n,m}$ as the unique function
$\hf g_{n,m} : V^n_\omega \to V_\omega$ such that~\eqref{eq:ax:hf:vomega}
holds.
Note that such functions are computable.
Otherwise, we interpret $\const g_{n,m}$ as the function with constant value 
$\emptyset$.
\end{conv}

%%%%%%%%%%%%%%%%%%%%%%%%%%%%%%%%%%%%%%%%%%%%%%%%%%%%%%%%%%%%%%%%%%%%%%%%%%%
\begin{rem}
\label{rem:ax:hf:compl}
%%%%%%%%%%%%%%%%%%%%%%%%%%%%%%%%%%%%%%%%%%%%%%%%%%%%%%%%%%%%%%%%%%%%%%%%%%%
Note each $\HF$-Function symbol is interpreted by a recursive function
in Convention~\ref{conv:ax:hf:ac}.
However, since~\eqref{eq:ax:hf:zfc} is undecidable, 
there is no algorithm taking $(n,m) \in \NN^2$ to
the interpretation of $\const g_{n,m}$.
This point is further discussed in~\S\ref{sec:remcompl},
where a natural workaround is proposed,
as well as some explanations for our present choice of Axioms on $\HF$-Sets.
\end{rem}

We now discuss some consequences of these axioms.
%First of all, we have the usual relation between inclusion and membership:
%\[
%  k \Sle \ell
%  ~~\liff~~
%  (\forall m \In k)(m \In \ell)
%\]
First note that if $\varphi$ is a \emph{closed} $\HF$-formula,
then it is provable in $\Sk(\ZFCM)$ if and only if it holds in $V_\omega$.
We state this fact as a Remark for the record, and also to reiterate how much deductive power underlies the axioms on $\HF$-sets.

%%%%%%%%%%%%%%%%%%%%%%%%%%%%%%%%%%%%%%%%%%%%%%%%%%%%%%%%%%%%%%%%%%%%%%%%%%%
\begin{rem}
\label{rem:ax:hf:vomega}
%%%%%%%%%%%%%%%%%%%%%%%%%%%%%%%%%%%%%%%%%%%%%%%%%%%%%%%%%%%%%%%%%%%%%%%%%%%
Given a closed $\HF$-formula $\varphi$,
\[
\FSOD \thesis \varphi
\qquad\text{whenever}\qquad
V_\omega \models \varphi
\]
\end{rem}

\cnote{\CR:NOTES
\begin{itemize}
\item
This is for Remark~\ref{rem:ax:hf:vomega}
that we require an $\exists!$ in~\eqref{eq:ax:hf:zfc}
(Otherwise, $\HF$-Functions may be underspecified).

\item
Also~\S\ref{sec:games:parity} introduces Functions on $\HF$-Sets.
\end{itemize}}

\noindent
Moreover, we have all instances of the following:
\begin{enumerate}[(a)]
\item 
\label{item:ax:hf:first}
\emph{Extensionality}.
\[
(\forall m \In k) (m \In \ell)
~~\limp~~
(\forall m \In \ell) (m \In k)
~~\limp~~
k \Eq \ell
\]

\item \emph{Finite sets}.
For each $n \in \NN$ we have an $n$-ary $\HF$-Function symbol
$\{-,\dots,-\}$ such that 
\[
  \bigconj_{1 \leq i \leq n}
  (k_i \In \{k_1,\dots,k_n\})
  \quad\land\quad
  (\forall m \In \{k_1,\dots,k_n\})
  \bigdisj_{1 \leq i \leq n}(m \Eq k_i)
    %(m \Eq k \lor m \Eq l)
\]

\noindent
We have in particular singletons $\{-\}$ and unordered pairs
$\{-,-\}$.
Using Extensionality, $\FSOD$ proves that
\[
\{\{k\},\{k,\ell\}\}
\Eq
\{\{k'\},\{k',\ell'\}\}
\quad\longliff\quad
%\left[
k \Eq k' \land \ell \Eq \ell'
%\right]
\]
We use the following shorthand:
\[
(k,\ell) \quad\deq\quad 
\{\{k\},\{k,\ell\}\}
\]

\item \emph{Union}.
We have an $\HF$-Function symbol $\cup(-)$ such that 
\[
  (\forall \ell \In \cup(k))
  (\exists m \In k) (\ell \In m)
\quad\land\quad
  (\forall \ell \In k) (\forall m \In \ell)(m \In \cup(k))
\]

\item \emph{Powerset.}
\label{item:ax:hf:po}
We have an $\HF$-Function symbol $\Po(-)$ such that 
\[
  (\forall \ell \In \Po(k)) (\forall m \In \ell)(m \In k)
\quad\land\quad
  (\forall \ell \Sle k)(\ell \In \Po(k))
\]

\noindent
The powerset is the reason for our introduction of inclusion ($\Sle$)
as an atomic formula:
It is well-known that the powerset cannot be defined by a $\Delta_0 (\In)$-formula.
A possible formula defining it is:
\[
  (\forall \ell \In \Po(k)) (\forall m \In \ell)(m \In k)
\quad\land\quad
  (\forall \ell) \big[
  (\forall m \In \ell)(m \In k) ~~\longlimp~~ \ell \In \Po(k)
  \big]
\]
The quantification $\forall \ell$ in the right conjunct is not $\In$-bounded,
and cannot be so.
In addition, we also have an $\HF$-Function symbol $\Pne(-)$
for the \emph{non-empty powerset}, that is such that
\[
k \In \Pne(\ell)
\quad\longliff\quad
\big(
k \In \Po(\ell) ~\land~ (\exists m \In k)
\big)
\]

\item \emph{Comprehension.}
Given an $\HF$-formula $\varphi$ with free variables
among $k_1,\dots,k_n,k$, we have an $\HF$-Function
$\{k \In (-) \st \varphi[-,\dots,-,k]\}$
such that 
\[
m \In \{k \In k_0 \st \varphi[k_1,\dots,k_n,k]\}
\quad\longliff\quad
m \In k_0 \land \varphi[k_1,\dots,k_n,m]
\]

\item \emph{Products.}
\label{item:ax:hf:fun:prod}
We have a binary $\HF$-Function $(-) \times (-)$
such that
\[
k \times \ell
~~\deq~~
\left\{ m \In \Po(\Po(k \cup \ell))
   \st (\exists k_0 \In k) (\exists \ell_0 \In \ell) \big[ m \Eq (k_0,\ell_0) \big]
\right\}
\]

\noindent
Moreover, we have binary projections given by $\HF$-Functions 
$\pi_1^{-,-}(-)$ and $\pi_2^{-,-}(-)$
such that 
\[
\pi_1^{k,\ell}(m)
~~=~~
\left\{
n \in \ell \st m \Sle k \times \ell ~\land~ (\exists n' \in k)\big[(n,n') \In m \big]
\right\}
\]
and similarly for $\pi_2^{-,-}(-)$.
Whenever possible, we write $\pi_i(-)$ instead of $\pi_i^{k,\ell}(-)$.
Note that by composing binary projections $\pi_1$ and $\pi_2$
we obtain projections
\[
\pi^n_i ~~:~~
k_1 \times \dots \times k_n ~~\longto~~ k_i
\]
for any $k_1,\dots,k_n$ and $i \in \{1,\dots,n\}$

\item \emph{Function Spaces and Application.}
\label{item:ax:hf:fun}
We have an exponent $\HF$-Function $(-)^{(-)}$
%We can define the exponential $\ell^k$ as
such that
\[
\ell^k ~~\deq~~
\left\{ 
  m \In \Po(k \times \ell) \st 
  (\forall k_0 \In k) (\exists ! \ell_0 \In \ell) \big[ (k_0,\ell_0) \In m \big]
\right\}
\]
%where $k \times \ell$ is
%\[
%k \times \ell
%\quad\deq\quad
%\{ m \In \Po(\Po(k \cup \ell))
%   \st \exists k_0 \In k.~ \exists \ell_0 \In \ell.~ m \Eq (k_0,\ell_0)
%\}
%\]

\noindent
Moreover, function application is given by the $\HF$-Function
$@_{-,-}(-,-)$ with
\[
@_{k,\ell}(f,a) ~~\deq~~
\left\{
  m \In \cup(\ell) \st (\exists \ell_0 \In \ell)
  \big[ m \In \ell_0 \land (a,\ell_0) \In f \big]
\right\}
\]
(here $f$ and $a$ are $\HF$-variables).
%Note that $@_{k,l}(f,a)$ is formally 
%%a constant symbol of arity $4$.
%an $\HF$-Function symbol of arity $4$.
%However, 
Whenever possible, we omit the subscripts $k,\ell$ of $@_{k,\ell}(f,a)$
and write simply $f(a)$ for $@_{k,\ell}(f,a)$.

\item
\label{item:ax:hf:last}
\emph{Disjoint Unions.}
\label{item:ax:hf:coprod}
We have a binary $\HF$-Function $(-) + (-)$ with
\[
k + \ell ~~\Eq~~ (\{0\} \times k) \cup (\{1\} \times \ell)
\]
(see Convention~\ref{conv:games:colors}
for a further discussion on finite ordinals in our context).
We moreover have $\HF$-Functions
$\inj^{k,\ell}_k(-)$, $\inj^{k,\ell}_\ell(-)$, and $[-,-]_{k,\ell}^\imath$
such that (dropping subscripts and superscripts)
\[
\inj(m) \Eq (0,m)
\qquad
\inj(n) \Eq (1,n)
\qquad
[f,g](0,m) \Eq f(m)
\qquad
[f,g](1,n) \Eq g(n)
\]
for $m \In k$, $\ell \In n$ and $f \In \imath^k$, $g \In \imath^\ell$.
\end{enumerate}

\cnote{\CR:NOTE\quad It is crucial here \wrt\@ the notion of $\HF$-closed
automata (\S\ref{sec:aut} and~\S\ref{sec:sim}) that projections as well as
operations on disjoint unions are given as (closed) $\HF$-Functions.}

%%%%%%%%%%%%%%%%%%%%%%%%%%%%%%%%%%%%%%%%%%%%%%%%%%%%%%%%%%%%%%%%%%%%%%%%%%%
\begin{conv}
\label{conv:ax:hf:fun}
%%%%%%%%%%%%%%%%%%%%%%%%%%%%%%%%%%%%%%%%%%%%%%%%%%%%%%%%%%%%%%%%%%%%%%%%%%%
Regarding function spaces as in~\ref{item:ax:hf:fun} above,
$\FSOD$ proves
\[
(k^\ell)^m \quad\iso\quad k^{\ell \times m}
\]
In the following, we reason modulo that bijection, and simply
identify $(k^\ell)^m$ with $k^{\ell \times m}$.
\end{conv}

%%%%%%%%%%%%%%%%%%%%%%%%%%%%%%%%%%%%%%%%%%%%%%%%%%%%%%%%%%%%%%%%%%%%%%%%%%%
\begin{rem}
\label{rem:ax:hf:well-order-hf}
%%%%%%%%%%%%%%%%%%%%%%%%%%%%%%%%%%%%%%%%%%%%%%%%%%%%%%%%%%%%%%%%%%%%%%%%%%%
An $\HF$-relation $\mathord\preceq \sle K\times K$ is a \emph{partial order}
on an $\HF$-term $K$, if the formula $\po(\preceq, K)$ holds in $V_\omega$, where
$\po(\preceq,K)$ is the $\HF$-formula:
\[
(\forall k,\ell,m \In K)
\Big[
%[k\leq l \lor l \leq k] \land
\big( k \preceq \ell ~~\limp~~ \ell \preceq k ~~\limp~~ k \Eq \ell \big)
~~\land~~
\big( k \preceq \ell ~~\limp~~ \ell \preceq m ~~\limp~~ k \preceq m \big)
\Big]
\]

\noindent
A partial order $\mathord\preceq \sle K \times K$ is a \emph{well-order}
if every subset of $K$ has a $\preceq$-least element,
that is, if the following formula $\formfont{WO}(\preceq,K)$ holds
in $V_\omega$:
\[
  (\forall \ell\Sle K)
  \big[
    \ell \neq \emptyset ~~\limp~~
    (\exists m \In \ell) (\forall n \In \ell) (m \preceq n)
  \big]
\]

\noindent
Since every $\HF$-set is finite and can be well-ordered, we have
\[
V_\omega 
\quad\models\quad
(\forall k)
\underbrace{
(\exists \mathord{\preceq} \Sle k\times k)
\big[
  \po(\preceq, k)
~~\land~~
  \formfont{WO}(\preceq,k)
~~\land~~
  \formfont{WO}(\succeq,k)
\big]}_{\varphi(k)}
\]
Since $\varphi(k)$ is an $\HF$-formula $\varphi(k)$,
it follows that $\FSOD$ proves $\varphi(k)$,
hence in particular that every $\HF$-set is well-ordered.
\qed
\end{rem}

%%%%%%%%%%%%%%%%%%%%%%%%%%%%%%%%%%%%%%%%%%%%%%%%%%%%%%%%%%%%%%%%%%%%%%%%%%%
\subsubsection{Functional Choice Axioms}
\label{sec:ax:choice}
%%%%%%%%%%%%%%%%%%%%%%%%%%%%%%%%%%%%%%%%%%%%%%%%%%%%%%%%%%%%%%%%%%%%%%%%%%%
We have the following functional choice axiom schemes.
\begin{itemize}
\item \emph{$\HF$-Bounded Choice for $\HF$-Sets}.
\[
  (\forall k \In K) (\exists \ell \In L) \varphi(k,\ell)
  ~~\longlimp~~
  \big( \exists f \In L^K\big) (\forall k \In K)
  \varphi(k,f(k))
\]

\item \emph{$\HF$-Bounded Choice for Functions}.
\[
  (\forall x) (\exists k \In K) \varphi(x,k)
  ~~\longlimp~~
  (\exists F : K) (\forall x) (\exists k \In K)
  \left(F(x) \Eq k ~\land~ \varphi(x,k) \right)
\]

\item
\emph{Iterated $\HF$-Bounded Choice}.
\[
(\forall k \In K) (\exists F : L)\varphi(k,F)
~~\longlimp~~
\big( \exists G : L^K \big) (\forall k \In K)
\varphi(k,F)[G(k) \sslash F]
\]

\noindent
where the substitution $[G(k)\sslash F]$ is defined as 
the usual substitution operation
but with
\[
(F(t) \Eq M)[G(k)\sslash F]
~~\deq~~
\big( \exists f \In L^K \big)
\left(
G(t) \Eq f
~\land~
f(k) \Eq M
\right)
\]
\end{itemize}

%Note
We insist
that none of these axioms create choice functions
for the \emph{individuals} of $\FSOD$
(cf Remark~\ref{rem:fso:choice}).
Despite their common shape, these three axiom schemes are actually
of different nature.
First, the axiom of
\emph{$\HF$-Bounded Choice for Functions}
\begin{equation}
\label{eq:ax:choice:fun}
  (\forall x) (\exists k \In K) \varphi(x,k)
  ~~\longlimp~~
  (\exists F : K) (\forall x) (\exists k \In K)
  \left(F(x) \Eq k ~\land~ \varphi(x,k) \right)
\end{equation}
is a counterpart in $\FSOD$ of the Comprehension Scheme of $\MSOD$.
Recalling the informal discussion in~\S\ref{sec:fso:intro}
%the introduction of this Section
and anticipating~\S\ref{sec:std} and~\S\ref{sec:cons},
let us assume a translation $\MI{-}$ from ($\HF$-closed)
$\FSOD$-formulae to $\MSOD$-formulae, and
let us assume that $K$ is a closed $\HF$-term
representing the $\HF$-set $\{\kappa_1,\dots,\kappa_n\}$.
Then the premise of~\eqref{eq:ax:choice:fun}
can be read as
\[
(\forall x) \bigdisj_{1 \leq i \leq n} \MI{\varphi(x,\kappa_i)}
\]

\noindent
The conclusion easily follows from the fact that
using Comprehension, one can 
define in $\MSOD$ a partition $X_1,\dots,X_n$ of $\univ$
such that
\[
(\forall x) \bigdisj_{1 \leq i \leq n}
\Big(
X_i(x) ~~\land~~ \MI{\varphi(x,\kappa_i)}
\Big)
\]

%\noindent
%This easily gives the conclusion.

Second, \emph{$\HF$-Bounded Choice for $\HF$-Sets}
\begin{equation}
\label{eq:ax:choice:hf}
  (\forall k \In K) (\exists \ell \In L) \varphi(k,\ell)
  ~~\longlimp~~
  \big( \exists f \In L^K\big) (\forall k \In K)
  \varphi(k,f(k))
\end{equation}
may look similar to the Axioms on $\HF$-Sets of~\S\ref{sec:ax:hf}.
The %crucial
differences are that the formula $\varphi$ here
is an \emph{arbitrary} formula of $\FSOD$,
not necessarily an $\HF$-formula in the sense of Definition~\ref{def:ax:hf:form},
and moreover that this axiom only involves $\FSOD$ (\ie\@ \emph{bounded})
quantifications, contrary to~\eqref{eq:ax:hf:zfc}.
Note that for $\HF$-formulae $\varphi$,
this axiom is indeed an instance of the axioms of~\S\ref{sec:ax:hf}.
In the general case, this axiom can be seen as following
from the fact that
quantifications over $\HF$-Sets in $\FSOD$ are
ultimately interpreted in propositional logic.
Assume that the $\HF$-terms $K$ and $L$ are closed, and correspond
to the $\HF$-sets $\kappa$ and $\lambda$ respectively.
Then the premise of~\eqref{eq:ax:choice:hf}
can be read as
\[
\bigconj_{\kappa_0 \in \kappa}
\bigdisj_{\lambda_0 \in \lambda}
\MI{\varphi(\kappa_0,\lambda_0)}
\]
which is equivalent in propositional logic to the interpretation of the conclusion
\[
\bigdisj_{f \in \lambda^\kappa} \bigconj_{\kappa_0 \in \kappa}
\MI{\varphi(\kappa_0,f(\kappa_0))}
\]

Similarly, \emph{Iterated $\HF$-Bounded Choice}
reduces to an equivalence %commutation
of the form
\[
\bigconj_{1 \leq i \leq n} (\exists \vec X)
\varphi(\kappa_i,\vec X)
\quad\longliff\quad
(\exists \vec X_1) \dots (\exists \vec X_n)
\bigconj_{1 \leq i \leq n}
\varphi(\kappa_i,\vec X_i)
\]
and follows from Comprehension.

The definition of $\FSOD$ is now complete.

%%%%%%%%%%%%%%%%%%%%%%%%%%%%%%%%%%%%%%%%%%%%%%%%%%%%%%%%%%%%%%%%%%%%%%%%%%%
\begin{nota}
%%%%%%%%%%%%%%%%%%%%%%%%%%%%%%%%%%%%%%%%%%%%%%%%%%%%%%%%%%%%%%%%%%%%%%%%%%%
Similarly as with $\MSOD$, we shall write $\FSO$ for $\FSOD$
when the set of tree directions $\Dir$ is clear from the context.
\end{nota}

%%%%%%%%%%%%%%%%%%%%%%%%%%%%%%%%%%%%%%%%%%%%%%%%%%%%%%%%%%%%%%%%%%%%%%%%%%%
\subsection{The Standard Model of $\FSO$}
\label{sec:std}
%%%%%%%%%%%%%%%%%%%%%%%%%%%%%%%%%%%%%%%%%%%%%%%%%%%%%%%%%%%%%%%%%%%%%%%%%%%

\noindent
The \emph{standard model} $\Std$
of $\FSOD$ is the full $\Dir$-ary tree $\univ$
%with strict prefix order denoted $<$, and
equipped with suitable domains for each sort:
\begin{itemize}
\item \emph{$\HF$-Sets} range over $V_\omega$, and each constant
$\const\kappa$ is interpreted by the corresponding $\HF$-set $\kappa \in V_\omega$.

\item \emph{Individuals} range over $\univ$, 
the constant $\Root$
is interpreted by the empty sequence $\root \in \univ$
and $\Succ_d$ as the map 
taking $p \in \univ$ to $p.d \in \univ$.
Moreover, we write $<$ for the strict prefix order on $\univ$.
%$p \in \univ \mapsto p.d \in \univ$.

\item \emph{Functions}
range over
\[
%\bigcup_{A \sle \univ}
\bigcup_{\kappa \in V_\omega}\left(
%A \quad\to\quad \kappa
\univ ~~\longto~~ \kappa
\right)
\]

\item
For each $n,m \in \NN$, the $\HF$-Function $\const g_{n,m}$
(of arity $n$) is interpreted as the function 
\[
\hf g_{n,m} ~~:~~ V_\omega^n ~~\longto~~ V_\omega
\]
fixed in Convention~\ref{conv:ax:hf:ac}.
%obtained by applying the axiom of Countable Axiom of Choice to
%\[ V_\omega \quad\models\quad \forall k_1,\dots,k_n. \exists! l.\varphi_{n,m} \]
\end{itemize}

%%%%%%%%%%%%%%%%%%%%%%%%%%%%%%%%%%%%%%%%%%%%%%%%%%%%%%%%%%%%%%%%%%%%%%%%%%%
\begin{rem}
%%%%%%%%%%%%%%%%%%%%%%%%%%%%%%%%%%%%%%%%%%%%%%%%%%%%%%%%%%%%%%%%%%%%%%%%%%%
Note that $\Std$ has the same individuals as the standard model of $\MSOD$.
Moreover we write $\Std$ for both the standard model
of $\FSOD$ and that of $\MSOD$, as an abuse of notation.
%In the following, we shall keep on writing $\Std$ for both standard models.
%the standard model of $\FSOD$ and that of $\MSOD$.
\end{rem}

We have the usual 
interpretation $\SI{t} \in \univ$ for each closed
individual term $t$ with parameters in $\Std$,
and
an interpretation $\SI{K} \in V_\omega$ for each closed $\HF$-term $K$
with parameters in~$\Std$.
The relation $\Std \models \varphi$, 
for a closed $\FSO$-formula $\varphi$ with parameters in $\Std$,
is defined by induction on $\varphi$ as follows:
\[
\begin{array}{!{\Std \models} l !{\quad\text{iff}\quad} l !{\quad} l}
  K \mathrel{\const\GC} L
& \SI K \GC \SI L
& \text{(for $\GC$ either $\eq$, $\in$, or $\sle$)}
\\
  t \mathrel{\const\GC} u
& \SI t \GC \SI u
& \text{(for $\GC$ either $\eq$ or $\lt$)}
\\
  \varphi \lor \psi
& (\Std \models \varphi) \text{ or } (\Std \models \psi)
\\
  \lnot \varphi
& \Std \not\models \varphi
\\
  (\exists x)\varphi
& \Std \models \varphi[p/x] ~~\text{for some $p \in \univ$}
\\
  (\exists k \mathrel{\const\GC} L) \varphi
& \Std \models \varphi[\kappa/k] ~~\text{for some $\kappa \GC \SI{L}$}
& \text{(for $\GC$ either $\in$ or $\sle$)}
\\
  (\exists F:L)\varphi
& \Std \models \varphi[\hf F/F] ~~\text{for some $\hf F \in \SI L^{\univ}$}
\end{array}
\]

\noindent
By a routine induction argument, we can show the soundness of $\FSOD$ \wrt\@ $\Std$:

%%%%%%%%%%%%%%%%%%%%%%%%%%%%%%%%%%%%%%%%%%%%%%%%%%%%%%%%%%%%%%%%%%%%%%%%%%%
\begin{prop}
\label{prop:std:cor}
%%%%%%%%%%%%%%%%%%%%%%%%%%%%%%%%%%%%%%%%%%%%%%%%%%%%%%%%%%%%%%%%%%%%%%%%%%%
Given $\FSO$-formulae $\psi_1,\dots,\psi_n,\varphi$
with free $\HF$-variables among $\vec k$,
free Individual variables among $\vec x$,
and free Function variables among $\vec F$,
if
\[
\psi_1,\dots,\psi_n \thesis_{\FSO} \varphi
\]
then
for all $\HF$-Sets $\vec \kappa$,
all $\vec p \in \univ$
and all
$\vec {\hf F} \in \bigcup_{\kappa \in V_\omega}\left(\univ \to \kappa\right)$,
we have
\[
\Std \models
\varphi[\vec \kappa/\vec k,\vec p/\vec x,\vec{\hf F}/\vec F]
\qquad\text{whenever}\qquad
\Std \models
\bigconj_{1 \leq i \leq n}
\psi_i[\vec \kappa/\vec k,\vec p/\vec x,\vec{\hf F}/\vec F]
\]
\end{prop}

%%%%%%%%%%%%%%%%%%%%%%%%%%%%%%%%%%%%%%%%%%%%%%%%%%%%%%%%%%%%%%%%%%%%%%%%%%%
\begin{rem}
\label{rem:ax:std:compl}
%%%%%%%%%%%%%%%%%%%%%%%%%%%%%%%%%%%%%%%%%%%%%%%%%%%%%%%%%%%%%%%%%%%%%%%%%%%
It follows from Remark~\ref{rem:ax:hf:compl} that the map
$\SI{-}$ is not computable on $\HF$-terms.
We refer to~\S\ref{sec:remcompl} for a discussion and a workaround.
\end{rem}

%%% Local Variables:
%%% mode: latex
%%% TeX-master: "main.tex"
%%% End:

%%%%%%%%%%%%%%%%%%%%%%%%%%%%%%%%%%%%%%%%%%%%%%%%%%%%%%%%%%%%%%%%%%%%%%%%%%%
%\subsection{Interpreting $\FSO$ in $\MSO$ and $\MSO$ in $\FSO$}
\subsection{Mutual Interpretability of $\FSO$ and $\MSO$}
\label{sec:cons}
%%%%%%%%%%%%%%%%%%%%%%%%%%%%%%%%%%%%%%%%%%%%%%%%%%%%%%%%%%%%%%%%%%%%%%%%%%%
While $\FSO$ seems more expressive than $\MSO$ (and, indeed, is easier to work with),
the two theories can mutually interpret each other 
via two formula-level translations:
\[
\MI{-} ~~:~~ \FSO ~~\longto~~ \MSO
\qquad\text{and}\qquad
(-)^\circ ~~:~~ \MSO ~~\longto~~ \FSO
\]

\noindent
%We discuss these two translations separately in~\S\ref{sec:cons:fso:mso}
%and~\S\ref{sec:cons:mso:fso} below.
%\full{In both cases, detailed proofs are deferred to Appendix~\ref{sec:app:cons}.}%
%\short{Full proofs are detailed in~\cite[App.\@ A]{dr20full}.}
%
As we shall see, both translations preserve and
reflect provability:
\begin{gather}
\FSO \thesis \varphi
\quad\text{if and only if}\quad
\MSO \thesis \MI\varphi
\tag{$\varphi$ closed $\FSO$-formula}
\\
\MSO \thesis \varphi
\quad\text{if and only if}\quad
\FSO \thesis \varphi^\circ
\tag{$\varphi$ closed $\MSO$-formula}
\end{gather}

\noindent
The interpretation $(-)^\circ$
%$(-)^\circ : \MSO \to \FSO$
of $\MSO$ in $\FSO$
simply amounts to simulate the (Monadic) Predicate variables of $\MSO$
by $\FSO$-Function variables $\univ \to \two$.
We therefore see $(-)^\circ$ as an embedding, and see $\FSO$
as a conservative extension of $\MSO$ which is faithfully interpretable in $\MSO$.
This property is not only a sanity check:
we actually rely on it in our completeness argument
(see Rem.~\ref{rem:cons:compl}).
We discuss the translations $\MI{-}$ and $(-)^\circ$
separately in~\S\ref{sec:cons:fso:mso}
and~\S\ref{sec:cons:mso:fso} below.
\full{In both cases, detailed proofs are deferred to Appendix~\ref{sec:app:cons}.}%
\short{Full proofs are detailed in~\cite[App.\@ A]{dr20full}.}

%We see the interpretation $(-)^\circ$
%%$(-)^\circ : \MSO \to \FSO$
%of $\MSO$ in $\FSO$ essentialy as an embedding: it
%simply amounts to simulate the (Monadic) Predicate variables of $\MSO$
%by $\FSO$-Function variables $\univ \to \two$.
%%
%The `conservativity' of $\FSO$ over $\MSO$ is not only a sanity check:
%we actually rely on it in our completeness argument
%(see Rem.~\ref{rem:cons:compl}).

%%%%%%%%%%%%%%%%%%%%%%%%%%%%%%%%%%%%%%%%%%%%%%%%%%%%%%%%%%%%%%%%%%%%%%%%%%%
\subsubsection{From $\FSO$ to $\MSO$}
\label{sec:cons:fso:mso}
%%%%%%%%%%%%%%%%%%%%%%%%%%%%%%%%%%%%%%%%%%%%%%%%%%%%%%%%%%%%%%%%%%%%%%%%%%%
The translation
$\MI{-} : \FSO \to \MSO$
interprets the $\HF$-part of $\FSO$ using propositional logic.
It is essentially straightforward, except for the case of Functions,
which require some care.
We will work with the following convention:

%%%%%%%%%%%%%%%%%%%%%%%%%%%%%%%%%%%%%%%%%%%%%%%%%%%%%%%%%%%%%%%%%%%%%%%%%%%
\begin{conv}
\label{conv:cons:interp:enum}
%%%%%%%%%%%%%%%%%%%%%%%%%%%%%%%%%%%%%%%%%%%%%%%%%%%%%%%%%%%%%%%%%%%%%%%%%%%
We assume that each $\HF$-set $\kappa$ comes with a fixed enumeration
$\vec \kappa = \kappa_1,\dots,\kappa_n$ of its elements.
\end{conv}

\noindent
The translation $\MI{-}$ will map an $\HF$-closed $\FSO$-formula $\varphi$
without free Function variables
to an $\MSO$-formula $\MI\varphi$.
As stated earlier, quantifications over $\HF$-Sets will be interpreted using
propositional logic.
For instance we have,
\[
\bigMI{(\exists k \In K)\varphi}
~~=~~
\bigdisj_{\kappa \in \SI K} \MI{\varphi[\kappa/k]}
\]
where $\SI K \in V_\omega$ is the standard interpretation
of the closed $\HF$-term $K$ defined in~\S\ref{sec:std}.
As a consequence, the translation $\MI{-}$ is \emph{non-uniform}
\wrt\@ $\HF$-Sets.
In particular, for an $\FSO$-formula $\varphi$
with free $\HF$-variables among $\vec k = k_1,\dots,k_p$,
each tuple of $\HF$-sets $\vec\kappa = \kappa_1,\dots,\kappa_p$
will induce a specific $\MSO$-formula
$\MI{\varphi[\vec\kappa/\vec k]}$.

The interpretation of Function variables is more complex.
Consider a closed $\HF$-term $K$
and assume $\SI K = \{\kappa_1,\dots,\kappa_c\}$.
Then a Function $(F:K)$ can be seen as a function
\[
F ~~:~~ \univ ~~\longto~~ \{\kappa_1,\dots,\kappa_c\}
\]

\noindent
As indicated in~\S\ref{sec:fso:intro},
we interpret $F$ as a tuple $X_1,\dots,X_c$ of Monadic variables
such that
\begin{equation}
\tag{for $i = 1,\dots,c$}
x \in X_i \quad\text{iff}\quad
F(x) = \kappa_i
\end{equation}

\noindent
In other words, 
$F : \univ \to \{\kappa_1,\dots,\kappa_c\}$
is seen as a partition $X_1,\dots,X_c$ of $\univ$.
%To compute $\MI\varphi$ it is actually convenient to temporarily work
To handle the interpretation of Functions in the inductive definition
of $\MI{-}$,
it is actually convenient to temporarily work
in an extension of $\FSO$ with the following atomic formulae:
\begin{itemize}
\item
$\extatom{X_1\dots X_n(t) \Eq_{\vec\kappa} L}$
where $\vec\kappa = \kappa_1,\dots,\kappa_n$
enumerates an $\HF$-set
and $X_1,\dots,X_n$ are monadic variables of $\MSO$.
%where $\kappa$ is an hereditarily finite set of the form
%$\{\lambda_1,\dots,\lambda_n\}$.
\end{itemize}

\noindent
\emph{Extended} $\FSO$-formulae are built just like
$\FSO$-formulae, but possibly using the atomic formulae above.
Extended atomic formulae are useful for dealing with 
$\HF$-bounded quantifications over Functions, say
$(\exists F : K)\varphi$.
The point is that
%Note that 
$F$ occurs
in subformulae of $\varphi$ of the form $F(t) \Eq L$,
where the $\HF$-term $L$ may contain free $\HF$-variables.
Hence the value of $L$ is not known when the translation
of $(\exists F :K)$ has to be computed.
Extended atomic formulae allow us to delay the interpretation of $F(t) \Eq L$
until $\SI L$ is known.

%We assume given,
%for each hereditarily finite set $\kappa = \{\lambda_1,\dots,\lambda_n\}$,
%a bijection $\jmath_\kappa : \kappa \eto \{1,\dots,n\}$.

The interpretation of an extended
$\HF$-closed $\FSO$-formula $\varphi$
\emph{without free Function variables}
is the $\MSO$-formula
$\MI{\varphi}$ defined by induction on $\varphi$ as follows:
\[
\begin{array}{r !{~~}@{}c@{}!{~~} l !{\qquad} l}
  %\MI{\extatom{f(t) \Eq L}^\kappa_{X_1,\dots,X_n}}
  \MI{\extatom{X_1\dots X_n(t) \Eq_{\vec \kappa} L}}
& \deq
&
  \bigdisj_{1 \leq i \leq n \ \&\ \kappa_i = \SI L} X_i(t)
\\

  \MI{K \mathrel{\const\GC} L}
& \deq
& \left\{
\begin{array}{l l}
\True & \text{if $\SI K \GC \SI L$} \\
\False & \text{otherwise}
\end{array}
\right.
& \text{(where $\mathord{\GC} \in \{=,\in,\sle\}$)}
\\

  \MI{t \GC u}
& \deq
& t \GC u
& \text{(where $\mathord{\GC} \in \{\Eq,\Lt\}$)}
\\

  \MI{\lnot \varphi}
& \deq
& \lnot \MI{\varphi}
\\

  \MI{\varphi \lor \psi}
& \deq
& \MI{\varphi} \lor \MI{\psi}
\\

  \MI{(\exists k \mathrel{\const\GC} K)\varphi}
& \deq
& \bigdisj_{\kappa \GC \SI K} \MI{\varphi[\kappa/k]}
& \text{(where $\mathord{\GC} \in \{\in,\sle\}$)}
\\

  \MI{(\exists x)\varphi}
& \deq
& (\exists x)\MI{\varphi}
\\

  \MI{(\exists F : K) \varphi}
& \deq
& (\exists X_1) \dots (\exists X_{c})\quad
\\
&
& 
\multicolumn{2}{r}{
\left\{
\begin{array}{c l}
& \Part_c(X_1,\dots,X_c)
\\
  \land
& \MI{\varphi
  [\extatom{X_1\dots X_c(t)
     \Eq_{\vec \kappa} L}
    ~/~ (F(t) \Eq L)]}
\end{array}
\right.}
\end{array}
\]

\noindent
where in the last clause,
${\SI K}$ is enumerated by $\vec \kappa = \kappa_1,\dots,\kappa_c$,
%the $\HF$-closed term
%$K^{\vec K}$ is the set of $\HF$-Functions $\{\hf g_1,\dots,\hf g_n\}$.
and $\Part_c(X_1,\dots,X_c)$
is the following $\MSO$-formula,
expressing that $X_1,\dots,X_c$ form a partition of $\univ$:
\[
\Part_c(X_1,\dots,X_c)
~~\deq~~
(\forall x)
\left[
\bigdisj_{1 \leq i \leq c} 
  \left( 
    X_i(x) \quad\land\quad
      \bigconj_{j \neq i} \lnot X_j(x)
  \right)
\right]
\]

\noindent
Note that in the definition of $\MI\varphi$ above,
since $\varphi$ is assumed to be $\HF$-closed,
the displayed $\HF$-terms $K$ and $L$ are closed,
so that
their $\Std$-interpretation $\SI K,\SI L \in V_\omega$ is defined
(see~\S\ref{sec:std}).

%%%%%%%%%%%%%%%%%%%%%%%%%%%%%%%%%%%%%%%%%%%%%%%%%%%%%%%%%%%%%%%%%%%%%%%%%%%
\begin{rem}
\label{rem:cons:dec}
%%%%%%%%%%%%%%%%%%%%%%%%%%%%%%%%%%%%%%%%%%%%%%%%%%%%%%%%%%%%%%%%%%%%%%%%%%%
Since it involves the standard interpretation map $\SI{-}$
on $\HF$-terms,
it follows from Remark~\ref{rem:ax:std:compl} (\S\ref{sec:std})
that the interpretation $\MI{-}$ is not recursive.
%As shown in~\S\ref{sec:compl:intro}, this does not impact
%the decision procedure for $\MSO$ given by the completeness
%of $\FSO + (\PosDet)$.
%This, however, makes the axiom set of 
%$\MSO + \MI{(\PosDet)}$ non-recursive.
We refer to~\S\ref{sec:posdet:mso}
and~\S\ref{sec:remcompl}
for discussions and workarounds.
\end{rem}

%%%%%%%%%%%%%%%%%%%%%%%%%%%%%%%%%%%%%%%%%%%%%%%%%%%%%%%%%%%%%%%%%%%%%%%%%%%
\begin{thm}
\label{thm:cons:fso:mso}
%%%%%%%%%%%%%%%%%%%%%%%%%%%%%%%%%%%%%%%%%%%%%%%%%%%%%%%%%%%%%%%%%%%%%%%%%%%
For every closed $\FSO$-formula $\varphi$, we have
\[
\MSO \thesis \MI{\varphi}
\qquad\text{whenever}\qquad
\FSO \thesis \varphi
\]
\end{thm}

\noindent
The proof of Theorem~\ref{thm:cons:fso:mso} is deferred
\full{to Appendix~\ref{sec:app:cons}.}%
\short{to~\cite[App.\@ A]{dr20full}.}
The logical rules of $\FSO$ are handled routinely.
The interpretations of most of the axioms of $\FSO$
are almost trivially provable from the corresponding axioms of $\MSO$.
The Functional Choice Axioms are dealt-with essentially as explained
in~\S\ref{sec:ax:choice}.
%The only non trivial cases are that of the Functional Choice Axioms
%of $\FSO$ (\S\ref{sec:ax:choice}).

%%%%%%%%%%%%%%%%%%%%%%%%%%%%%%%%%%%%%%%%%%%%%%%%%%%%%%%%%%%%%%%%%%%%%%%%%%%
\subsubsection{From $\MSO$ to $\FSO$}
\label{sec:cons:mso:fso}
%%%%%%%%%%%%%%%%%%%%%%%%%%%%%%%%%%%%%%%%%%%%%%%%%%%%%%%%%%%%%%%%%%%%%%%%%%%
The translation
$(-)^\circ : \MSO \to \FSO$
is much simpler than $\MI{-}$.
Assume given 
a $\FSO$-Function variable $F_X$
for each monadic $\MSO$-variable $X$.
The map $(-)^\circ$ is inductively defined as follows:
\[
\begin{array}{r !{~~}@{}c@{}!{~~} l !{\qquad\qquad} r !{~~\deq~~} l}
  (X(t))^\circ
& \deq
& F_X(t) \Eq 1
& (\varphi \lor \psi)^\circ
& \varphi^\circ \lor \psi^\circ
\\
  (t \Eq u)^\circ
& \deq
& t \Eq u
& (\lnot \varphi)^\circ
& \lnot (\varphi^\circ)
\\
  (t \Leq u)^\circ
& \deq
& t \Leq u
& ((\exists x)\varphi)^\circ
& (\exists x) \varphi^\circ
\\
&
&
& ((\exists X)\varphi)^\circ
& (\exists F_X:\two)\varphi^\circ
\end{array}
\]

\noindent
It is easy to see that $(-)^\circ$ is truth preserving (and reflecting) \wrt\@ the standard model $\Std$,
by a direct induction on formulae
relying on the bijection $\Po(\univ) \iso \two^{\univ}$.

%%%%%%%%%%%%%%%%%%%%%%%%%%%%%%%%%%%%%%%%%%%%%%%%%%%%%%%%%%%%%%%%%%%%%%%%%%%
\begin{lem}
\label{lem:cons:fso:mso:std}
%%%%%%%%%%%%%%%%%%%%%%%%%%%%%%%%%%%%%%%%%%%%%%%%%%%%%%%%%%%%%%%%%%%%%%%%%%%
Given a closed $\MSO$-formula $\varphi$, we have
\[
\Std \models \varphi 
\qquad\text{if and only if}\qquad
\Std \models \varphi^\circ
\]
\end{lem}

\noindent
The main result on $(-)^\circ$ is the following.
\full{Its proof is deferred to Appendix~\ref{sec:app:cons:mso}.}%
\short{Its proof is deferred to~\cite[App.\@ A]{dr20full}.}

%%%%%%%%%%%%%%%%%%%%%%%%%%%%%%%%%%%%%%%%%%%%%%%%%%%%%%%%%%%%%%%%%%%%%%%%%%%
\begin{thm}
\label{thm:cons:mso:fso:mso}
%%%%%%%%%%%%%%%%%%%%%%%%%%%%%%%%%%%%%%%%%%%%%%%%%%%%%%%%%%%%%%%%%%%%%%%%%%%
Given a closed $\MSO$-formula $\varphi$,
\[
\FSO \thesis \varphi^\circ
\qquad\text{if and only if}\qquad
\MSO \thesis \varphi
\]
\end{thm}

\noindent
Theorem~\ref{thm:cons:mso:fso:mso} can actually be extended to $\FSO$ formulae.
This is essentially the content of the following result.

%%%%%%%%%%%%%%%%%%%%%%%%%%%%%%%%%%%%%%%%%%%%%%%%%%%%%%%%%%%%%%%%%%%%%%%%%%%
\begin{prop}
\label{prop:cons:mso:cons}
%%%%%%%%%%%%%%%%%%%%%%%%%%%%%%%%%%%%%%%%%%%%%%%%%%%%%%%%%%%%%%%%%%%%%%%%%%%
For a closed $\FSO$-formula $\varphi$, we have the following.
\begin{gather}
\label{eq:cons:mso:cons:trans}
%\FSO \thesis\quad \big(\varphi ~~\liff~~ \MI{\varphi}^\circ \big)
\FSO \thesis\quad \varphi ~~\longliff~~ \MI{\varphi}^\circ
\\
\label{eq:cons:mso:cons:fso:mso}
\FSO \thesis \varphi \quad\text{iff}\quad \MSO \thesis \MI{\varphi}
\\
\label{prop:cons:mso:cons:fso:mso:std}
\Std \models \varphi \quad\text{iff}\quad \Std \models \MI{\varphi}
\end{gather}
\end{prop}

%%%%%%%%%%%%%%%%%%%%%%%%%%%%%%%%%%%%%%%%%%%%%%%%%%%%%%%%%%%%%%%%%%%%%%%%%%%
\begin{rem}
\label{rem:cons:compl}
%%%%%%%%%%%%%%%%%%%%%%%%%%%%%%%%%%%%%%%%%%%%%%%%%%%%%%%%%%%%%%%%%%%%%%%%%%%
Theorem~\ref{thm:cons:mso:fso:mso}
and Proposition~\ref{prop:cons:mso:cons} will be used
in our completeness argument (\S\ref{sec:compl}) in two different ways:
\begin{enumerate}
\item We first obtain completeness of $\FSO$
(augmented with the Axiom $(\PosDet)$ of~\S\ref{sec:posdet})
for $\MSO$ formulae via a usual translation of formulae to automata.
From this result, completeness of $\FSO + (\PosDet)$
follows by Proposition~\ref{prop:cons:mso:cons},
while completeness of $\MSO$ 
%(augmented with the $\MI{-}$-translations of all closed instances of $(\PosDet)$)
(augmented with $\MI{-}$-translations of suitable instances of $(\PosDet)$)
follows by Theorem~\ref{thm:cons:mso:fso:mso}.

\item
In addition, 
we will use Proposition~\ref{prop:cons:mso:cons} in~\S\ref{sec:msow}
in order to import the $\MSO$-theory of $\NN$ for the infinite paths of $\Std$.
We rely on this for 
the version of the Büchi-Landweber Theorem
(namely that $\FSO$ decides parity games on finite graphs)
used in the completeness argument of~\S\ref{sec:compl},
as well as for the Simulation Theorem in~\S\ref{sec:sim}.
\end{enumerate}
\end{rem}

%%% Local Variables:
%%% mode: latex
%%% TeX-master: "main.tex"
%%% End:

%%%%%%%%%%%%%%%%%%%%%%%%%%%%%%%%%%%%%%%%%%%%%%%%%%%%%%%%%%%%%%%%%%%%%%%%%%%
\subsection{Notations}
\label{sec:not}
%%%%%%%%%%%%%%%%%%%%%%%%%%%%%%%%%%%%%%%%%%%%%%%%%%%%%%%%%%%%%%%%%%%%%%%%%%%
We now introduce some notation that we will use throughout our formalization
of games and automata in $\FSO$.

%%%%%%%%%%%%%%%%%%%%%%%%%%%%%%%%%%%%%%%%%%%%%%%%%%%%%%%%%%%%%%%%%%%%%%%%%%%
\subsubsection{$\FSO$ with Extended $\HF$-Terms}
%%%%%%%%%%%%%%%%%%%%%%%%%%%%%%%%%%%%%%%%%%%%%%%%%%%%%%%%%%%%%%%%%%%%%%%%%%%
First, recall that the syntax of $\FSO$ formally disallows Functions in $\HF$-terms.
We propose here a notation system that allows them in some circumstances.
%First, while $\FSO$ formally disallow Functions in $\HF$-terms,
%it is actually possible to allow them in some circumstances.
For instance, assuming $(F:K)$,
we can use the notation
\[
(F(t) \In L)
~~\deq~~
(\exists k \In K) \big( F(t) \Eq k ~~\land~~ k \In L \big)
\]

\noindent
More generally, consider an atomic formula
\begin{equation*}
\tag{for $\mathord{\GC} \in \{\Eq,\In,\Sle\}$}
M \GC N
\end{equation*}
with $M$ and $N$ terms on the following grammar:
\begin{equation}
\label{eq:not:exthf}
M,N \quad\bnf\quad
    k 
\gs \const\kappa
\gs F(t)
%\gs g(L_1,\dots,L_n)
\gs \const g_{n,m}(L_1,\dots,L_n)
\end{equation}

\noindent
%Formulae of the form $(M \GC N)$ can be interpreted in $\FSO$,
Such formulae can be interpreted in $\FSO$,
provided one assumes bounds for the Function variables occurring in them.
Let $M$ and $N$ be as above, and assume their free Function variables
to be among $\vec F = F_1,\dots,F_n$.
Note that there are (proper) $\HF$-terms $M'$ and $N'$ such that
\[
M = M'[\vec{F(t)} / \vec \ell]
\qquad\text{and}\qquad
N = N'[\vec{F(t)} / \vec \ell]
\]
for some $\HF$-variables $\vec \ell = \ell_1,\dots,\ell_c$
and where $\vec{F(t)} = F_{i_1}(t_1),\dots,F_{i_c}(t_c)$.
Given proper $\HF$-terms $K_1,\dots,K_n$,
assuming $\vec{F}: \vec{K}$, one can let
\[
M \GC N
~~\deq~~
\big( \exists \vec \ell \in \vec L \big)
\left(
\vec{F(t)} \Eq \vec \ell
~~\land~~
M' \GC N'
\right)
\]

\noindent
where $\vec L = L_1,\dots,L_c$ is such that $L_j = K_i$
iff the $j$th element of $\vec{F(t)}$ is $F_i(t_j)$.
Note however that the above defined formula $M \GC N$
actually depends on the choice of $\vec K$, so we rather write it as:
\[
(M \GC N)_{\vec F,\vec K}
\]

\noindent
Generalizing further we can, with the above method, interpret in $\FSO$
formulae build with $\HF$-terms in the sense of~\eqref{eq:not:exthf}.
The interpretation in $\FSO$ of such a formula $\varphi$ with free Function variables among $\vec F = F_1,\dots,F_n$
is defined by induction, and depends on a choice of proper 
$\HF$-terms $\vec K = K_1,\dots,K_n$.
Using notation as above, we arrive at the following definition:
\[
\begin{array}{r !{~~\deq~~} l @{~~} l}
  (t \GC u)_{\vec F,\vec K}
& (t \GC u)
& \text{(for $\mathord{\GC} \in \{\Eq,\In\}$)}
\\

  (M \GC N)_{\vec F,\vec K}
& \left\{
\begin{array}{l l}
  G(u) \Eq N
& \text{if $(M \GC N) = (G(u) \Eq N)$}
\\
& \text{with $N$ a proper $\HF$-term}
\\
\multicolumn{2}{l}{
  (\exists \vec \ell \In \vec L)
  \left(\vec{F(t)} \Eq \vec \ell ~~\land~~ M' \GC N'\right)}
\\
& \text{otherwise}
\end{array}
\right.
& \text{(for $\mathord{\GC} \in \{\Eq,\In,\Sle\}$)}
\\

  (\lnot\varphi)_{\vec F,\vec K}
& \lnot (\varphi_{\vec F,\vec K})
\\
  (\varphi \lor \psi)_{\vec F,\vec K}
& \varphi_{\vec F,\vec K} \lor \psi_{\vec F,\vec K}
\\

  ((\exists x)\varphi)_{\vec F,\vec K}
& (\exists x) \varphi_{\vec F,\vec K}
\\

  ((\exists m \GC M)\varphi)_{\vec F,\vec K}
& %\multicolumn{2}{l}{
  (\exists \vec \ell \In \vec L)
  (\exists m \GC M')
    \left(\vec{F(t)} \Eq \vec \ell ~~\land~~
     \varphi_{\vec F,\vec K} \right)
  %}
& \text{(for $\mathord{\GC} \in \{\In,\Sle\}$)}
\\

  ((\exists G : M)\varphi)_{\vec F,\vec K}
& \multicolumn{2}{l}{
  (\exists \vec \ell \In \vec L)
  (\exists G : M') 
  \left( \vec{F(t)} \Eq \vec \ell ~~\land~~ \varphi_{G\vec F,M'\vec K} \right)}
\end{array}
\]

\noindent
Beware that $(\varphi)_{\vec F,\vec K}$ only makes sense under the assumptions
$\vec{F}: \vec{K}$.
Keeping this in mind we may obtain, for instance, the following formulations
of the Functional Choice Axioms of~\S\ref{sec:ax:choice}.
\begin{itemize}
%\item \emph{$\HF$-Bounded Choice for $\HF$-Sets}.
%\[
%  (\forall k \In K) (\exists \ell \In L) \varphi(k,\ell)
%  ~~\limp~~
%  \left(\exists f \In L^K\right) (\forall k \In K)
%  \varphi(k,f(k))
%\]

\item \emph{$\HF$-Bounded Choice for Functions}.
\[
  (\forall x) (\exists k \In K) \varphi(x,k)
  ~~\longlimp~~
  (\exists F:K) (\forall x)
  \varphi(x,F(x))
\]

\item
\emph{Iterated $\HF$-Bounded Choice}.
\[
(\forall k \In K) (\exists F:L) \varphi(k,F)
~~\longlimp~~
\left(\exists F : L^K\right) (\forall k \In K) \varphi(k , F(-,k))
\]
\end{itemize}

%%%%%%%%%%%%%%%%%%%%%%%%%%%%%%%%%%%%%%%%%%%%%%%%%%%%%%%%%%%%%%%%%%%%%%%%%%%
\subsubsection{Notations for Products and Functions}
\label{sec:funto}
%%%%%%%%%%%%%%%%%%%%%%%%%%%%%%%%%%%%%%%%%%%%%%%%%%%%%%%%%%%%%%%%%%%%%%%%%%%
We now introduce notation for a form of \emph{product type}, based
on the function spaces and application functions
of~\S\ref{sec:ax:hf}.\ref{item:ax:hf:fun}.
The main idea is to be able to manipulate a Function variable
\[
F ~~:~~ K^L
\]
as a function
\[
F ~~:~~ \univ \times L ~~\longto~~ K
\]

\noindent
Furthermore, it is convenient to allow such $F$ to have a domain
defined by an $\FSO$ formula $\psi(-)$, and to write
\[
F ~~:~~ \psi(-) \times L ~~\longto~~ K
\qquad\text{for}\qquad
(\forall x)\big(\psi(x) ~~\longlimp~~ F(x) \In K^L \big)
\]

\noindent
We develop here a notation system for such ``function'' and ``product types''.
In~\S\ref{sec:funto:choice}, we discuss formulations 
of the Functional Choice Axioms of~\S\ref{sec:ax:choice}
induced by this notation.
In order not to overload the arrow symbol $\longto$
(which will be used with games later on),
we will write typing declarations as
\[
\funto{F}{\univ \times L}{K}
\qquad\text{instead of}\qquad
F: \univ \times L \longto K
\]

%%%%%%%%%%%%%%%%%%%%%%%%%%%%%%%%%%%%%%%%%%%%%%%%%%%%%%%%%%%%%%%%%%%%%%%%%%%
\begin{nota}[Product Types]%[Products]
%\label{not:funto}
%%%%%%%%%%%%%%%%%%%%%%%%%%%%%%%%%%%%%%%%%%%%%%%%%%%%%%%%%%%%%%%%%%%%%%%%%%%
\emph{Product types} are given by the following grammar,
where $\psi(-)$ is an $\FSO$ formula of an individual variable
(with possibly other free variables of any sort),
and where $K$ is an $\HF$-term.
\[
\Pi \quad\bnf\quad
  %\univ \gs
  \psi(-) \gs
  K \gs
  \Pi \times K
\]

\noindent
The \emph{arity} of a product type $\Pi$ is:
\begin{itemize}
\item $(1,n)$ if $\Pi$ is of the form
$\psi(-) \times K_1 \times \dots \times K_n$,

\item $(0,n)$ if $\Pi$ is of the form
$K_1 \times \dots \times K_n$.
\end{itemize}

%The \emph{arity} of a product type $\Pi$ is defined as follows:
%\begin{itemize}
%\item the arity of $\psi(-) \times K_1 \times \dots \times K_n$
%%and $\univ \times K_1 \times \dots \times K_n$
%is $(1,n)$,
%\item the arity of $K_1 \times \dots \times K_n$
%is $(0,n)$.
%\end{itemize}
%\[
%\Pi \quad\bnf\quad
%    \psi(-)
%\gs \Pi \times K
%\]

\noindent
Product types are to be used with the following defined formulae.
\[
\begin{array}{r !{~~\deq~~} l}
  (t,\vec K) \PEq (u,\vec L)
& t \Eq u ~~\land~~ \vec K \Eq \vec L
\\
  (t,\vec K) \PIn \psi(-) \times \vec L
& \psi(t) ~~\land~~ \vec K \In \vec L
\\[1em]
  \funto{f}{K}{L}
& f \In L^K
\\
  \funto{F}{\psi(-)}{L}
& (\forall x)(\psi(x) ~~\limp~~ F(x) \In L)
\\
  \funto{\ptfun{f}}{\Pi \times K}{L}
& \funto{\ptfun{f}}{\Pi}{L^K}
\end{array}
\]

\noindent
Here $\ptfun{f}$ stands for a Function variable $F$ if the arity of $\Pi$
is of the form $(1,n)$,
and for an $\HF$-variable $f$ if the arity of $\Pi$ is of the form $(0,n)$.
Moreover, 
for $\Pi = \psi(-) \times K_1 \times \dots \times K_n$,
we let
\[
(\exists \funto{F}{\Pi}{L})\varphi
~~\deq~~
\left(\exists F : L^{K_n \times \dots \times K_1}\right)
  \big[\funto{F}{\Pi}{L} ~~\land~~ \varphi \big]
\]
\end{nota}

%%%%%%%%%%%%%%%%%%%%%%%%%%%%%%%%%%%%%%%%%%%%%%%%%%%%%%%%%%%%%%%%%%%%%%%%%%%
\begin{rems}
\label{rem:funto}
%%%%%%%%%%%%%%%%%%%%%%%%%%%%%%%%%%%%%%%%%%%%%%%%%%%%%%%%%%%%%%%%%%%%%%%%%%%
\hfill
\begin{enumerate}
\item
\label{rem:funto:tot}
Thanks to Rem.~\ref{rem:ax:hf:well-order-hf},
using the Axioms of $\HF$-Bounded Choice
%(for Functions or $\HF$-Sets)
(\S\ref{sec:ax:choice}),
we have
\[
\big( \funto{\ptfun{f}}{\Pi}{L} ~~\land~~ \varphi \big)
~~\longlimp~~
(\exists \funto{\ptfun{f}}{\Pi}{L}) \varphi
\]

\item
\label{eq:funto:prod}
Using Convention~\ref{conv:ax:hf:fun}, 
for each product type $\Pi$ we have
\[
  (\funto{\ptfun{f}}{\Pi \times K_1 \times \dots \times K_n}{K})
\quad\longliff\quad
  \left(
  \funto{\ptfun{f}}{\Pi} {K^{K_n \times \dots \times K_1}}
  \right)
\]
It follows that
for each %$1 \leq i \leq n$ 
product type $\Pi$
and each formula $\varphi$
we have
\[
(\exists\funto{\ptfun{f}}{\Pi \times K_1 \times \dots \times K_n}{K})\varphi
\quad\longliff\quad
\left(\exists\funto{\ptfun{f}}{\Pi}{K^{K_n \times \dots \times K_1}}\right)\varphi
\]
\end{enumerate}
\end{rems}

%%%%%%%%%%%%%%%%%%%%%%%%%%%%%%%%%%%%%%%%%%%%%%%%%%%%%%%%%%%%%%%%%%%%%%%%%%%
\begin{nota}
%\label{not:funto}
%%%%%%%%%%%%%%%%%%%%%%%%%%%%%%%%%%%%%%%%%%%%%%%%%%%%%%%%%%%%%%%%%%%%%%%%%%%
In the following, given a product type $\Pi$,
we use the notation $\tilde t : \Pi$,
where $\tilde t$ stands for a tuple of the form $(t,K_1,\dots,K_n)$
if $\Pi$ has arity $(1,n)$,
or of the form $(K_1,\dots,K_n)$ if $\Pi$ has arity $(0,n)$.
When $\Pi$ is clear from the context, we write $\tilde t$
instead of $\tilde t : \Pi$, and furthermore 
we may omit the overset tilde, writing $t$ instead of $\tilde t$.

Write $\tilde x$ for tuples of variables of the form
$(x,k_1,\dots,k_n)$ or of the form $(k_1,\dots,k_n)$.
%\hfill
\begin{enumerate}
%\item
%Given a product type $\Pi$,
%we use the notation $\tilde t : \Pi$,
%where $\tilde t$ stands for a tuple of the form $(t,K_1,\dots,K_n)$
%if $\Pi$ has arity $(1,n)$,
%or of the form $(K_1,\dots,K_n)$ if $\Pi$ has arity $(0,n)$.
%When $\Pi$ is clear from the context, we write $\tilde t$
%instead of $\tilde t : \Pi$, and furthermore 
%we may omit the overset tilde, writing $t$ instead of $\tilde t$.
%
%Write $\tilde x$ for tuples of variables of the form
%$(x,k_1,\dots,k_n)$ or of the form $(k_1,\dots,k_n)$.

\item
If $\Pi = \psi(-) \times K_1 \times \dots \times K_n$
and $\tilde t = (t,L_1,\dots,L_n)$,
we write $F^{\Pi \to K}(\tilde t)$ for the $\HF$-term
\[
@_{K_1 \times \dots \times K_n,K}(F(t),(L_1,\dots,L_n))
\]
If $\Pi = K_1 \times \dots \times K_n$
and $\tilde t = (L_1,\dots,L_n)$,
we write $f^{\Pi \to K}(\tilde t)$ for the $\HF$-term
\[
@_{K_1 \times \dots \times K_n,K}(f,(L_1,\dots,L_n))
\]

\noindent
When $\Pi$ and $K$ are
clear from the context, in either case above we write
$\ptfun f(\tilde t)$ 
for
$\ptfun f^{\Pi \to K}(\tilde t)$.

\item
We furthermore write $\tilde t \In \ptfun f$ or even $\ptfun f(\tilde t)$
for the formula
\[
\ptfun f(\tilde t) \Eq 1
\]

\item
We extend product types as follows
\[
\Pi \quad\bnf\quad \ldots \gs \univ \gs X
\]
where $X$ is a Function variable. We let
\[
\begin{array}{r !{~~\deq~~} l}
  \funto{F}{X}{L}
& (\forall x)(X(x) ~~\limp~~ F(x) \In L)
\\
  \funto{F}{\univ}{L}
& %F:L \quad=\quad
  (\forall x)(F(x) \In L)
\\[1em]

  (t,\vec K) \PIn \univ \times \vec L
& \True ~~\land~~ \vec K \In \vec L
\\
  (t,\vec K) \PIn X \times \vec L
& X(t) \Eq 1 ~~\land~~ \vec K \In \vec L
 
\end{array}
\]

\noindent
Note that
\[
\begin{array}{r !{\quad\longliff\quad} l}
  \funto{F}{(\univ \times K_1 \times \dots \times K_n)}{L}
& \funto{F}{(\True \times K_1 \times \dots \times K_n)}{L}
\\
  (\exists \funto{F}{\univ}{L})\varphi
& (\exists F : L)\varphi
\end{array}
\]
%
%\noindent
%We shall not use the notation
%$(\exists F : L)\varphi$
%again, except when discussing meta-theoretical properties of
%$\FSOD$.
\end{enumerate}
\end{nota}

%%%%%%%%%%%%%%%%%%%%%%%%%%%%%%%%%%%%%%%%%%%%%%%%%%%%%%%%%%%%%%%%%%%%%%%%%%%
\subsubsection{Choice and Comprehension}
\label{sec:funto:choice}
%%%%%%%%%%%%%%%%%%%%%%%%%%%%%%%%%%%%%%%%%%%%%%%%%%%%%%%%%%%%%%%%%%%%%%%%%%%
We list here some important straightforward consequences 
of the Functional Choice Axioms of~\S\ref{sec:ax:choice} pertaining to Product Types.

%%%%%%%%%%%%%%%%%%%%%%%%%%%%%%%%%%%%%%%%%%%%%%%%%%%%%%%%%%%%%%%%%%%%%%%%%%%
\begin{thm}
\label{thm:funto:choice}
%%%%%%%%%%%%%%%%%%%%%%%%%%%%%%%%%%%%%%%%%%%%%%%%%%%%%%%%%%%%%%%%%%%%%%%%%%%
$\FSOD$ proves the following generalizations of the 
Functional Choice Axioms of~\S\ref{sec:ax:choice}:
\begin{itemize}
\item \emph{$\HF$-Bounded Choice}.
\[
  (\forall \tilde x \In \Pi) (\exists k \In L) \varphi(\tilde x,k)
  ~~\longlimp~~
  (\exists \funto{\ptfun f}{\Pi}{L}) (\forall \tilde x \In \Pi)
  \varphi(\tilde x, \ptfun f(\tilde x))
\]

\item
\emph{Iterated $\HF$-Bounded Choice}.
\[
(\forall k \In K) (\exists \funto{\ptfun f}{\Pi}{L}) \varphi(k,\ptfun f)
~~\longlimp~~
\left( \exists \funto{\ptfun f}{\Pi}{L^K} \right)
(\forall k \In K) \varphi(k , \ptfun f(-,k))
\]
\end{itemize}
\end{thm}

\begin{fullproof}
\hfill
\begin{itemize}
\item \emph{Iterated $\HF$-Bounded Choice for $\Pi$ of arity $(0,n)$}.\\
Let $\Pi = K_1 \times \dots \times K_n$.
We have to prove
\[
  (\forall k \In K) (\exists \funto{f}{\Pi}{L})
  \varphi(k,f)
  ~~\longlimp~~
  \left( \exists \funto{f}{\Pi}{L^K} \right) (\forall k \In K)
  \varphi(k,f((-),k))
\]
which by
Remark~\ref{rem:funto}.\eqref{eq:funto:prod}
amounts to
\[
  (\forall k \In K) \left( \exists f \In L^{K_n \times \dots \times K_1} \right)
  \varphi(k,f)
  ~~\longlimp~~
  \left( \exists f \In L^{K \times K_n \times \dots \times K_1} \right)
  (\forall k \In K)
  \varphi(k,f((-),k))
\]

\noindent
But by the axioms on $\HF$-Sets this is equivalent to
\[
  (\forall k \In K) \left(\exists f \In L^{K_n \times \dots \times K_1} \right)
  \varphi(k,f)
  ~~\longlimp~~
  \left( \exists f \In L^{K_n \times \dots \times K_1 \times K} \right(
  (\forall k \In K)
  \varphi(k,f(k,(-)))
\]
which follows from the axiom of $\HF$-Bounded Choice for $\HF$-Sets.

\item \emph{$\HF$-Bounded Choice for $\Pi$ of arity $(0,n)$}.\\
Let $\Pi = K_1 \times \dots \times K_n$.
We have to prove
\[
  (\forall \vec k \In \vec K) (\exists k \In L) \varphi(\vec k,k)
  ~~\longlimp~~
  (\exists \funto{f}{\Pi}{L}) (\forall \vec k \In \vec K)
  \varphi(\vec k,f(\vec k))
\]
which by
Remark~\ref{rem:funto}.\eqref{eq:funto:prod}
amounts to
\[
  (\forall \vec k \In \vec K) (\exists k \In L) \varphi(\vec k,k)
  ~~\longlimp~~
  \left( \exists f \In K^{K_n \times \dots \times K_1} \right)
  (\forall \vec k \In \vec K)
  \varphi(\vec k,f(\vec k))
\]
We can then conclude by Iterated $\HF$-Bounded Choice for $\HF$-Sets
(\ie\@ for $\Pi$ of arity $(0,n)$).

\item \emph{$\HF$-Bounded Choice for $\Pi$ of the form
  $\Pi = \univ \times K_1 \times \dots \times K_n$}.\\
Using Remark~\ref{rem:funto}.\eqref{eq:funto:prod},
we have to prove
\[
  (\forall x)(\forall \vec k \In \vec K)( \exists k \In L) \varphi(\vec k,k)
  ~~\longlimp~~
  \left(\exists F : L^{K_n \times \dots \times K_1}\right)
  (\forall x) (\forall \vec k \In \vec K)
  \varphi(x,\vec k,F(x,\vec k))
\]

\noindent
But by the axiom of Iterated $\HF$-Bounded Choice for $\HF$-Sets, we have
\[
  (\forall x) (\forall \vec k \In \vec K) (\exists k \In L) \varphi(\vec k,k)
  ~~\longlimp~~
  (\forall x)
  \left( \exists f \In L^{K_n \times \dots \times K_1} \right)
  (\forall \vec k \In \vec K)
  \varphi(x,\vec k,f(\vec k))
\]
and we conclude by the axiom of $\HF$-Bounded Choice for Functions.

\item \emph{Iterated $\HF$-Bounded Choice for $\Pi$ of the form
  $\Pi = \univ \times K_1 \times \dots \times K_n$}.\\
Using Remark~\ref{rem:funto}.\eqref{eq:funto:prod},
we have to prove
\[
(\forall k \In K) \left( \exists F : L^{K_n \times \dots \times K_1} \right) \varphi(k,F)
~~\longlimp~~
\left( \exists F : L^{K_n \times \dots \times K_1 \times K} \right)
(\forall k \In K) \varphi(k , F(-,k))
\]
and we conclude by Iterated $\HF$-Choice for Functions.

\item \emph{$\HF$-Bounded Choice for $\Pi$ of the form
  $\Pi = \psi(-) \times K_1 \times \dots \times K_n$}.\\
Let $\Pi_0 \deq \univ \times K_1 \times \dots \times K_n$.
We have to show
\[
  (\forall \tilde x \In \Pi) (\exists k \In L) \varphi(\tilde x,k)
  \quad\longlimp\quad
  (\exists \funto{F}{\Pi}{L}) (\forall \tilde x \In \Pi)
  \varphi(\tilde x,F(\tilde x))
\]

\noindent
but this follows from
\begin{multline*}
  (\forall \tilde x \In \Pi_0) (\exists k \In L)
  \big(
    \tilde x \In \Pi
    ~\limp~
    \varphi(\tilde x,k)
  \big)
  \quad\longlimp
\\
  (\exists \funto{F}{\Pi_0}{L}) (\forall \tilde x \In \Pi_0)
  \big(
    \tilde x \In \Pi
    ~\limp~
    \varphi(\tilde x,F(\tilde x))
  \big)
\end{multline*}

\item \emph{$\HF$-Bounded Choice for $\Pi$ of the form
  $\Pi = \psi(-) \times K_1 \times \dots \times K_n$}.\\
Let $\Pi_0 \deq \univ \times K_1 \times \dots \times K_n$.
We have to show
\[
  (\forall k \In K) (\exists \funto{F}{\Pi}{L}) \varphi(k,F)
  ~~\longlimp~~
  \left(\exists \funto{F}{\Pi}{L^K}\right) \left(\forall k \In K\right)
  \varphi(k,F(-,k))
\]

\noindent
but this follows from
\begin{multline*}
  (\forall k \In K) (\exists \funto{F}{\Pi_0}{L})
  \big(
    \funto{F}{\Pi}{L}
    ~~\land~~
    \varphi(k,F)
  \big)
  \quad\longlimp
\\
  \left(\exists \funto{F}{\Pi_0}{L^K}\right) (\forall k \In K)
  \big(
      \funto{F(-,k)}{\Pi}{L} 
  ~~\land~~
  \varphi(k,F(-,k))
  \big)
\end{multline*}
\end{itemize}
\end{fullproof}

%%%%%%%%%%%%%%%%%%%%%%%%%%%%%%%%%%%%%%%%%%%%%%%%%%%%%%%%%%%%%%%%%%%%%%%%%%%
\begin{thm}[Comprehension for Product Types]
\label{thm:funto:ca}
%%%%%%%%%%%%%%%%%%%%%%%%%%%%%%%%%%%%%%%%%%%%%%%%%%%%%%%%%%%%%%%%%%%%%%%%%%%
$\FSOD$
proves the following form of Comprehension, %for subsets of product types
where $V$ does not occur free in $\varphi$:
\[
  (\exists \funto{V}{\Pi}{\two})
  (\forall \tilde x \In \Pi)
  \big( V(\tilde x) ~~\longliff~~ \varphi(\tilde x) \big)
\]
\end{thm}

%%%%%%%%%%%%%%%%%%%%%%%%%%%%%%%%%%%%%%%%%%%%%%%%%%%%%%%%%%%%%%%%%%%%%%%%%%%
\begin{proof}
%%%%%%%%%%%%%%%%%%%%%%%%%%%%%%%%%%%%%%%%%%%%%%%%%%%%%%%%%%%%%%%%%%%%%%%%%%%
We require
\[
(\exists \funto{V}{\Pi}{\two}) (\forall \tilde x \In \Pi )
\big( V(\tilde x) \Eq 1 ~~\longliff~~ \varphi(x) \big)
\]

\noindent
By excluded middle, bounded existentials and generalization we have,
\[
(\forall \tilde x\In \Pi) (\exists k \In \two) 
  \big( k \Eq 1  ~~\longliff~~ \varphi(x) \big)
\]

\noindent
and we conclude by $\HF$-Bounded Choice.
\end{proof}

%%%%%%%%%%%%%%%%%%%%%%%%%%%%%%%%%%%%%%%%%%%%%%%%%%%%%%%%%%%%%%%%%%%%%%%%%%%
\begin{rems}
\label{rem:hfchoice}
%%%%%%%%%%%%%%%%%%%%%%%%%%%%%%%%%%%%%%%%%%%%%%%%%%%%%%%%%%%%%%%%%%%%%%%%%%%
\hfill
\begin{enumerate}
\item
In the case of $\Pi = \univ$, 
Theorem~\ref{thm:funto:ca} gives Comprehension for characteristic functions:
\begin{equation*}
\tag{$X$ not free in $\varphi$}
  (\exists \funto{X}{\univ}{\two})
  (\forall x)  \big(x \in X ~~\longliff~~ \varphi(x) \big)
%\qquad\text{($X$ not free in $\varphi$)}
\end{equation*}

\item
We have the following form of Comprehension
for $\HF$-Sets:
\begin{equation*}
\tag{$\ell$ not free in $\varphi$}
  (\exists \ell \Sle K)
  (\forall k \In K) \big(
k \In \ell ~~\longliff~~ \varphi(k) \big)
%\qquad\text{($\ell$ not free in $\varphi$)}
\end{equation*}
\begin{fullproof}
By Theorem~\ref{thm:funto:ca} we have
\begin{equation*}
\tag{$g$ not free in $\varphi$}
  (\exists \funto{g}{K}{\two})
  (\forall k \In K) \left[g(k)\Eq 1 ~~\longliff~~ \varphi(k)\right]
%\qquad\text{($g$ not free in $\varphi$)}
\end{equation*}
By the axiom on $\HF$-Functions, we get that
\begin{equation*}
\tag{$g$ not free in $\varphi$}
  \left(\exists \funto{g}{\one}{\two^K}\right)
  (\forall k \In K) \left[
g(0,k)\Eq 1 ~~\longliff~~ \varphi(k)\right]
%\qquad\text{($g$ not free in $\varphi$)}
\end{equation*}
and therefore
\begin{equation*}
\tag{$g$ not free in $\varphi$}
  \left(\exists g \In \two^K \right)
  (\forall k \In K) \left[
g(k)\Eq 1 ~~\longliff~~ \varphi(k)\right]
%\qquad\text{($g$ not free in $\varphi$)}
\end{equation*}
But now, we have
\[
V_\omega \models\quad
\left(\forall g \In \two^K \right)
  (\exists \ell \Sle K) (\forall k \In K)\left[
  k \In \ell ~~\longliff~~ g(k) \Eq \one
  \right]
\]
It then follows from the Axioms on $\HF$-Sets 
(Remark~\ref{rem:ax:hf:vomega})
that
\begin{equation*}
\tag{$\ell$ not free in $\varphi$}
  (\exists \ell \Sle K)
  (\forall k \In K) \left[
k \In \ell ~~\longliff~~ \varphi(k)\right]
%\qquad\text{($\ell$ not free in $\varphi$)}
\end{equation*}
\end{fullproof}

\item
\label{eq:hfchoice:ind}
Using Comprehension for $\HF$-Sets, the well-orders on $\HF$-Sets
given by Remark~\ref{rem:ax:hf:well-order-hf}
give the following Induction Scheme for $\HF$-Sets:
\[
(\forall k \In K) \Big[
(\forall \ell \In K) \big(
\ell \prec_K k ~\limp~ \varphi(\ell)
\big)
~~\longlimp~~
\varphi(k)
\Big]
\quad\longlimp\quad
(\forall k \In K) \varphi(k)
\]

\noindent
where $\preceq_K$ is the well-order on $K$ given by Remark~\ref{rem:ax:hf:well-order-hf}.

\begin{fullproof}
Assume toward a contradiction that there is some $k \In K$ such that
$\lnot\varphi(k)$.
Then, by Comprehension for $\HF$-Sets, for the set of all $k \In K$
such that $\lnot\varphi(k)$.
Since this set is non-empty, it has a $\preceq$-least element, say $k_0$.
But then $\varphi(\ell)$ holds for all $\ell \prec k_0$, so we get
$\varphi(k)$, a contradiction.
\end{fullproof}
\end{enumerate}
\end{rems}

%%% Local Variables:
%%% mode: latex
%%% TeX-master: "main.tex"
%%% End:

%%%%%%%%%%%%%%%%%%%%%%%%%%%%%%%%%%%%%%%%%%%%%%%%%%%%%%%%%%%%%%%%%%%%%%%%%%%
\section{Game Positions}
\label{sec:pos}
%%%%%%%%%%%%%%%%%%%%%%%%%%%%%%%%%%%%%%%%%%%%%%%%%%%%%%%%%%%%%%%%%%%%%%%%%%%

This Section and the next one describe our setting for games. 
The games we consider ultimately aim at formalizing acceptance games of tree automata
(\S\ref{sec:aut}),
and thus must encompass acceptance games for \emph{non-deterministic} tree automata.
We shall therefore give a setting for infinite games, with players
\emph{Proponent} $\Prop$ (corresponding to \emph{Automaton} or $\exists$loïse)
and \emph{Opponent} $\Opp$ (corresponding to \emph{Pathfinder} or $\forall$bérlard).
In the case of acceptance games, $\Prop$ plays for acceptance and $\Opp$
plays for rejection, and in the particular case of 
%acceptance games for
non-deterministic automata, $\Prop$ chooses transitions
from the non-deterministic transition relation, while $\Opp$
chooses tree directions $d \in \Dir$, with the aim of building an infinite path.
This leads to an inherent asymmetry in the very notion of games, where,
from a game position with a given tree position $x \in \univ$,
$\Prop$ can only go to game positions with tree position $x$,
while $\Opp$ must go to a game position with tree position a successor of $x$.

Due to the fact that we cannot access the usual primitive recursive codings in
the monadic language, we will only consider games that are `superposed' onto
the infinite $\Dir$-tree,
with only boundedly many positions associated with each tree node.
Such a setting indeed suffices for the case of acceptance games arising from tree automata.
Assume that we are given disjoint non-empty $\HF$-Sets $\PL\G$ and $\OL\G$ of
\emph{Proponent} and \emph{Opponent} labels respectively.
Intuitively, \emph{Proponent} will play from game positions of the form
\[
\univ \times \PL\G
\]
while \emph{Opponent} will play from positions of the form
\[
\univ \times \OL\G
\]

\noindent
A game will be given by specifying edge relations of the form
\[
(x,k) \edge{}{\Prop} (x,\ell)
\quad\text{or}\quad
(x,\ell) \edge{}{\Opp} (x.d,k)
\qquad\text{where $k \in \PL\G$, $\ell \in \OL\G$ and $d \in \Dir$.}
\]
So $\Prop$ can only move to a game position with the same
underlying tree position,
while $\Opp$ is forced to move to a game position with a
successor underlying tree position.
This induces a dag structure on game positions,
whose underlying partial order $\gle_\G$
is the lexicographic product of the usual tree order with
the one setting $\Prop$-labels smaller than all $\Opp$-labels.
The games we shall consider will all be subrelations of $\gle_\G$. %this partial order.

This Section is devoted to the definition of this dag structure.
We shall also prove some basic results 
related to induction in~\S\ref{sec:pos:ind}
and to infinite paths in~\S\ref{sec:pos:path}.
These will help proving some similar results for games in~\S\ref{sec:games},
for which arguments are more naturally given at the level of $\gle_\G$.
%which rely on arguments which more naturally live at the level of $\gle_\G$.

%This Section is devoted to the definition of this dag structure.
%We shall also prove some basic results 
%(namely related to induction in~\S\ref{sec:pos:ind}
%and to infinite paths in~\S\ref{sec:pos:path})
%which will entail similar results for games in~\S\ref{sec:games},
%while relying on arguments which more naturally live
%%but which rely on arguments 
%%but whose proofs are more naturally 
%at the level of $\gle_\G$.

%%%%%%%%%%%%%%%%%%%%%%%%%%%%%%%%%%%%%%%%%%%%%%%%%%%%%%%%%%%%%%%%%%%%%%%%%%%
\subsection{A Partial Order of Game Positions}
\label{sec:pos:po}
%%%%%%%%%%%%%%%%%%%%%%%%%%%%%%%%%%%%%%%%%%%%%%%%%%%%%%%%%%%%%%%%%%%%%%%%%%%

We first introduce the formal notion of labels of game positions.

%%%%%%%%%%%%%%%%%%%%%%%%%%%%%%%%%%%%%%%%%%%%%%%%%%%%%%%%%%%%%%%%%%%%%%%%%%%
\begin{defi}[Labels of Game Positions]
%%%%%%%%%%%%%%%%%%%%%%%%%%%%%%%%%%%%%%%%%%%%%%%%%%%%%%%%%%%%%%%%%%%%%%%%%%%
\emph{Labels of game positions} are pairs $(\PL\G,\OL\G)$
of $\HF$-terms satisfying the following formula:
\[
\Labels(\PL\G,\OL\G)
\quad\deq\quad
 \lnot (\exists k \In \PL\G \cap \OL\G)
  ~~\land~~
  (\exists k \In \PL\G)
  ~~\land~~
  (\exists \ell \In \OL\G)
\]
We write $\PO\G$ for $\PL\G \cup \OL\G$.
When no ambiguity arises, we write $\PL{}$, $\OL{}$ and $\PO{}$
for $\PL\G$, $\OL\G$ and $\PO\G$ respectively.
%When no ambiguity occurs, we write $\G$ for $\PL\G \cup \OL\G$.
\end{defi}

%Informally, a game position is a pair $(x,k)$
%with $x \in \univ$ and $k \in \G$, for labels of game positions $(\PL\G,\OL\G)$.

Assume $(\PL{},\OL{})$ are labels of game positions.
Intuitively, game positions are pairs $(x,k)$ with $x \in \univ$ and $k \in \PO{}$,
Proponent's positions are game positions with $k \in \PL{}$
and Opponent's positions are game positions with $k \in \OL{}$.
To summarize, we have the informal correspondence:
\[
\begin{array}{l !{\qquad} l}
  \univ \times
  \PO{}
  %(\PL\G \cup \OL\G)
& \text{Game positions}
\\
  \univ \times \PL{}
& \text{Proponent's positions}
\\
  \univ \times \OL{}
& \text{Opponent's positions}
\end{array}
\]

%\[
%\begin{array}{l !{\quad=\quad} l !{\quad\deq\quad} l !{\qquad} l}
%  \PosSet
%& \PosSet(\PM,\OM)
%& \univ \times \Moves
%& \text{(\emph{Game positions})}
%\\
%  \PropPosSet
%& \PropPosSet(\PM,\OM)
%& \univ \times \PM
%& \text{(\emph{Proponent's positions})}
%\\
%  \OppPosSet
%& \OppPosSet(\PM,\OM)
%& \univ \times \OM
%& \text{(\emph{Opponent's positions})}
%\end{array}
%\]

We will throughout the paper use the following notation
to manipulate game positions and sets of game positions.

%%%%%%%%%%%%%%%%%%%%%%%%%%%%%%%%%%%%%%%%%%%%%%%%%%%%%%%%%%%%%%%%%%%%%%%%%%%
\begin{nota}[Game Positions]
%\label{def:games:pos}
\label{not:games:pos}
%%%%%%%%%%%%%%%%%%%%%%%%%%%%%%%%%%%%%%%%%%%%%%%%%%%%%%%%%%%%%%%%%%%%%%%%%%%
We introduce the following notation,
assuming $\Labels(\PL\G,\OL\G)$.
\begin{enumerate}
\item
\label{eq:games:pos:indiv}
Variables, written $u,v,w,\etc$, range over game positions,
that is over pairs $(x,k)$ with $x$ an Individual variable and
$k$ an $\HF$-variable ranging over $\PO\G$.

\item
Sets of game positions, written $U,V,W,\etc$,
range over %subsets of $\PosSet$ represented as
Functions $\univ \times \PO\G \fsoto \two$.
We will systematically use the following notation:
\[
\begin{array}{l !{~~\deq~~} l !{\qquad} l}
  \funto{V}{\G_{\phantom{\Opp}}}{\two}
& \funto{V}{\univ \times \PO{\G}}{\two}
& \text{(sets of Game positions)}
\\
  \funto{V}{\PP\G}{\two}
& \funto{V}{\univ \times \PL{\G}\phantom{\Opp}}{\two}
& \text{(sets of Proponent's positions)}
\\
  \funto{V}{\OP\G}{\two}
& \funto{V}{\univ \times \OL{\G}\phantom{\Prop}}{\two}
& \text{(sets of Opponent's positions)}
\end{array}
\]

\noindent
We often write $v  \in V$ or $V(v)$ for $V(v) \Eq \one$.

\item For a set of game positions $V$,
we write $\PP V$ and $\OP V$
for the $\Prop$ and $\Opp$ subsets of $V$ respectively.
This amounts to the following abbreviations:
\[
\begin{array}{r !{~~\deq~~} l}
  v \in \PP V
& v \in V ~~\land~~ v \in (\univ \times \PL\G) 
\\
  v \in \OP V
& v \in V ~~\land~~ v \in (\univ \times \OL\G)
\end{array}
\]
%so that, morally, $\PP V =  V\cap \PropPosSet$ and $\OP V = V \cap \OppPosSet$.
Intuitively,
$\PP V$ represents $V \cap (\univ \times \PL\G)$
while
$\OP V$ represents $V \cap (\univ \times \OL\G)$.

\item
In formulae we interpret quantifiers over (sets of) game positions as follows:
\[
\begin{array}{r !{~~\deq~~} l}
  (\exists v) \varphi
& (\exists x) (\exists \ell \In \PO\G)  \varphi[(x,\ell)/v]
\\
  (\exists V) \varphi
& (\exists \funto{V}{\G}{\two}) \varphi
\end{array}
\]
where, in the $\exists v$ case, we choose $x,\ell$ not free in $\varphi$.
\end{enumerate}
\end{nota}

\noindent
We now introduce the partial order $\gle$ on game positions.
%It underlies all edge relations of games.

%%%%%%%%%%%%%%%%%%%%%%%%%%%%%%%%%%%%%%%%%%%%%%%%%%%%%%%%%%%%%%%%%%%%%%%%%%%
\begin{defi}[Partial Order on Game Positions]
\label{def:games:order}
%%%%%%%%%%%%%%%%%%%%%%%%%%%%%%%%%%%%%%%%%%%%%%%%%%%%%%%%%%%%%%%%%%%%%%%%%%%
The relations
$\glt_{(\PL\G,\OL\G)}$, $\gle_{(\PL\G,\OL\G)}$ and $\gsucc_{(\PL\G,\OL\G)}$,
where $\PL\G$ and $\OL\G$ are $\HF$-variables,
are defined as follows:
\[
\begin{array}{r @{~~}c@{~~} l c l}
  (x,k)
& \glt_{(\PL\G,\OL\G)}
& (y,\ell)
& \deq
& x \Lt y ~\lor~ (x \Eq y ~\land~ k\In \PL\G ~\land~ \ell \In \OL\G)
\\
  u
& \gle_{(\PL\G,\OL\G)}
& v
& \deq
& u \glt_{(\PL\G,\OL\G)} v ~\lor~ u = v
\\
  \multicolumn{3}{r}{\gsucc_{(\PL\G,\OL\G)} (u,v)}
& \deq
& u \glt_{(\PL\G,\OL\G)} v
  ~\land~
  \lnot (\exists w) \left(u \glt_{(\PL\G,\OL\G)} w \glt_{(\PL\G,\OL\G)} v \right)
\end{array}
\]
%where $x,y$ are Individual variables and $k,\ell$ are $\HF$-variables. 

\noindent
When no ambiguity arises, we write
$\glt_{\G}$, $\gle_{\G}$ and $\gsucc_\G$,
or even $\glt$, $\gle$ and $\gsucc$
for
$\glt_{(\PL\G,\OL\G)}$, $\gle_{(\PL\G,\OL\G)}$ and $\gsucc_{(\PL\G,\OL\G)}$
respectively.
\end{defi}

Note that the formula $\gsucc(-,-)$ is actually bounded,
since by Notation~\ref{not:games:pos}.\eqref{eq:games:pos:indiv},
the variable $w$ ranges over game positions,
so that $(\exists w)$ stands for
$(\exists w \in \univ \times \PO{})$.

%%%%%%%%%%%%%%%%%%%%%%%%%%%%%%%%%%%%%%%%%%%%%%%%%%%%%%%%%%%%%%%%%%%%%%%%%%%%
%\begin{rem}
%%%%%%%%%%%%%%%%%%%%%%%%%%%%%%%%%%%%%%%%%%%%%%%%%%%%%%%%%%%%%%%%%%%%%%%%%%%%
%The relation $\glt$ has particular dag structure.
%First, note that if a position $u$ has more than one predecessor
%(\ie\@ there are at least two $v$'s such that $\gsucc(v,b)$),
%then $u$ must be a $\Prop$-position.
%\end{rem}

\noindent
We note a number of useful properties of $\glt$,
in particular that it is a discrete partial order.

%%%%%%%%%%%%%%%%%%%%%%%%%%%%%%%%%%%%%%%%%%%%%%%%%%%%%%%%%%%%%%%%%%%%%%%%%%%
\begin{prop} %[Discreteness]
\label{prop:games:gle}
%%%%%%%%%%%%%%%%%%%%%%%%%%%%%%%%%%%%%%%%%%%%%%%%%%%%%%%%%%%%%%%%%%%%%%%%%%%
$\FSOD$ proves following, under the assumption $\Labels(\PL{},\OL{})$.
\begin{enumerate}
\item
$u \glt v \glt w ~~\longlimp~~ u \glt w$

%\item
%$u \gle v \gle w ~\limp~ u \gle w$

\item
$\lnot(u \glt u)$

\item
$u \gle v \gle u ~~\longlimp~~ u = v$

\item
\label{eq:games:gle:decomp}
$u \glt v
  \quad\longliff\quad (\exists w \gle v) (\gsucc (u,w))
  \quad\longliff\quad (\exists w' \gge u) (\gsucc (w',v))$

\item
$(\forall k \in \PL{})\big( u \gle (\Root,k) ~~\longlimp~~ u = (\Root,k) \big)$
%\item
%$\lnot \exists u(u \glt v) ~\limp~ \forall u(v \gle u)$
\end{enumerate}
\end{prop}

\begin{fullproof}
\hfill\begin{enumerate}
\item 
Assume $(x,k) \glt (y,\ell) \glt (z,m)$.
Note that we must have $x \Leq y \Leq z$.
If $x \Lt z$, then $(x,k) \glt (z,m)$ and we are done.
Otherwise, $x \Eq z$ and the Tree Axioms (\S\ref{sec:ax:tree})
imply that $x \Eq y \Eq z$.
But in this case, by definition of $\glt$, we must have $\ell \in \PL{} \cap \OL{}$,
contradicting the disjointness of $\PL{}$ and $\OL{}$.

\item
If $(x,k) \glt (x,k)$,
since by the Tree Axioms $\lnot(x \Lt x)$,
then we must have $k \in \PL{} \cap \OL{}$,
a contradiction.

\item
$u \glt v \glt u$ would imply $u \glt u$, a contradiction.
So $u \gle v \gle u$ implies $u = v$.

%Assume $(x,k) \glt (y,\ell) \glt (x,k)$.
%By the Tree Axioms (\S\ref{sec:ax:tree}) we cannot have
%$x \Lt y \Lt x$, so that $x \Eq y$ and
%we must have $\ell \in \PL{} \cap  \OL{}$,
%again contradicting $\PL{} \cap \OL{} = \emptyset$.

\item
The implications
$\exists w' \gge u.\gsucc (w',v)
~\limp~
\exists w \gle v.\gsucc (u,w)
~\limp~
u \glt v$
are trivial.

As for the other implications,
first note that $(x,k) \glt (x,\ell)$ implies $\gsucc((x,k),(x,\ell))$.
Indeed, using the Tree Axioms (\S\ref{sec:ax:tree}),
$(x,k) \glt (y,m) \glt (x,\ell)$ implies $y \Eq x$, so that
$m \in \PL{} \cap \OL{}$, contradicting $\PL{} \cap \OL{} = \emptyset$.
Hence $\lnot \exists w((x,k) \glt w \glt (x,\ell))$ and
$\gsucc((x,k),(x,\ell))$.

\begin{itemize}
\item
We show
$u \glt v
~\limp~
\exists w \gge u.\gsucc (w,v)$.

Let $u = (x,k)$ and $v = (y,\ell)$. 
If $x \Eq y$ then we are done.
Otherwise, we must have $x \Lt y$.
If $\ell \in \OL{}$, then for any $m \in \PL{}$ we have
$u \glt (y,m) \glt (y,\ell)$, so that $\gsucc((y,m),(y,\ell))$.
Since $\PL{}$ is assumed to be non-empty,
we are done by taking $(y,m)$ for $w$.

It remains the case of $\ell \in \PL{}$.
It follows from the Induction Scheme of $\FSOD$
(\S\ref{sec:ax:ind}) that either $y \Eq \Root$
or there is some $z$ and some $d \in \Dir$ such that $y = \Succ_d(z)$.
Moreover, it follows from Proposition~\ref{prop:ax:tree}
that $x \Lt y$ implies $\lnot(y \Eq \Root)$.
Then the Tree Axioms give $x \Leq z$ from $x \Lt \Succ_d(z)$.
Let now $m \in \OL{}$. We thus have $(x,k) \gle (z,m) \glt (y,\ell)$.
Moreover, assuming $(z,m) \glt (z',m') \glt (y,\ell)$,
then $z' \Lt y$ implies $z' \Leq z$ hence $z' \Eq z$
and therefore $m \in \PL{} \cap \OL{}$, a contradiction.
Hence $z' \Eq y$, so that $\ell \in \OL{} \cap \PL{}$, again a contradiction.

\item The case of 
$u \glt v
~\limp~
\exists w \gle v.\gsucc (u,w)$
is dealt-with similarly, using the Tree Axioms (\S\ref{sec:ax:tree})
and Proposition~\ref{prop:ax:tree}.

Let $u = (x,k)$ and $v = (y,\ell)$.
First, if $x = y$ then we have $\gsucc(u,v)$ as shown above.
Assume now that $x \Lt y$ and that $k \in \PL{}$.
Consider any $m \in \OL{}$ (recall that $\OL{}$ is non-empty).
Then, again as above we have $\gsucc((x,k),(x,m))$
and we are done since $x \Lt y$ implies
$(x,m) \glt (y,\ell)$.
It remains the case of $k \in \OL{}$.
Since $x \Lt y$, by Proposition~\ref{prop:ax:tree}
there is some $d \in \Dir$ such that $\Succ_d(x) \gle y$.
Take any $m \in \PL{}$ (which is assumed to be non-empty).
Then we have $(x,k) \glt (\Succ_d(x),m) \glt (y,\ell)$.
Hence we are done as soon as we show $\gsucc((x,k),(\Succ_d(x),m))$.
Given $z$ and $n$ with
$(x,k) \glt (z,n) \glt (\Succ_d(x),m)$,
the Tree Axioms imply that either $x \Eq z$ or $z \Eq \Succ_d(x)$.
But $x \Eq z$ implies $k \in \OL{} \cap \PL{}$, a contradiction,
while $z \Eq \Succ_d(x)$ implies $m \in \PL{} \cap \OL{}$, also a contradiction.
Hence $\lnot\exists w[(x,k) \glt w \glt (\Succ_d(x),m)]$, as required.
\end{itemize}

\item
Assume $(x,\ell) \gle (\Root,k)$.
By definition of $\gle$ we have $x \Leq \Root$.
Then Proposition~\ref{prop:ax:tree} implies $x = \Root$,
so that $k \in \PL{}$ implies $\ell = k$.
\qedhere
\end{enumerate}
\end{fullproof}

%%%%%%%%%%%%%%%%%%%%%%%%%%%%%%%%%%%%%%%%%%%%%%%%%%%%%%%%%%%%%%%%%%%%%%%%%%%
\subsection{Induction and Recursion}
\label{sec:pos:ind}
%%%%%%%%%%%%%%%%%%%%%%%%%%%%%%%%%%%%%%%%%%%%%%%%%%%%%%%%%%%%%%%%%%%%%%%%%%%
We now present some basic results on induction and recursion \wrt\@
the partial order on game positions.
%We begin by induction and recursion \wrt\@ the partial order
%on game positions.
%Induction \wrt\@ edge relations $\edge{}{}$ of games easily follow.
%As for recursion, it turns out that we only need it \wrt\@ $\glt$.

\cnote{\CR:NOTES
\begin{itemize}
\item
We Recursion for
Proposition~\ref{prop:sim:sim:der}
and
Proposition~\ref{prop:sim:par:der}.
\end{itemize}}

We can show that $\glt$
%this order on game positions 
satisfies well-founded induction from the induction principle on the underlying tree.

%%%%%%%%%%%%%%%%%%%%%%%%%%%%%%%%%%%%%%%%%%%%%%%%%%%%%%%%%%%%%%%%%%%%%%%%%%%
\begin{thm}[$\glt$-Induction]
\label{thm:games:ind:pos}
%%%%%%%%%%%%%%%%%%%%%%%%%%%%%%%%%%%%%%%%%%%%%%%%%%%%%%%%%%%%%%%%%%%%%%%%%%%
$\FSOD$ proves the following, under the assumption $\Labels(\PL{},\OL{})$.
\[
(\forall V)
\Big(
(\forall v) \big[
(\forall u \glt v)(u \in V) ~~\longlimp~~  v \in V
\big]
~~\longlimp~~
(\forall v)(v \in V)
\Big)
\]
\end{thm}

\begin{proof}
Let $V$ be such that, for any game position $v$:
\begin{equation}
\label{eq:games:ind:pos}
(\forall u \glt v)~ 
\big( V(u) ~~\longlimp~~  V(v) \big)
\end{equation}
We show that
\[
(\forall x) (\forall y \Leq x) (\forall \ell \In \PO{})
\big( (y,\ell) \in V \big)
\]
by induction on $x$, whence the theorem will follow.
	
Suppose that $x = \Root$, and so $y = \Root$.
We first prove the statement for arbitrary $\ell \In \PL{}$;
in this case notice that there is no $u$ such that $u \glt (\Root,\ell)$,
and so we vacuously satisfy the LHS of~\eqref{eq:games:ind:pos} above.
Therefore we have that $(\Root,\ell) \in V$.
Otherwise $\ell \In \OL{}$ and every $u\glt (\Root, \ell)$ is of the form $(\Root, k)$
for some $k \In \PL{}$, and we have just shown that such $u$ must be contained in $V$.
Therefore we can conclude that $(\Root, \ell) \in V$,
again by~\eqref{eq:games:ind:pos}, as required.

Now we consider the inductive step, assuming the statement above is already true
for $x$ and considering the case of $\Succ_d x$.
If $y \Leq \Succ_d x$ then either $y \Leq x$ or $y \Eq \Succ_d x$.
In the former case we have by the inductive hypothesis that,
for any $\ell \In \PO{}$, $(y,\ell) \in V$.
So assume that $y \Eq \Succ_d x$.
Again we distinguish when $\ell \In \PL{}$
and when $\ell \In \OL{}$ in order to exhibit the
LHS of~\eqref{eq:games:ind:pos} above.
In the former case, notice that any $(z,k) \glt (y,\ell)$ is such that $z \Leq x$,
and so we have that $(z,k) \in V$ by the inductive hypothesis;
thus $(y,\ell) \in V$ by~\eqref{eq:games:ind:pos}.
In the latter case (when $\ell \In \OL{}$) we have for any $(z,k) \glt (y,\ell)$
either $z \Leq x$ or ($z \Eq \Succ_d x$ and $k\In \PL{}$).
In both cases we have seen that $(z,k)\in V$,
and so again we have that $(y,\ell)\in V$ by~\eqref{eq:games:ind:pos}. 
%
%Now we consider the inductive step, assuming the statement above is already true
%when $x = z$ and considering the case of $x = \Succ_d z$.
%If $y \Leq \Succ_d z$ then either $y \Leq z$ or $y \Eq \Succ_d z$.
%In the former case we have by the inductive hypothesis that,
%for any $l\In \PO{}$, $(y,l) \in V$.
%So assume that $y \Eq \Succ_d z$.
%Again we distinguish when $l \In \PL{}$ and when $l\In \OL{}$ in order to exhibit the
%LHS of~\eqref{eq:games:ind:pos} above.
%In the former case, notice that any $(x,k) \glt (y,l)$ is such that $x \leq z$,
%and so we have that $(x,k) \in V$ by the inductive hypothesis;
%thus $(y,l) \in V$ by~\eqref{eq:games:ind:pos}.
%In the latter case (when $l \In \OL{}$) we have that any $(x,k) \glt (y,l)$
%is such that $x \leq z$ or $x = \Succ_d z$ and $k\In \PL{}$.
%In both cases we have seen that $(x,k)\in V$,
%and so again we have that $(y,l)\in V$ by~\eqref{eq:games:ind:pos}. 
\end{proof}

Since $\glt$ is a partial order with induction,
comprehension (Theorem~\ref{thm:funto:ca})
gives a \emph{Recursion Theorem},
which allows us to
define a set of game positions $V$ by induction on game positions.
This requires the value of $V$ at a position $v$
to be determined by its values at positions $u \glt v$.
Thus, if the value of $V$ at $v$ is given by a formula $\varphi(V,v)$,
we assume that the following
formula %$\Rec(\varphi)$
holds
\[
\Rec(\varphi) ~~\deq~~
(\forall v)
(\forall V, V')
\Big[
(\forall w \glt v) \big(V w  ~\liff~ V' w \big)
~~\longlimp~~
\big(
\varphi(V,v) ~\liff~
\varphi(V',v)
\big)
\Big]
\]

\noindent
The Recursion Theorem says that, assuming $\Rec(\varphi)$,
the set of game positions $V$ given by
\[
%(\forall v)
%\bigg(
V v \quad\longliff\quad
(\forall U)
\Big[
(\forall u \gle v) \big(U u ~\liff~ \varphi(U,u) \big) 
~~\longlimp~~ U v
\Big]
%\bigg)
\]
is the unique set of game positions such that
\[
%(\forall v)\big(
V v \quad\longliff\quad \varphi(V,v)
%\big)
\]

%%%%%%%%%%%%%%%%%%%%%%%%%%%%%%%%%%%%%%%%%%%%%%%%%%%%%%%%%%%%%%%%%%%%%%%%%%%
\begin{prop}[Recursion Theorem]
\label{prop:games:rec}
%%%%%%%%%%%%%%%%%%%%%%%%%%%%%%%%%%%%%%%%%%%%%%%%%%%%%%%%%%%%%%%%%%%%%%%%%%%
$\FSOD$ proves that 
$\Labels(\PL{},\OL{}) \land \Rec(\varphi)$ implies
\[
\begin{array}{l r !{~~\longlimp~~} l}
\multicolumn{2}{r}{
(\forall v)
\bigg(
V v ~\longliff~
(\forall U)
\Big[
(\forall u \gle v) \big(U u ~\liff~ \varphi(U,u) \big) 
~~\limp~~ U v
\Big]
\bigg)
~~\longlimp}
&
(\forall v)\left(
V v ~\liff~ \varphi(V,v)
\right)
\\

\land
&
(\forall v)(V v ~\liff~ \varphi(V,v))
~~\longlimp~~
(\forall v)(U v ~\liff~ \varphi(U,v))
&
(\forall v)\left( V v ~\liff~ U v \right)
\end{array}
\]
\end{prop}

%%%%%%%%%%%%%%%%%%%%%%%%%%%%%%%%%%%%%%%%%%%%%%%%%%%%%%%%%%%%%%%%%%%%%%%%%%%
\begin{proof}
%%%%%%%%%%%%%%%%%%%%%%%%%%%%%%%%%%%%%%%%%%%%%%%%%%%%%%%%%%%%%%%%%%%%%%%%%%%
Consider a formula $\varphi(V,v)$ and assume $\Rec(\varphi)$
and $\Labels(\PL{},\OL{})$.
We begin with the second part of the statement, namely the uniqueness part.
Fix $V,U$.
By $\glt$-induction on $v$, we show that 
$\FSO$ proves the following formula $\psi(v) = \psi(V,U,v)$:
\[
(\forall u \gle v)(V u \liff \varphi(V,u))
~~\longlimp~~
(\forall u \gle v)(U u \liff \varphi(U,u))
~~\longlimp~~
(\forall u \gle v)(V u \liff U u)
\]

\noindent
Let $v$ and assume both premises of $\psi(v)$,
as well as $\psi(w)$ for all $w \glt v$.
The premises of $\psi(v)$ imply those of $\psi(w)$ for $w \glt v$,
so that we have $(V w \liff U w)$ for all $w \glt v$.
Hence, given $u \gle v$, if $u \glt v$ then we are done.
It thus remains to show $(V v \liff U v)$.
Thanks to the premises of $\psi(v)$, this amounts to showing
$\varphi(V,v) \liff \varphi(U,v)$,
which itself follows from $\Rec(\varphi)$,
since $(V w \liff U w)$ for all $w \glt v$.

We now turn to the first part of the statement.
Let $V$ such that
\[
V v \quad\longliff\quad
(\forall U)
\Big[
(\forall u \gle v) \big(U u ~\liff~ \varphi(U,u) \big) 
~~\longlimp~~ U v
\Big]
\]
By $\glt$-induction on $v$, we show that $\FSO$ proves the following formula
\[
\theta(v) \quad\deq\quad
(\forall u \gle v)
\underbrace{
\big(
V u
~~\longliff~~
\varphi(V,u)
\big)}_{\vartheta(u)}
\]

\noindent
So let $v$ and assume $\theta(w)$ for all $w \glt v$.
Given $u \gle v$, if $u \glt v$ then 
$\vartheta(u)$ follows from $\theta(u)$.
It thus remains to show $\vartheta(v)$.
We consider the two implications separately.
\begin{itemize}
\item
\emph{Case of $\varphi(V,v) \longlimp Vv$.}
Assume $\varphi(V,v)$.
By definition of $V$, we are done if we show
\[
(\forall U)
\Big[
(\forall u \gle v) \big(U u ~\liff~ \varphi(U,u) \big) 
~~\longlimp~~ U v
\Big]
\]

\noindent
Given $U$ such that
$\big(U u ~\liff~ \varphi(U,u) \big)$
for all $u \gle v$, we obtain $U v$ from $\varphi(U,v)$,
which itself follows $\varphi(V,v)$ and $\Rec(\varphi)$.
The premise of $\Rec(\varphi)$ follows from
$(\forall w \glt v)\psi(V,U,w)$,
whose premises are in turn given by resp.\@
$(\forall w \glt v)\vartheta(w)$ and the assumption on $U$.

\item
\emph{Case of $V v \longlimp \varphi(V,v)$.}
Assume $V v$.
By comprehension (Theorem~\ref{thm:funto:ca})
let $U$ such that
\[
U u
\quad\longliff\quad
\Big[
(u \glt v ~\land~ V u)
~\lor~
\big( u = v ~\land~ \varphi(V,v) \big)
\Big]
\]

\noindent
We obtain $\varphi(V,v)$ from $U v$,
which in turn by def.\@ of $V$ 
follows from
$(\forall u \gle v)\big(U u ~\liff~ \varphi(U,u) \big)$.
In order to show the latter,
note that by definition of $U$ we have
$(U u \liff V u)$ for all $u \glt v$.
Hence $\Rec(\varphi)$ gives
$\varphi(U,v) \liff \varphi(V,v)$
and we get
$(U v \liff \varphi(U,v))$ from the definition of $U$.
In the case of $u \glt v$, namely $(U u \liff \varphi(U,u))$,
we have $(\forall w \gle u)(U w \liff V w)$
so that $\Rec(\varphi)$ implies
$\varphi(U,u) \liff \varphi(V,u)$
and the result follows form $\vartheta(u)$.
\qedhere
\end{itemize}
\end{proof}

%%%%%%%%%%%%%%%%%%%%%%%%%%%%%%%%%%%%%%%%%%%%%%%%%%%%%%%%%%%%%%%%%%%%%%%%%%%
\subsection{Infinite Paths}
\label{sec:pos:path}
%%%%%%%%%%%%%%%%%%%%%%%%%%%%%%%%%%%%%%%%%%%%%%%%%%%%%%%%%%%%%%%%%%%%%%%%%%%
We develop here a notion of infinite paths
(\ie\@ unbounded linearly order sets) for the partial order $\gle$
on game positions.
This material will be useful in Section~\ref{sec:games:plays}
to handle properties of infinite plays in games
which intrinsically rely on the particular structure of
the relation $\gle$ on game positions.
A typical example is the Predecessor Lemma~\ref{lem:games:predplays}.

%%%%%%%%%%%%%%%%%%%%%%%%%%%%%%%%%%%%%%%%%%%%%%%%%%%%%%%%%%%%%%%%%%%%%%%%%%%
\begin{defi}[Game Paths]
\label{def:games:path}
%%%%%%%%%%%%%%%%%%%%%%%%%%%%%%%%%%%%%%%%%%%%%%%%%%%%%%%%%%%%%%%%%%%%%%%%%%%
Let $\PL{},\OL{}$ be $\HF$-variables.
Given a game position $u$ and a set of game positions $U$,
we say that $U$ is a \emph{path from $u$} if the following
formula $\Path(\PL{},\OL{},u,U)$ holds:
\[
\Path(\PL{},\OL{},u,U)
\quad\deq\quad
\left\{
\begin{array}{c l}
& u \in U
\\
  \land
& (\forall v \in U) (u \gle v)
\\
  \land
& (\forall v \in U)(\exists w \in U) (\gsucc(v,w))
\\
  \land
& (\forall v,w \in U)(w \glt v ~\lor~ v = w ~\lor~ v \glt w)
\end{array}
\right.
\]
We write $\Path(u,U)$ when $\PL{}$ and $\OL{}$ are clear from the context.
\end{defi}

As a preparation to the Predecessor Lemma~\ref{lem:games:predplays}
for Infinite Plays, we prove here the analogous property
for infinite paths.

%%%%%%%%%%%%%%%%%%%%%%%%%%%%%%%%%%%%%%%%%%%%%%%%%%%%%%%%%%%%%%%%%%%%%%%%%%%
\begin{lem}[Predecessor Lemma for Game Paths]
\label{lem:games:predpath}
%%%%%%%%%%%%%%%%%%%%%%%%%%%%%%%%%%%%%%%%%%%%%%%%%%%%%%%%%%%%%%%%%%%%%%%%%%%
$\FSOD$ proves the following.
Assuming that $\Labels(\PL{},\OL{})$ and
that $\Path(\PL{},\OL{},u_0,U)$ hold
for a game position $u_0$ and a set of game positions $U$,
we have
\[
(\forall v \in U)
\left[
  u_0 \glt v
  ~~\limp~~
  (\exists u \in U)  (\gsucc(u,v))
\right]
\]
\end{lem}

The proof of Lemma~\ref{lem:games:predpath}
relies on the following usual maximality principle 
for non-empty linearly-ordered bounded sets.

%%%%%%%%%%%%%%%%%%%%%%%%%%%%%%%%%%%%%%%%%%%%%%%%%%%%%%%%%%%%%%%%%%%%%%%%%%%
\begin{lem}
\label{lem:games:maxlinbounded}
%%%%%%%%%%%%%%%%%%%%%%%%%%%%%%%%%%%%%%%%%%%%%%%%%%%%%%%%%%%%%%%%%%%%%%%%%%%
$\FSOD$ proves the following, assuming $\Labels(\PL{},\OL{})$.
Given a set of game positions $V$, if
$V$
is bounded
(\ie\@ $(\exists u)(\forall v \in V)(v \glt u)$),
%(\ie\@ $(\forall v \in V)(v \glt b)$ for some $b \in \PO\G$),
non-empty
and linearly ordered,
then $V$ has a maximum element:
$(\exists u \in V) (\forall v \in V)(v \gle u)$.
\end{lem}

\begin{proof}
By $\glt$-induction, we prove the following property:
\begin{enumerate}[$(\star)$]
\item\label{item:game-maximal-element} For all $u$,
for all $V$,
if $V$ is non-empty, linearly ordered by $\glt$
and such that $\forall v \in V(v \gle u)$,
then $V$ has a $\glt$-maximal element.
\end{enumerate}

\noindent
Let $u$ and $V$ satisfy the assumptions from \ref{item:game-maximal-element} above,
and assume \ref{item:game-maximal-element} for all $c \glt u$.
First, if $u = v$ for some $v \in V$, then $u$ is indeed the maximal element
of $V$.
So we can assume $v \glt u$ for all $v \in V$.
%Moreover, if there is $b' \glt b$ such that $\forall v \in V(v \gle b')$,
%then we conclude by induction hypothesis.
%So we assume that $b$ is $\glt$-minimal such that $v \glt b$ for all $v \in V$.

By Comprehension for Product Types (Thm.~\ref{thm:funto:ca}),
let $U$ be the set of all 
$w$ such that $\gsucc(w,u)$
and such that $v \gle w$ for some $v \in V$.
For each $v \in V$, 
it follows from Proposition~\ref{prop:games:gle}.\eqref{eq:games:gle:decomp}
that there is some $w \in U$ such that $v \gle w$.
In particular, $U$ is non-empty since $V$ is non-empty.
%Since $V$ is non-empty, there is some $v \in V$ such that $v \glt b$,
%and it follows from
%Proposition~\ref{prop:games:gle}.\eqref{eq:games:gle:decomp}
%that there is some $w \gge v$ such that $\gsucc(w,b)$.
%So $U$ is non-empty.

We claim the following:
%%%%%%%%%%%%%%%%%%%%%%%%%%%%%%%%%%%%%%%%%%%%%%%%%%%%%%%%%%%%%%%%%%%%%%%%%%%
\begin{subclm}
\label{clm:games:maxlinbounded}
%%%%%%%%%%%%%%%%%%%%%%%%%%%%%%%%%%%%%%%%%%%%%%%%%%%%%%%%%%%%%%%%%%%%%%%%%%%
\[
(\forall w \in U)(\exists! \tilde w \in V) 
\underbrace{
%\left[
%\tilde w \glt b
%~~\land~~
(\forall v \in V)
\Big(
v \gle w ~~\limp~~ v \gle \tilde w
\Big)
%\right]
}_{\vartheta(w,\tilde w)}
\]
\end{subclm}

%%%%%%%%%%%%%%%%%%%%%%%%%%%%%%%%%%%%%%%%%%%%%%%%%%%%%%%%%%%%%%%%%%%%%%%%%%%
\begin{subproof}[Proof of Claim \thesubclm]
%%%%%%%%%%%%%%%%%%%%%%%%%%%%%%%%%%%%%%%%%%%%%%%%%%%%%%%%%%%%%%%%%%%%%%%%%%%
Let $w \in U$.
By Comprehension for Product Types (Thm.~\ref{thm:funto:ca}),
let $W$ be the set of all $v \in V$ such that $v \gle w$.
Note that $W$ is non-empty by definition of $U$.
It is inherits the property of being linearly ordered from $V$,
and by construction it is bounded by $w$ with $w \glt u$.
By induction hypothesis, $W$ has a maximal element, say $\tilde w$.
We indeed have $\tilde w \in V$
%$\tilde w \glt b$,
and $v \gle \tilde w$
for all $v \in V$ with $v \gle w$.
Since $\tilde w \gle w$, uniqueness follows from the antisymmetry of $\gle$.
\end{subproof}

The remainder of the argument relies on the particular structure of $\glt$.
Using Comprehension on $\HF$-Sets,
it follows from the definition of $\glt$ that there is some $x \in \univ$
and some $\HF$-Set $k$ such that 
$U$ is exactly the set of all $(x,\ell)$ with $\ell \in k$.
This observation allows us to show
%that there is some $w_0 \in U$ such that

%%%%%%%%%%%%%%%%%%%%%%%%%%%%%%%%%%%%%%%%%%%%%%%%%%%%%%%%%%%%%%%%%%%%%%%%%%%
\begin{subclm}
%%%%%%%%%%%%%%%%%%%%%%%%%%%%%%%%%%%%%%%%%%%%%%%%%%%%%%%%%%%%%%%%%%%%%%%%%%%
\[
(\exists \tilde w_m \in V)
\underbrace{
(\forall w \in U)
(\forall \tilde w \in V)
\big(
\vartheta(w,\tilde w) ~~\limp~~ \tilde w \gle \tilde w_m
\big)
}_{\varphi(\tilde w_m)}
\]
\end{subclm}

%%%%%%%%%%%%%%%%%%%%%%%%%%%%%%%%%%%%%%%%%%%%%%%%%%%%%%%%%%%%%%%%%%%%%%%%%%%
\begin{subproof}[Proof of Claim \thesubclm]
%%%%%%%%%%%%%%%%%%%%%%%%%%%%%%%%%%%%%%%%%%%%%%%%%%%%%%%%%%%%%%%%%%%%%%%%%%%
Write $\preceq$ for the well-order on $k$ given by Remark~\ref{rem:ax:hf:well-order-hf}.
By $\preceq$-Induction
(Remark~\ref{rem:hfchoice}.\eqref{eq:hfchoice:ind})
we show the following:
\[
(\forall \ell \in k)
(\exists m \in k)
\underbrace{
(\forall n \preceq \ell)
(\forall \tilde w_n, \tilde w_m \in V)
\big(
  \vartheta((x,n),\tilde w_n)
  ~~\limp~~
  \vartheta((x,m),\tilde w_m)
  ~~\limp~~
  \tilde w_n \gle \tilde w_m
\big)}_{\psi(\ell,m)}
\]

\noindent
Let $\ell \in k$ be such that the property holds for all $\ell' \prec \ell$.
If $\ell$ is $\preceq$-minimal, the result follows from the
existence of a unique $\tilde w$ such that $\vartheta((x,\ell),\tilde w)$.
Otherwise, let $\ell'$ be the $\preceq$-predecessor of $\ell$,
and let $m \in k$ such that $\psi(\ell',m)$ be given by induction hypothesis.
By Claim~\ref{clm:games:maxlinbounded}, let 
$\tilde w_\ell,\tilde w_m$ be the unique elements of $V$ such that
$\vartheta((x,\ell),\tilde w_\ell)$ and $\vartheta((x,m),\tilde w_m)$.
Since $V$ is linearly ordered, we have either that 
$\tilde w_m \gle \tilde w_\ell$ or that $\tilde w_m \gle \tilde w_\ell$.
In the former case, we take $\ell$ for the new $m$, and in the latter
we keep the same $m$.

Since $U$ is non-empty, there is a $\preceq$-maximal $\ell \in k$.
Let $m \in k$ such that $\psi(\ell,m)$,
and by Claim~\ref{clm:games:maxlinbounded},
let $\tilde w_m \in V$
such that $\vartheta((x,m),\tilde w_m)$.
By definition of $k$, we do have
$\tilde w \gle \tilde w_m$
for all $\tilde w \in V$ with $\vartheta(w,\tilde w)$
for some $w \in U$.
Hence we have that $\varphi(\tilde w_m)$.
\end{subproof}

Consider now $\tilde w_m \in V$ such that $\varphi(\tilde w_m)$.
As noted above, for all $v \in V$ there is some $w \in U$ such that $v \glt w$.
But we also have $v \gle \tilde w$ where $\tilde w$ is unique such that
$\vartheta(w,\tilde w)$.
It thus follows
%from $\varphi(\tilde w_m)$
that $v \gle \tilde w_m$ for all $v \in V$.

This concludes the proof of Lemma~\ref{lem:games:maxlinbounded}.
\end{proof}

We can now prove Lemma~\ref{lem:games:predpath}.

%%%%%%%%%%%%%%%%%%%%%%%%%%%%%%%%%%%%%%%%%%%%%%%%%%%%%%%%%%%%%%%%%%%%%%%%%%%
\begin{proof}[Proof of Lemma~\ref{lem:games:predpath}]
%%%%%%%%%%%%%%%%%%%%%%%%%%%%%%%%%%%%%%%%%%%%%%%%%%%%%%%%%%%%%%%%%%%%%%%%%%%
Fix $v \in U$ with $u_0 \glt v$.
By Comprehension for Product Types (Thm.~\ref{thm:funto:ca}),
let $W$ be the set of all $w \in U$ such that $w \glt v$.
Since $u_0 \glt v$ and $\Path(u_0,U)$, the set $W$
is non-empty, linearly ordered and bounded by~$v$.
By Lemma~\ref{lem:games:maxlinbounded}, it has a maximal element, say $w$.
We have $u_0 \gle w$ and $w \glt v$.
Moreover, by $\Path(u_0,U)$
there is some $\tilde w \in U$ such that $\gsucc(w,\tilde w)$.
Again by $\Path(u_0,U)$, we have
\[
\left(
  \tilde w \glt v \quad\lor\quad \tilde w = v \quad\lor\quad v \glt \tilde w
\right)
\]

\noindent
But $\tilde w \glt v$ implies $\tilde w \gle w$, a contradiction,
while $v \glt \tilde w$ implies $w \glt v \glt \tilde w$,
contradicting $\gsucc(w,\tilde w)$.
It thus follows that $\tilde w = v$
and we are done.
\end{proof}

%%% Local Variables:
%%% mode: latex
%%% TeX-master: "main.tex"
%%% End:

%%%%%%%%%%%%%%%%%%%%%%%%%%%%%%%%%%%%%%%%%%%%%%%%%%%%%%%%%%%%%%%%%%%%%%%%%%%
\section{Infinite Two-Player Games}
%\section{Basic Facts on Games}
\label{sec:games}
%%%%%%%%%%%%%%%%%%%%%%%%%%%%%%%%%%%%%%%%%%%%%%%%%%%%%%%%%%%%%%%%%%%%%%%%%%%

\noindent
This Section is devoted to definitions and basic properties relating to games,
building on~\S\ref{sec:pos}.
We will use these games in~\S\ref{sec:aut} and~\S\ref{sec:sim} to formalize
a basic theory of tree automata in $\FSO$.

Our games are played on bipartite dags (with partial order $\gle_\G$)
induced by labels of game positions $(\PL\G,\OL\G)$ in the sense of~\S\ref{sec:pos}.
%As indicated in the Introduction of~\S\ref{sec:pos},
Continuing~\S\ref{sec:pos},
\emph{Proponent} will play from positions of the form
\[
\PP\G \quad=\quad \univ \times \PL\G
\]
while \emph{Opponent} will play from positions of the form
\[
\OP\G \quad=\quad \univ \times \OL\G
\]

\noindent
A game will be given by specifying edge relations of the form
\[
(x,k) \edge{}{\Prop} (x,\ell)
\quad\text{and}\quad
(x,\ell) \edge{}{\Opp} (x.d,k)
\qquad\text{where $k \in \PL\G$, $\ell \in \OL\G$ and $d \in \Dir$,}
\]
so that, for $\player$ either $\Prop$ or $\Opp$,
\[
u \edge{}{\player} v
\qquad\text{implies}\qquad
u \glt v
\]
(actually even $\gsucc(u,v)$).
We insist on the fact that $\Prop$ can only move to a game position with the same
underlying tree position,
while $\Opp$ is forced to move to a game position with a
successor tree position.

We first give basic definitions and results on games
(\S\ref{sec:games:games})
and infinite plays
(\S\ref{sec:games:plays}).
Besides the above mentioned constraints on the shape of games,
these notions are standard.
Our notion of strategy is presented in~\S\ref{sec:strat}.
A crucial point here is that, \wrt\@ our games,
the monadic language imposes all strategies to be
by construction \emph{positional} in the usual sense
(see \eg~\cite{thomas97handbook}).
Finally,~\S\ref{sec:games:win} briefly discusses our setting
for \emph{winning} in games, and~\S\ref{sec:games:parity} presents in more detail
the important particular case of \emph{parity} conditions.
Parity conditions are one of the prominent formulations of winning conditions
for $\omega$-regular games.
%in particular because they are \emph{positionally} determined,
This is
in particular due to the fact that they are \emph{positionally} determined,
\ie\@ the winner of a parity game can always win with a \emph{positional}
winning strategy~\cite{ej91focs}
(see also~\cite{thomas97handbook,walukiewicz02tcs,pp04book}).
This is of crucial importance in our setting as all our strategies
are inherently positional, due to the underlying limits on expressiveness in the language of $\MSO$.
Finally, the Axiom $(\PosDet)$ of Positional Determinacy
of Parity Games is formulated in~\S\ref{sec:posdet}.

%%%%%%%%%%%%%%%%%%%%%%%%%%%%%%%%%%%%%%%%%%%%%%%%%%%%%%%%%%%%%%%%%%%%%%%%%%%
%\subsection{Games and Strategies}
\subsection{Games}
\label{sec:games:games}
%%%%%%%%%%%%%%%%%%%%%%%%%%%%%%%%%%%%%%%%%%%%%%%%%%%%%%%%%%%%%%%%%%%%%%%%%%%
%We now define our notion of games.
%As already indicated in the introduction of this Section, the idea is that 

A game $\G$ will be given by labels of game positions
$\PL\G$ and $\OL\G$ together with Functions
\[
\funto{\EP}{\PP\G}{\Pne(\OL\G)}
\qquad\text{and}\qquad
\funto{\EO}{\OP\G}{\Pne(\Dir \times \PL\G)}
\]
where $\Pne(-)$ is the $\HF$-Function of~\S\ref{sec:ax:hf}.\ref{item:ax:hf:po}.
Such Functions $\EP,\EO$ induce edge relations
$\edge{}{\PL\G}$ and $\edge{}{\OL\G}$
given by
\[
\begin{array}{l !{\quad\text{iff}\quad} l}
  (x,k) \edge{}{\PL\G} (x,\ell)
& \ell \in \EP(x,k)
\\
  (x,\ell) \edge{}{\OL\G} (x.d,k)
& (d,k) \in \EO(x,\ell)
\end{array}
\]

\noindent
We make this formal in the following definition.

%%%%%%%%%%%%%%%%%%%%%%%%%%%%%%%%%%%%%%%%%%%%%%%%%%%%%%%%%%%%%%%%%%%%%%%%%%%
\begin{defi}[Games and Edge Relations]
\label{def:games:games}
%%%%%%%%%%%%%%%%%%%%%%%%%%%%%%%%%%%%%%%%%%%%%%%%%%%%%%%%%%%%%%%%%%%%%%%%%%%
\hfill
\begin{enumerate}
\item
A \emph{game} $\G$ 
is given by $\HF$-terms $\PL\G,\OL\G$ and Functions $\EGP\G,\EGO\G$
which satisfy the following formula
%%has the form $(\PL\G,\OL\G,\EGP\G,\EGO\G)$
%where $\PL\G,\OL\G$ are $\HF$-terms and where
%$\EGP\G,\EGO\G$ are Functions, 
%and which satisfy the following formula
\[
\Game(\PL\G,\OL\G,\EGP\G,\EGO\G) \quad\deq\quad
\left\{
\begin{array}{c l}
& \Labels(\PL\G,\OL\G)
\\
  \land
& \funto{\EGP\G}{\PP\G}{\Pne(\OL\G)}
\\
  \land
& \funto{\EGO\G}{\OP\G}{\Pne(\Dir \times \PL\G)}
\end{array}
\right.
\]

\noindent
We often write $\Game(\G)$ for
$\Game(\PL\G,\OL\G,\EGP\G,\EGO\G)$.
Moreover, when no ambiguity arises, we abbreviate
$\G = (\PL\G,\OL\G,\EGP\G,\EGO\G)$ as
$\G = (\PL\G,\OL\G,\EP,\EO)$
or even 
$\G = (\PL\G,\OL\G,\e)$
or
$\G = (\Prop,\Opp,\e)$.
%Moreover, whenever possible we omit $\PM,\OM$ 
%and write $\G = e$ for $\G = (\PM,\OM)$.

\item
The \emph{edge relations} induced by
$\G = (\Prop,\Opp,\EP,\EO)$
%$\G = (\PL\G,\OL\G,\e)$
are defined as follows:
\[
\begin{array}{c !{\quad\deq\quad} l}
  (x,k) \edge{}{\PL\G} (y,\ell)
& k \In \Prop ~\land~ x \Eq y ~\land~ \ell \In \EP(x,k)
\\
  (x,\ell) \edge{}{\OL\G} (y,k)
& \ell \In \Opp ~\land~
  \bigdisj_{d \in \Dir} 
  \big(y \Eq \Succ_d(x) ~\land~ (d,k) \In \EO(x,\ell)\big)
\\
  u \edge{}{\G} v
& u \edge{}{\PL\G} v ~\lor~ u \edge{}{\OL\G} v
\end{array}
\]

\noindent
When no ambiguity arises, we write
$\edge{}{\Prop}$, 
$\edge{}{\Opp}$
and
$\edge{}{}$, 
for
$\edge{}{\PL\G}$, 
$\edge{}{\OL\G}$
and
$\edge{}{\G}$.
\end{enumerate}
\end{defi}

Note that $\Game(\G)$
implies that the edge relation $\edge{}{}$ has no dead ends,
\ie\@ that from any position, a move can always be made by one of the players.
It follows that
the edge relation $\edge{}{}$ induces
an unbounded partial order.
(Note that it already follows from the structure of $\edge{}{}$ that 
it induces a partial order.)

%%%%%%%%%%%%%%%%%%%%%%%%%%%%%%%%%%%%%%%%%%%%%%%%%%%%%%%%%%%%%%%%%%%%%%%%%%%
\begin{lem}
%%%%%%%%%%%%%%%%%%%%%%%%%%%%%%%%%%%%%%%%%%%%%%%%%%%%%%%%%%%%%%%%%%%%%%%%%%%
$\FSOD$ proves
\[
\Game(\G)
\quad\longlimp\quad
(\forall u) (\exists v) \left(u \edge{}{} v\right)
\]
\end{lem}

Games are equipped with a natural notion of subgame.
In this paper we will use subgames to ease some reasoning on automata
(in particular in~\S\ref{sec:sim}),
and also to more easily define certain strategies that are more naturally seen as concepts at the game level
(see~\S\ref{sec:strat}).
We only need the following weak notion of subgame.

%%%%%%%%%%%%%%%%%%%%%%%%%%%%%%%%%%%%%%%%%%%%%%%%%%%%%%%%%%%%%%%%%%%%%%%%%%%
\begin{defi}[Subgame]
\label{def:games:sub}
%%%%%%%%%%%%%%%%%%%%%%%%%%%%%%%%%%%%%%%%%%%%%%%%%%%%%%%%%%%%%%%%%%%%%%%%%%%
We say that $\G'$ is a \emph{subgame} of $\G$ whenever the following
formula holds
\[
\Sub(\G',\G)
\quad\deq\quad
\PL{\G'} \Eq \PL{\G}
~~\land~~
\OL{\G'} \Eq \OL{\G}
~~\land~~
(\forall u , v) \left(
  u \edge{}{\G'} v ~~\limp~~ u \edge{}{\G} v
\right)
\]
\end{defi}

%%%%%%%%%%%%%%%%%%%%%%%%%%%%%%%%%%%%%%%%%%%%%%%%%%%%%%%%%%%%%%%%%%%%%%%%%%%
\begin{rem}
\label{rem:games:sub}
%%%%%%%%%%%%%%%%%%%%%%%%%%%%%%%%%%%%%%%%%%%%%%%%%%%%%%%%%%%%%%%%%%%%%%%%%%%
Let $\G = (\PL\G,\OL\G,\EGP\G,\EGO\G)$ with $\Game(\G)$.
Then we have $\Sub(\G,\G(\gle))$,
where $\G(\gle)$
stands for the game
\[
(\PL\G,\OL\G,\EO,\EP)
\]
in which by $\HF$-Bounded Choice we let
\[
\EP(x,k) \deq \OL\G
\qquad\text{and}\qquad
\EO(x,\ell) \deq (\Dir \times \PL\G)
\]
%for all $(x,k) \in \univ \times \PO\G$.
Note that the edge relation of $\G(\gle)$
is precisely the relation $\gle_{(\PL\G,\OL\G)}$
of Definition~\ref{def:games:order},
hence the notation.
\end{rem}

%%%%%%%%%%%%%%%%%%%%%%%%%%%%%%%%%%%%%%%%%%%%%%%%%%%%%%%%%%%%%%%%%%%%%%%%%%%
\renewcommand\fntext{It is well known (see \eg~\cite[Chap.\@ 4]{libkin04book})
that transitive closure in graphs
is not expressible in first-order logic over the edge relation.}
%%%%%%%%%%%%%%%%%%%%%%%%%%%%%%%%%%%%%%%%%%%%%%%%%%%%%%%%%%%%%%%%%%%%%%%%%%%
The edge relation $\edge{}{}$ of a game $\G$ only specifies the \emph{moves} of $\G$.
In order to manipulate plays (\ie\@ sequences of moves)
we define the reflexive-transitive closure $\edge{*}{}$
and the transitive closure $\edge{+}{}$ of $\edge{}{}$.
As expected, these are second-order notions.\fn

%%%%%%%%%%%%%%%%%%%%%%%%%%%%%%%%%%%%%%%%%%%%%%%%%%%%%%%%%%%%%%%%%%%%%%%%%%%
\begin{defi}
%%%%%%%%%%%%%%%%%%%%%%%%%%%%%%%%%%%%%%%%%%%%%%%%%%%%%%%%%%%%%%%%%%%%%%%%%%%
Let $\G = (\Prop,\Opp,\EP,\EO)$
where $\Prop,\Opp$ are $\HF$-variables and $\EP,\EO$ are Function variables.
We define the following formulae.
\[
\begin{array}{r !{\quad\deq\quad} l !{\qquad} l}
  \DC_\G (V)
%& \forall v \in V .\  \forall u \edge{}{e} v . u \in V
& (\forall v \in V) (\forall u)
  \left(u \edge{}{\G} v ~~\limp~~ u \in V \right)
& \text{\emph{($V$ is downward-closed)}}
\\
  u \edge{*}{\G} v
& (\forall V)
  \big(\DC_\G (V)
  ~~\limp~~
  v \in V
  ~~\limp~~
  u \in V \big)
\\
  u \edge{+}{\G} v
& u \edge{*}{\G} v ~\land~ \lnot (u = v)
\end{array}
\]

\noindent
Whenever possible, we write %$\DC$,
$\edge{*}{}$ and $\edge{+}{}$
for %$\DC_\G$,
$\edge{*}{\G}$ and $\edge{+}{\G}$.
\end{defi}

%Note that $\edge{*}{}$ is a second-order notion.

The relations $\edge{*}{}$ and $\edge{+}{}$
satisfy properties analogous to those of Proposition~\ref{prop:games:gle}:

%%%%%%%%%%%%%%%%%%%%%%%%%%%%%%%%%%%%%%%%%%%%%%%%%%%%%%%%%%%%%%%%%%%%%%%%%%%
\begin{prop}[Properties of Edge Relations]
\label{prop:games:edges}
%%%%%%%%%%%%%%%%%%%%%%%%%%%%%%%%%%%%%%%%%%%%%%%%%%%%%%%%%%%%%%%%%%%%%%%%%%%
$\FSOD$ proves the following, under the assumption $\Game(\G)$.
\begin{enumerate}
\item 
\label{eq:games:edges:gsucc}
$u \edge{}{} v ~~\limp~~ \gsucc (u,v)$

\item
\label{eq:games:edges:irrasym}
$\edge{}{}$ is irreflexive and asymmetric.

\item
\label{eq:games:edges:reftrans}
$\edge{*}{}$
is reflexive and transitive.

\item
\label{eq:games:edges:decomp}
%$u \edge{*}{e} v \liff u = v \lor \exists w .~ u \edge{*}{e} w \edge{}{e} v$
$u \edge{*}{} v \quad\liff\quad 
  u = v \lor (\exists w)\left( u \edge{*}{} w \edge{}{} v \right)
\quad\liff\quad
  u = v \lor (\exists w) \left( u \edge{}{} w \edge{*}{} v \right)$

\item
$u \edge{*}{} v ~~\limp~~ u \gle v$

\item
\label{eq:games:edges:antisym}
$\edge{*}{}$ is antisymmetric.

\item
\label{eq:games:edges:glt}
$u \edge{+}{} v ~~\limp~~ u \glt v$

\item
$\edge{+}{}$ is irreflexive and transitive.

\item
\label{eq:games:edges:root}
$(\forall k \in \Prop) \left( u \edge{*}{} (\Root,k) ~~\limp~~ u = (\Root,k) \right)$
\end{enumerate}
\end{prop}

\begin{fullproof}
\begin{enumerate}
\item
Assume $u \edge{}{} v$ and let $u = (x,k)$ and $v = (y,\ell)$.
If $k \in \Prop$, then we must have $y \Eq x$ and $\ell \in \Opp$, so that $\gsucc(u,v)$.
Otherwise, we must have $u \edge{}{\Opp} v$ which implies $y = \Succ_d(x)$
for some $d \in \Dir$.
Moreover, reasoning as in the proof of Proposition~\ref{prop:games:gle},
we get $\gsucc(u,v)$ from the facts that $k \in \Opp$ and $\ell \in \Prop$.

\item
Inherited from the same properties for $\gsucc$.

\item
Reflexivity follows from reflexivity of implication
and transitivity follows from transitivity of implication.

\item
Assume first
$u \edge{*}{} w \edge{}{} v$.
Given $\DC(V)$ such that $v \in V$, we must have $w \in V$.
But $u \edge{*}{} w$ implies $u \in V$.
Hence $u \edge{*}{} v$.
Similarly, if
$u \edge{}{} w \edge{*}{} v$,
given $\DC(V)$ such that $v \in V$,
we have $w \in V$ by definition of $\edge{*}{}$
and we get $u \in V$ since $\DC(V)$.

Assume conversely that $u \edge{*}{} v$ with $u \neq v$.
We first show that $u \edge{*}{} w \edge{}{} v$ for some $w \in \PosSet$.
By Comprehension for Product Types (Theorem~\ref{thm:funto:ca}),
let $W$ be the set of all $\tilde w$
such that $\tilde w \edge{*}{} w$ for some $w \edge{}{} v$,
and let $V$ be the union of $W$ with $\{v\}$.
We claim that $V$ is downward closed.
Indeed, assume given $w' \in V$ and $w'' \edge{}{} w'$.
If $w' = v$, then (by reflexivity of $\edge{*}{}$),
$w'' \edge{}{} v$ implies $w'' \in W \sle V$.
Otherwise, we must have $w' \in W$, so that $w' \edge{*}{} \tilde w$
for some $\tilde w \edge{}{} w$.
But by transitivity of $\edge{*}{}$, we have $w'' \edge{*}{} \tilde w$,
so that $w'' \in W$.
Since $v \in V$ and $V$ is downward closed, 
$u \edge{*}{} v$ implies $u \in V$, and $u \neq v$
implies $u \in W$, so that $u \edge{*}{} w$ for some $w \edge{}{} v$.

We reason similarly in order to show
$u \edge{}{} w \edge{*}{} v$ for some $w \in \PosSet$.
Again by Comprehension, let $W$ be the set of all $\tilde w$
such that $\tilde w \edge{}{} w$ for some $w \edge{*}{} v$,
and let $V = W \cup \{v\}$.
Again, $V$ is downward closed, since
$w \edge{}{} v$ implies $w \edge{}{} v \edge{*}{}v$,
and given $w' \edge{}{} \tilde w \in W$ with
$\tilde w \edge{}{} w \edge{*}{} v$
we have
$w' \edge{}{} \tilde w \edge{*}{} v$.
Now, $u \edge{*}{} v$ implies $u \in V$, but since $u \neq v$,
we have $u \in W$ and we are done.

\item
For the right-to-left direction, assume $u \edge{*}{} w \edge{}{} v$.
Given $\DC(V)$ such that $v \in V$, we must have $w \in V$.
But $u \edge{*}{} w$ implies $u \in V$. Hence $u \edge{*}{} v$.

For the left-to-right direction, assume $u \edge{*}{} v$ with $\lnot(u = v)$.
By Comprehension for Product Types (Theorem~\ref{thm:funto:ca}),
let $W$ be the set of all $\tilde w$
such that $\tilde w \edge{*}{} w$ for some $w \edge{}{} v$,
and let $V$ be the union of $W$ with $\{v\}$.
We claim that $V$ is downward closed.
Indeed, assume given $w' \in V$ and $w'' \edge{}{} w'$.
If $w' = v$, then (by reflexivity of $\edge{*}{}$),
$w'' \edge{}{} v$ implies $w'' \in W \sle V$.
Otherwise, we must have $w' \in W$, so that $w' \edge{*}{} \tilde w$
for some $\tilde w \edge{}{} w$.
But by transitivity of $\edge{*}{}$, we have $w'' \edge{*}{} \tilde w$,
so that $w'' \in W$.

Since $v \in V$ and $V$ is downward closed, 
$u \edge{*}{} v$ implies $u \in V$, and $\lnot(u = v)$
implies $u \in W$, so that $u \edge{*}{} w$ for some $w \edge{}{} v$.

\item
Assume given $u$.
By $\glt$-induction, we show that $\forall v(u \edge{*}{} v \limp u \gle v)$.
So let $v$ such that the property holds for all $w \glt v$,
and assume $u \edge{*}{} v$.
If $u = v$ then we are done.
Otherwise, there is some $w \edge{}{} v$ such that $u \edge{*}{} w$.
But we have seen that $w \edge{}{} v$ implies $w \glt v$,
so that the induction hypothesis gives $u \gle w$,
and we are done by transitivity of $\gle$.

\item
Inherited from the same property for $\gle$.

\item
If $u \edge{+}{} v$ then $u \edge{*}{} v$ and $u \neq v$,
so that $u \gle v \land u \neq v$, which implies $u \glt v$
by definition of $\gle$.

\item
Inherited from the same properties for $\glt$.

\item
Inherited from the same property for $\gle$.
\qedhere
\end{enumerate}
\end{fullproof}

Induction for games (\ie\@ \wrt\@ edge relations)
is an immediate corollary to Theorem~\ref{thm:games:ind:pos}
and Proposition~\ref{prop:games:edges}.

%%%%%%%%%%%%%%%%%%%%%%%%%%%%%%%%%%%%%%%%%%%%%%%%%%%%%%%%%%%%%%%%%%%%%%%%%%%
\begin{cor}[Game Induction]
\label{cor:games:edges:ind}
%%%%%%%%%%%%%%%%%%%%%%%%%%%%%%%%%%%%%%%%%%%%%%%%%%%%%%%%%%%%%%%%%%%%%%%%%%%
$\FSOD$ proves the following,
under the assumption $\Game(\G)$.
\[
(\forall V)
\left(
(\forall v)
\left[
\left(\forall u \edge{+}{} v \right)(u \in V)
~~\longlimp~~
v \in V
\right]
~~\longlimp~~
(\forall v) \big( v \in V \big)
\right)
\]
\end{cor}

%%%%%%%%%%%%%%%%%%%%%%%%%%%%%%%%%%%%%%%%%%%%%%%%%%%%%%%%%%%%%%%%%%%%%%%%%%%
\subsection{Infinite Plays}
\label{sec:games:plays}
%%%%%%%%%%%%%%%%%%%%%%%%%%%%%%%%%%%%%%%%%%%%%%%%%%%%%%%%%%%%%%%%%%%%%%%%%%%
We now define our notion of infinite play.
They are sets of game positions which are unbounded and
linearly ordered \wrt\@ $\edge{}{}$.
Infinite plays will allow us to define winning in games (\S\ref{sec:games:win})
and thus acceptance for tree automata (\S\ref{sec:aut}).
Furthermore, we prove a number of basic properties on infinite plays
on which we rely for the formalization of usual operations
on tree automata.

In the following, given $\G=(\Prop,\Opp,\e)$,
we write $\Path(\G,u,U)$ for $\Path(\Prop,\Opp,u,U)$,
where $\Path$ is as in Definition~\ref{def:games:path}.

%%%%%%%%%%%%%%%%%%%%%%%%%%%%%%%%%%%%%%%%%%%%%%%%%%%%%%%%%%%%%%%%%%%%%%%%%%%
\begin{defi}[Infinite Plays]
%%%%%%%%%%%%%%%%%%%%%%%%%%%%%%%%%%%%%%%%%%%%%%%%%%%%%%%%%%%%%%%%%%%%%%%%%%%
Let $\G = (\Prop,\Opp,\EP,\EO)$
where $\Prop,\Opp$ are $\HF$-variables and $\EP,\EO$ are Function variables.
Given a position $u$ and a set of game positions $U$,
we say that $U$ is an \emph{infinite play in $\G$ from $u$} when the following
formula $\Play(\G,u,U)$ holds:
\[
\Play(\G,u,U)
\quad\deq\quad
\left\{
\begin{array}{c l}
& %U \sle V ~~\land~~
  (u \in U)
%~~\land~~
\\
  \land
& (\forall v \in U) \big( u \edge{*}{\G} v \big)
\\
  \land
& (\forall v \in U) (\exists w \in U) \big(v \edge{}{\G} w \big)
\\
  \land
& (\forall v,w \in U) \big(v \edge{+}{\G} w ~\lor~ v = w ~\lor~ w \edge{+}{\G} v \big)
\end{array}
\right.
\]
\end{defi}

\noindent
Note that $\Play(\G,u,U)$ is literally just the formula $\Path(\G,u,U)$
in which $\edge{*}{\G}$ replaces $\gle$,
$\edge{}{\G}$ replaces $\gsucc(-,-)$ and
$\edge{+}{\G}$ replaces $\glt$.
It follows from Proposition~\ref{prop:games:edges}
that $\Play(\G,u,U)$ implies $\Path(\G,u,U)$.
In other words, an infinite play in $\G = (\Prop,\Opp,\e)$
is simply an infinite path of the underlying
partial order $\gle_{(\Prop,\Opp)}$
which respects the transitions of $\G$ induced by $\e$.
Also, if $\G'$ is a subgame of $\G$,
then $\Play(\G',u,U)$ implies $\Play(\G,u,U)$.

We now gather some basic properties on infinite plays.
The first one will help to show that a set of game positions
is linearly ordered.

%%%%%%%%%%%%%%%%%%%%%%%%%%%%%%%%%%%%%%%%%%%%%%%%%%%%%%%%%%%%%%%%%%%%%%%%%%%
\begin{prop}
\label{prop:games:edges:lin}
%%%%%%%%%%%%%%%%%%%%%%%%%%%%%%%%%%%%%%%%%%%%%%%%%%%%%%%%%%%%%%%%%%%%%%%%%%%
$\FSOD$ proves the following, assuming $\Game(\G)$.
Let $V$ and $u_0 \in V$ be such that
\[
\left\{
\begin{array}{c l}
& (\forall v \in V)(u_0 \edge{*}{} v)
\\
  \land
& (\forall u \in V)(\exists! v \in V)(u \edge{}{} v)
\\
  \land
& (\forall v \in V) \left[
  v \neq u_0 ~~\limp~~ (\exists u \in V)(u \edge{}{} v)
  \right]
\end{array}
\right.
\]

\noindent
Then
\[
 (\forall v,w \in V)\left(
  v \edge{+}{} w \quad\lor\quad v = w \quad\lor\quad w \edge{+}{} v
\right)
\]
\end{prop}

\begin{proof}
First, 
it follows from Proposition~\ref{prop:games:edges}.\eqref{eq:games:edges:antisym}
that $u_0$ is unique such that 
$(\forall v \in V)(u_0 \edge{*}{} v)$.
By induction on the edge relation $\edge{+}{}$
(cf.~Corollary~\ref{cor:games:edges:ind}) we show
\[
(\forall u \in V)
\underbrace{(\forall v \in V)
\left(
u \edge{+}{} v ~~\lor~~ u = v ~~\lor~~ v \edge{+}{} u
\right)}_{\theta(u)}
\]

\noindent
Let $u \in V$, and assume that $\theta(v)$ holds for all $v \in V$
such that $v \edge{+}{} u$.
If $u = u_0$ then we are done since
%$\forall v \in U(u_0 \edge{*}{} v)$.
$u_0 \edge{*}{} v$ for all $v \in V$.
Otherwise, by assumption there is $v \in V$
with $v \edge{}{} u$, and moreover such that 
%$u$ is unique such that $(u \in U) \land (v \edge{}{} u)$.
$u$ is the unique $\edge{}{}$-successor of $v$ in $U$.

Note that $v \edge{}{} u$ implies 
$v \edge{+}{} u$
(Proposition~\ref{prop:games:edges},
\eqref{eq:games:edges:irrasym} \& \eqref{eq:games:edges:decomp}),
so that $\theta(v)$ follows from the induction hypothesis.
Given $w \in V$, if $w \edge{*}{} v$ then we get $w \edge{*}{} u$ and we are done.
Otherwise, since $\theta(v)$ implies $v \edge{+}{} w$,
we may appeal to the following.

%%%%%%%%%%%%%%%%%%%%%%%%%%%%%%%%%%%%%%%%%%%%%%%%%%%%%%%%%%%%%%%%%%%%%%%%%%%
\begin{subclm}
%%%%%%%%%%%%%%%%%%%%%%%%%%%%%%%%%%%%%%%%%%%%%%%%%%%%%%%%%%%%%%%%%%%%%%%%%%%
\label{clm:games:edges:lin}
\[
(\forall w \in V)\left(
  v \edge{+}{} w ~~\limp~~ u \edge{*}{} w
\right)
\]
\end{subclm}

\begin{subproof}[Proof of Claim \thesubclm]
We reason by induction on $\edge{+}{}$.
So let $w \in V$ with $v \edge{+}{} w$
and such that
\[
(\forall w' \in V)
\left(w' \edge{+}{} w ~~\limp~~ v \edge{+}{} w' ~~\limp~~ u \edge{*}{} w' \right)
\]

\noindent
Since %$w' \in \tilde U$ and $w \edge{+}{e} w'$,
$u_0 \edge{*}{} v \edge{+}{} w$
we have $w \neq u_0$ by 
Proposition~\ref{prop:games:edges}.\eqref{eq:games:edges:antisym},
so that there is $w' \in V$ with $w' \edge{}{} w$.
If $v \edge{+}{} w'$ then the induction hypothesis implies $u \edge{*}{} w'$,
so that $u \edge{+}{} w$ and we are done.
Otherwise $\theta(v)$ implies $w' \edge{*}{} v$.
Assume for contradiction that $w' \edge{+}{} v$.
We thus have 
\[
w' \edge{+}{} v \edge{+}{} w
\]

\noindent
Proposition~\ref{prop:games:edges}.\eqref{eq:games:edges:glt}
then gives 
$w' \glt v \glt w$.
But this contradicts
$w' \edge{}{} w$
since the latter implies
$\gsucc(w',w)$
by
Proposition~\ref{prop:games:edges}.\eqref{eq:games:edges:gsucc}.
Hence $w'=v$. But then $v = w' \edge{}{} w \in V$ and, since $u$
is the unique $\edge{}{}$-successor of $v$ in $V$,
we have $u = w$, as required.
\end{subproof}

This concludes the proof of Proposition~\ref{prop:games:edges:lin}.
\end{proof}

Proposition~\ref{prop:games:edges:lin} is a useful tool to prove
that given sets of game positions are infinite plays.
Some constructions on automata
(see~\S\ref{sec:aut}, \S\ref{sec:sim})
furthermore require us to build plays
in one game from plays in another game.
To this end, we note here the following property,
which we informally see as a partial converse to
Proposition~\ref{prop:games:edges:lin}.

%%%%%%%%%%%%%%%%%%%%%%%%%%%%%%%%%%%%%%%%%%%%%%%%%%%%%%%%%%%%%%%%%%%%%%%%%%%
\begin{lem}[Predecessor Lemma for Infinite Plays]
\label{lem:games:predplays}
%%%%%%%%%%%%%%%%%%%%%%%%%%%%%%%%%%%%%%%%%%%%%%%%%%%%%%%%%%%%%%%%%%%%%%%%%%%
$\FSOD$ proves the following.
Assuming $\Game(\G)$ and $\Play(\G,u_0,U)$,
we have
\[
(\forall v \in U)\left[
  u_0 \edge{+}{} v
  ~~\longlimp~~
  (\exists u \in U)
  \left( u \edge{}{} v \right)
\right]
\]
\end{lem}

\begin{proof}
First, it follows from Proposition~\ref{prop:games:edges}
that $\Play(\G,u_0,U)$ implies $\Path(\G,u_0,U)$.
We invoke the Predecessor Lemma~\ref{lem:games:predpath} for Game Paths.
Assuming $u_0 \edge{+}{} v$, Proposition~\ref{prop:games:edges}
implies $u_0 \glt v$, 
so there is $u \in U$ such that $\gsucc(u,v)$.
Since $U$ is an infinite play, $u \in U$ has an $\edge{}{}$-successor in $U$,
\ie\@ there is some $u' \in U$ such that $u \edge{}{} u'$.
Again since $U$ is an infinite play, we have
\[
\left(
  v \edge{+}{} u' \quad\lor\quad u' = u \quad\lor\quad u' \edge{+}{} v
\right)
\]

\noindent
But by Proposition~\ref{prop:games:edges} again,
$v \edge{+}{} u'$ 
implies $u \glt v \glt u'$, contradicting $\gsucc(u,u')$,
while $u' \edge{+}{} v$ implies $u \glt u' \glt v$,
contradicting $\gsucc(u,v)$.
Hence $u' = v$ and we are done.
\end{proof}

Next, we show that games have infinite plays from any position, relying on Remark~\ref{rem:ax:hf:well-order-hf}.

%%%%%%%%%%%%%%%%%%%%%%%%%%%%%%%%%%%%%%%%%%%%%%%%%%%%%%%%%%%%%%%%%%%%%%%%%%%
\begin{lem}
\label{lem:games:infplay:pos}
%%%%%%%%%%%%%%%%%%%%%%%%%%%%%%%%%%%%%%%%%%%%%%%%%%%%%%%%%%%%%%%%%%%%%%%%%%%
$\FSOD$ proves that $\Game(\G)$
implies
\[
(\forall v) (\exists U)
\Play(\G,v,U)
\]
\end{lem}

\begin{proof}
Let $\G = (\Prop,\Opp,\EP,\EO)$.
Fix $v \in V$.
Using Remark~\ref{rem:ax:hf:well-order-hf},
let $\preceq$ be a well-order on $\Prop \cup \Opp$.
We extend the relation $\preceq$ to $\univ \times (\Prop \cup \Opp)$ by setting:
\[
(x,k) \prec (y,\ell)
\quad\deq\quad
\left\{
\begin{array}{l}
  (x \Eq y ~\land~ k \preceq \ell) \quad\lor
\\
  (\exists z) \bigdisj_{d < d' \in \Dir}\left(
      x \Eq \Succ_d(z) \land y \Eq \Succ_{d'}(z)
  \right)
\end{array}
\right.
\]

\noindent
Remark~\ref{rem:ax:hf:well-order-hf} implies that
every non-empty $W$ such that
\[
(\forall (x,k) \in W)(x = \Root)
\quad\lor\quad
(\exists z) (\forall (x,k) \in W) \bigdisj_{d \in \Dir}(x = \Succ_d(z))
\]
has a $\preceq$-least element.
By $\HF$-Bounded Choice (Theorem~\ref{thm:funto:choice}), we define
\[
\funto{\e'_\Prop}{\univ \times \Prop}{\Pne(\Opp)}
\qquad\text{and}\qquad
\funto{\e'_\Opp}{\univ \times \Opp}{\Pne(\Dir \times \Prop)}
\]
by setting, for $\player$ either $\Prop$ or $\Opp$,
\[
\e'_\player(u) \quad\deq\quad \{\text{the $\preceq$-least element of $\e_\player(u)$}\}
\]
Let $\G' \deq (\Prop,\Opp,\e')$.
Note that we have $\Game(\G')$
and that
\begin{equation}
\label{eq:lem:infpath:incl}
(\forall u , v) \left(
  u \edge{}{\G'} v ~~\limp~~ u \edge{}{\G} v
\right)
\end{equation}

\noindent
By Comprehension for Product Types (Theorem~\ref{thm:funto:ca}),
we then let 
\[
U \quad\deq\quad
\{ u \st v \edge{*}{\G'} u \}
\]

It remains to show
\[
\Play(\G,v,U)
\]

\noindent
First, we have $v \in U$ by reflexivity of $\edge{*}{\G'}$
(Proposition~\ref{prop:games:edges}.\eqref{eq:games:edges:reftrans}),
and $(\forall u \in U)\left(v \edge{*}{\G} u\right)$
follows from~\eqref{eq:lem:infpath:incl}.
Moreover, we have
\[
(\forall u \in U) (\exists w \in U) \left(u \edge{}{\G} w \right)
\]
thanks to~\eqref{eq:lem:infpath:incl},
since this property already holds for $\G'$.
It remains to show that $U$ is linearly ordered \wrt\@ $\edge{*}{\G}$.
We invoke Proposition~\ref{prop:games:edges:lin}:
its first premise has already been discussed,
its second follows from the definition of $\e'$,
and its last one is
Proposition~\ref{prop:games:edges}.\eqref{eq:games:edges:decomp}.
\end{proof}

Finally, in some situations
(typically for the \emph{Simulation Theorem} in~\S\ref{sec:sim}),
it is convenient to build infinite plays from
paths (in the sense of Definition~\ref{def:games:path}).

%%%%%%%%%%%%%%%%%%%%%%%%%%%%%%%%%%%%%%%%%%%%%%%%%%%%%%%%%%%%%%%%%%%%%%%%%%%
\begin{lem}[Infinite Plays From Paths]
\label{lem:games:pathplays}
%%%%%%%%%%%%%%%%%%%%%%%%%%%%%%%%%%%%%%%%%%%%%%%%%%%%%%%%%%%%%%%%%%%%%%%%%%%
Assume $\Game(\G)$
and let $u_0$ and $U$
be such that
\[
  \Path(\G,u_0,U)
~\land~
  (\forall u,v \in U)\left[\gsucc(u,v) ~\limp~ u \edge{}{} v\right]
\]
Then $\FSOD$ proves $\Play(\G,u_0,U)$.
\end{lem}

\begin{proof}
Thanks to Proposition~\ref{prop:games:edges},
the result directly follows from the fact that
\[
(\forall u,v \in U) \left(u \gle v ~~\longlimp~~ u \edge{*}{} v \right)
\]

\noindent
Fix $u \in U$. By $\glt$-induction we show 
$(\forall v \in U)(u \gle v ~~\limp~~ u \edge{*}{} v)$.
So let $v \in U$ such that the property holds for all $w \glt v$,
and assume $u \gle v$.
If $u = v$ then we are done.
Otherwise, by the Predecessor Lemma~\ref{lem:games:predpath} for Paths,
we have $\gsucc(w,v)$ for some $w \in U$ with $u \gle w$.
By induction hypothesis we get $u \edge{*}{} w \edge{}{} v$
and we conclude by Proposition~\ref{prop:games:edges}.
\end{proof}

\cnote{\CR:
\begin{itemize}
\item 
Lemma~\ref{lem:games:pathplays} says that 
``being an infinite play'' is a first order notion in $\FSO$
(from the primitive symbol $\Lt$).

\item \TODO: 
Put a remark on this, possibly when defining the notion
of infinite plays.
\end{itemize}}

%%%%%%%%%%%%%%%%%%%%%%%%%%%%%%%%%%%%%%%%%%%%%%%%%%%%%%%%%%%%%%%%%%%%%%%%%%%
\subsection{Strategies}
\label{sec:strat}
%\subsubsection*{Strategies}
%%%%%%%%%%%%%%%%%%%%%%%%%%%%%%%%%%%%%%%%%%%%%%%%%%%%%%%%%%%%%%%%%%%%%%%%%%%
We now turn to strategies.
Our strategies are Functions from the positions of one player to the set
of labels of the other player,
which
%Strategies 
must respect the edge relations.
This implies that all our strategies are, by definition, positional.

%%%%%%%%%%%%%%%%%%%%%%%%%%%%%%%%%%%%%%%%%%%%%%%%%%%%%%%%%%%%%%%%%%%%%%%%%%%
\begin{defi}[Strategies]
\label{def:games:strat}
%%%%%%%%%%%%%%%%%%%%%%%%%%%%%%%%%%%%%%%%%%%%%%%%%%%%%%%%%%%%%%%%%%%%%%%%%%%
Let $\G = (\Prop,\Opp,\EP,\EO)$
where $\Prop,\Opp$ are $\HF$-variables and where $\EP,\EO$ are Function variables.
\begin{enumerate}
\item A \emph{$\Prop$-strategy on $\G$}
is a Function $\strat$
which satisfies the formula 
\[
  \Strat_\Prop(\G,\strat)
\quad\deq\quad
\funto{\strat}{\PP\G}{\Opp}
~~\land~~
(\forall v)
  \left( \strat(v) \in \EP(v) \right)
\]
%such that for all $u \in \PP V$ we have
%\[
%s(u) \In e_\Prop(u)
%\]

\item
An \emph{$\Opp$-strategy on $\G$} is a Function
$\strat$
which satisfies the formula 
\[
  \Strat_\Opp(\G,\strat)
\quad\deq\quad
\funto{\strat}{\OP\G}{\Dir \times \Prop}
~~\land~~
(\forall v)
  \left( \strat(v) \in \EO(v) \right)
\]
\end{enumerate}
\end{defi}

Strategies naturally induce subgames in the sense of Definition~\ref{def:games:sub}.
This will allow us to lift to strategies notions
which are more naturally defined at the level of games.

%%%%%%%%%%%%%%%%%%%%%%%%%%%%%%%%%%%%%%%%%%%%%%%%%%%%%%%%%%%%%%%%%%%%%%%%%%%
\begin{defi}[Subgame induced by a Strategy]
%%%%%%%%%%%%%%%%%%%%%%%%%%%%%%%%%%%%%%%%%%%%%%%%%%%%%%%%%%%%%%%%%%%%%%%%%%%
Given a player $\player$ (either $\Prop$ or $\Opp$)
and a $\player$-strategy $\strat$ on $\G$,
we let
\[
\G\restr\{\strat\}_\player
\quad\deq\quad
(\PL\G \,,\, \OL\G \,,\, \EG\G\restr\{\strat\}_\player)
\]
where
\[
\EG\G\restr\{\strat\}_\Prop ~~\deq~~ (\{\strat\}_\Prop,\EGO\G)
\qquad\text{and}\qquad
\EG\G\restr\{\strat\}_\Opp ~~\deq~~ (\EGP\G,\{\strat\}_\Opp)
\]
and where $\{\strat\}_\player \sle \EG\G_\player$
is defined by $\HF$-Bounded Choice to be the Function
taking $u \in \univ \times \G_\player$
to the singleton $\{\strat(u)\}$.
%and where for all $u \in \univ \times \G_\player$,
%$\{\strat\}_\player \sle e_\player$
%is defined by $\HF$-Bounded Choice to 
%be the singleton $\{\strat(u)\}$.

Whenever possible, we write $\G\restr\{\strat\}$
or even just $\strat$ for $\G\restr\{\strat\}_\player$, when it is unambiguous.
\end{defi}

%%%%%%%%%%%%%%%%%%%%%%%%%%%%%%%%%%%%%%%%%%%%%%%%%%%%%%%%%%%%%%%%%%%%%%%%%%%
\begin{lem}
%%%%%%%%%%%%%%%%%%%%%%%%%%%%%%%%%%%%%%%%%%%%%%%%%%%%%%%%%%%%%%%%%%%%%%%%%%%
$\FSOD$ proves the following, where $\player$ is a player
(either $\Prop$ or $\Opp$):
\[
\big(
\Game(\G) ~~\land~~ \Strat_\player(\G,\strat)
\big)
\quad\longlimp\quad
\Game(\strat)
\]
\end{lem}

\noindent
This in particular allows us to speak of the infinite plays of a strategy
$\strat$ on $\G$ simply as infinite plays of the game $\G\restr\{\strat\}$.

%%%%%%%%%%%%%%%%%%%%%%%%%%%%%%%%%%%%%%%%%%%%%%%%%%%%%%%%%%%%%%%%%%%%%%%%%%%
\subsection{Winning}
\label{sec:games:win}
%%%%%%%%%%%%%%%%%%%%%%%%%%%%%%%%%%%%%%%%%%%%%%%%%%%%%%%%%%%%%%%%%%%%%%%%%%%
In order to deal with acceptance for automata, we equip games
with a notion of \emph{winning}.
Given a game $\G$, a \emph{winning condition} on $\G$
is a formula $\Win(U)$ where $U$ is intended to range over the infinite
plays of $\G$.
As usual a $\Prop$-strategy $\strat$ on $(\G,\Win)$
is winning from a position $v$
whenever all the infinite plays $U$ of $\strat$ from $v$ satisfy $\Win(U)$.
Dually, an $\Opp$-strategy is winning from $v$
when all its infinite plays $U$ from $v$ satisfy $\lnot \Win(U)$.

We formally proceed as follows.

%%%%%%%%%%%%%%%%%%%%%%%%%%%%%%%%%%%%%%%%%%%%%%%%%%%%%%%%%%%%%%%%%%%%%%%%%%%
\begin{defi}
\label{def:games:win}
%%%%%%%%%%%%%%%%%%%%%%%%%%%%%%%%%%%%%%%%%%%%%%%%%%%%%%%%%%%%%%%%%%%%%%%%%%%
Let $\G = (\Prop,\Opp,\EP,\EO)$
where $\Prop,\Opp$ are $\HF$-variables and $\EP,\EO$ are Function variables.
Let $\Win(U)$ be a given $\FSO$-formula where $U$ is a Function variable.
\begin{enumerate}
\item
We define the following formulae.
\[
\begin{array}{r !{\quad\deq\quad} l}
  \WonGame_\Prop(\G,v,\Win)
& (\forall U)
  \big(\Play(\G,v,U) ~~\limp~~ \Win(U) \big)
\\
  \WonGame_\Opp(\G,v,\Win)
& (\forall U)
  \big(\Play(\G,v,U) ~~\limp~~ \lnot \Win(U) \big)
\end{array}
\]

\item
Given a player $\player$ (either $\Prop$ or $\Opp$),
we say that a $\player$-strategy $\strat$ is
\emph{winning in $(\G,\Win)$ from $v$}
if the game $(\G\restr \{\strat\}_\player,\Win)$ is won by $\player$ from $v$,
\ie\@ if
the following formula holds
%$\WinStrat_\player(\G,\strat,v,\Win)$ holds:
\[
\WinStrat_\player(\G,\strat,v,\Win)
\quad\deq\quad
  \WonGame_\player(\G\restr\{\strat\}_\player, v, \Win)
\]
\end{enumerate}
\end{defi}

\noindent
Strictly speaking, in Definition~\ref{def:games:win} above,
$\WonGame_\player$
and
$\WinStrat_\player$
are actually families of $\FSO$ formulae, parametrized by the choice of $\FSO$-formula $\Win$.
%not formulae of $\FSO$.
%But for each $\FSO$-formula $\Win$, 
%$\WonGame_\player(\G,u,\Win)$
%and
%$\WinStrat_\player(\G,\strat,v,\Win)$
%are formulae of $\FSO$.

As expected, a game position cannot be winning for both players.

%%%%%%%%%%%%%%%%%%%%%%%%%%%%%%%%%%%%%%%%%%%%%%%%%%%%%%%%%%%%%%%%%%%%%%%%%%%
\begin{lem}
\label{lem:games:notbothwin}
%%%%%%%%%%%%%%%%%%%%%%%%%%%%%%%%%%%%%%%%%%%%%%%%%%%%%%%%%%%%%%%%%%%%%%%%%%%
$\FSOD$ proves the following.
\begin{multline*}
\Game(\G)
\quad\longlimp\quad
\Strat_\Prop(\G,\strat_\Prop)
\quad\longlimp\quad
\Strat_\Opp(\G,\strat_\Opp)
\quad\longlimp
\\
\lnot (\exists v)\bigg[
\WinStrat_\Prop(\G,\strat_\Prop,v,\Win)
~~\land~~
\WinStrat_\Opp(\G,\strat_\Opp,v,\Win)
\bigg]
\end{multline*}
\end{lem}

\begin{proof}
Assume for contradiction that for some $v$ we have
\[
\WinStrat_\Prop(\G,\strat_\Prop,v,\Win)
~~\land~~
\WinStrat_\Opp(\G,\strat_\Opp,v,\Win)
\]
that is
\[
  (\forall U)
  \Big[ \Play(\G\restr\{\strat_\Prop\},v,U) ~~\limp~~ \Win(U) \Big]
~~\land~~
  (\forall U)
  \Big[ \Play(\G\restr\{\strat_\Opp\},v,U) ~~\limp~~ \lnot \Win(U) \Big]
\]

\noindent
Consider the game
\[
\G' \quad\deq\quad
(\Prop,\, \Opp,\, \{\strat_\Prop\}_\Prop,\, \{\strat_\Opp\}_\Opp)
\]
Note that $\G'$ is a subgame of both $\G\restr\{\strat_\Prop\}$
and $\G\restr\{\strat_\Opp\}$.
We thus get
\[
  (\forall U)
  \Big[ \Play(\G',v,U) ~~\limp~~ \Win(U) \land \lnot \Win(U) \Big]
\]
which implies that there is no $U$
such that $\Play(\G',v,U)$,
contradicting Lemma~\ref{lem:games:infplay:pos}.
\end{proof}

%%%%%%%%%%%%%%%%%%%%%%%%%%%%%%%%%%%%%%%%%%%%%%%%%%%%%%%%%%%%%%%%%%%%%%%%%%%
\subsection{Parity Conditions}
\label{sec:games:parity}
%%%%%%%%%%%%%%%%%%%%%%%%%%%%%%%%%%%%%%%%%%%%%%%%%%%%%%%%%%%%%%%%%%%%%%%%%%%
In this paper, we mostly consider winning conditions expressed as
\emph{parity conditions}.
Parity conditions are defined from \emph{colorings}
of game positions by natural numbers from a given finite interval.
We represent natural numbers and the operations and relations on them
using the Functions on $\HF$-Sets of $\FSO$ and the axioms
of~\S\ref{sec:ax:hf}.

%%%%%%%%%%%%%%%%%%%%%%%%%%%%%%%%%%%%%%%%%%%%%%%%%%%%%%%%%%%%%%%%%%%%%%%%%%%
\begin{conv}
\label{conv:games:colors}
%%%%%%%%%%%%%%%%%%%%%%%%%%%%%%%%%%%%%%%%%%%%%%%%%%%%%%%%%%%%%%%%%%%%%%%%%%%
In order to conveniently manipulate colorings and parity conditions, 
we will use the following functions on finite ordinals
(a.k.a.\@ natural numbers),
obtained from the \emph{Axioms on $\HF$-Functions} (see~\S\ref{sec:ax:hf}).
We rely on the well-known fact that ``\emph{$n$ is an ordinal}''
can be expressed by an $\HF$-formula
$\Ord(n)$
(see \eg~\cite[Lemma 12.10]{jech06set}).
\begin{enumerate}[(1)]
\item
We consider unary $\HF$-Functions
\[
[0,-] \,,\, [0,-) \,,\, (0,-]
~~:~~ V_\omega ~~\longto~~ V_\omega
\]
such that for all finite ordinals $n$, we have
\[
%V_\omega\models\quad
\Sk(\ZFCM)\thesis\quad
  [0,n] ~\Eq~ \{0,\dots,n\}
~~\land~~
  [0,n) ~\Eq~ \{0,\dots,n-1\}
~~\land~~
  (0,n] ~\Eq~ \{1,\dots,n\}
\]

\item
We consider binary $\HF$-Functions
\[
\const g_\leq \,,\, \const g_< \,,\, \const g_\geq \,,\, \const g_>
~~:~~
V_\omega \times V_\omega ~~\longto~~ \two
\]
such that for finite ordinals $n,m$
\[
\begin{array}{c !{\qquad} c}
  \const g_\leq(n,m) = 1 ~~\text{iff}~~ n \leq m
& \const g_\geq(n,m) = 1 ~~\text{iff}~~ n \geq m
\\
  \const g_<(n,m) = 1 ~~\text{iff}~~ n < m
& \const g_>(n,m) = 1 ~~\text{iff}~~ n > m
\end{array}
\]
In $\FSO$-formulae, we write
$n \leq m$ for the formula $\const g_\leq(n,m) \Eq 1$, and so on.

\item
We consider a unary $\HF$-Function
\[
\even
~~:~~ V_\omega ~~\longto~~ V_\omega
\]

\noindent
such that for each ordinal $n$,
$\even(n)$ is the set of ordinals $m \in [0,n]$
such that $m$ represents an even number.

\item
\label{item:games:colors:arith}
We consider $\HF$-Functions $\max(-,-)$ and $(-) + 1$,
computing respectively the maximum
of two finite ordinals and the successor ordinal of an ordinal.
\end{enumerate}
\end{conv}

%%%%%%%%%%%%%%%%%%%%%%%%%%%%%%%%%%%%%%%%%%%%%%%%%%%%%%%%%%%%%%%%%%%%%%%%%%%
\begin{rem}
%%%%%%%%%%%%%%%%%%%%%%%%%%%%%%%%%%%%%%%%%%%%%%%%%%%%%%%%%%%%%%%%%%%%%%%%%%%
Even if ``\emph{$n$ is an ordinal}'' can be expressed by an
$\HF$-formula, quantification over all finite ordinals
 \emph{cannot} be expressed in $V_\omega$ by an $\HF$-formula,
since for each finite ordinal $n>0$ 
we have $n \in V_{n} \setminus V_{n-1}$.
In particular, induction over finite $\HF$-ordinals cannot be expressed
by an $\HF$-formula.
\end{rem}

%%%%%%%%%%%%%%%%%%%%%%%%%%%%%%%%%%%%%%%%%%%%%%%%%%%%%%%%%%%%%%%%%%%%%%%%%%%
\begin{defi}[Parity Conditions]
\label{def:games:parity}
%%%%%%%%%%%%%%%%%%%%%%%%%%%%%%%%%%%%%%%%%%%%%%%%%%%%%%%%%%%%%%%%%%%%%%%%%%%
Let $\G = (\Prop,\Opp,\EP,\EO)$
where $\Prop,\Opp$ are $\HF$-variables and $\EP,\EO$ are Function variables.
\begin{enumerate}
\item
A \emph{coloring} is given by a Function $\col$ and an $\HF$-term $n$
satisfying the following formula
\[
\Coloring(\G,\col,n)
\quad\deq\quad
\Ord(n) ~~\land~~ \funto{\col}{\G}{[0,n]}
\]

\item
We define the following formula:
\[
\Par(\G,\col,n,U)
\quad\deq\quad
(\exists m \In \even(n))
\left[
  \begin{array}{r l}
  & (\forall u \in U) (\exists v \in U) (u \edge{+}{} v ~\land~ \col(v) \Eq m)
  \\
    \land
  & (\exists u \in U) (\forall v \in U) (u \edge{+}{} v ~\limp~ \col(v) \geq m)
  \end{array}
\right]
\]
\end{enumerate}
\end{defi}

%%%%%%%%%%%%%%%%%%%%%%%%%%%%%%%%%%%%%%%%%%%%%%%%%%%%%%%%%%%%%%%%%%%%%%%%%%%
\begin{rem}
%%%%%%%%%%%%%%%%%%%%%%%%%%%%%%%%%%%%%%%%%%%%%%%%%%%%%%%%%%%%%%%%%%%%%%%%%%%
The formula $\Par(\G,\col,n,U)$ will be used to say that an infinite play $U$
satisfies the (min) parity condition induced by the
coloring $\funto{\col}{\G}{[0,n]}$.
In the standard model $\Std$, if $U$ is an infinite play in $\G$, then
$\Par(\G,\col,n,U)$ holds if and only if there is an even $m \leq n$
such that
$U$ has infinitely many positions colored by $m$, and
$U$ has only finitely many positions colored by any $k < m$.
Also, notice that any $U$ (not necessarily a play) satisfying $\Par(\G,\col,n,U)$ in $\Std$ is infinite.
\end{rem}

%%%%%%%%%%%%%%%%%%%%%%%%%%%%%%%%%%%%%%%%%%%%%%%%%%%%%%%%%%%%%%%%%%%%%%%%%%%
\begin{rem}
\label{rem:games:parity:sub}
%%%%%%%%%%%%%%%%%%%%%%%%%%%%%%%%%%%%%%%%%%%%%%%%%%%%%%%%%%%%%%%%%%%%%%%%%%%
Assume that $\G'$ is a subgame of $\G$
(in the sense of Definition~\ref{def:games:sub}).
Note that $\FSO$ proves
\[
\Coloring(\G,\col,n) ~~\longliff~~
\Coloring(\G',\col,n)
\]

\noindent
Furthermore, as noted earlier, every infinite play in $\G'$ 
is an infinite play in $\G$.
It follows that $\FSO$ proves 
\begin{multline*}
  \Game(\G)
~~\longlimp~~
  \Coloring(\G,\col,n)
~~\longlimp~~
  \Game(\G')
~~\longlimp~~
  \Sub(\G',\G)
~~\longlimp~~
\\
(\forall \funto{U}{\G'}{\two})(\forall v)
\Big[
\Play(\G',v,U)
~~\longlimp~~
\big(
\Par(\G',\col,n,U)
~~\liff~~
\Par(\G,\col,n,U)
\big)
\Big]
\end{multline*}
\end{rem}

%%%%%%%%%%%%%%%%%%%%%%%%%%%%%%%%%%%%%%%%%%%%%%%%%%%%%%%%%%%%%%%%%%%%%%%%%%%
\begin{rem}
\label{rem:games:parity:gle}
%%%%%%%%%%%%%%%%%%%%%%%%%%%%%%%%%%%%%%%%%%%%%%%%%%%%%%%%%%%%%%%%%%%%%%%%%%%
When considering parity automata in~\S\ref{sec:aut},
it will actually be convenient to define acceptance via
the formula $\Par$
for games of the form $\G(\gle)$ in the sense of Remark~\ref{rem:games:sub}.
It follows from Remarks~\ref{rem:games:sub} and~\ref{rem:games:parity:sub}
that $\FSO$ proves
\begin{multline*}
  \Game(\G)
~~\longlimp~~
  \Coloring(\G(\gle),\col,n)
~~\longlimp~~
\\
(\forall \funto{U}{\G}{\two})(\forall v)
\Big[
\Play(\G,v,U)
~~\longlimp~~
\Big(
\Par(\G,\col,n,U)
~~\liff~~
\Par(\G(\gle),\col,n,U)
\Big)
\Big]
\end{multline*}
\end{rem}

\noindent
We use the following more succinct notation for winning in the case parity games.

%%%%%%%%%%%%%%%%%%%%%%%%%%%%%%%%%%%%%%%%%%%%%%%%%%%%%%%%%%%%%%%%%%%%%%%%%%%
\begin{nota}[Winning in Parity Games]
%%%%%%%%%%%%%%%%%%%%%%%%%%%%%%%%%%%%%%%%%%%%%%%%%%%%%%%%%%%%%%%%%%%%%%%%%%%
Let $\G = (\Prop,\Opp,\EP,\EO)$
where $\Prop,\Opp$ are $\HF$-variables and $\EP,\EO$ are Function variables.
Let $\col$ be a Function variable and $n$ be an $\HF$-variable.
We write the following,
where $\player$ is a player (either $\Prop$ or $\Opp$).
\[
\begin{array}{r !{\quad\deq\quad} l}
  \WonGame_\player(\G,v,\col,n)
& \WonGame_\player(\G,v,\Par(\G,\col,n,-))
\\
  \WinStrat_\player(\G,\strat,v,\col,n)
& \WinStrat_\player(\G,\strat,v,\Par(\G,\col,n,-))
\end{array}
\]
\end{nota}

%%%%%%%%%%%%%%%%%%%%%%%%%%%%%%%%%%%%%%%%%%%%%%%%%%%%%%%%%%%%%%%%%%%%%%%%%%%
\subsection{The Axiom of Positional Determinacy of Parity Games.}
\label{sec:posdet}
%%%%%%%%%%%%%%%%%%%%%%%%%%%%%%%%%%%%%%%%%%%%%%%%%%%%%%%%%%%%%%%%%%%%%%%%%%%
We now formulate the axiom scheme $(\PosDet)$,
which states the (positional) determinacy of 
parity games.
Intuitively $(\PosDet)$ should consist of all formulae
of the form
\begin{multline*}
\Game(\G) ~~\limp~~ \Coloring(\G,\col,n)
%\quad\longlimp
~~\limp~~
\\
(\forall v \in \G)
\left[
\begin{array}{c r l}
& (\exists \funto{\strat_\Prop}{\PP\G}{\Opp})
& \left(
  \begin{array}{c l}
  & \Strat_\Prop(\G,\strat_\Prop)
  \\
    \land
  %& \WonGame_\Prop((V,e\restr\{\strat_\Prop\}_\Prop),v,c,n)
  & \WinStrat_\Prop(\G,\strat_\Prop,v,\col,n)
  \end{array}
  \right)
\\\\
  \lor
& (\exists \funto{\strat_\Opp}{\OP\G}{\Dir \times \Prop})
& \left(
  \begin{array}{c l}
  & \Strat_\Opp(\G,\strat_\Opp)
  \\
    \land
  %& \WonGame_\Opp((V,e\restr\{\strat_\Opp\}_\Opp),v,c,n)
  & \WinStrat_\Opp(\G,\strat_\Opp,v,\col,n)
  \end{array}
  \right)
\end{array}
\right]
\end{multline*}

\noindent
But note that these formulae are open, and in particular
\[
\G = (\Prop,\Opp,\EP,\EO)
\quad\text{and}\quad
\col
\]
contain free Function variables.
On the other hand, when formulating our completeness results in~\S\ref{sec:compl},
it will be interesting to have translations of instances of $(\PosDet)$ in $\MSO$,
based on the map $\MI{-} : \FSO \to \MSO$ of~\S\ref{sec:cons}.
However, the translation $\MI{-}$ only handles $\HF$-closed
formulae without free Function variables.
We therefore officially let $(\PosDet)$ consist
of all formulae %of the form
$\formfont{PosDet}(\Prop,\Opp,n)$,
for $\Prop$, $\Opp$ and $n$ ranging over $\HF$-terms (see~\S\ref{sec:fso:fso}),
where
$\formfont{PosDet}(\Prop,\Opp,n)$ is the formula
\begin{multline*}
%\formfont{PosDet}(\Prop,\Opp,n) \quad\deq\quad
\Labels(\Prop,\Opp)
~~\limp~~
\Ord(n)
~~\limp~~
\\
\Big( \forall \funto{\EP}{\PP\G}{\Pne(\Opp)} \Big)
\Big( \forall \funto{\EO}{\OP\G}{\Pne(\Dir \times \Prop)} \Big)
\Big( \forall \funto{\col}{\G}{[0,n]} \Big)
\\
(\forall v \in \G)
\left[
\begin{array}{c r l}
& (\exists \funto{\strat_\Prop}{\PP\G}{\Opp})
& \left(
  \begin{array}{c l}
  & \Strat_\Prop(\G,\strat_\Prop)
  \\
    \land
  %& \WonGame_\Prop((V,e\restr\{\strat_\Prop\}_\Prop),v,c,n)
  & \WinStrat_\Prop(\G,\strat_\Prop,v,\col,n)
  \end{array}
  \right)
\\\\
  \lor
& (\exists \funto{\strat_\Opp}{\OP\G}{\Dir \times \Prop})
& \left(
  \begin{array}{c l}
  & \Strat_\Opp(\G,\strat_\Opp)
  \\
    \land
  %& \WonGame_\Opp((V,e\restr\{\strat_\Opp\}_\Opp),v,c,n)
  & \WinStrat_\Opp(\G,\strat_\Opp,v,\col,n)
  \end{array}
  \right)
\end{array}
\right]
\end{multline*}

%\noindent
%where $\Prop$, $\Opp$ and $n$ are $\HF$-terms (see~\S\ref{sec:fso:fso}).

It follows from the positional determinacy of parity games~\cite{ej91focs}
(see also~\cite{thomas97handbook,walukiewicz02tcs,pp04book})
that all instances of $(\PosDet)$ hold in the standard model
$\Std$ of $\FSO$.
We can thus extend Proposition~\ref{prop:std:cor}
to the following.

%%%%%%%%%%%%%%%%%%%%%%%%%%%%%%%%%%%%%%%%%%%%%%%%%%%%%%%%%%%%%%%%%%%%%%%%%%%
\begin{prop}
\label{prop:games:std:cor}
%%%%%%%%%%%%%%%%%%%%%%%%%%%%%%%%%%%%%%%%%%%%%%%%%%%%%%%%%%%%%%%%%%%%%%%%%%%
For each closed $\FSO$-formula $\varphi$,
\[
\Std \models \varphi
\qquad\text{whenever}\qquad
\FSO + (\PosDet) \thesis \varphi
\]
\end{prop}

%%%%%%%%%%%%%%%%%%%%%%%%%%%%%%%%%%%%%%%%%%%%%%%%%%%%%%%%%%%%%%%%%%%%%%%%%%%
\subsubsection{The Axiom of Positional Determinacy in $\MSO$.}
\label{sec:posdet:mso}
%%%%%%%%%%%%%%%%%%%%%%%%%%%%%%%%%%%%%%%%%%%%%%%%%%%%%%%%%%%%%%%%%%%%%%%%%%%
In order to obtain a complete axiomatization of $\MSOD$ from
the completeness of $\FSOD+(\PosDet)$ (see \S\ref{sec:compl}),
we extend the axioms of $\MSOD$ with sufficiently many
translated instantiations $\MI{\formfont{PosDet}(\Prop,\Opp,n)}$
for $\Prop$, $\Opp$ and $n$ \emph{closed} $\HF$-terms.
However, in general these terms may contain arbitrary $\HF$-Functions
symbols, which make the translation 
$\MI{\formfont{PosDet}(\Prop,\Opp,n)}$
%their $\MI{-}$-translations
in general uncomputable from $\Prop$, $\Opp$ and $n$
(see Remark~\ref{rem:cons:dec} and \S\ref{sec:ax:hf}).
However, for each closed $\HF$-\emph{terms} $\Prop$, $\Opp$
and $n$, there are constant symbols for $\HF$-\emph{sets}
$\const\Prop$, $\const\Opp$ and $\const n$ such that the formulae
$\MI{\formfont{PosDet}(\Prop,\Opp,n)}$
and
$\MI{\formfont{PosDet}(\const\Prop,\const\Opp,\const n)}$
are syntactically identical.
We therefore officially take the following version of $(\PosDet)$
in $\MSOD$.

%%%%%%%%%%%%%%%%%%%%%%%%%%%%%%%%%%%%%%%%%%%%%%%%%%%%%%%%%%%%%%%%%%%%%%%%%%%
\begin{defi}[The Axiom of Positional Determinacy in $\MSO$]
\label{def:posdet:mso}
%%%%%%%%%%%%%%%%%%%%%%%%%%%%%%%%%%%%%%%%%%%%%%%%%%%%%%%%%%%%%%%%%%%%%%%%%%%
We let $\MI{\PosDet}$ consist of all formulae of the form
$\MI{\formfont{PosDet}(\const\Prop,\const\Opp,\const n)}$,
for $\const\Prop$, $\const\Opp$ and $\const n$
ranging over constant symbols for $\HF$-sets.
%(see \S\ref{sec:fso:fso}).
\end{defi}

%Note that this axiom is an open formula $\FSO$,
%because all quantifications over Functions and $\HF$-Sets 
%in $\FSO$ should be bounded.
%Let $\G = (\Prop,\Opp,\EP,\EO)$
%where $\Prop,\Opp$ are $\HF$-variables and $\EP,\EO$ are Function variables.
%Let further $\col$ be a Function and $n$ be an $\HF$-variable.
%The axiom $(\PosDet)$ is the following formula.
%\begin{multline*}
%\Game(\G) ~~\land~~ \Coloring(\G,\col,n)
%\quad\longlimp
%\\
%(\forall v \in \G)
%\left[
%\begin{array}{c r l}
%& (\exists \funto{\strat_\Prop}{\PP\G}{\Opp})
%& \left(
%  \begin{array}{c l}
%  & \Strat_\Prop(\G,\strat_\Prop)
%  \\
%    \land
%  %& \WonGame_\Prop((V,e\restr\{\strat_\Prop\}_\Prop),v,c,n)
%  & \WinStrat_\Prop(\G,\strat_\Prop,v,\col,n)
%  \end{array}
%  \right)
%\\\\
%  \lor
%& (\exists \funto{\strat_\Opp}{\OP\G}{\Dir \times \Prop})
%& \left(
%  \begin{array}{c l}
%  & \Strat_\Opp(\G,\strat_\Opp)
%  \\
%    \land
%  %& \WonGame_\Opp((V,e\restr\{\strat_\Opp\}_\Opp),v,c,n)
%  & \WinStrat_\Opp(\G,\strat_\Opp,v,\col,n)
%  \end{array}
%  \right)
%\end{array}
%\right]
%\end{multline*}

%%% Local Variables:
%%% mode: latex
%%% TeX-master: "main.tex"
%%% End:

%%%%%%%%%%%%%%%%%%%%%%%%%%%%%%%%%%%%%%%%%%%%%%%%%%%%%%%%%%%%%%%%%%%%%%%%%%%
\section{Alternating Tree Automata}
\label{sec:aut}
%%%%%%%%%%%%%%%%%%%%%%%%%%%%%%%%%%%%%%%%%%%%%%%%%%%%%%%%%%%%%%%%%%%%%%%%%%%

\noindent
We detail in this Section a representation of alternating tree automata
in $\FSO$.
We closely follow the presentation of~\cite{walukiewicz02tcs}.
Our main motivation to consider alternating automata is that when formulating
acceptance with (parity) games (of the kind of~\S\ref{sec:games}),
complementation follows from (positional) determinacy
(\ie\@ in our setting from the Axiom $(\PosDet)$).
Let us recall the main ideas underlying alternating automata.
The original formulation, as %stated
in \eg~\cite{ms87tcs,ms95tcs},
is for an automaton $\At A$ with state set $Q$
to have transitions with values in the free distributive
lattice over $\Dir \times Q$
(in other words, transitions have positive Boolean formulae over $\Dir \times Q$
as values).
%(or equivalently).
%But recall from \eg~\cite[Lemma I.4.8]{johnstone86stone}
%that free distributive lattices are given by irredundant disjunctive
%normal forms.
%Actually, following~\cite{walukiewicz02tcs} we can give up irredundancy.
%We thus simply assume that transitions are of the form
Actually, following~\cite{walukiewicz02tcs} we can
simply assume that transitions are of the form
\[
\trans ~~:~~ Q \times \Sigma
~~\longto~~ \Pne(\Pne(\Dir \times Q))
\]
and we read $\trans(q,\al a)$
as the disjunctive normal form
\[
\bigdisj\limits_{\Conj \in \trans(q, \al a)}
\bigconj\limits_{(d,q') \in \Conj} (d,q')
\]
This results in acceptance games 
where intuitively $\Prop$ plays from disjunctions while $\Opp$
plays from conjunctions.
In the following we often call the $\Conj \in \trans(q, \al a)$
\emph{conjunctions}.

We begin by giving basic definitions in~\ref{sec:aut:alt}.
Because our setting is restricted to only describe
\emph{positional} strategies, and because parity
games are positionally determined, we give a special emphasis
to \emph{parity} automata, whose acceptance conditions
are parity conditions generated from a coloring of their states.
We then present a series of operations on automata,
on which we rely in~\S\ref{sec:compl:aut} for the interpretation
of $\MSO$ formulae as automata.
We recapitulate them in Table~\ref{tab:aut:op}.
First,~\S\ref{sec:aut:subst} and~\S\ref{sec:aut:disj} present
two simple constructions
implementing respectively a substitution and a disjunction
operation.
We discuss in~\S\ref{sec:aut:nd} and~\S\ref{sec:nd:proj} the important
special case of \emph{non-deterministic} automata.
Non-deterministic automata are important because they allow us,
via the usual \emph{projection} operation (\S\ref{sec:nd:proj}),
to interpret the existential quantifier of $\MSO$ (see~\S\ref{sec:compl:aut}).
To this end, an important result of the theory of automata on infinite trees
is the \emph{Simulation Theorem}~\cite{ej91focs,ms95tcs},
which states that each alternating automaton is equivalent to a
non-deterministic one.
The formalization of this result in $\FSO$ is deferred to~\S\ref{sec:sim}.
This is the only part of this paper where we shall (momentarily)
use automata with acceptance conditions which are not parity conditions.
This result moreover relies on the complete axiomatization of
$\MSO$ on $\omega$-words for paths of $\FSO$
(to be discussed in~\S\ref{sec:msow}).
Finally, in~\S\ref{sec:neg} we discuss complementation in the setting of $\FSO$,
and show that alternating automata can be complemented in $\FSO$
when we assume the Axiom $(\PosDet)$ of Positional Determinacy of Parity Games.

%We fix for this Section a non-empty $\HF$-set $\Dir$ of tree directions.

\cnote{\CR:NOTE\quad We require $\Sigma$ to be non-empty
so that $\G(\At A)$ is always defined.}

%%%%%%%%%%%%%%%%%%%%%%%%%%%%%%%%%%%%%%%%%%%%%%%%%%%%%%%%%%%%%%%%%%%%%%%%%%%
\begin{table}[tbp]
%%%%%%%%%%%%%%%%%%%%%%%%%%%%%%%%%%%%%%%%%%%%%%%%%%%%%%%%%%%%%%%%%%%%%%%%%%%
\[
\begin{array}{l l !{\quad} l !{\quad} l l}
\toprule
  \textbf{Name}
& \textbf{Notation}
& \textbf{Requirements} 
& \multicolumn{2}{l}{\textbf{Location in text}}
\\
\midrule

  \text{Substitution}
& \At A[f] : \Gamma
& \At A : \Sigma
  \quad\text{and}\quad \funto{f}{\Gamma}{\Sigma}
& \S\ref{sec:aut:subst}
& \text{(Lem.~\ref{lem:aut:subst})}
\\
%\hline

  \text{Disjunction}
& (\At A_0 \oplus \At A_1) : \Sigma
& \At A_i : \Sigma
  \quad(i = 0,1)
& \S\ref{sec:aut:disj}
& \text{(Lem.~\ref{lem:aut:disj})}
\\
%\hline

  \text{Complementation}
& \aneg \At A:\Sigma
& \At A :\Sigma
  \quad\text{parity}
& \S\ref{sec:neg}
& \text{(Thm.~\ref{thm:neg})}
\\
&
& \text{($\FSO$ + Axiom $(\PosDet)$)}
&
&
\\
%\hline

  \text{Projection}
& (\exists_\Gamma \At A) : \Sigma
& \At A : \Sigma \times \Gamma
  \quad\text{non-deterministic}
& \S\ref{sec:nd:proj}
& \text{(Prop.~\ref{prop:aut:nd:proj})}
\\
%\hline

  \text{Simulation}
& \ND(\At A) : \Sigma
& \At A : \Sigma
  \quad\text{$\HF$-closed}
& \S\ref{sec:sim}
& \text{(Thm.~\ref{thm:sim})}
\\
\toprule
\end{array}
\]
\caption{Operations on Automata.\label{tab:aut:op}}
\end{table}

%%%%%%%%%%%%%%%%%%%%%%%%%%%%%%%%%%%%%%%%%%%%%%%%%%%%%%%%%%%%%%%%%%%%%%%%%%%%
%\begin{table}[tbp]
%%%%%%%%%%%%%%%%%%%%%%%%%%%%%%%%%%%%%%%%%%%%%%%%%%%%%%%%%%%%%%%%%%%%%%%%%%%%
%\[
%\begin{array}{| l l @{\quad}| l | l l |}
%\cline{1-3}
%  \multicolumn{2}{| c |}{\text{Operation}}
%& \multicolumn{1}{c |}{\text{Requirement}}
%& %\S
%& %\text{Correctness}
%  \multicolumn{1}{c}{}
%\\
%\hline
%
%  \text{Substitution}
%& \At A[f] : \Gamma
%& \At A : \Sigma
%  \quad\text{and}\quad \funto{f}{\Gamma}{\Sigma}
%& \S\ref{sec:aut:subst}
%& \text{(Lem.~\ref{lem:aut:subst})}
%\\
%\hline
%
%  \text{Disjunction}
%& (\At A_0 \oplus \At A_1) : \Sigma
%& \At A_i : \Sigma
%  \quad(i = 0,1)
%& \S\ref{sec:aut:disj}
%& \text{(Lem.~\ref{lem:aut:disj})}
%\\
%\hline
%
%  \text{Complementation}
%& \aneg \At A:\Sigma
%& \At A :\Sigma
%  \quad\text{parity}
%& \S\ref{sec:neg}
%& \text{(Thm.~\ref{thm:neg})}
%\\
%&
%& \text{($\FSO$ + Axiom $(\PosDet)$)}
%&
%&
%\\
%\hline
%
%  \text{Projection}
%& (\exists_\Gamma \At A) : \Sigma
%& \At A : \Sigma \times \Gamma
%  \quad\text{non-deterministic}
%& \S\ref{sec:nd:proj}
%& \text{(Prop.~\ref{prop:aut:nd:proj})}
%\\
%\hline
%
%  \text{Simulation}
%& \ND(\At A) : \Sigma
%& \At A : \Sigma
%  \quad\text{$\HF$-closed}
%& \S\ref{sec:sim}
%& \text{(Thm.~\ref{thm:sim})}
%\\
%\hline
%\end{array}
%\]
%%\hrule
%\caption{Operations on Automata.\label{tab:aut:op}}
%\end{table}

%%%%%%%%%%%%%%%%%%%%%%%%%%%%%%%%%%%%%%%%%%%%%%%%%%%%%%%%%%%%%%%%%%%%%%%%%%%
\subsection{Alternating Tree Automata in $\FSOD$}
\label{sec:aut:alt}
%%%%%%%%%%%%%%%%%%%%%%%%%%%%%%%%%%%%%%%%%%%%%%%%%%%%%%%%%%%%%%%%%%%%%%%%%%%
We present here a representation of alternating tree automata
in $\FSO$.

%%%%%%%%%%%%%%%%%%%%%%%%%%%%%%%%%%%%%%%%%%%%%%%%%%%%%%%%%%%%%%%%%%%%%%%%%%%
\begin{defi}[Alternating Tree Automata]
\label{def:aut:alt}
%%%%%%%%%%%%%%%%%%%%%%%%%%%%%%%%%%%%%%%%%%%%%%%%%%%%%%%%%%%%%%%%%%%%%%%%%%%
Given an $\HF$-Set $\Sigma$,
an \emph{Alternating Tree Automaton}
(or simply \emph{Automaton}) $\At A$ on $\Sigma$
(notation $\At A : \Sigma$)
is given by $\HF$-terms $Q_{\At A}$, $\init q_{\At A}$
and $\trans_{\At A}$
together with
an $\FSO$-formula $\Omega_{\At A}(U)$
of a Function variable $U$,
which are required to satisfy the following formula:
\[
\Aut(\Sigma, Q_{\At A},\init q_{\At A},\trans_{\At A})
\quad\deq\quad
(\exists \al a \In \Sigma)
~~\land~~
\init q_{\At A} \In Q_{\At A}
~~\land~~
\funto{\trans_{\At A}}{Q_{\At A} \times \Sigma}
{\Pne(\Pne(\Dir \times Q_{\At A}))}
\]
where $\Pne(-)$ is the $\HF$-Function of~\S\ref{sec:ax:hf}.\ref{item:ax:hf:po}.
We write
\[
\At A:\Sigma
\quad=\quad (Q_{\At A},\, \init q_{\At A},\, \trans_{\At A}, \, \Omega_{\At A})
\]
and adopt the following terminology:
%the $\HF$-Set
$\Sigma$ is the \emph{input alphabet}
of $\At A$, $Q_{\At A}$ is its set of \emph{states}
(with $\init q_{\At A}$ \emph{initial}),
$\trans_{\At A}$ is the \emph{transition function} of $\At A$
and $\Omega_{\At A}$ is its \emph{acceptance condition}.

We often write $\Aut(\At A:\Sigma)$ or even $\Aut(\At A)$ for 
$\Aut(\Sigma, Q_{\At A},\init q_{\At A},\trans_{\At A})$.
\end{defi}

\noindent
An automaton $\At A : \Sigma$ is intended to run over $\Sigma$-labeled
$\Dir$-ary trees, represented as Functions $F:\Sigma$
(equivalently $\funto{F}{\univ}{\Sigma}$, following~\S\ref{sec:not}).
As usual, acceptance is modeled using games,
which we formalize in the setting of~\S\ref{sec:games}.

%The idea is that transitions
%\[
%\funto{\trans_{\At A}}{Q_{\At A} \times \Sigma}
%{\Pne(\Pne(\Dir \times Q_{\At A}))}
%\]
%are seen as a set representation of disjunctive normal forms (DNF's)
%over $\Dir \times Q_{\At A}$

%%%%%%%%%%%%%%%%%%%%%%%%%%%%%%%%%%%%%%%%%%%%%%%%%%%%%%%%%%%%%%%%%%%%%%%%%%%
\begin{defi}[Acceptance Games]
\label{def:aut:games}
%%%%%%%%%%%%%%%%%%%%%%%%%%%%%%%%%%%%%%%%%%%%%%%%%%%%%%%%%%%%%%%%%%%%%%%%%%%
Given an automaton $\At A : \Sigma$ and a Function $F:\Sigma$
we define the \emph{acceptance game} $\G(\At A,F)$ as follows:
\[
\PL{\G(\At A,F)} \quad\deq\quad Q_{\At A}
\qquad\qquad
\OL{\G(\At A,F)} \quad\deq\quad Q_{\At A} \times \Pne(\Dir \times Q_{\At A})
\]
and $\EGP{\G(\At A,F)}$, $\EGO{\G(\At A,F)}$
are defined by
$\HF$-Bounded Choice for Product Types (Theorem~\ref{thm:funto:choice})
and Comprehension for $\HF$-Sets (Remark~\ref{rem:hfchoice}) as
\[
\begin{array}{l !{\qquad} r !{\quad\text{iff}\quad} l}
& (q',\Conj) \in \EGP{\G(\At A,F)}(x,q) 
& q' \Eq q ~\land~ \Conj \in \trans_{\At A}(q,F(x))
\\
  \text{and}
& (d,q') \in \EGO{\G(\At A,F)}(x,(q,\Conj))
& (d,q') \in \Conj
\end{array}
\]
\end{defi}

%%%%%%%%%%%%%%%%%%%%%%%%%%%%%%%%%%%%%%%%%%%%%%%%%%%%%%%%%%%%%%%%%%%%%%%%%%%
\begin{rem}
%%%%%%%%%%%%%%%%%%%%%%%%%%%%%%%%%%%%%%%%%%%%%%%%%%%%%%%%%%%%%%%%%%%%%%%%%%%
Note that $\Aut(\At A)$ implies $\Game(\G(\At A,F))$ for
$F:\Sigma$.
The edge relations of $\G(\At A,F)$
(in the sense of Definition~\ref{def:games:games})
are given by
\[
\begin{array}{r c l !{\quad\text{iff}\quad} l}
  (x,q)
& \edge{}{\Prop}
& (x,(q,\Conj))
& \Conj \in \trans_{\At A}(q,F(x))
\\
  (x,(q,\Conj))
& \edge{}{\Opp}
& (\Succ_d(x),q')
& (d,q') \in \Conj
\end{array}
\]

\noindent
Note also that an $\Opp$-position $(x,(q,\Conj))$
is equipped with the information
$(x,q) \in \PP{\G(\At A,F)}$.
It follows that
%so that 
an $\Opp$-position has at most one predecessor.
This is useful when complementing automata (\S\ref{sec:neg}).
\end{rem}

%%%%%%%%%%%%%%%%%%%%%%%%%%%%%%%%%%%%%%%%%%%%%%%%%%%%%%%%%%%%%%%%%%%%%%%%%%%
\begin{conv}
%%%%%%%%%%%%%%%%%%%%%%%%%%%%%%%%%%%%%%%%%%%%%%%%%%%%%%%%%%%%%%%%%%%%%%%%%%%
It the rest of this paper, unless explicitly stated otherwise,
when speaking of an infinite play in an acceptance game $\G(\At A,F)$
(including infinite plays in strategies in such games),
we always mean an infinite play from position $(\Root,\init q_{\At A})$.
\end{conv}

Given $\At A : \Sigma$ and $F:\Sigma$,
the acceptance condition $\Omega_{\At A}(-)$ of $\At A$
induces a winning condition in the sense of~\S\ref{sec:games:win}
in the game $\G(\At A,F)$.
%The game $\G(\At A,F)$ may be equipped with a winning condition $\Omega_{\At A}(-)$
%as in~\S\ref{sec:games:win}.
This gives the following notions of tree
acceptance and language generated by an automaton.

%%%%%%%%%%%%%%%%%%%%%%%%%%%%%%%%%%%%%%%%%%%%%%%%%%%%%%%%%%%%%%%%%%%%%%%%%%%
\begin{defi}[Language of an Automaton]
\label{def:aut:lang}
%%%%%%%%%%%%%%%%%%%%%%%%%%%%%%%%%%%%%%%%%%%%%%%%%%%%%%%%%%%%%%%%%%%%%%%%%%%
Given an automaton $\At A : \Sigma$, a winning condition $\Omega_{\At A}$
(in the sense of Definition~\ref{def:games:win}) and $F : \Sigma$,
we say that \emph{$\At A$ accepts $F$}
when the following formula $F \in \Lang(\At A)$
holds.
\[
F \in \Lang(\At A) \quad\deq\quad
  (\exists \funto{\strat_\Prop}{\PP{\G(\At A,F)}}{\Opp})
  \left(
  \begin{array}{c l}
  & \Strat_\Prop(\G(\At A,F),\strat_\Prop)
  \\
    \land
  & \WinStrat_\Prop(\G(\At A,F),\strat_\Prop,v,\Omega_{\At A})
  \end{array}
  \right)
\]
\end{defi}

Recall that the formulae $\Strat$ and $\WinStrat$ are defined 
in Def.~\ref{def:games:strat} (\S\ref{sec:strat})
and Def.~\ref{def:games:win} (\S\ref{sec:games:win}) respectively.
In words, the formula $F \in \Lang(\At A)$ 
of Definition~\ref{def:aut:lang}
states that
$\Prop$ has a winning strategy from position
$(\Root,\init q_{\At A})$ in the game $\G(\At A,F)$.

%%%%%%%%%%%%%%%%%%%%%%%%%%%%%%%%%%%%%%%%%%%%%%%%%%%%%%%%%%%%%%%%%%%%%%%%%%%%
%\begin{defi}[Language of an Automaton]
%%\label{def:aut:lang}
%%%%%%%%%%%%%%%%%%%%%%%%%%%%%%%%%%%%%%%%%%%%%%%%%%%%%%%%%%%%%%%%%%%%%%%%%%%%
%Given an automaton $\At A$ and a winning condition $\Omega_{\At A}$
%(in the sense of Definition~\ref{def:games:win}), we say that \emph{$\At A$ accepts $F$}
%and write $F \in \Lang(\At A)$,
%when $\Prop$ has a winning strategy from position
%$(\Root,\init q_{\At A})$ in the game $\G(\At A,F)$,
%that is when the following formula holds.
%\[
%  (\exists \funto{\strat_\Prop}{\PP\G}{\Opp})
%  \left(
%  \begin{array}{c l}
%  & \Strat_\Prop(\G,\strat_\Prop)
%  \\
%    \land
%  %& \WonGame_\Prop((V,e\restr\{\strat_\Prop\}_\Prop),v,c,n)
%  & \WinStrat_\Prop(\G,\strat_\Prop,v,\col,n)
%  \end{array}
%  \right)
%\]
%The \emph{language} of $\At A$ is the set $\Lang(\At A)$
%defined by Comprehension for Product Types
%(Theorem~\ref{thm:funto:ca})
%as the set of all $F:\Sigma$
%such that $\At A$ accepts $F$.
%\end{defi}

Except for the Simulation Theorem in~\S\ref{sec:sim},
we shall only consider automata whose acceptance conditions
are given by \emph{parity conditions}
in the sense of~\S\ref{sec:games:parity}.
Recall from Definition~\ref{def:games:parity}
that a parity condition on a game $\G$ is given by the formula
\[
\Par(\G,\col,n,U)
\]
which depends on $\G$.
However, it is desirable that automata come, as in Definition~\ref{def:aut:alt},
with acceptance conditions which are independent from any particular
acceptance game.
Note that for a given automaton $\At A$, all acceptance games $\G(\At A,F)$
have the same sets of $\Prop$ and $\Opp$ labels and positions;
the input trees $F$ can only induce different edge relations.
Recall now the games $\G(\gle)$ from Remark~\ref{rem:games:sub}.
The game $\G(\gle)$ has the same labels and positions as $\G$,
but its edge relation is exactly the partial order $\gle$ discussed
in~\S\ref{sec:pos}.
It follows that for a fixed automaton $\At A : \Sigma$,
all acceptance games $\G(\At A,F)$ for $F:\Sigma$
induce the same $\G(\At A,F)(\gle)$, that we shall write
\begin{equation}
\label{eq:aut:gle}
\G(\At A)(\gle)
\end{equation}

%%%%%%%%%%%%%%%%%%%%%%%%%%%%%%%%%%%%%%%%%%%%%%%%%%%%%%%%%%%%%%%%%%%%%%%%%%%
\begin{defi}[Parity Automata]
\label{def:aut:parity}
%%%%%%%%%%%%%%%%%%%%%%%%%%%%%%%%%%%%%%%%%%%%%%%%%%%%%%%%%%%%%%%%%%%%%%%%%%%
Let 
\[
\At A:\Sigma
\quad=\quad
(Q_{\At A} \,,\, \init q_{\At A} \,,\, \trans_{\At A} \,,\, \Omega_{\At A})
\]
We say that $\At A$ is a \emph{parity} automaton if
$\At A$ comes equipped with $\HF$-terms $n_{\At A}$ and $\col_{\At A}$
such that the two following conditions are satisfied.
\begin{enumerate}
\item The following formula holds
\[
\PAut(\At A,\col_{\At A},n_{\At A})
\quad\deq\quad
\Aut(\At A)
~~\land~~
\Ord(n_{\At A})
~~\land~~
\funto{\col_{\At A}}{Q_{\At A}}{[0,n_{\At A}]}
\]

\item The formula $\Omega_{\At A}(U)$
is $\Par(\G(\At A)(\gle),\hat \col_{\At A},n_{\At A})$,
where
\[
\hat \col_{\At A}(x,k) \quad\deq\quad
\left\{
\begin{array}{l !{\quad} l !{\qquad} l}
  \col_{\At A}(q)
& \text{if $k=q \in Q_{\At A}$}
& \text{($\Prop$-position)}
\\
  %n_{\At A}
  \col_{\At A}(q)
& \text{if $k=(q,\Conj) \in Q_{\At A} \times \Pne(\univ \times Q_{\At A})$}
& \text{($\Opp$-position)}
\end{array}
\right.
\]
\end{enumerate}

\noindent
We write
\[
\At A \quad=\quad
(Q_{\At A} \,,\, \init q_{\At A} \,,\, \trans_{\At A} \,,\, \col_{\At A} \,,\, n_{\At A})
\]
for a parity automaton $\At A$ with $\col_{\At A}$ and $n_{\At A}$ as above.
Furthermore, we write
$\Par(\At A,\hat \col_{\At A},n_{\At A},U)$
or even 
$\Par(\At A,U)$
for the formula
$\Par(\G(\At A)(\gle),\hat \col_{\At A},n_{\At A},U)$.
\end{defi}

In Definition~\ref{def:aut:parity},
the purpose of the coloring $\hat \col_{\At A}$
is to equip the game $\G(\At A)(\gle)$ with a coloring in the sense of
Def.~\ref{def:games:parity} (\S\ref{sec:games:parity}), namely a coloring of the
positions of the game, while the coloring $\col_{\At A}$
only colors the states of $\At A$.

Note that it follows 
from Remarks~\ref{rem:games:sub} and~\ref{rem:games:parity:gle}
that $\FSO$ proves
%that assuming $\PAut(\At A :\Sigma)$, $\FSOD$ proves
\[
\PAut(\At A :\Sigma)
\quad\longlimp\quad
(\forall F:\Sigma)
\bigg( \Sub\Big(\G(\At A,F) ,\, \G(\At A)(\gle) \Big) \bigg)
\]

\noindent
where the formula $\Sub(\G,\G')$ (stating that $\G$ is a subgame of $\G'$)
is defined in Def.~\ref{def:games:sub} (\S\ref{sec:games:games}),
and
\begin{multline*}
\PAut(\At A :\Sigma)
\quad\longlimp\quad
(\forall F:\Sigma)
(\forall \funto{U}{\G(\At A,F)}{\two})
\\
\bigg(
\Play\big( \G(\At A,F),(\Root, \init q_{\At A}),U \big)
~~\limp~~
\Big[
\Par\big( \G(\At A,F),\hat \col_{\At A},n_{\At A},U \big)
~~\liff~~
\Par\big(\At A,\hat\col_{\At A},n_{\At A},U \big)
\Big]
\bigg)
\end{multline*}

\noindent
The following simple fact will be useful when proving the Simulation Theorem
in~\S\ref{sec:sim}.

%%%%%%%%%%%%%%%%%%%%%%%%%%%%%%%%%%%%%%%%%%%%%%%%%%%%%%%%%%%%%%%%%%%%%%%%%%%
\begin{rem}
\label{rem:aut:propplays}
%%%%%%%%%%%%%%%%%%%%%%%%%%%%%%%%%%%%%%%%%%%%%%%%%%%%%%%%%%%%%%%%%%%%%%%%%%%
Given two plays $U$ and $V$ in $\G(\At A)(\gle)$, if $\PP U = \PP V$
then
\[
\Par\big( \G(\At A)(\gle),U \big)
~~\liff~~
\Par\big( \G(\At A)(\gle),V \big)
\]
\end{rem}

Other than the Simulation Theorem
in~\S\ref{sec:sim},
all constructions we need on automata can be performed on automata
$\At A:\Sigma$
where $\Sigma$, $Q_{\At A}$, $\init q_{\At A}$, $\trans_{\At A}$, $\col_{\At A}$
and $n_{\At A}$ are given by 
%\emph{not necessarily closed}
arbitrary
$\HF$-terms.
However, our completeness result (\S\ref{sec:compl})
ultimately relies, via Proposition~\ref{prop:bfsos:bfso},
on the completeness of $\FSOW$ over $\omega$-words (\S\ref{sec:msow})
and requires automata to be given by \emph{closed} $\HF$-terms.
In addition, our proof of the 
Simulation Theorem %in~\S\ref{sec:sim}
uses McNaughton's Theorem~\cite{mcnaughton66ic},
and imports it into $\FSO$ by Proposition~\ref{prop:bfsos:bfso},
which also requires automata to be closed objects.
This leads to the following. %notion.

%However, for the Simulation Theorem in~\S\ref{sec:sim}, we rely 
%(via Proposition~\ref{prop:bfsos:bfso}) on
%Proposition~\ref{prop:msow:bfsos}
%(\ie\@ Proposition~\ref{prop:cons:mso:cons})
%which does require to manipulate closed (and in particular $\HF$-closed) objects.
%Moreover, Completeness (\S\ref{sec:compl})
%also requires via Proposition~\ref{prop:bfsos:bfso}
%automata to be made of closed $\HF$-terms.
%%
%This leads to the following notion.

%%%%%%%%%%%%%%%%%%%%%%%%%%%%%%%%%%%%%%%%%%%%%%%%%%%%%%%%%%%%%%%%%%%%%%%%%%%
\begin{defi}
\label{def:aut:hfclosed}
%%%%%%%%%%%%%%%%%%%%%%%%%%%%%%%%%%%%%%%%%%%%%%%%%%%%%%%%%%%%%%%%%%%%%%%%%%%
A parity automaton $\At A : \Sigma$ is \emph{$\HF$-closed}
if $\Sigma$, $Q_{\At A}$, $\init q_{\At A}$, $\trans_{\At A}$,
$\col_{\At A}$ and $n_{\At A}$ are closed $\HF$-terms.
\end{defi}

%%%%%%%%%%%%%%%%%%%%%%%%%%%%%%%%%%%%%%%%%%%%%%%%%%%%%%%%%%%%%%%%%%%%%%%%%%%
\begin{rem}
\label{rem:aut:hfclosed}
%%%%%%%%%%%%%%%%%%%%%%%%%%%%%%%%%%%%%%%%%%%%%%%%%%%%%%%%%%%%%%%%%%%%%%%%%%%
For each of our constructions on automata (see Table~\ref{tab:aut:op}),
the alphabets, states and colorings of new automata
will be obtained by composing simple Functions on $\HF$-Sets
from~\S\ref{sec:ax:hf}
and Convention~\ref{conv:games:colors}.
In particular this means that the obtained automata
have $\HF$-closed alphabets, states and coloring
provided we started from $\HF$-closed ones.

On the other hand, transition functions may be more complex
(see~\S\ref{sec:neg} or~\S\ref{sec:sim}),
and we often present them in a way suggesting 
the use of the Axiom of $\HF$-Bounded Choice for $\HF$-Sets
(\S\ref{sec:ax:choice}).
This is unproblematic when $\HF$-closedness is not at issue.
To preserve $\HF$-closedness, 
starting from $\HF$-closed automata, the transition functions
of the newly built automata must always be read as being
constructed from \emph{concrete} $\HF$-sets.
\end{rem}

%%%%%%%%%%%%%%%%%%%%%%%%%%%%%%%%%%%%%%%%%%%%%%%%%%%%%%%%%%%%%%%%%%%%%%%%%%%
\begin{conv}
%%%%%%%%%%%%%%%%%%%%%%%%%%%%%%%%%%%%%%%%%%%%%%%%%%%%%%%%%%%%%%%%%%%%%%%%%%%
In the rest of this paper, whenever we speak
of a (parity) automaton $\At A$ in formal statements, we always mean
that the formula $\Aut(\At A)$ (resp.\@ $\PAut(\At A)$) holds.
(By contrast, $\HF$-closedness is an external notion.)
\end{conv}

%%%%%%%%%%%%%%%%%%%%%%%%%%%%%%%%%%%%%%%%%%%%%%%%%%%%%%%%%%%%%%%%%%%%%%%%%%%
\subsection{Substitution}
\label{sec:aut:subst}
%%%%%%%%%%%%%%%%%%%%%%%%%%%%%%%%%%%%%%%%%%%%%%%%%%%%%%%%%%%%%%%%%%%%%%%%%%%
Let $\At A:\Sigma$ be an automaton and let $\Gamma$ and
$\funto{f}{\Gamma}{\Sigma}$ be $\HF$-sets.
The automaton $\At A[f] : \Gamma$ is defined to have
the same states and acceptance condition as $\At A : \Sigma$,
and its transitions are given
by
\[
(q,\al b) 
\quad\longmapsto\quad
\trans_{\At A}(q,f(\al b))
\]

\noindent
Note that $\Aut(\At A) \land (\exists \al b \In \Gamma)$ implies $\Aut(\At A[f])$.
Also, $\At A[f]$ is a parity automaton
whenever $\At A$ is.
Furthermore, it follows from Remark~\ref{rem:aut:hfclosed} that
$\At A[f] : \Gamma$ is $\HF$-closed when $\At A:\Gamma$
is $\HF$-closed
and in addition $\Gamma$ and $f$ are closed $\HF$-terms.
A typical use of substitution,
on which we rely when translating formulae to automata in~\S\ref{sec:compl:aut},
is to enlarge the input alphabet of an automaton.
For instance, given $\HF$-closed $\Sigma_1,\dots,\Sigma_n$
and an $\HF$-closed $\At A : \Sigma_i$,
we obtain an $\HF$-closed
\[
\At A[\pi^n_i] : \Sigma_1 \times \dots \times \Sigma_n
\]
where
\[
\pi^n_i ~~:~~ \Sigma_1 \times \dots \times \Sigma_n ~~\longto~~ \Sigma_i
\]
is a projection $\HF$-Function of~\S\ref{sec:ax:hf}.\ref{item:ax:hf:fun:prod}.

%%%%%%%%%%%%%%%%%%%%%%%%%%%%%%%%%%%%%%%%%%%%%%%%%%%%%%%%%%%%%%%%%%%%%%%%%%%
\begin{lem}
\label{lem:aut:subst}
%%%%%%%%%%%%%%%%%%%%%%%%%%%%%%%%%%%%%%%%%%%%%%%%%%%%%%%%%%%%%%%%%%%%%%%%%%%
Given $\Gamma$, $f$ and $\At A$ as above,
$\FSOD$ proves the following. %assuming $\Aut(\At A)$.
\[
(\forall H:\Gamma)
\Bigg[
H \in \Lang(\At A[f])
~~\liff~~
(\forall F:\Sigma)
\bigg(
(\forall x) \Big[ F(x) = f(H(x)) \Big]
~\limp~
F \in \Lang(\At A)
\bigg)
\Bigg]
\]
\end{lem}

\noindent
Note that by $\HF$-Bounded Choice, $\FSOD$ proves that
\[
(\forall H:\Gamma)
(\exists F:\Sigma)
(\forall x)
\Big(
F(x) = f(H(x))
\Big)
\]
so the above Lemma could have equivalently been stated with an existentially bound $F$.

%%%%%%%%%%%%%%%%%%%%%%%%%%%%%%%%%%%%%%%%%%%%%%%%%%%%%%%%%%%%%%%%%%%%%%%%%%%
\subsection{Disjunction}
\label{sec:aut:disj}
%%%%%%%%%%%%%%%%%%%%%%%%%%%%%%%%%%%%%%%%%%%%%%%%%%%%%%%%%%%%%%%%%%%%%%%%%%%
We use here the $\HF$-Functions from~\S\ref{sec:ax:hf}.\ref{item:ax:hf:coprod}
and Convention~\ref{conv:games:colors}.\ref{item:games:colors:arith}.
Given parity automata $\At A_0,\At A_1 : \Sigma$,
the parity automaton $\At A_0 \oplus \At A_1 : \Sigma$
has state set
\[
Q_{\At A_0} + Q_{\At A_1} + \{\init q\}
\]
with $\init q$ initial,
transitions given 
by
\[
\begin{array}{r !{\quad}c!{\quad} l !{\qquad} l}
  (\init q,\al a)
& \longmapsto
& \trans_{\At A_0}(\init q_{\At A_0},\al a) + \trans_{\At A_1}(\init q_{\At A_1},\al a)
& \text{(modulo $Q_{\At A_i} \hookrightarrow Q_{\At A_0 \oplus \At A_1}$)}
\\
  (q_{\At A_i},\al a)
& \longmapsto
& \trans_{\At A_i}(q_{\At A_i}, \al a)
& \text{(for $q_{\At A_i} \in Q_{\At A_i}$)} 
\end{array}
\]
and coloring $\funto{\col}{Q_{\At A_0 \oplus \At A_1}}{[0,n]}$
(where $n = \max(n_{\At A_0},n_{\At A_1})$)
given by
\[
\begin{array}{r !{\quad\deq\quad} l !{\qquad} l}
  \col(\init q)
& n
\\
  \col(q_{\At A_i})
& \col_{\At A_i}(q_{\At A_i})
& \text{(for $q_{\At A_i} \in Q_{\At A_i}$)} 
\end{array}
\]

\noindent
We have
\[
  \Aut(\At A_0)
\quad\longlimp\quad
  \Aut(\At A_1)
\quad\longlimp\quad
  \Aut(\At A_0 \oplus \At A_1)
\]
Moreover, if follows from Remark~\ref{rem:aut:hfclosed} that
$\At A_0 \oplus \At A_1 : \Sigma$ is 
$\HF$-closed
whenever $\At A_0 : \Sigma$ and $\At A_1 : \Sigma$ are.

%%%%%%%%%%%%%%%%%%%%%%%%%%%%%%%%%%%%%%%%%%%%%%%%%%%%%%%%%%%%%%%%%%%%%%%%%%%
\begin{rem}
%%%%%%%%%%%%%%%%%%%%%%%%%%%%%%%%%%%%%%%%%%%%%%%%%%%%%%%%%%%%%%%%%%%%%%%%%%%
Even in our positional setting,
strictly speaking the automaton $\At A_0 \oplus \At A_1$ does not
require $\At A_0$ and $\At A_1$ to be \emph{parity} automata
(see Table~\ref{tab:aut:op}).
However, the acceptance condition of $\At A_0 \oplus \At A_1$
is actually simpler to define when both $\At A_0$ and $\At A_1$
are parity automata.
Since we shall only need $\At A_0 \oplus \At A_1$
for parity automata,
we only formally define disjunction
%we only define its acceptance condition 
in this setting.
\end{rem}

%%%%%%%%%%%%%%%%%%%%%%%%%%%%%%%%%%%%%%%%%%%%%%%%%%%%%%%%%%%%%%%%%%%%%%%%%%%
\begin{lem}
\label{lem:aut:disj}
%%%%%%%%%%%%%%%%%%%%%%%%%%%%%%%%%%%%%%%%%%%%%%%%%%%%%%%%%%%%%%%%%%%%%%%%%%%
%Given $\At A_0,\At A_1:\Sigma$, 
$\FSOD$ proves the following.
\[
(F:\Sigma)
\Big(
F \in \Lang(\At A_0 \oplus \At A_1)
~~\longliff~~
\Big(
F \in \Lang(\At A_0)
~\lor~
F \in \Lang(\At A_1)
\Big)
\Big)
\]
\end{lem}

\begin{proof}
Assume first that $F \in \Lang(\At A_0 \oplus \At A_1)$
for $F:\Sigma$,
and consider a winning $\Prop$-strategy
$\strat$ in the acceptance game $\G(\At A_0 \oplus \At A_1,F)$.
We first look at the move of $\strat$ on the initial
position $(\Root,\init q)$.
By definition of $\At A_0 \oplus \At A_1$ we have
\[
\strat(\Root,\init q) = (q,\Conj)
\qquad\text{with}\qquad
\Conj \in 
\Big(
  \trans_{\At A_0}\big( \init q_{\At A_0},F(\Root) \big)
+ \trans_{\At A_1}\big( \init q_{\At A_1},F(\Root) \big)
\Big)
\]

\noindent
Assume $\Conj \in \trans_{\At A_i}(\init q_{\At A_i},F(\Root))$.
Then $\strat$ induces a $\Prop$-strategy $\strat_i$ in
$\G(\At A_i,F)$.
The strategy $\strat_i$ is defined using
$\HF$-Bounded Choice for Product Types (Theorem~\ref{thm:funto:choice})
by putting
\[
\strat_i(x,q_{\At A_i})
\quad=\quad
\left\{
\begin{array}{l l}
  \strat(\Root,\init q)
& \text{if $(x,q_{\At A_i}) = (\Root,\init q)$}
\\
  \strat(x,q_{\At A_i})
& \text{otherwise}
\end{array}
\right.
\]

\noindent
It remains to show that $\strat_i$ is winning, that is
\[
(\forall \funto{V}{\G(\At A_i,F)}{\two})
\Big(
\Play(\strat_i,\iota,V)
~~\limp~~
\Par(\At A_i,V)
\Big)
\]
for $\iota = (\Root,\init q_{\At A_i})$.
Consider an infinite play $V$ of $\strat_i$ from $\iota$.
Then by Comprehension for Product Types (Theorem~\ref{thm:funto:ca}),
define $\funto{U}{\G(\At A_0 \oplus \At A_1,F)}{\two}$
as the set of all $(x,\ell)$ such that either
$(x,\ell) = (\Root,\init q)$ or $(x,\ell) \in V$.
It is clear that
\[
\Par(\At A_i, V)
~~\liff~~
\Par(\At A_1 \oplus \At A_2,U)
\]

The converse is proved similarly.
\end{proof}

%%%%%%%%%%%%%%%%%%%%%%%%%%%%%%%%%%%%%%%%%%%%%%%%%%%%%%%%%%%%%%%%%%%%%%%%%%%
\subsection{Non-Deterministic Automata}
\label{sec:aut:nd}
%%%%%%%%%%%%%%%%%%%%%%%%%%%%%%%%%%%%%%%%%%%%%%%%%%%%%%%%%%%%%%%%%%%%%%%%%%%
We turn to the important class of alternating automata known as
\emph{non-deterministic} automata.
Non-deterministic automata are important because they allow us,
via the usual \emph{projection} operation (\S\ref{sec:nd:proj}),
to interpret the existential quantifier of $\MSO$ (see~\S\ref{sec:compl}).
An important result in the theory of automata on infinite
trees is the \emph{Simulation Theorem}~\cite{ej91focs,ms95tcs}
(addressed in~\S\ref{sec:sim}),
stating that each alternating automata can be simulated by a non-deterministic one.

Intuitively, an automaton $\At A$ is non-deterministic
if in acceptance games $\Opp$ can only explicitly choose tree directions
but not states.
%The formal definition is as follows.

%%%%%%%%%%%%%%%%%%%%%%%%%%%%%%%%%%%%%%%%%%%%%%%%%%%%%%%%%%%%%%%%%%%%%%%%%%%
\begin{defi}[Non-Deterministic Automata]
\label{def:aut:nd}
%%%%%%%%%%%%%%%%%%%%%%%%%%%%%%%%%%%%%%%%%%%%%%%%%%%%%%%%%%%%%%%%%%%%%%%%%%%
An automaton $(\At A :\Sigma)$ in the sense of Definition~\ref{def:aut:alt},
with
\[
\funto{\trans_{\At A}}{Q_{\At A} \times \Sigma}{\Pne(\Pne(\Dir \times Q_{\At A}))}
\]
is \emph{non-deterministic} if for every $q \in Q_{\At A}$,
every $\al a \in \Sigma$, every $\Conj \in \trans_{\At A}(q,\al a)$, and every tree direction $d \in \Dir$, there is at most one
$q' \in Q_{\At A}$ such that $(d,q') \in \Conj$.
%the following holds:
%\begin{itemize}
%\item for every tree direction $d \in \Dir$, there is at most one
%$q' \in Q_{\At A}$ such that $(d,q') \in \Conj$.
%\end{itemize}

%\noindent
%A non-deterministic automaton $\At A$ is \emph{complete} if the condition
%above holds for \emph{exactly one} $q' \in Q_{\At A}$ for each $d \in \Dir$
%and each $\Conj$ in the range of $\trans_{\At A}$.
\end{defi}

%The notion of complete automaton actually makes sense for
%alternating automata, but we shall only need it for non-deterministic ones.
%By possibly adding an accepting state, each (non-deterministic)
%automaton is equivalent to a complete (non-deterministic) automaton.
%
%%%%%%%%%%%%%%%%%%%%%%%%%%%%%%%%%%%%%%%%%%%%%%%%%%%%%%%%%%%%%%%%%%%%%%%%%%%%
%\begin{lem}
%\label{lem:aut:nd:compl}
%%%%%%%%%%%%%%%%%%%%%%%%%%%%%%%%%%%%%%%%%%%%%%%%%%%%%%%%%%%%%%%%%%%%%%%%%%%%
%For each non-deterministic automaton $\At A : \Sigma$
%one can define a non-deterministic \emph{complete} $\At C : \Sigma$
%such that $\Lang(\At A) = \Lang(\At C)$.
%\end{lem}

\noindent
The key property of non-deterministic automata
is that in each play of a $\Prop$-strategy $\strat$ in an acceptance game,
the sequence of states is uniquely determined from the tree positions.
We formally state this as follows.

%The key property of non-deterministic automata is the following.
%%that in one $\Prop$-strategy
%%$\strat$, for every tree position $x \in \univ$, there is at most one play
%%$V$ of $\strat$ such that $(x,q) \in V$ for some $q \in Q_{\At A}$.

%%%%%%%%%%%%%%%%%%%%%%%%%%%%%%%%%%%%%%%%%%%%%%%%%%%%%%%%%%%%%%%%%%%%%%%%%%%
\begin{lem}
\label{lem:aut:nd:unique}
%%%%%%%%%%%%%%%%%%%%%%%%%%%%%%%%%%%%%%%%%%%%%%%%%%%%%%%%%%%%%%%%%%%%%%%%%%%
Consider a non-deterministic automaton $\At A : \Sigma$,
and let $F:\Sigma$.
Furthermore let $\strat$ be a $\Prop$-strategy in $\G(\At A,F)$.
Then $\FSOD$ proves that 
for all $x \in \univ$
and all infinite plays $V$ and $V'$ of $\strat$, if
\[
(\exists q \in Q_{\At A})
%\Big[
(x,q) \in V
%\Big]
~~\land~~
(\exists q' \in Q_{\At A})
%\Big[
(x,q') \in V'
%\Big]
%\exists q \in Q_{\At A} \left(
%   V : (\Root,\init q_{\At A}) \edge{*}{e_{\At A}\restr\{\strat\}} (x,q)
%  \right)
%\quad\land\quad
%\exists q' \in Q_{\At A} \left(
%   V' : (\Root,\init q_{\At A}) \edge{*}{e_{\At A}\restr\{\strat\}} (x,q')
%  \right)
\]
then
for all $y \Leq x$, all $q \in Q_{\At A}$,
and all $\Conj \in \Pne(\Dir \times Q_{\At A})$,
we have
\[
%\forall y \Leq x~ \forall q \in Q_{\At A}~ \forall \Conj \in \Pne(\Dir \times Q_{\At A})
%\left(
\Big[ (y,q) \in V ~~\liff~~ (y,q) \in V' \Big]
\quad\land\quad
\Big[ (y,(q,\Conj)) \in V ~~\liff~~ (y,(q,\Conj)) \in V' \Big]
%\right)
\]
\end{lem}

\begin{proof}
Fix $\strat$ and $V,V'$ as in the statement of the Lemma
and let $x \in \univ$.
First, note that for every $y \Leq x$
we have

%%%%%%%%%%%%%%%%%%%%%%%%%%%%%%%%%%%%%%%%%%%%%%%%%%%%%%%%%%%%%%%%%%%%%%%%%%%
\begin{subclm}
\label{clm:aut:nd:unique:exists}
%%%%%%%%%%%%%%%%%%%%%%%%%%%%%%%%%%%%%%%%%%%%%%%%%%%%%%%%%%%%%%%%%%%%%%%%%%%
\[
(\exists q \in Q_{\At A}) \big( (y,q) \in V \big)
~~\land~~
(\exists q \in Q_{\At A}) \big( (y,q) \in V' \big)
\]
\end{subclm}

\begin{subproof}[Proof of Claim \thesubclm]
We use the Induction Axiom of $\FSO$ (\S\ref{sec:ax:ind}).
The property holds for $\Root \Leq x$ since
$(\Root,\init q_{\At A})$ belongs to both $V$ and $V'$.
Assume now the property for $y \Leq x$, and consider some tree direction $d \in \Dir$
such that $\Succ_d(y) \Leq x$.
By assumption, we have some $q \in Q_{\At A}$ such that $(y,q) \in V$,
and by using $\Game(\strat)$ twice, 
we get some $q' \in Q_{\At A}$ and some $d' \in \Dir$
such that
\[
\exists q'\left(
(y,q)
\edge{}{\strat}
\edge{}{\strat}
(\Succ_{d'}(y),q')
\right)
\]

\noindent
But since $V$ is a play of $\strat$, 
by Proposition~\ref{prop:games:edges}
we must have $\Succ_{d'}(y) \gle x$, so that $d' = d$
and we are done.
The same reasoning gives the result for $V'$.
\end{subproof}

Using the Induction Axiom of $\FSO$ (\S\ref{sec:ax:ind}), we now show that
\[
(\forall y \Leq x) (\forall q \in Q_{\At A})
\Big[
(y,q) \in V ~~\liff~~ (y,q) \in V'
\Big]
\]
First, 
we have
\[
(\Root,q) \in V
\quad\liff\quad
(\Root,q)\in V'
\quad\liff\quad
q = \init q_{\At A}
\]

\noindent
Assume now the property for $y \Leq x$ and let us prove it for
$\Succ_d(y)$ with $\Succ_d(y) \Leq x$.
It follows from the induction hypothesis and
Claim~\ref{clm:aut:nd:unique:exists}
that
we have $(y,q) \in V$ and $(y,q) \in V'$
for some $q \in Q_{\At A}$.
Again by Claim~\ref{clm:aut:nd:unique:exists},
let $q',q'' \in Q_{\At A}$ such that
$(\Succ_d(y),q') \in V$ and $(\Succ_d(y),q'') \in V'$.
Now since $V$ and $V'$ are plays of $\strat$,
there are $\Conj$, $\Conj'$ such that
$(y,(q,\Conj)) \in V$ and $(y,(q,\Conj')) \in V'$,
and
we necessarily have
\[
(q,\Conj) \quad=\quad (q,\Conj') \quad=\quad \strat(x,q)
\]
so that $\Conj = \Conj'$.
Moreover, we have $(d,q'),(d,q'') \in \Conj$,
but this implies $q' = q''$ since $\At A$ is non-deterministic.

This concludes the proof of Lemma~\ref{lem:aut:nd:unique}.
\end{proof}

\cnote{\CR:NOTE\quad It is not possible to
strengthen Corollary~\ref{cor:aut:nd:unique}
to an ``exists unique''.
This would require all non-deterministic automata to be ``total'',
and therefore to allow for empty sets as states of $\oc \At A$
in~\S\ref{sec:sim}}

%%%%%%%%%%%%%%%%%%%%%%%%%%%%%%%%%%%%%%%%%%%%%%%%%%%%%%%%%%%%%%%%%%%%%%%%%%%
\begin{cor}
\label{cor:aut:nd:unique}
%%%%%%%%%%%%%%%%%%%%%%%%%%%%%%%%%%%%%%%%%%%%%%%%%%%%%%%%%%%%%%%%%%%%%%%%%%%
Given $\At A$, $F$ and $\strat$ as in Lemma~\ref{lem:aut:nd:unique},
$\FSOD$ proves that for each $x \in \univ$ there is at most
one $q \in Q_{\At A}$
such that
\[
(\exists \funto{U}{\G(\At A,F)}{\two})
\Big(
  \Play(\strat,\, (\Root,\init q_{\At A}),\, U)
~~\land~~
  (x,q) \in U
\Big)
\]
\end{cor}

We now state the \emph{Simulation Theorem}~\cite{ej91focs,ms95tcs}.
Its proof in $\FSOD$, requiring $\HF$-closedness of automata,
is deferred to~\S\ref{sec:sim}.

%%%%%%%%%%%%%%%%%%%%%%%%%%%%%%%%%%%%%%%%%%%%%%%%%%%%%%%%%%%%%%%%%%%%%%%%%%%
\begin{thm}[Simulation]
\label{thm:aut:nd:sim}
%%%%%%%%%%%%%%%%%%%%%%%%%%%%%%%%%%%%%%%%%%%%%%%%%%%%%%%%%%%%%%%%%%%%%%%%%%%
For each $\HF$-closed parity automaton $\At A : \Sigma$
there is a non-deterministic $\HF$-closed parity automaton
$\ND(\At A) : \Sigma$
such that 
%$\FSOD$ proves that 
\[
\FSO
\thesis\quad \Lang(\ND(\At A)) = \Lang(\At A)
\]
\end{thm}

%%%%%%%%%%%%%%%%%%%%%%%%%%%%%%%%%%%%%%%%%%%%%%%%%%%%%%%%%%%%%%%%%%%%%%%%%%%
\subsection{Projection}
\label{sec:nd:proj}
%%%%%%%%%%%%%%%%%%%%%%%%%%%%%%%%%%%%%%%%%%%%%%%%%%%%%%%%%%%%%%%%%%%%%%%%%%%
We now discuss the usual operation of \emph{projection},
which allows us to interpret (existential) quantification in $\MSO$
(see~\S\ref{sec:compl:aut}).
This operation is defined on arbitrary alternating automata,
but it only correctly computes the appropriate projection for non-deterministic ones.

Given an automaton $\At A : \Sigma \times \Gamma$
as in Definition~\ref{def:aut:alt},
we define its \emph{projection} on $\Sigma$
to be the automaton $\exists_\Gamma \At A : \Sigma$ with
\[
\exists_\Gamma \At A
\quad\deq\quad
(Q_{\At A},\, \init q_{\At A},\, \trans_{\exists_\Gamma \At A},\,
 \col_{\At A},\, n_{\At A})
\]
where
\[
\trans_{\exists_\Gamma \At A} ~~:~~ 
Q_{\At A} \times \Sigma
~~\longto~~
\Pne(\Pne(\Dir \times Q_{\At A}))
\]
is given
by
\[
\trans_{\exists_\Gamma \At A}(q,\al a)
\quad\deq\quad
\bigcup_{\al b \in \Gamma}
\trans_{\At A}(q,(\al a,\al b))
\]
%Note that $\exists_\Gamma \At A$ is complete if $\At A$ is complete.

\noindent
Note that $\Aut(\At A : \Sigma \times \Gamma)$
implies $\Aut(\exists_\Gamma \At A)$.
Moreover, $\exists_\Gamma \At A : \Sigma$ is an ($\HF$-closed) parity
automaton whenever so is $\At A : \Sigma \times \Gamma$.

We shall now prove that $\exists_\Gamma \At A : \Sigma$
indeed implements the projection of $\At A : \Sigma \times \Gamma$.
This involves a notion of pairing for trees.
Given 
$F:\Sigma$ and $G:\Gamma$,
we let
$\pair{F,G} : \Gamma \times \Sigma$
be given
(using the axiom of $\HF$-Bounded Choice for $\HF$-Functions~(\S\ref{sec:ax:choice}))
by
\[
\pair{F,G}(x) \quad\deq\quad (F(x),G(x))
\]

%%%%%%%%%%%%%%%%%%%%%%%%%%%%%%%%%%%%%%%%%%%%%%%%%%%%%%%%%%%%%%%%%%%%%%%%%%%
\begin{prop}
\label{prop:aut:nd:proj}
%%%%%%%%%%%%%%%%%%%%%%%%%%%%%%%%%%%%%%%%%%%%%%%%%%%%%%%%%%%%%%%%%%%%%%%%%%%
Consider a \emph{non-deterministic} $\At A : \Sigma \times \Gamma$ and 
let $\exists_\Gamma \At A : \Sigma$
be as defined above.
Then $\FSOD$ proves the following.
\[
(\forall F:\Sigma)
\bigg[
F \in \Lang(\exists_\Gamma \At A)
~~\liff~~
(\exists G:\Gamma)
\Big(
\pair{F,G} \in \Lang(\At A)
\Big)
\bigg]
\]
\end{prop}

\begin{proof}
Given $G:\Gamma$
and a winning $\Prop$-strategy $\strat$ on $\G(\At A,\pair{F,G})$,
it is easy to see that $\strat$ is also a winning strategy on
$\G(\exists_\Gamma \At A,F)$.

Conversely, assume that $\strat$ is a winning $\Prop$-strategy on
$\G(\exists_\Gamma \At A,F)$.
%We invoke Remark~\ref{rem:ax:hf:well-order-hf}
%and fix a well-order $\preceq_\Gamma$ on $\Gamma$
%and a well-order $\preceq_{\Pne(\Dir \times Q_{\At A})}$
%on $\Pne(\Dir \times Q_{\At A})$.
%
We define a tree $G:\Gamma$
by $\HF$-Bounded Choice for $\HF$-Functions~(\S\ref{sec:ax:choice})
as follows:
\begin{itemize}
\item For $x \in \univ$, if there is some infinite play
$U$ of $\strat$ such that $(x,q) \in U$ for some state $q \in Q_{\At A}$,
then we let $G(x)$ be some $\al b \in \Gamma$
such that $\strat(x,q) \in \trans_{\At A}(q,(F(x),\al b))$.

\item Otherwise, we let $G(x)$ be any element of $\Gamma$.
\end{itemize}

\noindent
We now define a $\Prop$-strategy $\strat_G$ on $\G(\At A,\pair{F,G})$
as follows,
again using
$\HF$-Bounded Choice for $\HF$-Functions~(\S\ref{sec:ax:choice}).
\begin{itemize}
\item %Given $(x,q) \in \PP{\G(\At A)}$,
If $(x,q) \in U$ for some infinite play $U$ of $\strat$, then we let
$\strat_G(x,q) \deq \strat(x,q)$.

\item Otherwise, we let $\strat_G(x,q) = (q,\Conj)$, where
$\Conj \in \trans_{\At A}(q,\pair{F,G}(x))$.
\end{itemize}

We first check that $\strat_G$ is indeed a strategy on $\G(\At A,\pair{F,G})$,
namely that for all $(x,q) \in \univ \times Q_{\At A}$,
if $\strat_G(x,q) = (q,\Conj)$ then 
$\Conj \in \trans_{\At A}(q,\pair{F,G}(x))$.
If $(x,q)$ belongs to no infinite play of $\strat$, then the result follows
by definition of $\strat_G$.
Otherwise, by Corollary~\ref{cor:aut:nd:unique},
$q$ is unique in $Q_{\At A}$ such that $(x,q)$ belongs to an infinite play
of $\strat$, and we are done since
\[
\strat_G(x,q) = \strat(x,q) \in \trans_{\At A}(q,\pair{F,G}(x))
\]

In order to show that $\strat_G$ is winning, we show that any infinite play
of $\strat_G$ is also an infinite play of $\strat$.
So let $\funto{U}{\G(\At A,\pair{F,G})}{\two}$
such that
\[
\Play(\strat_G,\, (\Root,\init q_{\At A}),\, U)
\]

\noindent
We are done if we show that
\[
(\forall (x,q) \in U)
\left(
\strat(x,q) = \strat_G(x,q)
\right)
\]
which follows from the fact that

%%%%%%%%%%%%%%%%%%%%%%%%%%%%%%%%%%%%%%%%%%%%%%%%%%%%%%%%%%%%%%%%%%%%%%%%%%%
\begin{subclm}
%%%%%%%%%%%%%%%%%%%%%%%%%%%%%%%%%%%%%%%%%%%%%%%%%%%%%%%%%%%%%%%%%%%%%%%%%%%
\[
(\forall (x,q) \in U)
(\exists \funto{W}{\G(\At A,\pair{F,G})}{\two})
\Big(
\Play(\strat,\, (\Root,\init q_{\At A}),\, W)
~~\land~~ (x,q) \in W
\Big)
\]
\end{subclm}

\begin{subproof}[Proof of Claim \thesubclm]
We apply the Induction Scheme of $\FSOD$
(\S\ref{sec:ax:ind}).
In the base case $x = \Root$,
and we conclude by Lemma~\ref{lem:games:infplay:pos}.

For the induction step consider the case of $\Succ_d(x)$,
assuming the property for $x$.
So let $q' \in Q_{\At A}$ such that $(\Succ_d(x),q') \in U$.
First, by applying
twice the Predecessor Lemma~\ref{lem:games:predplays} for Infinite Plays,
we get some $q \in Q_{\At A}$ such that $(x,q) \in U$,
and by induction hypothesis, there is some infinite play $W$ of $\strat$
such that $(x,q) \in W$.
But then, by definition of $\strat_G$, we have $\strat(x,q) = \strat_G(x,q)$.
We thus have $(d,q') \in \Conj$, where $(q,\Conj) = \strat(x,q)$.
Using Lemma~\ref{lem:games:infplay:pos}, let now $W'$
be an infinite play of $\strat$ from position $(\Succ_d(x),q')$.
By Comprehension for Product Types
(Theorem~\ref{thm:funto:ca}),
we define an infinite play $W''$ of $\strat$ from position $(\Root,\init q_{\At A})$
as follows:
\begin{itemize}
\item
Given $u$ a position of $\G(\At A,\pair{F,G})$, if
$u \in W'$ then $u \in W''$.
Otherwise, we let $u \in W''$ iff $u \in W$ and
$u \edge{*}{\strat} (\Succ_d(x),q')$.
\end{itemize}
It is then easy to check that $W''$ is an infinite play of $\strat$.
\end{subproof}

This concludes the proof of Proposition~\ref{prop:aut:nd:proj}.
\end{proof}

%%% Local Variables:
%%% mode: latex
%%% TeX-master: "main.tex"
%%% End:

%%%%%%%%%%%%%%%%%%%%%%%%%%%%%%%%%%%%%%%%%%%%%%%%%%%%%%%%%%%%%%%%%%%%%%%%%%%
\subsection{Complementation}
\label{sec:neg}
%%%%%%%%%%%%%%%%%%%%%%%%%%%%%%%%%%%%%%%%%%%%%%%%%%%%%%%%%%%%%%%%%%%%%%%%%%%
It is known that, assuming the determinacy of acceptance games,
alternating tree automata are closed under complement~\cite{ms87tcs}.
On the other hand, our setting only allows us to manipulate
\emph{positional} strategies on acceptance games,
which leads us to formulate complementation for \emph{parity} automata,
since their acceptance games are always positionally determined.
Thus, in this section, we formalize the fact that, assuming the axiom
$(\PosDet)$, each alternating parity automaton has a complement in $\FSO$.
More precisely, we prove the following.

%%%%%%%%%%%%%%%%%%%%%%%%%%%%%%%%%%%%%%%%%%%%%%%%%%%%%%%%%%%%%%%%%%%%%%%%%%%
\begin{thm}[Complementation of Tree Automata]
\label{thm:neg}
%%%%%%%%%%%%%%%%%%%%%%%%%%%%%%%%%%%%%%%%%%%%%%%%%%%%%%%%%%%%%%%%%%%%%%%%%%%
For each ($\HF$-closed) parity automaton $\At A : \Sigma$,
there is an ($\HF$-closed) parity automaton $\aneg \At A : \Sigma$
such that 
\[
\FSO + (\PosDet) \thesis\quad
(\forall F:\Sigma)
\Big(
F \in \Lang(\aneg \At A) ~~\liff~~
F \notin \Lang(\At A)
\Big)
\]
\end{thm}

Alternating automata may be directly complemented in a locally syntactic fashion.
For an automaton $\At A:\Sigma$ we may define a complement automaton
$\aneg \At A:\Sigma$ 
with the same states as $\At A$, and such that $\Prop$-strategies
in acceptance games for $\aneg \At A$ correspond
(\wrt\@ the visited states in infinite plays)
to $\Opp$-strategies in acceptance games for $\At A$, and vice-versa.
Closely following~\cite{walukiewicz02tcs},
the basic idea 
%for complementing an alternating automaton $(\At A:\Sigma)$
is to see the transition function of $\At A$
\[
\trans_{\At A} ~~:~~ 
{Q_{\At A} \times \Sigma}
~~\longto~~ {\Pne(\Pne(\Dir \times Q_{\At A}))}
\]
as taking $(q,\al a)$
to the disjunctive normal form
\[
\bigdisj\limits_{\Conj \in \trans_{\At A}(q, \al a)}
\bigconj\limits_{(d,q') \in \Conj} (d,q')
\]
Then, for the complement $\aneg \At A:\Sigma$ of $\At A$,
we can let
\[
{\trans_{\aneg \At A}} ~~:~~
{Q_{\At A} \times \Sigma}
~~\longto~~
{\Pne(\Pne(\Dir \times Q_{\At A}))}
\]
take $(q,\al a)$
to the De Morgan dual of $\trans_{\At A}(q,\al a)$.

We now proceed to the formal definition.

%%%%%%%%%%%%%%%%%%%%%%%%%%%%%%%%%%%%%%%%%%%%%%%%%%%%%%%%%%%%%%%%%%%%%%%%%%%
\begin{defi}
\label{def:neg}
%%%%%%%%%%%%%%%%%%%%%%%%%%%%%%%%%%%%%%%%%%%%%%%%%%%%%%%%%%%%%%%%%%%%%%%%%%%
Given a parity automaton $\At A : \Sigma$, we define
the parity automaton
$\aneg \At A :\Sigma$ %following~\cite{walukiewicz02tcs}.
as follows.
The automaton $\aneg \At A$ has the same states and initial state as $\At A$.
Its transitions are defined 
as
\[
\trans_{\aneg \At A}(q,a)
\quad\deq\quad
\Big\{ \dual\Conj \in \Pne(\Dir \times Q_{\At A})
  ~~\st~~
  (\forall \Conj \in \trans_{\At A}(q,a))
    \big( \dual\Conj \cap \Conj \neq \emptyset \big)
\Big\}
\]
Its coloring is given as follows,
using Convention~\ref{conv:games:colors}.\ref{item:games:colors:arith}:
\[
\col_{\aneg \At A}(q) ~~\deq~~
\col_{\At A}(q) + 1
\]
\end{defi}

Note that by Remark~\ref{rem:aut:hfclosed},
$\aneg \At A:\Sigma$ is $\HF$-closed
whenever so is $\At A :\Sigma$.
We are now going to prove Theorem~\ref{thm:neg}.
To this end, fix a parity automaton $\At A :\Sigma$
and let $\aneg \At A : \Sigma$ be as in Definition~\ref{def:neg}.
Fix also some $F:\Sigma$.
We split Theorem~\ref{thm:neg} into the following statements.
%%%%%%%%%%%%%%%%%%%%%%%%%%%%%%%%%%%%%%%%%%%%%%%%%%%%%%%%%%%%%%%%%%%%%%%%%%%
\begin{prop}
\label{prop:neg:toneg}
%%%%%%%%%%%%%%%%%%%%%%%%%%%%%%%%%%%%%%%%%%%%%%%%%%%%%%%%%%%%%%%%%%%%%%%%%%%
%$\FSO + (\PosDet)$ proves
\(
\FSO + (\PosDet) \thesis~~
F \notin \Lang(\At A) ~~\longlimp~~
F \in \Lang(\aneg \At A)
\).
\end{prop}

%%%%%%%%%%%%%%%%%%%%%%%%%%%%%%%%%%%%%%%%%%%%%%%%%%%%%%%%%%%%%%%%%%%%%%%%%%%
\begin{prop}
\label{prop:neg:fromneg}
%%%%%%%%%%%%%%%%%%%%%%%%%%%%%%%%%%%%%%%%%%%%%%%%%%%%%%%%%%%%%%%%%%%%%%%%%%%
\(
\FSO\thesis~~
F \in \Lang(\aneg \At A) ~~\longlimp~~
F \notin \Lang(\At A)
\).
\end{prop}

The key
%for Propositions~\ref{prop:neg:toneg} and~\ref{prop:neg:fromneg}
is that
$\Prop$-strategies on $\G(\aneg \At A,F)$
correspond to $\Opp$-strategies on $\G(\At A,F)$, and vice-versa.
We make this formal in~\S\ref{sec:neg:toneg} and~\S\ref{sec:neg:fromneg} below.
First, notice that $Q_{\aneg \At A} = Q_{\At A}$,
so that
the games $\G(\At A,F)$ and $\G(\aneg \At A,F)$
have the same sets of labels
\[
\Prop ~~\deq~~ Q_{\At A}
\qquad\text{and}\qquad
\Opp ~~\deq~~ Q_{\At A} \times \Pne(\Dir \times Q_{\At A})
\]
In the following, we let
\[
\G ~~\deq~~ \univ \times \PO{}
\]
be the set of positions of the games
$\G(\At A,F)$ and $\G(\aneg \At A,F)$,
and we let
$\iota \deq (\Root,\init q_{\At A})$
be their (common) initial position.

%%%%%%%%%%%%%%%%%%%%%%%%%%%%%%%%%%%%%%%%%%%%%%%%%%%%%%%%%%%%%%%%%%%%%%%%%%%
\subsubsection{Proof of Proposition~\ref{prop:neg:toneg}.}
\label{sec:neg:toneg}
%%%%%%%%%%%%%%%%%%%%%%%%%%%%%%%%%%%%%%%%%%%%%%%%%%%%%%%%%%%%%%%%%%%%%%%%%%%
We are going to show that
$\FSO + (\PosDet)$ proves
\[
F \notin \Lang(\At A)
~~\longlimp~~
F \in \Lang(\aneg \At A)
\]

First, 
given an $\Opp$-strategy $\strat_\Opp$ on $\G(\At A,F)$,
we define a $\Prop$-strategy $\strat_\Prop$ on $\G(\aneg \At A,F)$.
Assuming that 
%the strategy
$\strat_\Opp$
satisfies
$\Strat_\Opp(\G(\At A,F),\strat_\Opp)$,
the strategy $\strat_\Prop$
will satisfy $\Strat_\Prop(\G(\aneg \At A,F),\strat_\Prop)$.
Recall that this in particular means
\[
\funto{\strat_\Opp}{\OP\G}{\Dir \times \Prop}
\qquad\text{and}\qquad
\funto{\strat_\Prop}{\PP\G}{\Opp}
\]

\noindent
By $\HF$-Bounded Choice for Product Types (Theorem~\ref{thm:funto:choice})
%and Comprehension for $\HF$-Sets (Remark~\ref{rem:hfchoice}),
we are going to define $\strat_\Prop$ such that 
$\strat_\Prop(x,q) \in \trans_{\aneg \At A}(q,F(x))$
for each
$(x,q) \in \univ \times Q_{\At A}$.
Assume fixed $(x,q) \in \univ \times Q_{\At A}$.
For all $\Conj \in \Pne(\Dir \times Q_{\At A})$
such that $\Conj \in \trans_{\At A}(q,F(x))$,
%by definition of $\G(\At A,F)$
we have $\strat_\Opp(x,(q,\Conj)) \in \Conj$.
By $\HF$-Comprehension (Remark~\ref{rem:hfchoice}),
let
\[
\dual\Conj \quad\deq\quad
\{\strat_\Opp(x,(q,\Conj)) \st \Conj \in \trans_{\At A}(q,F(x)) \}
\]

\noindent
By construction, we thus have $\dual\Conj \in \trans_{\aneg \At A}(q,F(x))$,
and
%by $\HF$-Bounded Choice for Product Types (Theorem~\ref{thm:funto:choice}),
we let %$\strat_\Prop(x,q)$ be $(q,\dual\Conj)$.
\[
\strat_\Prop(x,q) ~~\deq~~ (q,\dual\Conj)
\]
We trivially have $\Strat_\Prop(\G(\aneg \At A,F),\strat_\Prop)$.

%%%%%%%%%%%%%%%%%%%%%%%%%%%%%%%%%%%%%%%%%%%%%%%%%%%%%%%%%%%%%%%%%%%%%%%%%%%
\begin{lem}
\label{lem:neg:toneg}
%%%%%%%%%%%%%%%%%%%%%%%%%%%%%%%%%%%%%%%%%%%%%%%%%%%%%%%%%%%%%%%%%%%%%%%%%%%
Consider $\strat_\Opp$ and $\strat_\Prop$ as above.
For every infinite play $V$ of $\strat_\Prop$ in $\G(\aneg \At A,F)$
there is some infinite play $U$ of $\strat_\Opp$ in $\G(\At A,F)$
with $\PP V = \PP U$.
\end{lem}

\begin{proof}
We define $U$ by Comprehension for Product Types (Theorem~\ref{thm:funto:ca})
as follows.
\begin{itemize}
\item First, for
$(x,k) \in \PP{\G}$, if $(x,k) \in \PP V$ then we let $(x,k) \in \PP U$.

\item Consider $(x,(q,\Conj)) \in \OP\G$.
Using Remark~\ref{rem:ax:hf:well-order-hf},
let $\prec$ be a well-order on $\Pne(\Dir \times Q_{\At A})$.
Then we let $(x,(q,\Conj)) \in \OP U$ iff
$(x,q) \in \PP V$
and $\Conj$ is $\preceq$-minimal in $\trans_{\At A}(q,F(x))$
such that $(\Succ_d(x),q') \in \PP V$
for $(d,q') = \strat_\Opp(x,(q,\Conj))$.
%for some $(d,q') \in \Conj$.
\end{itemize}

\noindent
Note that consecutive $\Prop$-positions in $\PP U$ are indeed
connected by the edge relation of $\G(\At A,F)$:

%%%%%%%%%%%%%%%%%%%%%%%%%%%%%%%%%%%%%%%%%%%%%%%%%%%%%%%%%%%%%%%%%%%%%%%%%%%
\begin{subclm}
\label{clm:op:compl:toneg:between}
%%%%%%%%%%%%%%%%%%%%%%%%%%%%%%%%%%%%%%%%%%%%%%%%%%%%%%%%%%%%%%%%%%%%%%%%%%%
\[
%\forall x,q,q'\left[
(x,q),(\Succ_d(x),q') \in \PP U
~~\limp~~
(\exists ! u \in \OP U)
\left(
(x,q) \edge{}{\strat_\Opp} u
\edge{}{\strat_\Opp}
(\Succ_d(x),q')
\right)
%\right]
\]
\end{subclm}

\begin{subproof}[Proof of Claim \thesubclm]
We first show uniqueness. Let
$(y_0,(q_0,\Conj_0)),(y_1,(q_1,\Conj_1)) \in \OP U$
be between $(x,q)$ and $(\Succ_d(x),q')$.
Then we must have $y_0 = y_1 = x$ and $q_0 = q_1 = q$.
Hence, $\Conj_0$ and $\Conj_1$ are both $\preceq$-minimal in $\trans_{\At A}(q,F(x))$
such that $\strat_\Opp(x,(q,\Conj_0)) = \strat_\Opp(x,(q,\Conj_1)) = (d,q')$,
yielding $\Conj_0 = \Conj_1$ as required.

We now show the existence of an appropriate $(x,(q,\Conj)) \in \OP U$.
Since
$\Play(\strat_\Prop,\iota,V)$,
%$\Play(V,(\Root, \init q_{\At A}),e_{\aneg \At A}\restr\{\strat_\Prop\})$,
we have $(d,q') \in \dual\Conj$
with $(\ell,\dual\Conj) \in \strat_\Prop(y,\ell)$
for some $(y,\ell) \in \PP V$.
But 
$\Play(\strat_\Prop,\iota,V)$
moreover implies that either $(y,\ell) \glt (x,q)$ or $(x,q) \gle (y,\ell)$,
from which follows that $(y,\ell) = (x,q)$
and $(q,\dual\Conj) \in \strat_\Prop(x,q)$.
Since
\[
\dual\Conj \quad\deq\quad
\{\strat_\Opp(x,\Conj) \st \Conj \in \trans_{\At A}(q,F(x)) \}
\]
it follows that $(d,q') \in \strat_\Opp(x,\Conj)$ for some
$\Conj \in \trans_{\At A}(q,F(x))$, and we are done.
\end{subproof}

We now check that $U$ is indeed an infinite play of $\strat_\Opp$,
\ie\@ that 
$\Play(\strat_\Opp,\iota,U)$
holds.
First, we have $\iota \in U$.
Moreover,

%%%%%%%%%%%%%%%%%%%%%%%%%%%%%%%%%%%%%%%%%%%%%%%%%%%%%%%%%%%%%%%%%%%%%%%%%%%
\begin{subclm}
%%%%%%%%%%%%%%%%%%%%%%%%%%%%%%%%%%%%%%%%%%%%%%%%%%%%%%%%%%%%%%%%%%%%%%%%%%%
\[
(\forall u \in U)\left( \iota \edge{*}{\strat_\Opp} u \right)
\]
\end{subclm}

\begin{subproof}[Proof of Claim \thesubclm]
We reason by induction on
$\edge{}{\strat_\Opp}{}$
(Corollary~\ref{cor:games:edges:ind}).
First, if $u \in \OP U$, then $u$ is of the form $(x,(q,\Conj))$.
By definition of $\OP U$ we have $(x,q) \in \PP U$ with
$(x,q) \edge{}{\strat_\Opp} (x,(q,\Conj))$
and we conclude by induction hypothesis.

Consider now the case of $u \in \PP U = \PP V$.
In this case, $u$ of the form $(x,q)$.
We apply Proposition~\ref{prop:ax:tree}.\eqref{eq:ax:tree:rootorsucc},
stating that either $x \Eq \Root$ or $x = \Succ_d(y)$ for some $d$ and $y$.
In the former case, 
since $V$ is a play, we have
$\iota \edge{*}{\strat_\Prop} (x,q)$,
and
Proposition~\ref{prop:games:edges}.\eqref{eq:games:edges:root}
implies $u = \iota$.
In the latter case, assume $x$ is $\Succ_d(y)$.
We apply twice
the Predecessor Lemma~\ref{lem:games:predplays} for Infinite Plays,
which gives some $(y,q') \in \PP V$ such that
\[
(y,q')
  \edge{+}{\strat_\Prop}
  %\edge{}{e_{\aneg \At A}\restr\{\strat_\Prop\}}
(\Succ_d(y),q)
\]

\noindent
By induction hypothesis we get
$\iota \edge{*}{\strat_\Opp} (y,q')$
and we conclude by Claim~\ref{clm:op:compl:toneg:between}.
\end{subproof}

Also,

%%%%%%%%%%%%%%%%%%%%%%%%%%%%%%%%%%%%%%%%%%%%%%%%%%%%%%%%%%%%%%%%%%%%%%%%%%%
\begin{subclm}
%%%%%%%%%%%%%%%%%%%%%%%%%%%%%%%%%%%%%%%%%%%%%%%%%%%%%%%%%%%%%%%%%%%%%%%%%%%
\[
(\forall u \in U)(\exists! v \in U)
\left(
u \edge{}{\strat_\Opp} v
\right)
\]
\end{subclm}

\begin{subproof}[Proof of Claim \thesubclm]
The case of $u \in \PP U = \PP V$ follows directly from the definition
of $\OP U$ and the fact that $\funto{\strat_\Opp}{\OP\G}{\Dir \times \Prop}$
and 
$\Play(\strat_\Prop,\iota,V)$.
Consider now the case of $u \in \OP U$.
By definition of $\OP U$ there is some $v \in \PP U$
such that
$u \edge{}{\strat_\Opp} v$.
Uniqueness follows from the fact that $\PP U = \PP V$
and 
$\Play(\strat_\Prop,\iota,V)$.
\end{subproof}

In order to obtain
$\Play(\strat_\Opp,\iota,U)$,
we invoke Proposition~\ref{prop:games:edges:lin}
and it remains to show:

%%%%%%%%%%%%%%%%%%%%%%%%%%%%%%%%%%%%%%%%%%%%%%%%%%%%%%%%%%%%%%%%%%%%%%%%%%%
\begin{subclm}
%%%%%%%%%%%%%%%%%%%%%%%%%%%%%%%%%%%%%%%%%%%%%%%%%%%%%%%%%%%%%%%%%%%%%%%%%%%
\[
(\forall u \in U)\left[
u \neq \iota
~~\limp~~
(\exists v \in U) \left( v \edge{}{\strat_\Opp} u \right)
\right]
\]
\end{subclm}

\begin{subproof}[Proof of Claim \thesubclm]
The case of $u \in \OP U$ follows from the definition of $\OP U$.
The case of $u \in \PP U$ directly follow from
Claim~\ref{clm:op:compl:toneg:between}
(together with Proposition~\ref{prop:games:edges}.\eqref{eq:games:edges:root})
and
$\Play(\strat_\Prop,\iota,V)$.
\end{subproof}

This concludes the proof of Lemma~\ref{lem:neg:toneg}.
\end{proof}

We use the following simple fact
in order to obtain from Lemma~\ref{lem:neg:toneg}
that $\strat_\Prop$ is winning in $\G(\aneg \At A,F)$
whenever $\strat_\Opp$ is winning in $\G(\At A,F)$.

%%%%%%%%%%%%%%%%%%%%%%%%%%%%%%%%%%%%%%%%%%%%%%%%%%%%%%%%%%%%%%%%%%%%%%%%%%%
\begin{lem}
\label{lem:neg:win:toneg}
%%%%%%%%%%%%%%%%%%%%%%%%%%%%%%%%%%%%%%%%%%%%%%%%%%%%%%%%%%%%%%%%%%%%%%%%%%%
Given plays $\funto{U,V}{\G}{\two}$
as in Lemma~\ref{lem:neg:toneg},
we have
$\Par(\At A, U) ~\liff~ \lnot \Par(\aneg \At A, V)$.
\end{lem}

We now have everything we need to obtain Proposition~\ref{prop:neg:toneg},
namely
\[
\FSO + (\PosDet) \thesis\quad
F \notin \Lang(\At A)
~~\longlimp~~
F \in \Lang(\aneg \At A)
\]

\noindent
Assume $F \notin \Lang(\At A)$.
By Definition~\ref{def:aut:lang}, 
there is no winning $\Prop$-strategy in $\G(\At A,F)$.
By the axiom of positional determinacy of parity games $(\PosDet)$
there is a winning $\Opp$-strategy $\strat_\Opp$ in $\G(\At A,F)$,
so that
\begin{equation}
\label{eq:op:compl:toneg:win}
(\forall \funto{U}{\G}{\two})
\Big(
\Play(\strat_\Opp,\iota,U)
%\Play (U,(\Root,\init q_{\At A}),e_{\aneg \At A}\restr\{\strat_\Prop\})
~~\limp~~
\lnot \Par(\At A,U)
\Big)
\end{equation}

\noindent
Consider now the $\Prop$-strategy $\strat_\Prop$ on $\G(\aneg \At A,F)$
as defined above.
We claim that $\strat_\Prop$ is winning, that is

%%%%%%%%%%%%%%%%%%%%%%%%%%%%%%%%%%%%%%%%%%%%%%%%%%%%%%%%%%%%%%%%%%%%%%%%%%%
\begin{clm}
%%%%%%%%%%%%%%%%%%%%%%%%%%%%%%%%%%%%%%%%%%%%%%%%%%%%%%%%%%%%%%%%%%%%%%%%%%%
\[
(\forall \funto{V}{\G}{\two})
\Big(
\Play(\strat_\Prop,\iota,V)
%\Play (U,(\Root,\init q_{\At A}),e_{\aneg \At A}\restr\{\strat_\Prop\})
~~\limp~~
\Par(\aneg \At A,V)
\Big)
\]
\end{clm}

\begin{proof}[Proof of Claim \theclm]
Given an infinite play $V$ of $\strat_\Prop$,
by Lemma~\ref{lem:neg:toneg}
we can build an infinite play $U$ of $\strat_\Opp$,
which by~\eqref{eq:op:compl:toneg:win}
satisfies $\lnot \Par(\At A,-)$,
so that $V$ satisfies $\Par(\aneg \At A,-)$
thanks to Lemma~\ref{lem:neg:win:toneg}.
\end{proof}

We thus have $F \in \Lang(\aneg \At A,F)$.
This concludes the proof of Proposition~\ref{prop:neg:toneg}.

%%%%%%%%%%%%%%%%%%%%%%%%%%%%%%%%%%%%%%%%%%%%%%%%%%%%%%%%%%%%%%%%%%%%%%%%%%%
\subsubsection{Proof of Proposition~\ref{prop:neg:fromneg}.}
\label{sec:neg:fromneg}
%%%%%%%%%%%%%%%%%%%%%%%%%%%%%%%%%%%%%%%%%%%%%%%%%%%%%%%%%%%%%%%%%%%%%%%%%%%
We are now going to show that
$\FSO$ proves
\[
F \in \Lang(\aneg \At A) ~~\longlimp~~
F \notin \Lang(\At A)
\]

\noindent
We associate a (winning) $\Opp$-strategy 
$\strat_\Opp$ on $\G(\At A,F)$
to each (winning) $\Prop$-strategy $\strat_\Prop$ on $\G(\aneg \At A,F)$.
%
%Given a $\Prop$-strategy $\strat_\Prop$ on $\G(\aneg \At A,F)$,
%we thus define an $\Opp$-strategy $\strat_\Opp$ on $\G(\At A,F)$.
Assuming that the $\Prop$-strategy
satisfies
$\Strat_\Prop(\G(\aneg \At A,F),\strat_\Prop)$,
the $\Opp$-strategy
will satisfy
$\Strat_\Opp(\G(\At A,F),\strat_\Opp)$.
Note that
\[
\funto{\strat_\Prop}{\PP\G}{\Opp}
\qquad\text{and}\qquad
\funto{\strat_\Opp}{\OP\G}{\Dir \times \Prop}
\]

\noindent
We define $\strat_\Opp(x,(q,\Conj))$
for each position
\[
(x,(q,\Conj)) \in \univ \times (Q_{\At A} \times \Pne(\Dir \times Q_{\At A}))
\]
By definition of $\trans_{\aneg \At A}(q,F(p))$,
we have $\strat_\Prop(p,q) = (q,\dual\Conj)$
where $\dual\Conj$ intersects all $\Conj \in \trans_{\At A}(q,F(p))$.
%$\dual\Conj \cap \Conj \neq \emptyset$ for all $\Conj \in \trans_{\At A}(q,F(p))$.
So if $\Conj \in \trans_{\At A}(q,F(p))$,
by $\HF$-Bounded Choice for Product Types (Theorem~\ref{thm:funto:choice})
we let $\strat_\Opp(p,(q,\Conj))$ be some
$(d,q')$ such that $(d,q') \in \Conj \cap \dual\Conj$.
Otherwise, since $\Conj \neq \emptyset$, 
we let $\strat_\Opp(p,(q,\Conj))$ be some
$(d,q')$ such that $(d,q') \in \Conj$.

We also trivially have that $\Strat_\Opp(\G(\At A,F),\strat_\Opp)$.

%%%%%%%%%%%%%%%%%%%%%%%%%%%%%%%%%%%%%%%%%%%%%%%%%%%%%%%%%%%%%%%%%%%%%%%%%%%
\begin{lem}
\label{lem:neg:fromneg}
%%%%%%%%%%%%%%%%%%%%%%%%%%%%%%%%%%%%%%%%%%%%%%%%%%%%%%%%%%%%%%%%%%%%%%%%%%%
Consider a $\Prop$-strategy $\strat_\Prop$ and
an $\Opp$-strategy $\strat_\Opp$ as in above.
For every infinite play $V$
of $\strat_\Opp$ on $\G(\At A,F)$
there is some infinite play $U$ of $\strat_\Prop$ on $\G(\aneg \At A,F)$
with $\PP V = \PP U$.
\end{lem}

\begin{proof}
We define $U$ by Comprehension for Product Types (Theorem~\ref{thm:funto:ca})
as follows.
\begin{itemize}
\item \emph{Definition of $U$.}
For $(x,k) \in \PP\G$, if $(x,k) \in \PP V$ then we let $(x,k) \in \PP U$,
and for $(x,(q,\dual\Conj)) \in \OP\G$, we let $(x,(q,\dual\Conj)) \in \OP U$
iff $(q,\dual\Conj) = \strat_\Prop(x,q)$
for $(x,q) \in \PP U$.
\end{itemize}

\noindent
Similarly as in Lemma~\ref{lem:neg:toneg},
we have

%%%%%%%%%%%%%%%%%%%%%%%%%%%%%%%%%%%%%%%%%%%%%%%%%%%%%%%%%%%%%%%%%%%%%%%%%%%
\begin{subclm}
\label{clm:op:compl:fromneg:between}
%%%%%%%%%%%%%%%%%%%%%%%%%%%%%%%%%%%%%%%%%%%%%%%%%%%%%%%%%%%%%%%%%%%%%%%%%%%
\[
%\forall x,q,q'\left[
(x,q),(\Succ_d(x),q') \in \PP U
~~\limp~~
(\exists ! u \in \OP U)
\left(
(x,q) \edge{}{\strat_\Prop} u
\edge{}{\strat_\Prop}
(\Succ_d(x),q')
\right)
%\right]
\]
\end{subclm}

\begin{subproof}[Proof of Claim \thesubclm]
Uniqueness directly follows from the fact that $u = (x,\strat_\Prop(x,q))$.
As for existence, we directly have
$(x,q) \edge{}{\strat_\Prop} u$,
so it remains to show
$u \edge{}{\strat_\Prop} (\Succ_d(x),q')$,
which amounts to $(d,q') \in \dual\Conj$ for $(q,\dual\Conj) = \strat_\Prop(x,q)$.
But $(\Succ_d(x),q') \in \PP V$
with
$\Play(\strat_\Opp,\iota,V)$
%$\Play(V,(\Root, \init q_{\At A}),e_{\At A}\restr\{\strat_\Opp\})$
imply that $(d,q')= \strat_\Opp(x,(\ell,\Conj))$
for some $\ell$ such that $(x,\ell) \in \PP V$
and some $\Conj \in \trans_{\At A}(\ell,F(x))$.
Moreover, 
$\Play(\strat_\Opp,\iota,V)$
%$\Play(V,(\Root, \init q_{\At A}),e_{\At A}\restr\{\strat_\Opp\})$
implies $\ell = q$.
By definition of $\strat_\Opp$,
we thus have $(d,q') \in \Conj \cap \dual\Conj$ and we are done.
\end{subproof}

We now check that
$\Play(\strat_\Prop,\iota,U)$.
%$\Play((G,e_{\aneg \At A}\restr\{\strat_\Prop\}),\iota,U)$.
Note that $\iota \in U$.
Moreover, proceeding as in 
Lemma~\ref{lem:neg:toneg},
we have 

%%%%%%%%%%%%%%%%%%%%%%%%%%%%%%%%%%%%%%%%%%%%%%%%%%%%%%%%%%%%%%%%%%%%%%%%%%%
\begin{subclm}
%%%%%%%%%%%%%%%%%%%%%%%%%%%%%%%%%%%%%%%%%%%%%%%%%%%%%%%%%%%%%%%%%%%%%%%%%%%
\[
(\forall u \in U)
\left(
\iota \edge{*}{\strat_\Prop} u
\right)
\]
\end{subclm}

\begin{subproof}[Proof of Claim \thesubclm]
By induction on $\edge{}{\strat_\Prop}$
(Corollary~\ref{cor:games:edges:ind}).
The case of $u \in \OP U$ follows directly form the induction hypothesis
and the definition of $\OP U$.
As for $u \in \PP U$, we proceed as 
in Lemma~\ref{lem:neg:toneg},
using Claim~\ref{clm:op:compl:fromneg:between} and
Lemma~\ref{lem:games:predplays}.
\end{subproof}

%We now check
%$\Play(\G(\aneg \At A,F)\restr\{\strat_\Prop\},\iota,U)$
Continuing as in Lemma~\ref{lem:neg:toneg},
we now invoke Proposition~\ref{prop:games:edges:lin}
and we are left with showing

%%%%%%%%%%%%%%%%%%%%%%%%%%%%%%%%%%%%%%%%%%%%%%%%%%%%%%%%%%%%%%%%%%%%%%%%%%%
\begin{subclm}
%%%%%%%%%%%%%%%%%%%%%%%%%%%%%%%%%%%%%%%%%%%%%%%%%%%%%%%%%%%%%%%%%%%%%%%%%%%
\[
(\forall u \in U)(\exists! v \in U)
\left(
u \edge{}{\strat_\Prop} v
\right)
~~\land~~
(\forall u \in U)\left[
u \neq \iota
~~\limp~~
(\exists v \in U) \left( v \edge{}{\strat_\Prop} u \right)
\right]
\]
\end{subclm}

\begin{subproof}[Proof of Claim \thesubclm]
The cases of $u\in \PP U$ follow from the definition of $\OP U$,
and from Claim~\ref{clm:op:compl:toneg:between}
(together with Proposition~\ref{prop:games:edges}.\eqref{eq:games:edges:root})
and 
$\Play(\strat_\Opp,\iota,V)$.
%$\Play((G,e_{\At A}\restr\{\strat_\Opp\}),\iota,V)$.
%
Consider now $u \in \OP U$.
The predecessor property follows from the definition of $\OP U$.
The unique successor property
is obtained from Claim~\ref{clm:op:compl:fromneg:between}
together with 
$\Play(\strat_\Opp,\iota,V)$.
%$\Play((G,e_{\At A}\restr\{\strat_\Opp\}),\iota,V)$.
\end{subproof}

This concludes the proof of Lemma~\ref{lem:neg:fromneg}.
\end{proof}

Similarly as in~\S\ref{sec:neg:toneg},
we use the following simple fact.

%%%%%%%%%%%%%%%%%%%%%%%%%%%%%%%%%%%%%%%%%%%%%%%%%%%%%%%%%%%%%%%%%%%%%%%%%%%
\begin{lem}
\label{lem:neg:win:fromneg}
%%%%%%%%%%%%%%%%%%%%%%%%%%%%%%%%%%%%%%%%%%%%%%%%%%%%%%%%%%%%%%%%%%%%%%%%%%%
Given plays $\funto{U,V}{\G}{\two}$
as in Lemma~\ref{lem:neg:fromneg},
we have
$\Par(\aneg \At A, U) ~\liff~ \lnot \Par(\At A, V)$
\end{lem}

It is now easy to obtain Proposition~\ref{prop:neg:fromneg},
namely
\[
\FSO \thesis\quad
F \in \Lang(\aneg \At A) ~~\longlimp~~
F \notin \Lang(\At A)
\]

\noindent
Assume that $F \in \Lang(\aneg \At A)$.
By Definition~\ref{def:aut:lang}, we thus have a winning 
$\Prop$-strategy $\strat_\Prop$ in $\G(\aneg \At A,F)$,
so that
\[
\label{eq:op:compl:fromneg:win}
(\forall \funto{U}{\G}{\two})
\Big(
\Play(\strat_\Prop,\iota,U)
%\Play (U,(\Root,\init q_{\At A}),e_{\aneg \At A}\restr\{\strat_\Prop\})
~~\limp~~
\Par(\aneg \At A,U)
\Big)
\]

\noindent
Consider now the $\Opp$-strategy $\strat_\Opp$ on $\G(\At A,F)$
as defined above.
Reasoning as in the case $F \notin \Lang(\At A)$
(\S\ref{sec:neg:toneg}),
Lemmas~\ref{lem:neg:fromneg}
and~\ref{lem:neg:win:fromneg}
imply
\[
(\forall \funto{V}{\G}{\two})
\Big(
\Play(\strat_\Opp,\iota,V)
%\Play (U,(\Root,\init q_{\At A}),e_{\aneg \At A}\restr\{\strat_\Prop\})
~~\limp~~
\lnot \Par(\At A,V)
\Big)
\]

\noindent
It then follows from Lemma~\ref{lem:games:notbothwin}
that there is no winning $\Prop$-strategy on $\G(\At A,F)$,
so that $F \notin \Lang(\At A)$.

This concludes the proof of Proposition~\ref{prop:neg:fromneg}.

%%% Local Variables:
%%% mode: latex
%%% TeX-master: "main.tex"
%%% End:

%%%%%%%%%%%%%%%%%%%%%%%%%%%%%%%%%%%%%%%%%%%%%%%%%%%%%%%%%%%%%%%%%%%%%%%%%%%
\section{$\MSO$ on Infinite Words in Paths of $\FSOD$}
\label{sec:msow}
%%%%%%%%%%%%%%%%%%%%%%%%%%%%%%%%%%%%%%%%%%%%%%%%%%%%%%%%%%%%%%%%%%%%%%%%%%%

\noindent
We discuss here the theory of $\MSO$ over $\omega$-words
for the \emph{infinite paths} of $\FSOD$.
Since $\MSO$ on $\omega$-words admits a complete axiomatization~\cite{siefkes70lnm},
this will allow us to
freely import results on $\MSO$ over $\omega$-words
for the paths of $\FSOD$.
In particular, our completeness argument (\S\ref{sec:compl})
relies on a version of the Büchi-Landweber's Theorem~\cite{bl69tams}
formulated with $\MSO$ over $\omega$-words,
that we lift for free to $\FSOD$.
Also, to prove
the Simulation Theorem~\ref{thm:aut:nd:sim} in~\S\ref{sec:sim},
we use McNaughton's Theorem~\cite{mcnaughton66ic},
and similarly obtain it for free in $\FSOD$.

An obvious way to obtain $\MSO$ over $\omega$-words is to consider
the system $\MSO_\one$ (that is $\MSOD$ for $\Dir = \one$).
However, recall that we want to see each path 
of $\FSOD$ (in the sense of~\eqref{eq:bfsos:rpath} below)
as a model of $\MSO$ on $\omega$-words.
%However, we shall need a slightly different viewpoint,
%according to which
%each path
%(\ie\@ infinite set of individuals linearly ordered by the prefix relation)
%of $\FSOD$
%gives a model of $\MSO$ on $\omega$-words.
This is technically simpler if, following~\cite{riba12ifip},
one uses a version of $\MSO$ on $\omega$-words
over a purely relational vocabulary with only 
the strict order $\Lt$ on numbers as atomic relation
(besides equality $\Eq$).

%%%%%%%%%%%%%%%%%%%%%%%%%%%%%%%%%%%%%%%%%%%%%%%%%%%%%%%%%%%%%%%%%%%%%%%%%%%
\begin{defi}[The Theory $\FSOW$]
\label{def:bfsos}
%%%%%%%%%%%%%%%%%%%%%%%%%%%%%%%%%%%%%%%%%%%%%%%%%%%%%%%%%%%%%%%%%%%%%%%%%%%
The language of $\FSOW$ is the language of $\FSOD$
with the following restriction:
\begin{itemize}
\item the only \emph{Individual terms} of $\FSOW$
are the constant $\Root$ and
the individual variables ($x,y,z$ etc.)
\end{itemize}

\noindent
The deduction rules of $\FSOW$ are the same as the rules of
$\FSOD$.
The axioms of $\FSOW$ are
the Equality Axioms of~\S\ref{sec:ax:eq},
the Axioms on $\HF$-Sets of~\S\ref{sec:ax:hf},
the Functional Choice Axioms of~\S\ref{sec:ax:choice},
together with
the axioms displayed on Figure~\ref{fig:lt},
stating that $\Lt$ is a discrete unbounded strict linear order with $\Root$ as its minimal element
(see \eg~\cite{riba12ifip}),
and with 
the following induction scheme.
\begin{itemize}
\item 
\emph{Well-Founded Induction.}
For each formula $\varphi$, the axiom
\[
(\forall x)\big[
(\forall y \Lt x)(\varphi(y))
~~\longlimp~~ \varphi(x)
\big]
~~\longlimp~~
(\forall x) \varphi(x)
\]
\end{itemize}
\end{defi}

%%%%%%%%%%%%%%%%%%%%%%%%%%%%%%%%%%%%%%%%%%%%%%%%%%%%%%%%%%%%%%%%%%%%%%%%%%%
\begin{rem}
%%%%%%%%%%%%%%%%%%%%%%%%%%%%%%%%%%%%%%%%%%%%%%%%%%%%%%%%%%%%%%%%%%%%%%%%%%%
Note that all Individuals of $\FSOW$
are Individuals of $\FSOD$, but not conversely.
As a consequence, all $\HF$-terms of $\FSOW$
are $\HF$-terms of $\FSOD$, but not conversely.
Also, note that it may have seemed more natural not to include the individual
constant $\Root$ in the language of $\FSOW$.
We have included it because this eases our concrete uses of $\FSOW$
in~\S\ref{sec:compl:red} and~\S\ref{sec:sim:omega}.
\end{rem}

%%%%%%%%%%%%%%%%%%%%%%%%%%%%%%%%%%%%%%%%%%%%%%%%%%%%%%%%%%%%%%%%%%%%%%%%%%%
\begin{figure}[tbp]
%%%%%%%%%%%%%%%%%%%%%%%%%%%%%%%%%%%%%%%%%%%%%%%%%%%%%%%%%%%%%%%%%%%%%%%%%%%
\[
\begin{array}{c}
(\forall x)\left(\Root \Leq x\right)
\qquad
\lnot (\exists x) \left(x \Lt x \right)
\qquad
(\forall x) (\forall y) (\forall z) 
\left( x \Lt y ~~\limp~~ y \Lt z ~~\limp~~ x \Lt z \right)
\\\\
(\forall x) (\exists y) \left( x \Lt y \right)
\qquad
(\forall x) (\forall y) \left[
x \Lt y ~~\lor~~ x \Eq y ~~\lor~~ y \Lt x
\right]
\\\\
(\forall x) \Big[(\exists y \Lt x)
  ~~\limp~~
  (\exists y \Lt x) \lnot (\exists z) \big( y \Lt z \Lt x \big)
\Big]
%\\\\
%\forall X
%\left[
%\forall x \left(
%\forall y (y \Lt x \limp X y) \limp X x
%\right)
%\limp \forall x X x
%\right]
\\[1em]
\end{array}
\]
%\hrule
\caption{Axioms on the relation $\Lt$ of $\FSOW$.\label{fig:lt}}
\end{figure}

\noindent
Similarly to the case of $\Dir$-ary trees (\S\ref{sec:mso}),
the theory $\FSOW$ is intended to be interpreted
in a theory $\MSOW$.
Intuitively, $\MSOW$ is to $\FSOW$ what $\MSOD$ is to $\FSOD$.

%%%%%%%%%%%%%%%%%%%%%%%%%%%%%%%%%%%%%%%%%%%%%%%%%%%%%%%%%%%%%%%%%%%%%%%%%%%
\begin{defi}[The Theory $\MSOW$]
\label{def:msow}
%%%%%%%%%%%%%%%%%%%%%%%%%%%%%%%%%%%%%%%%%%%%%%%%%%%%%%%%%%%%%%%%%%%%%%%%%%%
The language of $\MSOW$ is the language of $\MSOD$
with the following restriction:
\begin{itemize}
\item the only \emph{Individual terms} of $\MSOW$
are individual variables ($x,y,z$ etc.).
\end{itemize}

\noindent
The axioms of $\MSOW$ are the equality axioms
and the comprehension scheme of $\MSOD$ (\S\ref{sec:mso:th}),
together with the induction scheme
and the axioms on $\Lt$ of $\FSOW$ 
displayed in Figure~\ref{fig:lt}.
%(Def.~\ref{def:bfsos}).
\end{defi}

\noindent
We write $\StdN$ both for the standard model of $\FSOW$
and for the standard model of $\MSOW$.
In the case of $\MSOW$, 
formulae are interpreted in $\StdN$ as expected:
individual variables range over $\NN$, monadic predicate variables
range over $\Po(\NN)$ and $\Lt$ is the standard order $<$ on $\NN$.
The interpretation of $\FSOW$-formulae in $\StdN$ is similar,
with the obvious changes \wrt~\S\ref{sec:std} for the interpretation
of terms, and where Functions range over
\[
\bigcup_{\kappa \in V_\omega} \left(\NN ~~\longto~~ \kappa \right)
\]

\noindent
The key property of $\MSOW$ we rely on is that it completely axiomatizes the
theory of the standard model $\StdN$
of $\omega$-words~\cite{siefkes70lnm}
(see also~\cite{riba12ifip}).

%%%%%%%%%%%%%%%%%%%%%%%%%%%%%%%%%%%%%%%%%%%%%%%%%%%%%%%%%%%%%%%%%%%%%%%%%%%
\begin{thm}[\cite{siefkes70lnm}]
\label{thm:msow:compl}
%%%%%%%%%%%%%%%%%%%%%%%%%%%%%%%%%%%%%%%%%%%%%%%%%%%%%%%%%%%%%%%%%%%%%%%%%%%
For every closed $\MSOW$-formula $\varphi$, %we have
\[
\StdN \models \varphi
\qquad\text{if and only if}\qquad
\MSOW \thesis \varphi
\]
\end{thm}

%$\MSOW$-formulae are interpreted in $\StdN$ as expected:
%individual variables range over $\NN$, monadic predicate variables
%range over $\Po(\NN)$ and $\Lt$ is the standard order $<$ on $\NN$.

\noindent
The formula translation from $\FSOD$ to $\MSOD$ of~\S\ref{sec:cons}
restricts to a translation of $\FSOW$-formulae
to $\MSOW$-formulae.
This easily extends to theories,
and we get the following version of Proposition~\ref{prop:cons:mso:cons}.

%%%%%%%%%%%%%%%%%%%%%%%%%%%%%%%%%%%%%%%%%%%%%%%%%%%%%%%%%%%%%%%%%%%%%%%%%%%
\begin{prop}
\label{prop:msow:bfsos}
%%%%%%%%%%%%%%%%%%%%%%%%%%%%%%%%%%%%%%%%%%%%%%%%%%%%%%%%%%%%%%%%%%%%%%%%%%%
For every closed $\FSOW$-formula $\varphi$,
\begin{gather}
\label{eq:msow:bfsos:cons}
\FSOW \thesis \varphi
\qquad\text{if and only if}\qquad
\MSOW \thesis \MI{\varphi}
\\
\label{eq:msow:bfsos:std}
\StdN \models \varphi
\qquad\text{if and only if}\qquad
\StdN \models \MI{\varphi}
\end{gather}
\end{prop}

\noindent
Thanks to~\eqref{eq:msow:bfsos:std},
%Proposition~\ref{prop:msow:bfsos}.\eqref{eq:msow:bfsos:std},
the completeness of $\MSOW$ directly gives the completeness
of $\MSOW$ \wrt\@ the translation closed of $\FSOW$-formulae $\varphi$:
\[
%\tag{$\varphi$ closed $\FSOW$-formula}
\FSOW \thesis \varphi
\qquad\text{if and only if}\qquad
\StdN \models \MI{\varphi}
\]

Our goal now is to prove that if a closed $\FSOW$ formula holds in the
standard model $\StdN$ of $\omega$-words, then $\FSOD$
proves its relativization to any rooted tree path.
Given a formula $\varphi$ of $\FSOW$ and a Function variable $P$,
write $\varphi^{P}$ for the $\FSOD$ formula 
obtained from $\varphi$
by relativizing all individual quantifications to $P$
and by replacing all Function quantifications
$F : K$ by $\funto{F}{P}{K}$.
Moreover, 
we say that $\funto{P}{\univ}{\two}$ is a \emph{rooted path}
when the following formula $\TPath(P)$ holds:
\begin{equation}
\label{eq:bfsos:rpath}
\TPath(P)
\quad\deq\quad
\left\{
\begin{array}{l l}
& (\Root \in P)
\\
  \land
& (\forall x,y \in P)\left(x \Lt y ~~\lor~~ x \Eq y ~~\lor~~ y \Lt x \right)
\\
  \land
& (\forall x \in P)(\exists y \in P)(\FSucc(x,y))
\end{array}
\right.
\end{equation}
where $\FSucc(x,y)$ stands for
\[
x \Lt y
~~\land~~
\lnot(\exists z)\big[x \Lt z \Lt y\big]
\]

\noindent
%Our goal is to prove that if a closed $\FSOW$ formula holds in the
%standard model $\StdN$ of $\omega$-words, then $\FSOD$
%proves its relativization to any rooted tree path:
We can now formally state the property we are targeting:
\begin{equation}
\label{eq:bfsos:bfso}
{\FSOD} \thesis
(\forall P : \two)\left(
\TPath(P) ~~\limp~~
\varphi^{P}
\right)
\qquad\text{whenever}\qquad
\StdN \models \varphi
\end{equation}

\noindent
The proof of~\eqref{eq:bfsos:bfso} is deferred to Proposition~\ref{prop:bfsos:bfso}.
It relies on two lemmas.
The first one is an adaptation of
Lemma~\ref{lem:games:predpath} (\S\ref{sec:pos:path})
to rooted tree paths, which will give the last axiom of Figure~\ref{fig:lt}
for rooted tree paths.
The second one is a weakening of~\eqref{eq:bfsos:bfso}
where $\FSOW \thesis \varphi$ is assumed instead of $\StdN \models \varphi$.

%%%%%%%%%%%%%%%%%%%%%%%%%%%%%%%%%%%%%%%%%%%%%%%%%%%%%%%%%%%%%%%%%%%%%%%%%%%
\begin{lem}
\label{lem:omega:pred}
%%%%%%%%%%%%%%%%%%%%%%%%%%%%%%%%%%%%%%%%%%%%%%%%%%%%%%%%%%%%%%%%%%%%%%%%%%%
$\FSOD$ proves the following, assuming $\funto{P}{\univ}{\two}$
and $\TPath(P)$:
\[
%\TPath(P) ~~\limp~~
(\forall x \in P) \Big[(\exists y \in P)(y \Lt x)
  ~~\limp~~
  (\exists y \in P)
  \big(y \Lt x ~\land~
     \lnot (\exists z \in P) (y \Lt z \Lt x)
  \big)
\Big]
\]
\end{lem}

\begin{fullproof}
The argument is exactly the same as in the proof of Lemma~\ref{lem:games:predpath},
as soon as we have an analogue to Lemma~\ref{lem:games:maxlinbounded}.
This amounts to the following, whose proof is considerably simpler
than that of Lemma~\ref{lem:games:maxlinbounded}.

%%%%%%%%%%%%%%%%%%%%%%%%%%%%%%%%%%%%%%%%%%%%%%%%%%%%%%%%%%%%%%%%%%%%%%%%%%%
\begin{subclm}
%%%%%%%%%%%%%%%%%%%%%%%%%%%%%%%%%%%%%%%%%%%%%%%%%%%%%%%%%%%%%%%%%%%%%%%%%%%
Assume that $(X:\two)$ is bounded
(\ie\@ $(\exists x)(\forall y \in X)(y \Lt x)$),
non-empty and linearly ordered.
Then $X$ has a maximal element:
$(\exists z \in X) (\forall y \in X)(y \Leq z)$.
\end{subclm}

\begin{subproof}[Proof of Claim \thesubclm]
Let $X$ be as in the statement and let $y \Lt x$
for all $y \in X$.
By Well-Founded Induction (Theorem~\ref{thm:ax:wfind})
we can assume $x$ to be minimal with this property,
in the sense that
\[
z \Lt x ~~\limp~~ (\exists y \in X)\lnot(y \Lt z)
\]

\noindent
Now, by Proposition~\ref{prop:ax:tree},
$x$ is either $\Root$ or a successor.
If $x \Eq \Root$, then Proposition~\ref{prop:ax:tree}
leads to a contradiction since $X$ is assumed to be non-empty.
So assume $x = \Succ_d(z)$.
It then follows from the Tree axioms of $\FSOD$ (Figure~\ref{fig:ded:tree})
that $y \Leq z$ for all $y \in X$,
and the minimality of $x$ implies $z \in X$.
\end{subproof}

This concludes the proof of Lemma~\ref{lem:omega:pred}.
\end{fullproof}

%%%%%%%%%%%%%%%%%%%%%%%%%%%%%%%%%%%%%%%%%%%%%%%%%%%%%%%%%%%%%%%%%%%%%%%%%%%
\begin{lem}
\label{lem:bfsos:bfso}
%%%%%%%%%%%%%%%%%%%%%%%%%%%%%%%%%%%%%%%%%%%%%%%%%%%%%%%%%%%%%%%%%%%%%%%%%%%
For all closed $\FSOW$-formula $\varphi$, we have
\[
\FSOD \thesis
\forall P:\two \left(
\TPath(P) ~~\limp~~
\varphi^{P}
\right)
\qquad\text{whenever}\qquad
\FSOW \thesis \varphi
\]
\end{lem}

\begin{proof}
The proof is by induction on derivations of $\FSOW$-formulae.
For formulae $\vec\psi,\varphi$ with free Function variables
$\vec F = F_1,\dots,F_p$ and free Individual variables
$\vec x = x_1,\dots,x_q$
(and possibly further free $\HF$-variables), we show that
for all $\HF$-terms $\vec K = K_1,\dots,K_p$ of $\FSOW$
we have
\[
\vec{F:K} ~,~ \vec\psi \thesis_{\FSOW} \varphi
\qquad\text{implies}\qquad
\TPath(P) ~,~ \vec{\funto{F}{P}{K}} ~,~ \vec x \In P ~,~
\vec{\psi^P}
\thesis_{\FSOD} \varphi^P
\]

\noindent
The cases for each inference rule are immediate from their respective
induction hypothesis, and we also easily obtain
the Equality Axioms (\S\ref{sec:ax:eq}),
the Axioms of $\HF$-Sets (\S\ref{sec:ax:hf})
and the Axiom of $\HF$-Bounded Choice for $\HF$-Sets (\S\ref{sec:ax:choice}).
We resort on Theorem~\ref{thm:funto:choice}
for the axioms of $\HF$-Bounded Choice for Functions 
and of Iterated $\HF$-Bounded Choice.
Moreover, the Induction axiom of $\FSOW$ on the formula $\varphi(x)$
directly follows from Well-Founded Induction in $\FSOD$
(Theorem~\ref{thm:ax:wfind})
on the formula
\[
\psi(x) \quad\deq\quad \left(x \in P ~~\limp~~ \varphi(x)\right)
\]

\noindent
It remains to deal with the $\Lt$-axioms of Figure~\ref{fig:lt}.
The first five axioms
(stating that $\Lt$ is an unbounded linear order)
directly follow from the Tree axioms of $\FSOD$ (Figure~\ref{fig:ded:tree})
and from relativization to $P$ with $\TPath(P)$.
Finally, we have to show that $\FSOD$ proves 
that the translation of the predecessor axiom
holds within $P$ whenever $\TPath(P)$ is assumed:
\[
\TPath(P)
\quad\limp\quad
\forall x \in P\left[\exists y \in P(y \Lt x)
  ~~\limp~~
  \exists y \in P
  \left(y \Lt x ~\land~
     \lnot \exists z \in P (y \Lt z \Lt x)
  \right)
\right]
\]
This is handled by Lemma~\ref{lem:omega:pred}.
%
%Finally, the translation of the
%predecessor axiom,
%\[
%\TPath(P) \thesis_{\FSOD} \quad
%\forall x \in P\left[\exists y \in P(y \Lt x)
%  ~~\limp~~
%  \exists y \in P
%  \left(y \Lt x ~\land~
%     \lnot \exists z \in P (y \Lt z \Lt x)
%  \right)
%\right]
%\]
%is handled by Lemma~\ref{lem:omega:pred}.
\end{proof}

We have now everything we need to prove~\eqref{eq:bfsos:bfso}.

%%%%%%%%%%%%%%%%%%%%%%%%%%%%%%%%%%%%%%%%%%%%%%%%%%%%%%%%%%%%%%%%%%%%%%%%%%%
\begin{prop}
\label{prop:bfsos:bfso}
%%%%%%%%%%%%%%%%%%%%%%%%%%%%%%%%%%%%%%%%%%%%%%%%%%%%%%%%%%%%%%%%%%%%%%%%%%%
Consider a closed formula $\varphi$ of $\FSOW$.
Then
\[
{\FSOD} \thesis
(\forall P : \two)\left(
\TPath(P) ~~\limp~~
\varphi^{P}
\right)
\qquad\text{whenever}\qquad
\StdN \models \varphi
\]
\end{prop}

\begin{proof}
Assume
$\StdN \models \varphi$.
By~\eqref{eq:msow:bfsos:std}
%Proposition~\ref{prop:msow:bfsos}.\eqref{eq:msow:bfsos:std}
we have
$\StdN \models \MI{{\varphi}}$.
Theorem~\ref{thm:msow:compl} then implies
$\MSOW \thesis \MI{{\varphi}}$,
and by~\eqref{eq:msow:bfsos:cons}
%Proposition~\ref{prop:msow:bfsos}.\eqref{eq:msow:bfsos:cons}
we have that
$\FSOW \thesis \varphi$.
We conclude by Lemma~\ref{lem:bfsos:bfso}.
\end{proof}

%%% Local Variables:
%%% mode: latex
%%% TeX-master: "main.tex"
%%% End:

%%%%%%%%%%%%%%%%%%%%%%%%%%%%%%%%%%%%%%%%%%%%%%%%%%%%%%%%%%%%%%%%%%%%%%%%%%%
\section{Completeness}
\label{sec:compl}
%%%%%%%%%%%%%%%%%%%%%%%%%%%%%%%%%%%%%%%%%%%%%%%%%%%%%%%%%%%%%%%%%%%%%%%%%%%

\noindent
This Section is devoted to the proof of our main result,
 the completeness of $\FSO + (\PosDet)$.

%%%%%%%%%%%%%%%%%%%%%%%%%%%%%%%%%%%%%%%%%%%%%%%%%%%%%%%%%%%%%%%%%%%%%%%%%%%
\begin{thm}[Main Theorem]
\label{thm:main}
%%%%%%%%%%%%%%%%%%%%%%%%%%%%%%%%%%%%%%%%%%%%%%%%%%%%%%%%%%%%%%%%%%%%%%%%%%%
For each closed formula $\varphi$ of $\FSO$,
\[
\FSO + (\PosDet) \thesis \varphi
\qquad\text{or}\qquad
\FSO + (\PosDet) \thesis \lnot \varphi
\]
\end{thm}

%%%%%%%%%%%%%%%%%%%%%%%%%%%%%%%%%%%%%%%%%%%%%%%%%%%%%%%%%%%%%%%%%%%%%%%%%%%
\subsection{Overview}
\label{sec:compl:intro}
%%%%%%%%%%%%%%%%%%%%%%%%%%%%%%%%%%%%%%%%%%%%%%%%%%%%%%%%%%%%%%%%%%%%%%%%%%%
%\noindent
%Theorem~\ref{thm:main} relies on the two main constructions of this paper.
The two main ingredients of Theorem~\ref{thm:main} are the following.
\begin{enumerate}
\item
The translations
\[
\MI{-} ~~:~~ \FSO ~~\longto~~ \MSO
\qquad\text{and}\qquad
(-)^\circ ~~:~~ \MSO ~~\longto~~ \FSO
\]
providing faithful mutual interpretations of $\FSO$ and $\MSO$
(\S\ref{sec:cons}, recapitulated in Table~\ref{tab:compl:cons}).

\item
The translation of $\MSO$-formulae to automata, that we detail
in~\S\ref{sec:compl:synt} and \S\ref{sec:compl:aut} below.
This translation relies on the correctness of the constructions
on automata of~\S\ref{sec:aut}, which are recapitulated in Table~\ref{tab:aut:op}.
In particular, we require the Axiom $(\PosDet)$
of positional determinacy of parity games (\S\ref{sec:posdet})
for the complementation of tree automata (Theorem~\ref{thm:neg}).
\end{enumerate}

%%%%%%%%%%%%%%%%%%%%%%%%%%%%%%%%%%%%%%%%%%%%%%%%%%%%%%%%%%%%%%%%%%%%%%%%%%%
\begin{table}[tbp]
%%%%%%%%%%%%%%%%%%%%%%%%%%%%%%%%%%%%%%%%%%%%%%%%%%%%%%%%%%%%%%%%%%%%%%%%%%%
\[
\begin{array}{l !{\quad\text{if and only if}\quad} l !{~~} !{{~~}} l}
\toprule
%  \multicolumn{2}{c}{\textbf{Property}}
%& \multicolumn{1}{c}{\textbf{Location in text}}
%\\
%\midrule

  \multicolumn{2}{c}{\FSO\thesis~~ \varphi \longliff \MI{\varphi}^\circ}
& \text{Proposition~\ref{prop:cons:mso:cons}, \eqref{eq:cons:mso:cons:trans}}
\\

  \FSO \thesis \varphi
& \MSO \thesis \MI{\varphi}
& \text{Proposition~\ref{prop:cons:mso:cons}, \eqref{eq:cons:mso:cons:fso:mso}}
\\

  \FSO \thesis \varphi^\circ
& \MSO \thesis \varphi
& \text{Theorem~\ref{thm:cons:mso:fso:mso}}
\\

  \Std \models \varphi^\circ
& \Std \models \varphi
& \text{Lemma~\ref{lem:cons:fso:mso:std}}
\\
\toprule
\end{array}
\]
\caption{\nameref{sec:cons} (\S\ref{sec:cons}).\label{tab:compl:cons}}
\end{table}

%%%%%%%%%%%%%%%%%%%%%%%%%%%%%%%%%%%%%%%%%%%%%%%%%%%%%%%%%%%%%%%%%%%%%%%%%%%%
%\begin{table}[tbp]
%%%%%%%%%%%%%%%%%%%%%%%%%%%%%%%%%%%%%%%%%%%%%%%%%%%%%%%%%%%%%%%%%%%%%%%%%%%%
%\[
%\begin{array}{| l !{\quad\text{if and only if}\quad} l !{~~}|!{{~~}} l |}
%\hline
%  \multicolumn{2}{| c |!{~~}}{\FSO\thesis~~ \varphi \longliff \MI{\varphi}^\circ}
%& \text{Proposition~\ref{prop:cons:mso:cons}, \eqref{eq:cons:mso:cons:trans}}
%\\
%\hline
%  \FSO \thesis \varphi
%& \MSO \thesis \MI{\varphi}
%& \text{Proposition~\ref{prop:cons:mso:cons}, \eqref{eq:cons:mso:cons:fso:mso}}
%\\
%\hline
%  \FSO \thesis \varphi^\circ
%& \MSO \thesis \varphi
%& \text{Theorem~\ref{thm:cons:mso:fso:mso}}
%\\
%\hline
%  \Std \models \varphi^\circ
%& \Std \models \varphi
%& \text{Lemma~\ref{lem:cons:fso:mso:std}}
%\\
%\hline
%\end{array}
%\]
%\hrule
%\caption{\nameref{sec:cons} (\S\ref{sec:cons})\label{tab:compl:cons}}
%\end{table}

\noindent
The mutual interpretability results of Table~\ref{tab:compl:cons} 
also allows us to obtain a completeness result for $\MSO$.
%Let us write $\MI{(\PosDet)}$
%for the axiom scheme consisting of the $\MI{-}$-translations
%of all closed instances of $(\PosDet)$.
Recall that $\MI{\PosDet}$ is defined in Definition~\ref{def:posdet:mso},
\S\ref{sec:posdet:mso}.
We then get the following corollary to Theorem~\ref{thm:main}.

%%%%%%%%%%%%%%%%%%%%%%%%%%%%%%%%%%%%%%%%%%%%%%%%%%%%%%%%%%%%%%%%%%%%%%%%%%%
\begin{cor}
%%%%%%%%%%%%%%%%%%%%%%%%%%%%%%%%%%%%%%%%%%%%%%%%%%%%%%%%%%%%%%%%%%%%%%%%%%%
For each closed formula $\varphi$ of $\MSO$,
\[
%\MSO + \MI{(\PosDet)} \thesis \varphi
\MSO + \MI{\PosDet} \thesis \varphi
\qquad\text{or}\qquad
\MSO + \MI{\PosDet} \thesis \lnot \varphi
%\MSO + \MI{(\PosDet)} \thesis \lnot \varphi
\]
\end{cor}

\begin{proof}
Consider a closed $\MSO$-formula $\varphi$.
Assume $\FSO + (\PosDet) \thesis \varphi^\circ$.
Let $\formfont{PosDet}(\Prop_i,\Opp_i,n_i)$ ($i=1,\dots,k$)
be the instances of $(\PosDet)$ used in the proof, so that
%$\FSO$ proves
\[
\FSO\thesis\quad
\land_{1 \leq i \leq k} \formfont{PosDet}(\Prop_i,\Opp_i,n_i)
~~\longlimp~~
\varphi^\circ
\]

\noindent
By~\eqref{eq:cons:mso:cons:trans} (Proposition~\ref{prop:cons:mso:cons}),
we get
%Then $\FSO$ proves 
\[
\FSO\thesis \quad
\land_{1 \leq i \leq k} \MI{\formfont{PosDet}(\Prop_i,\Opp_i,n_i)}^\circ
~~\longlimp~~
\varphi^\circ
\]
and since $(-)^\circ$ commutes over propositional connectives,
by Theorem~\ref{thm:cons:mso:fso:mso} we obtain 
\[
\MSO \thesis\quad
\land_{1 \leq i \leq k} \MI{\formfont{PosDet}(\Prop_i,\Opp_i,n_i)}
~~\longlimp~~
\varphi
\]

\noindent
Moreover, since $\varphi^\circ$ is $\HF$-closed, we can assume
the $\HF$-terms $\Prop_i$, $\Opp_i$ and $n_i$ to be closed.
It follows that there are constants for $\HF$-sets
$\const\Prop_i$, $\const\Opp_i$ and $\const n_i$ ($i = 1\dots,k$)
such that
each formula $\MI{\formfont{PosDet}(\Prop_i,\Opp_i,n_i)}$
is syntactically identical to
$\MI{\formfont{PosDet}(\const\Prop_i,\const\Opp_i,\const n_i)}$.
We thus obtain
\[
\MSO\thesis \quad
\land_{1 \leq i \leq k}
\MI{\formfont{PosDet}(\const\Prop_i,\const\Opp_i,\const n_i)}
~~\longlimp~~
\varphi
\]
which implies that $\MSO + \MI{\PosDet}$ proves $\varphi$.

If $\FSO + (\PosDet)$ does not prove $\varphi^\circ$,
Theorem~\ref{thm:main} gives
$\FSO + (\PosDet) \thesis \lnot(\varphi^\circ)$
and we conclude similarly.
\end{proof}

%\begin{proof}
%Consider a closed $\MSO$-formula $\varphi$.
%Assume $\FSO + (\PosDet) \thesis \varphi^\circ$.
%Then $\FSO$ proves $\MI{(\PosDet)}^\circ \limp \varphi^\circ$
%by~\eqref{eq:cons:mso:cons:trans} (Proposition~\ref{prop:cons:mso:cons}).
%Since $(-)^\circ$ commutes over propositional connectives,
%we obtain 
%$\MSO + \MI{(\PosDet)} \thesis \varphi$
%by Theorem~\ref{thm:cons:mso:fso:mso}.
%%
%Otherwise, Theorem~\ref{thm:main} gives
%$\FSO + (\PosDet) \thesis \lnot(\varphi^\circ)$
%and we conclude similarly.
%\end{proof}

\noindent
In particular, it follows from Proposition~\ref{prop:games:std:cor}
that $\FSO + (\PosDet)$ completely axiomatizes the standard
model $\Std$ of $\Dir$-ary trees.

%%%%%%%%%%%%%%%%%%%%%%%%%%%%%%%%%%%%%%%%%%%%%%%%%%%%%%%%%%%%%%%%%%%%%%%%%%%
\begin{cor}
\label{cor:compl:std}
%%%%%%%%%%%%%%%%%%%%%%%%%%%%%%%%%%%%%%%%%%%%%%%%%%%%%%%%%%%%%%%%%%%%%%%%%%%
\hfill
\begin{itemize}
\item For each closed formula $\varphi$ of $\FSO$,
\[
\Std\models \varphi
\qquad\text{if and only if}\qquad
\FSO + (\PosDet) \thesis \varphi
\]

\item For each closed formula $\varphi$ of $\MSO$,
\[
\Std\models \varphi
\qquad\text{if and only if}\qquad
\MSO + \MI{\PosDet} \thesis \varphi
%\MSO + \MI{(\PosDet)} \thesis \varphi
\]
\end{itemize}
\end{cor}

%%%%%%%%%%%%%%%%%%%%%%%%%%%%%%%%%%%%%%%%%%%%%%%%%%%%%%%%%%%%%%%%%%%%%%%%%%%
\begin{rem}
\label{rem:compl:rec}
%%%%%%%%%%%%%%%%%%%%%%%%%%%%%%%%%%%%%%%%%%%%%%%%%%%%%%%%%%%%%%%%%%%%%%%%%%%
Note that it follows from Remark~\ref{rem:ax:hf:rec}
that Theorem~\ref{thm:main} together with Corollary~\ref{cor:compl:std}
implies the decidability of $\FSO$ over its standard model $\Std$.
By Lemma~\ref{lem:cons:fso:mso:std} (see Table~\ref{tab:compl:cons})
we thus obtain a proof of Rabin's Tree Theorem~\cite{rabin69tams},
namely the decidability of $\MSO$ over $\Std$.
%We further elaborate on this in~\S\ref{sec:remcompl}.
However, even if provability in $\FSO$ is semi-recursive,
the axiom set of $\FSO$ is not recursive and the interpretation
of $\HF$-Functions is not computable
(see Remarks~\ref{rem:ax:hf:rec} and~\ref{rem:ax:hf:compl}
in~\S\ref{sec:ax:hf},
as well as Remark~\ref{rem:cons:dec} in~\S\ref{sec:cons:fso:mso}).
We further elaborate on this in~\S\ref{sec:remcompl}.
\end{rem}

\noindent
We will actually deduce Theorem~\ref{thm:main}
via Proposition~\ref{prop:cons:mso:cons},
\eqref{eq:cons:mso:cons:trans}
(see Table~\ref{tab:compl:cons})
from the following.

%%%%%%%%%%%%%%%%%%%%%%%%%%%%%%%%%%%%%%%%%%%%%%%%%%%%%%%%%%%%%%%%%%%%%%%%%%%
\begin{thm}
\label{thm:main:msoform}
%%%%%%%%%%%%%%%%%%%%%%%%%%%%%%%%%%%%%%%%%%%%%%%%%%%%%%%%%%%%%%%%%%%%%%%%%%%
For each closed formula $\varphi$ of $\MSO$,
\[
\FSO + (\PosDet) \thesis \varphi^\circ
\qquad\text{or}\qquad
\FSO + (\PosDet) \thesis \lnot \varphi^\circ
\]
%either
%$\FSO + (\PosDet) \thesis \varphi^\circ$
%or
%$\FSO + (\PosDet) \thesis \lnot \varphi^\circ$.
\end{thm}

\noindent
%Note that Theorem~\ref{thm:main:msoform}
%involves the translation
%$(-)^\circ : \MSO \to \FSO$ of~\S\ref{sec:cons}.
The proof of Theorem~\ref{thm:main:msoform}
%Its proof
proceeds as expected via a translation of $\MSO$-formulae
to automata.
As usual, such translations are easier to define when one starts from a version
of $\MSO$ with a purely relational and individual-free language.
We perform a translation of $\MSO$ to such a language in~\S\ref{sec:compl:synt}.
Then, the translation of formulae to automata is presented in~\S\ref{sec:compl:aut}.
It relies on the constructions of~\S\ref{sec:aut}.
We thus arrive at Proposition~\ref{prop:compl:aut},
namely that
for each closed formula $\varphi$ of $\MSO$
there is an $\HF$-closed parity automaton $\At A$ over the
singleton alphabet $\one$
such that
\[
\FSO + (\PosDet)\thesis\quad
\varphi^\circ ~~\longliff~~
(\exists F:\one)
\big(F \in \Lang(\At A) \big)
\]

\noindent
In order to obtain Theorem~\ref{thm:main:msoform},
it remains to show that 
$\FSO$ actually \emph{decides} the emptiness of such automata:
%for such automata, %$\At A$,
%%for an $\HF$-closed parity automaton $(\At A:\one)$,
%$\FSO$ actually \emph{decides} the emptiness of $\Lang(\At A)$:
\[
\FSO \thesis (\exists F:\one)\big(F \in \Lang(\At A))
\qquad\text{or}\qquad
\FSO \thesis \lnot (\exists F:\one)\big(F \in \Lang(\At A))
\]

\noindent
This is Proposition~\ref{prop:compl:aut:dec}.
Its proof relies on the fact that the acceptance games of $(\At A:\one)$
are actually generated from closed $\HF$-Sets.
We call such games \emph{reduced parity games}.
Section~\ref{sec:compl:red} is devoted to defining reduced parity games
and to showing that $\FSO$ decides winning for them
(Theorem~\ref{thm:compl:red}).
%Theorem~\ref{thm:compl:red} itself relies on the completeness
%of $\FSOW$ and on the lifting of its theory to the paths of $\FSO$ (\S\ref{sec:msow}).
%
This essentially amounts to
a version of the Büchi-Landweber Theorem~\cite{bl69tams}
(see also \eg~\cite{thomas97handbook,pp04book}),
the effective determinacy 
of parity games on finite graphs,
which is obtained thanks to the completeness of $\FSOW$ (\S\ref{sec:msow}).
Theorem~\ref{thm:compl:red} then follows from the lifting of $\FSOW$
to the paths of $\FSO$ (Proposition~\ref{prop:bfsos:bfso}).

%%%%%%%%%%%%%%%%%%%%%%%%%%%%%%%%%%%%%%%%%%%%%%%%%%%%%%%%%%%%%%%%%%%%%%%%%%%
\subsection{Restricted Languages for $\MSOD$}
\label{sec:compl:synt}
%%%%%%%%%%%%%%%%%%%%%%%%%%%%%%%%%%%%%%%%%%%%%%%%%%%%%%%%%%%%%%%%%%%%%%%%%%%

For the translation of formulae to automata, it is
useful and customary to work with formulae in a slightly different syntax,
based on a purely relational, individual-free vocabulary.

%%%%%%%%%%%%%%%%%%%%%%%%%%%%%%%%%%%%%%%%%%%%%%%%%%%%%%%%%%%%%%%%%%%%%%%%%%%
%\subsubsection{Relational syntax}
%\subsubsection{Translation to the Relational Language $\F^R_\Dir$}
\subsubsection{Restriction to a Relational Language}
\label{sec:compl:synt:rel}
%%%%%%%%%%%%%%%%%%%%%%%%%%%%%%%%%%%%%%%%%%%%%%%%%%%%%%%%%%%%%%%%%%%%%%%%%%%

We first restrict to a purely relational vocabulary,
based on the defined formulae
\begin{equation}
\tag{for each $d \in \Dir$}
\FSucc_d(x,y) \quad\deq\quad (\Succ_d(x) \Eq y)
%\qquad\text{(for each $d \in \Dir$)}
\end{equation}

\noindent
The \emph{relational formulae} $\varphi,\psi \in \F^R_\Dir$
are built from atomic formulae $X y$ and $\FSucc_d(x,y)$
by means of $\lnot$, $\lor$, $\exists x$ and $\exists X$.
To each $\MSO$-formula $\varphi \in \F$
we associate a formula $\varphi^R$ as follows.
For $t$ a term of $\MSO$,
define the formula $(z \eqcirc t)$ by structural induction on $t$:
\[
\begin{array}{r !{\quad\deq\quad} l}
  (z \eqcirc y)
& (z \Eq y)
\\
  (z \eqcirc \Root)
& \lnot (\exists z') \bigdisj_{d \in \Dir} \FSucc_d(z',z)
\\
  (z \eqcirc \Succ_d(t))
& (\exists z')\big( z' \eqcirc t \land \FSucc_d(z',z) \big)
\end{array}
\]
Note that 
\[
\MSO\thesis\quad (z \eqcirc t) ~~\liff~~ (z \Eq t)
\]

\noindent
Then, $\varphi^R$ is obtained from $\varphi$
by replacing each atomic formula $X t$, where $t$ is not a variable,
by $(\exists z) [(z \eqcirc t) \land X z]$, where $z$ is a fresh variable.

%%%%%%%%%%%%%%%%%%%%%%%%%%%%%%%%%%%%%%%%%%%%%%%%%%%%%%%%%%%%%%%%%%%%%%%%%%%
\begin{lem}
\label{lem:sketch:rel}
%%%%%%%%%%%%%%%%%%%%%%%%%%%%%%%%%%%%%%%%%%%%%%%%%%%%%%%%%%%%%%%%%%%%%%%%%%%
For every $\MSO$-formula $\varphi$, we have
$\MSO \thesis \varphi \liff \varphi^R$.
\end{lem}

%%%%%%%%%%%%%%%%%%%%%%%%%%%%%%%%%%%%%%%%%%%%%%%%%%%%%%%%%%%%%%%%%%%%%%%%%%%
%\subsubsection{Individual-free syntax}
\subsubsection{Restriction to an Individual-Free Language}
\label{sec:compl:synt:if}
%%%%%%%%%%%%%%%%%%%%%%%%%%%%%%%%%%%%%%%%%%%%%%%%%%%%%%%%%%%%%%%%%%%%%%%%%%%

The next step is to get rid of individual quantifiers.
Consider the defined formulae:
\[
\begin{array}{r !{\quad\deq\quad} l}
(X \Sle Y) & (\forall x) (X x ~~\limp~~ Y x)
\\
\FSucc_d(X,Y) &
(\exists x) (\exists y) \left[
X x ~~\land~~ Y y ~~\land~~ \FSucc_d(x,y)
\right]
\end{array}
\]
The \emph{individual-free formulae}
$\varphi,\psi \in \IF\F_\Dir$
are built from atomic formulae
$(X \Sle Y)$
and $\FSucc_d(X,Y)$
by means of negation, disjunction and
\emph{second-order monadic} quantification $\exists X$ only.
Let $\varphi \in \F^R$
with free variables among $x_1,\dots,x_p,Y_1,\dots,Y_q$.
We inductively associate to $\varphi$ a formula $\IF\varphi \in \IF\F$
with free variables among
$X_1,\dots,X_p,Y_1,\dots,Y_q$
%(with $\FV(\IF\varphi) = \{X_1,\dots,X_p,Y_1,\dots,Y_q\}$)
%inductively 
as follows.
Let
\[
\IF{((\exists x_{p+1}) \varphi)}
\quad\deq\quad
(\exists X_{p+1}) \left[\Sing(X_{p+1}) ~~\land~~ \IF\varphi \right]
\]

\noindent
where
\[
\begin{array}{r !{\quad\deq\quad} l}
  \Sing(X)
& \lnot(X \Eq \emptyset) 
  ~~\land~~
  (\forall Y)
  \left[
    Y \Sle X ~~\limp~~ \left( Y \Eq \emptyset ~~\lor~~ X \Sle Y \right)
  \right]
\\
  (X \Eq \emptyset) 
& (\forall Y)(X \Sle Y)
\end{array}
\]

\noindent
%where $\Sing(X)$
%is defined as
%%we let
%\[
%%\Sing(X) \ \deq\
%\lnot(X \Eq \emptyset) \land
%\forall Y\left[
%Y \Sle X \limp \left(
%Y \Eq \emptyset \lor X \Sle Y 
%\right)
%\right]
%\]
%with
%$(X \Eq \emptyset) \deq \forall Y(X \Sle Y)$.
The other inductive cases are given as follows:
\[
\begin{array}{r !{\quad\deq\quad}l !{\qquad} r !{\quad\deq\quad} l}
  \IF{(Y_j(x_i))}
& X_i \Sle Y_j
& \IF{(\FSucc_d(x_i,x_j))}
& \FSucc_d(X_i,X_j)
\\
  \IF{(\lnot \varphi)}
& \lnot \IF\varphi
& \IF{(\varphi \lor \psi)}
& \IF\varphi ~\lor~ \IF\psi
\\
  \IF{((\exists Y_{q+1}) \varphi)}
& (\exists Y_{q+1}) \IF\varphi
& \multicolumn{2}{c}{}
%\IF{(\exists x_{p+1} \varphi)} & 
%\exists X_{p+1}\, \Sing(X_{p+1}) \land \IF\varphi
\end{array}
\]

%with $(X \Eq \emptyset) \deq \forall Y(X \Sle Y)$.

%%%%%%%%%%%%%%%%%%%%%%%%%%%%%%%%%%%%%%%%%%%%%%%%%%%%%%%%%%%%%%%%%%%%%%%%%%%
\begin{lem}
%%%%%%%%%%%%%%%%%%%%%%%%%%%%%%%%%%%%%%%%%%%%%%%%%%%%%%%%%%%%%%%%%%%%%%%%%%%
For every formula $\varphi \in \F^R$
with free variables among $\vec x, \vec Y$,
we have
\[
\vec{X x} \,,\, \Sing(\vec X) \thesis_{\MSO}\quad
\varphi ~~\liff~~ \IF\varphi
\]
\end{lem}

\noindent
By composing the translations
$(-)^R : \F \to \F^R$
and
$\IF{(-)} : \F^R \to \IF\F$,
we obtain:

%%%%%%%%%%%%%%%%%%%%%%%%%%%%%%%%%%%%%%%%%%%%%%%%%%%%%%%%%%%%%%%%%%%%%%%%%%%
\begin{cor}
\label{cor:compl:restrsynt}
%%%%%%%%%%%%%%%%%%%%%%%%%%%%%%%%%%%%%%%%%%%%%%%%%%%%%%%%%%%%%%%%%%%%%%%%%%%
For every closed $\MSO$-formula $\varphi$,
there is a closed formula $\psi \in \IF\F$
such that $\MSO \thesis \varphi \liff \psi$.
\end{cor}

%%%%%%%%%%%%%%%%%%%%%%%%%%%%%%%%%%%%%%%%%%%%%%%%%%%%%%%%%%%%%%%%%%%%%%%%%%%
\subsection{From Formulae to Automata}
\label{sec:compl:aut}
%%%%%%%%%%%%%%%%%%%%%%%%%%%%%%%%%%%%%%%%%%%%%%%%%%%%%%%%%%%%%%%%%%%%%%%%%%%

\noindent
We are now going to associate to each formula $\varphi \in \IF\F$
with free variables among $X_1,\dots,X_p$
an $\HF$-closed parity automaton $\At A(\varphi) : \two^p$
such that
\[
\FSO + (\PosDet)
\thesis\quad
(\forall F_{X_1} : \two)\dots(\forall F_{X_p}:\two)
\Big(
\pair{F_{X_1},\dots,F_{X_p}} \in \Lang(\At A(\varphi))
~~\liff~~
\varphi^\circ
\Big)
\]

\noindent
Note that the correctness of $\At A(\varphi)$ \wrt\@ $\varphi$
is proved in $\FSO$ using the translation
$(-)^\circ : \MSO \to \FSO$
of~\S\ref{sec:cons:mso:fso}.
%of~\S\ref{sec:cons}.
Recall that $(-)^\circ$ replaces each monadic variable $X_i$ of $\varphi$
by a Function variable $(F_{X_i} : \two)$.
The construction of $\At A(\varphi)$ from $\varphi$
is done by induction on $\varphi$
using the operations on automata devised in~\S\ref{sec:aut}
(see Table~\ref{tab:aut:op}).
The base cases are provided by
the automata $\At A(X_i \Sle X_j)$ and $\At A(\FSucc_d(X_i,X_j))$
discussed in~\S\ref{sec:compl:aut:at} below
for the atomic formulae of $\IF\F$.
% are provided in~\S\ref{sec:compl:aut:at}.
The inductive cases are performed as follows,
where we implicitly apply substitutions (cf.~\S\ref{sec:aut:subst})
when necessary:
\[
\begin{array}{r !{\quad\deq\quad} l !{\qquad} l}
  \At A (\varphi \lor \psi)
& \At A(\varphi) \oplus \At A(\psi)
& \text{(Lemma~\ref{lem:aut:disj})}
\\
  \At A(\lnot \varphi)
& \aneg \At A(\varphi)
& \text{(Theorem~\ref{thm:neg})}
\\
  \At A((\exists X_{p+1})\varphi)
& \exists_\two \ND(\At A(\varphi))
& \text{(Proposition~\ref{prop:aut:nd:proj} ~\&~ Theorem~\ref{thm:aut:nd:sim})}
\end{array}
\]

\noindent
In particular, if $\varphi$ is closed then $\At A (\varphi)$
is an automaton over the singleton alphabet $\one$,
whence by Corollary~\ref{cor:compl:restrsynt}
we have:

%%%%%%%%%%%%%%%%%%%%%%%%%%%%%%%%%%%%%%%%%%%%%%%%%%%%%%%%%%%%%%%%%%%%%%%%%%%
\begin{prop}
\label{prop:compl:aut}
%%%%%%%%%%%%%%%%%%%%%%%%%%%%%%%%%%%%%%%%%%%%%%%%%%%%%%%%%%%%%%%%%%%%%%%%%%%
For each closed formula $\varphi$ of $\MSO$
there is an $\HF$-closed parity automaton $(\At A : \one)$
such that
\[
\FSO + (\PosDet)\thesis\quad
\varphi^\circ ~~\liff~~
(\exists F:\one)\big(F \in \Lang(\At A) \big)
\]
\end{prop}

In order to obtain Theorem~\ref{thm:main:msoform} from Proposition~\ref{prop:compl:aut},
it remains to show that $\FSO$ actually \emph{decides} the emptiness
of $\Lang(\At A)$ for an $\HF$-closed parity automaton $\At A$ over the singleton
alphabet $\one$.

%%%%%%%%%%%%%%%%%%%%%%%%%%%%%%%%%%%%%%%%%%%%%%%%%%%%%%%%%%%%%%%%%%%%%%%%%%%
\begin{prop}
\label{prop:compl:aut:dec}
%%%%%%%%%%%%%%%%%%%%%%%%%%%%%%%%%%%%%%%%%%%%%%%%%%%%%%%%%%%%%%%%%%%%%%%%%%%
Given an $\HF$-closed parity automaton $(\At A:\one)$,
\[
\FSO \thesis (\exists F:\one)\big(F \in \Lang(\At A))
\qquad\text{or}\qquad
\FSO \thesis \lnot (\exists F:\one)\big(F \in \Lang(\At A))
\]
\end{prop}

\noindent
Proposition~\ref{prop:compl:aut:dec} is proved in~\S\ref{sec:compl:red} below.

%\noindent
%Proposition~\ref{prop:compl:aut:dec}
%follows from the Büchi-Landweber Theorem~\cite{bl69tams}
%which states the \emph{effective} positional determinacy of parity games
%over \emph{finite} graphs.
%We formalize the Büchi-Landweber Theorem~\cite{bl69tams} 
%and prove
%Proposition~\ref{prop:compl:aut:dec} 
%in~\S\ref{sec:compl:red} below.

%%%%%%%%%%%%%%%%%%%%%%%%%%%%%%%%%%%%%%%%%%%%%%%%%%%%%%%%%%%%%%%%%%%%%%%%%%%
%\subsubsection{Automata for Atomic Formulae}
\subsubsection{Automata for Atomic Formulae}
\label{sec:compl:aut:at}
%%%%%%%%%%%%%%%%%%%%%%%%%%%%%%%%%%%%%%%%%%%%%%%%%%%%%%%%%%%%%%%%%%%%%%%%%%%

We provide $\HF$-closed parity automata for the atomic formulae $(X_1 \Sle X_2)$
and $\FSucc_d(X_1,X_2)$ of the individual-free syntax 
$\IF\F$ of $\MSO$.
\begin{itemize}
\item %\emph{$(X \Sle Y)$.}
The automaton $\At A(X_1 \Sle X_2)$ over $\two \times \two$
has state set $\Bool= \{\true,\false\}$, with $\true$ initial,
transitions given by
\[
%\begin{array}{r !{\quad}c!{\quad} l}
\begin{array}{r !{\quad}c!{\quad} l !{\quad} l}
%  \Bool \times (\two \times \two)
%& \longto
%& \Po(\Po(\Dir \times \Bool))
%\\

%  (\true, (i,j))
%& \longmapsto
%& \left\{
%  \begin{array}{l !{\quad} l}
%    \{\{(d,\false) \st d \in \Dir\}\}
%  & \text{if $i=0$ and $j = 1$}
%  \\
%    \{\{(d,\true) \st d \in \Dir\}\}
%  & \text{otherwise}
%  \end{array}
%  \right.
%\\

  (\true, (i,j))
& \longmapsto
& \{\{(d,\false) \st d \in \Dir\}\}
& \text{if $i=1$ and $j = 0$}
\\

  (\true, (i,j))
& \longmapsto
& \{\{(d,\true) \st d \in \Dir\}\}
& \text{otherwise}
\\

  (\false, (-,-))
& \longmapsto
& \{\{(d,\false) \st d \in \Dir\}\}
\end{array}
\]
and coloring $\funto{\col}{\Bool}{\two}$
given by
\[
\col(\true) \quad\deq\quad 0
\qquad\text{and}\qquad
\col(\false) \quad\deq\quad 1
\]

\item %\emph{$\FSucc_d(X,Y)$.}
For $d \in \Dir$,
the automaton $\At A(\FSucc_d(X_1,X_2))$ over $\two \times \two$
has state set $Q_{\FSucc} \deq \Bool + \{\al w\}$,
with $\false$ initial,
transitions given by
\[
\begin{array}{r !{\quad}c!{\quad} l}
%  Q_{\FSucc} \times (\two \times \two)
%& \longto
%& \Po(\Po(\Dir \times Q_{\FSucc}))
%\\
  (\false, (0,-))
& \longmapsto
& \{\{(d',\false) \} \st d' \in \Dir \}
\\
  (\false, (1,-))
& \longmapsto
& \{\{(d,\al w)\}\}
\\
  (\al w, (-,1))
& \longmapsto
& \{\{(d',\true) \st d' \in \Dir\}\}
\\
  (\al w, (-,0))
& \longmapsto
& \{\{(d',\false) \} \st d' \in \Dir\}
\\
  (\true, (-,-))
& \longmapsto
& \{\{(d',\true) \st d' \in \Dir\}\}
\end{array}
\]
and with coloring given by
\[
\col(\true) \quad\deq\quad 0
\qquad\qquad
\col(\false) \quad\deq\quad 1
\qquad\qquad
\col(\al w) \quad\deq\quad 0
\]
\end{itemize}

%%%%%%%%%%%%%%%%%%%%%%%%%%%%%%%%%%%%%%%%%%%%%%%%%%%%%%%%%%%%%%%%%%%%%%%%%%%
\begin{rem}
%%%%%%%%%%%%%%%%%%%%%%%%%%%%%%%%%%%%%%%%%%%%%%%%%%%%%%%%%%%%%%%%%%%%%%%%%%%
Recall from~\S\ref{sec:compl:synt:if}
that the formula $\FSucc_d(X,Y)$ of the individual-free syntax
$\IF\F$ amounts in $\MSO$ to the formula
$(\exists x) (\exists y)\left[X x \land Y y \land y \Eq \Succ_d(x) \right]$.
So the automaton $\At A(\FSucc_d(X,Y))$ only looks \emph{for some}
$x \in X$ and $y \in Y$ such that $y$ is the $d$-successor of $x$,
but is does not check whether $X$ and $Y$ are singletons.
\end{rem}

%%%%%%%%%%%%%%%%%%%%%%%%%%%%%%%%%%%%%%%%%%%%%%%%%%%%%%%%%%%%%%%%%%%%%%%%%%%
\begin{lem}
\label{lem:compl:aut:at}
%%%%%%%%%%%%%%%%%%%%%%%%%%%%%%%%%%%%%%%%%%%%%%%%%%%%%%%%%%%%%%%%%%%%%%%%%%%
$\FSO$ proves that
\[
\begin{array}{l @{} l !{~~\liff~~} l}
  (\forall F_{X_1}:\two) (\forall F_{X_2}:\two) \Big(
& \pair{F_{X_1},F_{X_2}} \in \Lang(\At A(X_1,X_2))
& (X_1 \Sle X_2)^\circ
  \Big)
\\
  (\forall F_{X_1}:\two) (\forall F_{X_2}:\two) \Big(
& \pair{F_{X_1},F_{X_2}} \in \Lang(\At A(\FSucc_d(X_1,X_2)))
& (\FSucc_d(X_1,X_2))^\circ
  \Big)
\end{array}
\]
\end{lem}

%%%%%%%%%%%%%%%%%%%%%%%%%%%%%%%%%%%%%%%%%%%%%%%%%%%%%%%%%%%%%%%%%%%%%%%%%%%
\subsection{Reduced Parity Games}
\label{sec:compl:red}
%%%%%%%%%%%%%%%%%%%%%%%%%%%%%%%%%%%%%%%%%%%%%%%%%%%%%%%%%%%%%%%%%%%%%%%%%%%
The goal of this Section is to prove Proposition~\ref{prop:compl:aut:dec},
namely that for 
an $\HF$-closed parity automaton $\At A$
over the singleton alphabet $\one$,
\[
\FSO \thesis (\exists F:\one)\big(F \in \Lang(\At A))
\qquad\text{or}\qquad
\FSO \thesis \lnot (\exists F:\one)\big(F \in \Lang(\At A))
\]

\noindent
Consider an $\HF$-closed automaton $\At A$ over the singleton alphabet $\one = \{0\}$.
Then for any $(F : \one)$
the game $\G \deq \G(\At A,F)$ has edge relations induced by functions
\begin{equation}
\label{eq:compl:red:edge}
\ep : \PL{\G} \fsoto \Pne(\OL{\G})
\qquad\text{and}\qquad
\eo : \OL{\G} \fsoto \Pne(\Dir \times \PL{\G})
\end{equation}
given (following Remark~\ref{rem:aut:hfclosed})
by
\[
\begin{array}{l !{\quad\text{iff}\quad} l}
  (q',\Conj) \In e_\Prop(q)
& q' \Eq q ~~\land~~ \Conj \in \trans_{\At A}(q,0)
\\
  (d,q') \In e_\Opp(q,\Conj)
& (d,q') \In \Conj
\end{array}
\]

\noindent
So in particular the edge relations of $\G(\At A,F)$
are independent from $F$.
But also, 
since
\[
\PL{\G} ~~\deq~~ Q_{\At A}
\qquad\text{and}\qquad
\OL{\G} ~~\deq~~ Q_{\At A} \times \Pne(\Dir \times Q_{\At A})
\]
the whole game $\G(\At A,F)$
is actually generated from $\HF$-Sets.

In this Section, we discuss games generated from $\HF$-Sets,
that we call \emph{reduced games}.
We show that for reduced parity games, winning can actually be defined 
within $\FSOW$.
Thanks to the completeness of $\FSOW$ \wrt\@ its standard model
(\S\ref{sec:msow}),
this implies that $\FSOW$ itself decides winning in such games.
This essentially amounts to a version of the Büchi-Landweber Theorem~\cite{bl69tams}
using the completeness of $\MSOW$ over its standard model.
%(see also \eg~\cite{thomas97handbook,pp04book}).
Using Proposition~\ref{prop:bfsos:bfso}
we can then lift this result to $\FSO$.

In~\S\ref{sec:compl:red:hf} and~\S\ref{sec:compl:red:omega}
we repeat some material of~\S\ref{sec:games},
but for the slightly different setting of reduced games.
We then obtain that $\FSOW$ decides winning in reduced parity games,
%(Proposition~\ref{prop:compl:red:omega}),
and we lift this to $\FSO$ in Theorem~\ref{thm:compl:red}, \S\ref{sec:compl:red:fso}.
%This gives Theorem~\ref{thm:compl:red},
%which easily entails Proposition~\ref{prop:compl:aut:dec}.
This directly entails Proposition~\ref{prop:compl:aut:dec}.

%%%%%%%%%%%%%%%%%%%%%%%%%%%%%%%%%%%%%%%%%%%%%%%%%%%%%%%%%%%%%%%%%%%%%%%%%%%
\subsubsection{Reduced Games as $\HF$-Sets}
\label{sec:compl:red:hf}
%%%%%%%%%%%%%%%%%%%%%%%%%%%%%%%%%%%%%%%%%%%%%%%%%%%%%%%%%%%%%%%%%%%%%%%%%%%
The purpose of this Section is to give
adaptations of the notions of~\S\ref{sec:games}
to those parity games which are entirely generated from $\HF$-Sets.
All the formulae of this Section are $\HF$-formulae
in the sense of Definition~\ref{def:ax:hf:form}.
Hence, thanks to the Axioms of $\HF$-Sets 
(Remark~\ref{rem:ax:hf:vomega}, \S\ref{sec:ax:hf})
their closed instances are provable (both in $\FSO$ and $\FSOW$)
if and only if they hold in $V_\omega$.

%%%%%%%%%%%%%%%%%%%%%%%%%%%%%%%%%%%%%%%%%%%%%%%%%%%%%%%%%%%%%%%%%%%%%%%%%%%
\begin{defi}[Reduced Games]
\label{def:compl:red:games}
%%%%%%%%%%%%%%%%%%%%%%%%%%%%%%%%%%%%%%%%%%%%%%%%%%%%%%%%%%%%%%%%%%%%%%%%%%%
A \emph{reduced game} $G$ 
is given by $\HF$-terms $\Prop,\Opp,\ep,\eo$
which satisfy the following formula
\[
\Game_0(\Prop,\Opp,\ep,\eo) \quad\deq\quad
\Big(
  \Labels(\Prop,\Opp)
~~\land~~
  \ep : \Prop \fsoto \Pne(\Opp)
~~\land~~
  \eo : \Opp \fsoto \Pne(\Dir \times \Prop)
\Big)
\]

\noindent
We often write $\Game_0(G)$ for
$\Game_0(\Prop,\Opp,\ep,\eo)$.
Moreover, when no ambiguity arises, we abbreviate
$G = (\Prop,\Opp,\ep,\eo)$ as
$G = (\Prop,\Opp,\egpo G)$.
\end{defi}

%%%%%%%%%%%%%%%%%%%%%%%%%%%%%%%%%%%%%%%%%%%%%%%%%%%%%%%%%%%%%%%%%%%%%%%%%%%
\begin{defi}[Reduced Subgame]
\label{def:compl:red:sub}
%%%%%%%%%%%%%%%%%%%%%%%%%%%%%%%%%%%%%%%%%%%%%%%%%%%%%%%%%%%%%%%%%%%%%%%%%%%
We say that $G'=(\Prop',\Opp',\ep',\eo')$
is a \emph{reduced subgame} of $G = (\Prop,\Opp,\ep,\eo)$ whenever the following
formula holds
\[
\Sub_0(G',G)
\quad\deq\quad
\left\{
\begin{array}{l l}
& \Prop' \Eq \Prop ~~\land~~ \Opp' \Eq \Opp
\\
  \land
& (\forall k \in \Prop') \big( \ep'(k) \Sle \ep(k) \big)
\\
  \land
& (\forall \ell \in \Opp') \big( \eo'(\ell) \Sle \eo(\ell) \big)
\end{array}
\right.
\]
\end{defi}

%%%%%%%%%%%%%%%%%%%%%%%%%%%%%%%%%%%%%%%%%%%%%%%%%%%%%%%%%%%%%%%%%%%%%%%%%%%
\begin{defi}[Reduced Strategies]
\label{def:compl:red:strat}
%%%%%%%%%%%%%%%%%%%%%%%%%%%%%%%%%%%%%%%%%%%%%%%%%%%%%%%%%%%%%%%%%%%%%%%%%%%
Let $G = (\Prop,\Opp,\ep,\eo)$
where $\Prop,\Opp,\ep,\eo$ are $\HF$-variables.
\begin{enumerate}
\item A \emph{reduced $\Prop$-strategy on $G$}
is an $\HF$-set $s$
which satisfies the formula 
\[
  \Strat^0_\Prop(G,s)
\quad\deq\quad
\funto{s}{\Prop}{\Opp}
~~\land~~
(\forall k \in \Prop)
  \left( s(k) \in \ep(k) \right)
\]
%such that for all $u \in V_\Prop$ we have
%\[
%s(u) \In e_\Prop(u)
%\]

\item
A \emph{reduced $\Opp$-strategy on $G$} is an $\HF$-set
$s$
which satisfies the formula 
\[
  \Strat^0_\Opp(G,s)
\quad\deq\quad
\funto{s}{\Opp}{\Dir \times \Prop}
~~\land~~
(\forall \ell \in \Opp)
  \left( s(\ell) \in \eo(\ell) \right)
\]
\end{enumerate}
\end{defi}

%%%%%%%%%%%%%%%%%%%%%%%%%%%%%%%%%%%%%%%%%%%%%%%%%%%%%%%%%%%%%%%%%%%%%%%%%%%
\begin{defi}[Reduced Subgame induced by a Reduced Strategy]
%%%%%%%%%%%%%%%%%%%%%%%%%%%%%%%%%%%%%%%%%%%%%%%%%%%%%%%%%%%%%%%%%%%%%%%%%%%
Given a player $\player$ (either $\Prop$ or $\Opp$)
and a $\player$-strategy $s$ on $G$,
we let
\[
G\restr\{s\}_\player
\quad\deq\quad
\big( \PL\G \,,\, \OL\G \,,\, \egpo{G}\restr\{\strat\}_\player \big)
\]

\noindent
where
\[
\egpo{G}\restr\{\strat\}_\Prop \quad\deq\quad \big( \{s\}_\Prop,\ego\G \big)
\qquad\text{and}\qquad
\egpo{G}\restr\{\strat\}_\Opp \quad\deq\quad \big( \egp\G,\{s\}_\Opp \big)
\]

\noindent
and where $\{s\}_\player \sle e_\player$
is defined
(following the method of Remark~\ref{rem:aut:hfclosed})
%$\HF$-Bounded Choice for $\HF$-sets 
to be the function
taking $k \in G_\player$
to the singleton $\{s(k)\}$.
%and where for all $u \in \univ \times \G_\player$,
%$\{\strat\}_\player \sle e_\player$
%is defined by $\HF$-Bounded Choice to 
%be the singleton $\{\strat(u)\}$.

Whenever possible, we write $G\restr\{s\}$
or even just $s$ for $G\restr\{s\}_\player$.
\end{defi}

%%%%%%%%%%%%%%%%%%%%%%%%%%%%%%%%%%%%%%%%%%%%%%%%%%%%%%%%%%%%%%%%%%%%%%%%%%%
\subsubsection{Reduced Games in $\FSOW$}
\label{sec:compl:red:omega}
%%%%%%%%%%%%%%%%%%%%%%%%%%%%%%%%%%%%%%%%%%%%%%%%%%%%%%%%%%%%%%%%%%%%%%%%%%%
In~\S\ref{sec:compl:red:hf} we gave notions
of reduced parity games and reduced strategies.
In this Section, we work within $\FSOW$ and show
that this setting suffices to define \emph{winning}
for reduced parity games.
Thanks to the completeness of $\FSOW$ \wrt\@ the standard model of $\omega$-words
(\S\ref{sec:msow}),
we obtain that $\FSOW$ \emph{decides} winning in such games.
This is essentially the Büchi-Landweber Theorem~\cite{bl69tams}.
%(see also \eg~\cite{thomas97handbook,pp04book}).

We use the following $\FSOW$-formula:
\[
\FSucc(x,y) \quad\deq\quad
x \Lt y
~~\land~~
\lnot(\exists z)\big[x \Lt z \Lt y\big]
\]

%%%%%%%%%%%%%%%%%%%%%%%%%%%%%%%%%%%%%%%%%%%%%%%%%%%%%%%%%%%%%%%%%%%%%%%%%%%
\begin{defi}[Infinite Plays in Reduced Games]
\label{def:compl:red:omega:infplay}
%%%%%%%%%%%%%%%%%%%%%%%%%%%%%%%%%%%%%%%%%%%%%%%%%%%%%%%%%%%%%%%%%%%%%%%%%%%
Working in $\FSOW$,
let $G = (\Prop,\Opp,\ep,\eo)$,
where $\Prop,\Opp,\ep,\eo$ are $\HF$-variables.
Given an $\HF$ set $K$ and a Function $(\tilde V:\Prop)$,
we say that
$\tilde V$ is an \emph{infinite play in $G$ from $K$} when the following
formula $\Play[\Lt](G,K,\tilde V)$ holds:
\[
  \tilde V(\Root) \Eq K
  ~~\land~~
  (\forall x)(\forall y)
\Big(
\FSucc(x,y) ~~\limp~~
\big( \exists \ell \In \ep(\tilde V(x)) \big)
\big( \exists d \In \Dir \big)
\Big[
(d, \tilde V(y)) \in \eo(\ell)
\Big]
\Big)
\]
\end{defi}

\noindent
Note that in Definition~\ref{def:compl:red:omega:infplay} above,
we use the notation $\tilde V$ for a play
in a reduced games,
to mark the difference with the notion of
plays (and more generally sets of game positions)
in the setting of~\S\ref{sec:pos}.

%%%%%%%%%%%%%%%%%%%%%%%%%%%%%%%%%%%%%%%%%%%%%%%%%%%%%%%%%%%%%%%%%%%%%%%%%%%
\begin{defi}[Parity Conditions for Reduced Games]
\label{def:compl:red:omega:parity}
%%%%%%%%%%%%%%%%%%%%%%%%%%%%%%%%%%%%%%%%%%%%%%%%%%%%%%%%%%%%%%%%%%%%%%%%%%%
Working in $\FSOW$,
let $G = (\Prop,\Opp,\ep,\eo)$,
where $\Prop,\Opp,\ep,\eo$ are $\HF$-variables.
\begin{enumerate}
\item
A \emph{coloring} is given by Function $\col$ and an $\HF$-set $n$
satisfying the following formula
\[
\Coloring_0(G,\col,n)
\quad\deq\quad
\Ord(n) ~~\land~~ \funto{\col}{\Prop}{[0,n]}
\]

\item
We define the following formula:
\[
\Par[\Lt](\col,n,\tilde V)
\quad\deq\quad
\big( \exists m \in \even(n) \big)
\left[
\begin{array}{r l}
& (\forall x) (\exists y)
  \big( x \Lt y ~~\land~~ \col(\tilde V(y)) \Eq m \big)
\\
  \land
& (\exists x) (\forall y)
  \big( x \Lt y ~~\limp~~ \col(\tilde V(y)) \geq m \big)
\end{array}
\right]
\]
\end{enumerate}
\end{defi}

%%%%%%%%%%%%%%%%%%%%%%%%%%%%%%%%%%%%%%%%%%%%%%%%%%%%%%%%%%%%%%%%%%%%%%%%%%%
\begin{defi}[Winning of Reduced Parity Games]
\label{def:compl:red:omega:winning}
%%%%%%%%%%%%%%%%%%%%%%%%%%%%%%%%%%%%%%%%%%%%%%%%%%%%%%%%%%%%%%%%%%%%%%%%%%%
Working in $\FSOW$,
let $G = (\Prop,\Opp,\ep,\eo)$,
where $\Prop,\Opp,\ep,\eo$ are $\HF$-variables.
Furthermore let $\col$ be a Function variable and $n$ be an $\HF$-variable.
\begin{enumerate}
\item
We define the following formulae.
\[
\begin{array}{r !{\quad\deq\quad} l}
  \WonGame[\Lt]_\Prop(G,K,\col,n)
& (\forall \tilde V:\Prop)
  \big(
    \Play[\Lt](G,K,\tilde V) ~~\limp~~ \Par[\Lt](\col,n,\tilde V)
  \big)
\\
  \WonGame[\Lt]_\Opp(G,K,\col,n)
& (\forall \tilde V:\Prop)
  \big(
    \Play[\Lt](G,K,\tilde V) ~~\limp~~ \lnot \Par[\Lt](\col,n,\tilde V)
  \big)
\end{array}
\]

\item
Given a player $\player$ (either $\Prop$ or $\Opp$),
we say that a $\player$-strategy $s$ is
\emph{winning in $(G,\col,n)$ from $K$}
if the game $(G\{s\}_\player,\Par[\Lt](\col,n,-))$ is won by $\player$ from $K$,
\ie\@ if
the following formula holds
%$\WinStrat_\player(\G,\strat,v,\Win)$ holds:
\[
\WinStrat[\Lt]_\player(G,s,K,\col,n)
\quad\deq\quad
  \WonGame[\Lt]_\player(G\restr\{s\}_\player, K, \col,n)
\]
\end{enumerate}
\end{defi}

\noindent
Note that in Definition~\ref{def:compl:red:omega:winning},
we have denoted strategies in reduced games with a lower case roman $s$.
This notation contrasts with our notation $\strat$ for games
in the sense of~\S\ref{sec:games} in order to insist on the fact that
strategies on reduced games are $\HF$-sets.

Consider now $G=(\Prop,\Opp,\ep,\eo)$
where $\Prop$, $\Opp$, $\ep$ and $\eo$ are closed $\HF$-terms
such that
\[
V_\omega \models \Game_0(G)
\]
Assume also given closed $\HF$-terms $n$ and $\col$
such that
\[
V_\omega \models \Coloring_0(G,\col,n)
\]

\noindent
Then the positional determinacy of parity games (cf.~\cite{ej91focs})
implies
that for every $\HF$-set $\kappa \in \Prop$,
the following holds in the standard model $\StdN$
of $\FSOW$:
\begin{itemize}
\item For some player $\player$ (either $\Prop$ or $\Opp$)
there an $\HF$-set $\hf s$ such that
\[
\StdN\models\quad
\Strat^0_\player(G,\hf s)
~~\land~~
\WinStrat[\Lt]_\player(G,\hf s,\kappa,\col,n)
\]
\end{itemize}

\noindent
Thanks to the completeness of $\FSOW$ \wrt\@ $\StdN$
(Theorem~\ref{thm:msow:compl} and Proposition~\ref{prop:msow:bfsos}),
we obtain the following result,
that may be viewed as a formulation of the Büchi-Landweber Theorem~\cite{bl69tams}
(see also \eg~\cite{thomas97handbook,pp04book}).
Recall that $\Strat^0_\player(G,\hf s)$ holds in
$\StdN$ (resp.\@ $\FSOW$, $\FSO$) if and only if it holds
in $V_\omega$.

%%%%%%%%%%%%%%%%%%%%%%%%%%%%%%%%%%%%%%%%%%%%%%%%%%%%%%%%%%%%%%%%%%%%%%%%%%%
\begin{prop}
\label{prop:compl:red:bl}
%%%%%%%%%%%%%%%%%%%%%%%%%%%%%%%%%%%%%%%%%%%%%%%%%%%%%%%%%%%%%%%%%%%%%%%%%%%
Assume given closed $\HF$-terms $G=(\Prop,\Opp,\ep,\eo)$, $n$ and $\col$ such that
\[
V_\omega \models\quad 
\Game_0(G) ~~\land~ \Coloring_0(G,\col,n)
\]
Then for every $\kappa \in \Prop$,
there is a player $\player$ (either $\Prop$ or $\Opp$)
and an $\HF$-set $\hf s$ such that
\[
\FSOW\thesis\quad
%\Strat^0_\player(G,\hf s)
%~~\land~~
\WinStrat[\Lt]_\player(G,\hf s,\kappa,\col,n)
\]
\end{prop}

\noindent
Proposition~\ref{prop:bfsos:bfso},
namely
\[
{\FSOD} \thesis
(\forall P : \two)\left(
\TPath(P) ~~\limp~~
\varphi^{P}
\right)
\qquad\text{whenever}\qquad
\StdN \models \varphi
\]
(for $\varphi$ a closed $\FSOW$-formula)
moreover gives the following.

%%%%%%%%%%%%%%%%%%%%%%%%%%%%%%%%%%%%%%%%%%%%%%%%%%%%%%%%%%%%%%%%%%%%%%%%%%%
\begin{prop}
\label{prop:compl:red:omega}
%%%%%%%%%%%%%%%%%%%%%%%%%%%%%%%%%%%%%%%%%%%%%%%%%%%%%%%%%%%%%%%%%%%%%%%%%%%
Assume given closed $\HF$-terms $G=(\Prop,\Opp,\ep,\eo)$, $n$ and $\col$ such that
\[
V_\omega \models\quad 
\Game_0(G) ~~\land~ \Coloring_0(G,\col,n)
\]
Then for every $\kappa \in \Prop$,
there is a player $\player$ (either $\Prop$ or $\Opp$)
and an $\HF$-set $\hf s$ such that
\[
\FSO\thesis\quad
(\forall X : \two)
\Big(
\TPath(X) ~~\limp~~
\WinStrat[\Lt]^X_\player(G,\hf s,\kappa,\col,n)
%\big[
%\Strat^0_\player(G,\hf s)
%~~\land~~
%\WinStrat[\Lt]_\player(G,\hf s,\kappa,\col,n)
%\big]^X
\Big)
\]
\end{prop}

%%%%%%%%%%%%%%%%%%%%%%%%%%%%%%%%%%%%%%%%%%%%%%%%%%%%%%%%%%%%%%%%%%%%%%%%%%%
\subsubsection{Reduced Games in $\FSO$}
\label{sec:compl:red:fso}
%%%%%%%%%%%%%%%%%%%%%%%%%%%%%%%%%%%%%%%%%%%%%%%%%%%%%%%%%%%%%%%%%%%%%%%%%%%
We now come back to $\FSO$.
In this Section, we show, using Proposition~\ref{prop:compl:red:omega},
that $\FSO$ decides winning for parity games
induced from reduced parity games
(Theorem~\ref{thm:compl:red}).
This directly gives Proposition~\ref{prop:compl:aut:dec}.

A reduced game $G = (\Prop,\Opp,\ep,\eo)$
induces a game $\G = (\Prop,\Opp,\EP,\EO)$ in the sense
of Definition~\ref{def:games:games},
where 
\[
\funto{\EP}{\univ \times \Prop}{\Pne(\Opp)}
\qquad\text{and}\qquad
\funto{\EO}{\univ \times \Opp}{\Pne(\Dir \times \Prop)}
\]
are defined using 
$\HF$-Bounded Choice for Product Types (Theorem~\ref{thm:funto:choice})
as
\[
\EP(x,k) \quad\deq\quad \ep(k)
\qquad\text{and}\qquad
\EO(x,\ell) \quad\deq\quad \eo(\ell)
\]

\noindent
Similarly, a strategy $s$ in a reduced game $G$ induces
a strategy $\strat$ in $\G$ in the sense of Definition~\ref{def:games:strat},
with
\[
\strat(x,k) \quad\deq\quad s(k)
\]

\noindent
As for colorings, from $(\col : \Prop \fsoto [0,n])$ we define
$\funto{\hat\col}{\G}{[0,n]}$ as in Definition~\ref{def:aut:parity}:
\[
\hat \col(x,k) \quad\deq\quad
\left\{
\begin{array}{l !{\quad} l !{\qquad} l}
  \col(k)
& \text{if $k \in \Prop$}
& \text{($\Prop$-position)}
\\
  n_{\At A}
& \text{if $k \in \Opp$}
& \text{($\Opp$-position)}
\end{array}
\right.
\]

\noindent
We clearly have the following:
\[
\begin{array}{l !{\qquad\text{whenever}\qquad} l}
  \FSO \thesis \Game(\G)
& V_\omega \models \Game_0(G)
\\
  \FSO \thesis \Strat_\player(\G,\strat)
& V_\omega \models \Strat^0_\player(G,s)
\\
  \FSO \thesis \Coloring(\G,\hat\col,n)
& V_\omega \models \Coloring_0(G,\col,n)
\end{array}
\]

%%%%%%%%%%%%%%%%%%%%%%%%%%%%%%%%%%%%%%%%%%%%%%%%%%%%%%%%%%%%%%%%%%%%%%%%%%%
\begin{thm}
\label{thm:compl:red}
%%%%%%%%%%%%%%%%%%%%%%%%%%%%%%%%%%%%%%%%%%%%%%%%%%%%%%%%%%%%%%%%%%%%%%%%%%%
Assume given closed $\HF$-terms $G=(\Prop,\Opp,\ep,\eo)$, $n$ and $\col$ such that
\[
V_\omega \models\quad 
\Game_0(G) ~~\land~ \Coloring_0(G,\col,n)
\]

\noindent
Then for every $\kappa \in \Prop$,
\[
\begin{array}{r !{\qquad} l r l}
  \text{either}
& \FSO\thesis
& (\exists \funto{\strat_\Prop}{\PP\G}{\Opp})
& \left(
  \begin{array}{c l}
  & \Strat_\Prop(\G,\strat_\Prop)
  \\
    \land
  & \WinStrat_\Prop(\G(\gle),\strat_\Prop,\kappa,\col,n)
  \end{array}
  \right)
\\\\
  \text{or}
& \FSO\thesis
& (\exists \funto{\strat_\Opp}{\OP\G}{\Dir \times \Prop})
& \left(
  \begin{array}{c l}
  & \Strat_\Opp(\G,\strat_\Opp)
  \\
    \land
  & \WinStrat_\Opp(\G(\gle),\strat_\Opp,\kappa,\col,n)
  \end{array}
  \right)
\end{array}
\]

%\noindent
%Then for every $\kappa \in \Prop$,
%there is a player $\player$ (either $\Prop$ or $\Opp$)
%such that
%\[
%\FSO \thesis\quad
%(\exists \strat)
%\Big(
%\Strat_\player(\G,\strat) ~~\land~~
%\WinStrat_\player(\G(\gle),\strat,\kappa,\col,n)
%\Big)
%\]
\end{thm}

\noindent
In the statement of Theorem~\ref{thm:compl:red},
$\G(\gle)$ refers to the the game of Remark~\ref{rem:games:sub}
(see also Remark~\ref{rem:games:parity:gle}).
%Note that strictly speaking the quantification
%\[
%(\exists \sigma)\big(\Strat_\player(\G,\strat) ~~\land~~ \dots \big)
%\]
%is not in the language of $\FSO$.

%%%%%%%%%%%%%%%%%%%%%%%%%%%%%%%%%%%%%%%%%%%%%%%%%%%%%%%%%%%%%%%%%%%%%%%%%%%
\begin{rem}
%%%%%%%%%%%%%%%%%%%%%%%%%%%%%%%%%%%%%%%%%%%%%%%%%%%%%%%%%%%%%%%%%%%%%%%%%%%
The crucial differences between Theorem~\ref{thm:compl:red}
and the axiom $(\PosDet)$
are the following.
On one hand, Theorem~\ref{thm:compl:red} 
allows us to derive $(\PosDet)$ for games on \emph{finite} graphs only,
while $(\PosDet)$ speaks about arbitrary $\FSO$-definable games
(in the sense of~\S\ref{sec:games}).
On the other hand,
Theorem~\ref{thm:compl:red}
says that $\FSO$ \emph{decides} winning for games on finite graphs,
while $(\PosDet)$ is a statement of \emph{determinacy}, i.e.\ that one of the players wins,
but not \emph{which} player wins.
\end{rem}

\begin{proof}[Proof of Theorem~\ref{thm:compl:red}]
Fix $G$, $n$, $\col$ and $\kappa$ as in the statement.
Let $\player$ and $\hf s$ be given by Proposition~\ref{prop:compl:red:omega},
and let $\strat$ be induced from $\hf s$ as above.
We are going to show that $\strat$ is winning in $\G$ from position $\kappa$:
\[
(\forall \funto{V}{\G}{\two})
\Big(
\Play(\strat,\kappa,V)
~~\limp~~
\Par(\G,\hat\col,n,V)
\Big)
\]

\noindent
So let $\funto{V}{\G}{\two}$ be an infinite play of $\strat$ from $\kappa$.
Our plan is to obtain $\Par(\G,\hat\col,n,V)$
from Proposition~\ref{prop:compl:red:omega}.
By Comprehension for Product Types (Theorem~\ref{thm:funto:ca}), 
let $\funto{|V|}{\univ}{\two}$
be the set of all $x \in \univ$ such that
$(x,k) \in V$ for some $k \in \Prop$.
Note that $\TPath(|V|)$ holds in $\FSO$.
Proposition~\ref{prop:compl:red:omega}
then gives
\[
\FSO\thesis\quad
\WinStrat[\Lt]^{|V|}_\player(G,\hf s,\kappa,\col,n)
\]

\noindent
Note that
\begin{multline*}
  \WinStrat[\Lt]^{|V|}_\Prop(G,\hf s,\kappa,\col,n)
~~\longliff
\\
(\forall \funto{\tilde V}{|V|}{\Prop})
  \Big(
    \Play[\Lt]^{|V|}(\hf s,\kappa,\tilde V) ~~\limp~~ \Par[\Lt]^{|V|}(\col,n,\tilde V)
  \Big)
\end{multline*}
and similarly for
$\WinStrat[\Lt]^{|V|}_\Opp(G,\hf s,\kappa,\col,n)$.
By $\HF$-Bounded Choice for Functions (\S\ref{sec:ax:choice}),
let $\funto{\tilde V}{|V|}{\Prop}$
take $x \in |V|$ to the unique $k \in \Prop$
such that $(x,k) \in V$.
Then we are done as soon as we show

%%%%%%%%%%%%%%%%%%%%%%%%%%%%%%%%%%%%%%%%%%%%%%%%%%%%%%%%%%%%%%%%%%%%%%%%%%%
\begin{subclm}
%%%%%%%%%%%%%%%%%%%%%%%%%%%%%%%%%%%%%%%%%%%%%%%%%%%%%%%%%%%%%%%%%%%%%%%%%%%
\[
  \Play[\Lt]^{|V|}(\hf s,\kappa,\tilde V)
\quad\land\quad
\Big(
  \Par(\G,\hat\col,n,V)
~~\liff~~
  \Par[\Lt]^{|V|}(\col ,n, \tilde V)
\Big)
\]
\end{subclm}

\begin{subproof}[Proof of Claim \thesubclm]
The property on parity conditions follows from the fact that
for all $m \in [0,n]$ we have
\[
\begin{array}{l @{\quad} l}
&
\left[
\begin{array}{r l}
& (\forall x \in |V|) (\exists y \in |V|)
  \big( x \Lt y ~~\land~~ \col(\tilde V(y)) \Eq m \big)
\\
  \land
& (\exists x \in |V|) (\forall y \in |V|)
  \big( x \Lt y ~~\limp~~ \col(\tilde V(y)) \geq m \big)
\end{array}
\right]
\\

\longliff
\\

&
\left[
  \begin{array}{r l}
  & (\forall u \in V) (\exists v \in V) (u \glt v ~\land~ \hat\col(v) \Eq m)
  \\
    \land
  & (\exists u \in V) (\forall v \in V) (u \glt v ~\limp~ \hat\col(v) \geq m)
  \end{array}
\right]
\end{array}
\]

\noindent
As for $\Play[\Lt]^{|V|}(\hf s,\kappa,\tilde V)$, 
note that it unfolds to
\[
\begin{array}{l}
  \tilde V(\Root) \Eq \kappa
  \quad\land
\\
  (\forall x \in |V|)(\forall y \in |V|)
\Big(
\FSucc^{|V|}(x,y) ~~\limp~~
\big( \exists \ell \In \egp{\hf s}(\tilde V(x)) \big)
\big( \exists d \In \Dir \big)
\Big[
(d, \tilde V(y)) \in \ego{\hf s}(\ell)
\Big]
\Big)
\end{array}
\]
where
\[
\FSucc^{|V|}(x,y)
\quad=\quad
(x \Lt y) ~~\land~~ \lnot(\exists z \in |V|)\big[x \Lt z \Lt y \big]
\]

\noindent
But this directly follows from the definition of $\strat$ from $\hf s$
together with the fact that $V$ is a play of $\strat$ from $\kappa$.
\end{subproof}

This concludes the proof of Theorem~\ref{thm:compl:red}.
\end{proof}

We are now ready to prove Proposition~\ref{prop:compl:aut:dec},
thus completing the proof of Theorem~\ref{thm:main:msoform}.
%%%%%%%%%%%%%%%%%%%%%%%%%%%%%%%%%%%%%%%%%%%%%%%%%%%%%%%%%%%%%%%%%%%%%%%%%%%
\begin{proof}[Proof of Proposition~\ref{prop:compl:aut:dec}]
%%%%%%%%%%%%%%%%%%%%%%%%%%%%%%%%%%%%%%%%%%%%%%%%%%%%%%%%%%%%%%%%%%%%%%%%%%%
We have to show that for 
an $\HF$-closed parity automaton $(\At A:\one)$,
\[
\FSO \thesis (\exists F:\one)\big(F \in \Lang(\At A))
\qquad\text{or}\qquad
\FSO \thesis \lnot (\exists F:\one)\big(F \in \Lang(\At A))
\]

\noindent
For any $(F : \one)$, 
the game $\G(\At A,F)$ is generated as above from the edge
relations~\eqref{eq:compl:red:edge}.
Moreover, recall from Definition~\ref{def:aut:parity}
that the winning condition of $\G(\At A,F)$
is generated, as in the statement of Theorem~\ref{thm:compl:red},
by the game $\G(\At A,F)(\gle)$ of Remark~\ref{rem:games:sub}.
We then conclude by Theorem~\ref{thm:compl:red},
and this completes the proof of Proposition~\ref{prop:compl:aut:dec}.
\end{proof}

%Consider now an $\HF$-closed parity automaton $(\At A:\one)$.
%Then for any $(F : \one)$, 
%the game $\G(\At A,F)$ is generated as above from the edge
%relations~\eqref{eq:compl:red:edge}.
%Moreover, recall from Definition~\ref{def:aut:parity}
%that the winning condition of $\G(\At A,F)$
%is generated, as in the statement of Theorem~\ref{thm:compl:red},
%by the game $\G(\At A,F)(\gle)$ of Remark~\ref{rem:games:sub}.
%We thus obtain Proposition~\ref{prop:compl:aut:dec}
%from Theorem~\ref{thm:compl:red},
%and this completes the proof of Theorem~\ref{thm:main:msoform}.
%
%%%%%%%%%%%%%%%%%%%%%%%%%%%%%%%%%%%%%%%%%%%%%%%%%%%%%%%%%%%%%%%%%%%%%%%%%%%%
%\begin{cor}[Proposition~\ref{prop:compl:aut:dec}]
%%%%%%%%%%%%%%%%%%%%%%%%%%%%%%%%%%%%%%%%%%%%%%%%%%%%%%%%%%%%%%%%%%%%%%%%%%%%
%Given an $\HF$-closed parity automaton $(\At A:\one)$,
%\[
%\FSO \thesis (\exists F:\one)\big(F \in \Lang(\At A))
%\qquad\text{or}\qquad
%\FSO \thesis \lnot (\exists F:\one)\big(F \in \Lang(\At A))
%\]
%\end{cor}

%%% Local Variables:
%%% mode: latex
%%% TeX-master: "main.tex"
%%% End:

%%%%%%%%%%%%%%%%%%%%%%%%%%%%%%%%%%%%%%%%%%%%%%%%%%%%%%%%%%%%%%%%%%%%%%%%%%%
\subsection{Remarks on Recursiveness}
\label{sec:remcompl}
%%%%%%%%%%%%%%%%%%%%%%%%%%%%%%%%%%%%%%%%%%%%%%%%%%%%%%%%%%%%%%%%%%%%%%%%%%%

We noted in Remark~\ref{rem:compl:rec}
that
the completeness of $\FSO + (\PosDet)$ indeed allows us
to decide $\FSO$ and $\MSO$ formulae in the standard model $\Std$ of~\S\ref{sec:std}.
This however comes with two apparent defects.
The first one is that the interpretation $\SI{-}$ of $\HF$-terms fixed
in Convention~\ref{conv:ax:hf:ac} is not computable
(see Remarks~\ref{rem:ax:hf:compl} and~\ref{rem:ax:std:compl}),
because provability in $\Sk(\ZFCM)$ is not decidable
(as this theory contains the $\Pi^0_1$ fragment of arithmetic).
The second one is that, although the axiom set $\MSO + \MI{\PosDet}$
is even polynomial-time recognizable (recall that $\MI{\PosDet}$ is defined in
Definition~\ref{def:posdet:mso}, \S\ref{sec:posdet:mso}),
the interpretation $\MI{-}$ for $\HF$-terms relies on Convention~\ref{conv:ax:hf:ac}
(fixing the interpretation of $\HF$-Functions),
and is thus not computable.
We discuss here a workaround for this involving a slightly different setting for $\FSO$.
We chose to not officially work in that setting because we found it less uniform
and elegant than the current presentation of $\FSO$,
which nonetheless still allows us to derive Rabin's Tree Theorem~\cite{rabin69tams}.

Rather than taking all the axioms on $\HF$-sets
of~\S\ref{sec:ax:hf}, in particular considering
the whole theory $\Sk(\ZFCM)$ there, we may work in systems
parametrized by chosen sets of $\HF$-Functions.
A way to implement this would be to consider systems $\FSO(\Skolem)$,
where the parameter $\Skolem$ specifies \emph{some} interpretations
$\hf g_{n,m}$ for constants $\const g_{n,m}$
such that~\eqref{eq:ax:hf:zfc} is assumed to hold.
Concretely, a specification $\Skolem$ consists of a set $\SK \sle \NN \times \NN$
together with functions
\begin{equation}
\tag{for each $(n,m) \in \SK$}
\hf g_{n,m} ~~:~~ V^n_\omega ~~\longto~~ V_\omega
\end{equation}
Given a set $\SK \sle \NN \times \NN$,
we let $\ZFCM(\SK)$ consist of $\ZFCM$ augmented with
the axioms
\begin{equation}
\tag{for each $(n,m) \in \SK$}
(\forall k_1,\dots,k_n)
(\exists! \ell)(\varphi_{n,m})
~~\limp~~
(\forall k_1,\dots,k_n)
\varphi_{n,m}[\const g_{n,m}(k_1,\dots,k_n)/\ell]
%(\forall k_1,\dots,k_n)
%\Big(
%(\exists! \ell)(\varphi_{n,m})
%~~\limp~~
%\varphi_{n,m}[\const g_{n,m}(k_1,\dots,k_n)/\ell]
%\Big)
\end{equation}

\noindent
We say that $\Skolem$ is a \emph{specification} if
\[
\Skolem \quad=\quad
(\SK,(\hf g_{n,m})_{(n,m) \in \SK})
\]
where,
for each $(n,m) \in \SK$,
\begin{itemize}
\item $\hf g_{n,m}$
is a computable function $V^n_\omega \to V_\omega$, and

\item for each each $\const g_{n',m'}$ occurring in $\varphi_{n,m}$,
we have $(n',m') \in \SK$, and

\item
$\ZFCM(\SK)\thesis (\forall k_1,\dots,k_n)(\exists !\ell)\varphi_{n,m}$,
and

\item
$V_\omega \models
(\forall k_1,\dots,k_n)\varphi_{n,m}[\hf g_{n,m}(k_1,\dots,k_n)/\ell]$
\end{itemize}

\noindent
Given a specification $\Skolem$, one can fix the interpretation
of all constants $(\const g_{n,m})_{n,m \in \NN}$
by taking for $\const g_{n,m}$ with $(n,m) \notin \SK$
%as 
the function $V^n_\omega \to V_\omega$
with constant value $\emptyset$.

For the formal definition of $\FSO(\Skolem)$,
instead of the Axioms on $\HF$-Sets of~\S\ref{sec:ax:hf}, one has the following.
\begin{itemize}
\item For each $(n,m) \in \SK$, and for all $\HF$-terms $\vec K = K_1,\dots,K_n$,
the axiom
\[
\varphi_{n,m}[\vec K/\vec k][\const g_{n,m}(\vec K) / \ell]
\]

\item For each closed $\HF$-formula $\varphi$ such that $V_\omega\models\varphi$,
the axiom
\[
\varphi
\]
\end{itemize}

Given a specification $\Skolem$, the interpretations $\SI{-}$
and $\MI{-}$ are computable.
%and so we see that $\MSO + \MI{(\PosDet)}$ is a recursive set of axioms.
All results of this paper
%(and in particular all those of~\S\ref{sec:compl})
hold for sufficiently large specifications.

%%%%%%%%%%%%%%%%%%%%%%%%%%%%%%%%%%%%%%%%%%%%%%%%%%%%%%%%%%%%%%%%%%%%%%%%%%%
\begin{thm}
%%%%%%%%%%%%%%%%%%%%%%%%%%%%%%%%%%%%%%%%%%%%%%%%%%%%%%%%%%%%%%%%%%%%%%%%%%%
Let $\Skolem$ be a specification defining all the $\HF$-Functions
of~\ref{item:ax:hf:first}--\ref{item:ax:hf:last}, \S\ref{sec:ax:hf},
as well as those of Convention~\ref{conv:games:colors},
\S\ref{sec:games:parity}.
Then all the results stated in~\S\ref{sec:compl}
hold for $\FSO(\Skolem)$ instead of $\FSO$.
\end{thm}

%%% Local Variables:
%%% mode: latex
%%% TeX-master: "main.tex"
%%% End:

%%%%%%%%%%%%%%%%%%%%%%%%%%%%%%%%%%%%%%%%%%%%%%%%%%%%%%%%%%%%%%%%%%%%%%%%%%%
%\subsection{Simulation and Projection}
\section{The Simulation Theorem}
\label{sec:sim}
%%%%%%%%%%%%%%%%%%%%%%%%%%%%%%%%%%%%%%%%%%%%%%%%%%%%%%%%%%%%%%%%%%%%%%%%%%%

\noindent
This Section is devoted to the proof of the
\emph{Simulation Theorem}, cf.~\cite{ej91focs,ms95tcs}.

%%%%%%%%%%%%%%%%%%%%%%%%%%%%%%%%%%%%%%%%%%%%%%%%%%%%%%%%%%%%%%%%%%%%%%%%%%%
\begin{thm}[Simulation Theorem~\ref{thm:aut:nd:sim}]
\label{thm:sim}
%%%%%%%%%%%%%%%%%%%%%%%%%%%%%%%%%%%%%%%%%%%%%%%%%%%%%%%%%%%%%%%%%%%%%%%%%%%
For each $\HF$-closed parity automaton $\At A : \Sigma$
there is a non-deterministic $\HF$-closed parity automaton
$\ND(\At A) : \Sigma$
such that 
%$\FSOD$ proves that 
\[
\FSO
\thesis\quad \Lang(\ND(\At A)) = \Lang(\At A)
\]
\end{thm}

We assume that $\At A$ is $\HF$-closed in Theorem~\ref{thm:sim}
because we rely on McNaughton's Theorem~\cite{mcnaughton66ic},
in the standard model for $\omega$-words,
which we import into $\FSO$ thanks to Proposition~\ref{prop:bfsos:bfso}.

Before a detailed exposition, let us explain the main idea
behind Theorem~\ref{thm:sim}.
We momentarily work in the usual mathematical universe
(\ie\@ not in the formal theory $\FSO$).
Recall that in a non-deterministic automaton $\At N$, $\Opp$ can only explicitly choose
tree directions, since for each possible $\Conj_{\At N}$ in the image of $\trans_{\At N}$,
if $(d,q),(d,q') \in \Conj_{\At N}$ then $q = q'$, by definition.
In order to obtain a non-deterministic automaton $\At N$
from an alternating automaton $\At A$,
the idea is to perform a subset construction,
such that each $\Conj_{\At N}$ in the image of $\trans_{\At N}$ is of the form
\[
\Conj_{\At N} \quad=\quad \{(d,S'_d) \st d \in \Dir \}
\]
where each $S'_d$ gathers states $q$ such that $(d,q) \in \Conj_{\At A}$
with $\Conj_{\At A}$ in the image of $\trans_{\At A}$.

More precisely, assuming $S \in \Pne(Q_{\At A})$, one may consider functions
\begin{equation}
\label{eq:sim:f}
\begin{array}{c !{~~}c!{~~} c !{~~}c!{~~} l}
  f
& :
& S
& \longto
& \Pne(\Dir \times Q_{\At A})
\\
&
& q
& \longmapsto
& \Conj_q \in \trans_{\At A}(q,\al a)
\end{array}
\end{equation}
Each such $f$ induces 
\[
\Conj_{\At N}(f) = \{(d,S'_d(f)) \st d \in \Dir \}
\quad\text{where}\quad
S'_d(f) = \{q \st (d,q) \in f(q) \} 
\]
and we can let
\[
\trans_{\At N}(S,\al a) \quad\deq\quad
\{\Conj_{\At N}(f) \st \text{$f$ is as in~\eqref{eq:sim:f}}\}
\]

\noindent
Then, for each $\Conj_{\At N}(f)$ in the image of $\trans_{\At N}$
and for each tree direction $d \in \Dir$, the set $S'_d$
is unique such that $(d,S'_d) \in \Conj_{\At N}(f)$,
and $\At N$ satisfies the property asked in Definition~\ref{sec:aut:nd}
to non-deterministic automata.

%\noindent
There is however a difficulty in the definition
of the acceptance condition of $\At N$.
We follow here the construction of~\cite{walukiewicz02tcs}
where the states of $\At N$, rather than being simply sets of states, are
sets of pairs of states $S \in \Po(Q_{\At A} \times Q_{\At A})$.
Then an infinite sequence of states $S_0,S_1,\ldots \in Q_{\At N}$
induces a set of \emph{traces} $q_0,q_1,\ldots \in Q_{\At A}$
with $(q_i,q_{i+1}) \in S_{i+1}$.
For $(S_n)_{n \in \NN} \in Q_{\At N}^\omega$ to be accepting, one may then require
all its traces $(q_n)_{n \in \NN} \in Q_{\At A}^\omega$
to be accepting.
We may obtain a parity condition for $\At N$
by noticing that its acceptance condition is $\omega$-regular
(\ie\@ definable in $\MSO$ over $\omega$-words).
This allows us to apply McNaughton's Theorem~\cite{mcnaughton66ic},
and to obtain a deterministic $\omega$-word parity automaton $\At D$
whose language is the set of accepting sequences 
$(S_n)_{n \in \NN} \in Q_{\At N}^\omega$.
A suitable product of $\At N$ with $\At D$ then gives a non-deterministic
parity automaton equivalent to $\At A$.

The organization of this Section follows the above construction.
Working in $\FSO$, consider a parity automaton $\At A :\Sigma$.
We will build a \emph{non-deterministic} automaton
$\ND(\At A): \Sigma$ with the same language.
%that is such that $F \in \Lang(\ND(\At A))$ iff $F \in \Lang(\At A)$.
%We follow the construction of~\cite{walukiewicz02tcs}.
The automaton $\ND(\At A)$ will be defined in three steps:
\begin{enumerate}
\item
We first define in~\S\ref{sec:sim:sim} a non-deterministic automaton $\oc \At A$
in the sense of Definition~\ref{def:aut:alt}.
The acceptance condition of $\oc \At A$ will be given by an $\FSO$-formula
with a free Function variable
(intended to be range over infinite plays) rather than a parity condition.

\item
For an \emph{$\HF$-closed} parity automaton $\At A$, 
the formula describing the acceptance condition of $\oc \At A$
is then transformed in~\S\ref{sec:sim:omega} to a $\FSOW$ formula
relativized to infinite rooted tree paths
(see~\S\ref{sec:msow}).
This construction relies,
via Proposition~\ref{prop:bfsos:bfso},
on
Proposition~\ref{prop:msow:bfsos}
(\ie\@ Proposition~\ref{prop:cons:mso:cons})
which requires the manipulation of closed (and in particular $\HF$-closed)
objects.

\item
Using the tools of~\S\ref{sec:msow},
and relying on McNaughton's Theorem~\cite{mcnaughton66ic}
(see also \eg~\cite{thomas97handbook,pp04book}),
in~\S\ref{sec:sim:par}
we will then turn $\oc \At A$ into an equivalent non-deterministic 
\emph{parity} automaton $\ND(\At A)$,
in the sense of Definition~\ref{def:aut:parity}.
\end{enumerate}

\cnote{\CR:NOTE\\
Beware that $\Dir$ is used both as an $\HF$-set and as a parameter in the definition
of $\FSOD/\MSOD$}

In this Section, it is convenient
to work with the following games.
%slight generalization of the acceptance games
%as defined in Definition~\ref{def:aut:games},
%for which we require no input tree.

%%%%%%%%%%%%%%%%%%%%%%%%%%%%%%%%%%%%%%%%%%%%%%%%%%%%%%%%%%%%%%%%%%%%%%%%%%%
\begin{defi}
\label{def:sim:inputgames}
%%%%%%%%%%%%%%%%%%%%%%%%%%%%%%%%%%%%%%%%%%%%%%%%%%%%%%%%%%%%%%%%%%%%%%%%%%%
Given an automaton $\At A :\Sigma$, we let $\G(\At A)$ be the game with
\[
\PL{\G(\At A)} ~~\deq~~ Q_{\At A}
\qquad\text{and}\qquad
\OL{\G(\At A)} ~~\deq~~ Q_{\At A} \times \Pne(\Dir \times Q_{\At A})
\]
and with transitions defined by 
$\HF$-Bounded Choice for Product Types (Theorem~\ref{thm:funto:choice})
and Comprehension for $\HF$-Sets (Remark~\ref{rem:hfchoice}) as
\[
\begin{array}{l !{\qquad} r !{\quad\text{iff}\quad} l}
& (q',\Conj) \in \EGP{\G(\At A)}(x,q) 
& q' \Eq q ~~\land~~
  (\exists \al a \in \Sigma)\left[\Conj \in \trans_{\At A}(q,\al a)\right]
\\
  \text{and}
& (d,q') \in \EGO{\G(\At A)}(x,(q,\Conj))
& (d,q') \in \Conj
\end{array}
\]
\end{defi}

\noindent
As for winning,
we will consider the game $\G(\At A)$ as being equipped
with the winning condition $\Omega_{\At A}$ in the sense of~\S\ref{sec:games:win}.
Note that for $F:\Sigma$, the acceptance game $\G(\At A,F)$
is a subgame of $\G(\At A)$ in the sense of Def.~\ref{def:games:sub}.

%%%%%%%%%%%%%%%%%%%%%%%%%%%%%%%%%%%%%%%%%%%%%%%%%%%%%%%%%%%%%%%%%%%%%%%%%%%
\begin{rem}
\label{rem:sim:inputgames}
%%%%%%%%%%%%%%%%%%%%%%%%%%%%%%%%%%%%%%%%%%%%%%%%%%%%%%%%%%%%%%%%%%%%%%%%%%%
Note that if $\Aut(\At A)$ then $\Sigma$ is non-empty,
so we indeed have $\Game(\G(\At A))$.
For each $F:\Sigma$,
the acceptance game
$\G(\At A,F)$ is a subgame of $\G(\At A)$
(in the sense of Def.~\ref{def:games:sub}).
In particular infinite plays in $\G(\At A,F)$ are infinite plays in $\G(\At A)$.
Moreover, it is easy to see that (winning) strategies on $\G(\At A,F)$
are (winning) strategies on $\G(\At A)$.

Furthermore, note that the game $\G(\At A)(\gle)$ induced by Remark~\ref{def:games:sub}
from Definition~\ref{def:sim:inputgames}
is precisely the game $\G(\At A)(\gle)$ of~\eqref{eq:aut:gle}.
It follows that in the case of a parity automaton $\At A$,
we unambiguously extend the notation of Definition~\ref{def:aut:parity}
and write
$\Par(\At A,\hat \col_{\At A},n_{\At A},U)$
or
$\Par(\At A,U)$
for the formula
$\Par(\G(\At A)(\gle),\hat \col_{\At A},n_{\At A},U)$.
\end{rem}

%%%%%%%%%%%%%%%%%%%%%%%%%%%%%%%%%%%%%%%%%%%%%%%%%%%%%%%%%%%%%%%%%%%%%%%%%%%
\subsection{The Construction of $\oc \At A$.}
\label{sec:sim:sim}
%%%%%%%%%%%%%%%%%%%%%%%%%%%%%%%%%%%%%%%%%%%%%%%%%%%%%%%%%%%%%%%%%%%%%%%%%%%
Consider an alternating parity automaton $\At A$,
in the sense of Definition~\ref{def:aut:parity}.
So we have
$\At A = (Q_{\At A},\, \init q_{\At A},\, \trans_{\At A},\, \col_{\At A},\, n_{\At A})$
where
\[
\trans_{\At A} ~~:~~
Q_{\At A} \times \Sigma ~~\longto~~ \Pne(\Pne(\Dir \times Q_{\At A}))
\qquad\text{and}\qquad
\funto{\col_{\At A}}{Q_{\At A}}{[0,n_{\At A}]}
\]

\noindent
We define the state set and the initial state of $\oc \At A$ as:
%as follows
\[
Q_{\oc \At A}
\quad\deq\quad
\Pne(Q_{\At A} \times Q_{\At A})
\qquad\text{and}\qquad
\init q_{\oc \At A}
\quad\deq\quad
(\init q_{\At A},\init q_{\At A})
\]

\noindent
The transition function of $\oc \At A$ is defined as follows,
using Remark~\ref{rem:aut:hfclosed}.
For $\al a \in \Sigma$ and $S \in Q_{\oc \At A}$
we let $\oc\Conj \in \trans_{\oc \At A}(S,\al a)$
if and only if there is some $\HF$-set
\[
\begin{array}{c c c c l}
  f
& :
& S
& \longto
& \Pne(\Dir \times Q_{\At A})
\\
&
& q
& \longmapsto
& \Conj_q \in \trans_{\At A}(q,\al a)
\end{array}
\]
%where $\pi_2$
%is a projection $\HF$-Function of~\S\ref{sec:ax:hf}.\ref{item:ax:hf:fun:prod}
%and
such that
$\oc\Conj = \{(d,S'_d) \st d \in \Dir \land S'_d \neq \emptyset\}$,
where
\begin{equation}
\label{eq:sim:trans}
S'_d \quad=\quad \{ (q,q') \st q \in \pi_2(S) ~~\land~~ (d,q') \in f(q) \}
\end{equation}
and where $\pi_2$
is a projection $\HF$-Function of~\S\ref{sec:ax:hf}.\ref{item:ax:hf:fun:prod}.

%%%%%%%%%%%%%%%%%%%%%%%%%%%%%%%%%%%%%%%%%%%%%%%%%%%%%%%%%%%%%%%%%%%%%%%%%%%
\begin{rem}
%%%%%%%%%%%%%%%%%%%%%%%%%%%%%%%%%%%%%%%%%%%%%%%%%%%%%%%%%%%%%%%%%%%%%%%%%%%
We indeed have
\[
\trans_{\oc \At A} ~~:~~
Q_{\oc\At A} \times \Sigma ~~\longto~~ \Pne(\Pne(\Dir \times Q_{\oc \At A}))
\]
since for $S \in Q_{\oc \At A} = \Pne(Q_{\At A} \times Q_{\At A})$,
by
$\HF$-Bounded Choice for $\HF$-Sets
(\S\ref{sec:ax:choice})
there is always some $f \in \Pne(\Dir \times Q_{\At A})^{\pi_2(S)}$
with $\forall q \in \pi_2(S)(f(q) \in \trans_{\At A}(q,\al a))$,
and moreover such that
$S'_d$ is non-empty for at least one $d \in \Dir$.

Note our unusual choice of taking \emph{non-empty} sets as states of $\oc\At A$.
It would have been more natural to allow the empty set as a state,
in particular because it would have allowed us to strengthen
Corollary~\ref{cor:aut:nd:unique} to an ``exists unique'' statement.
This could also have worked in our setting where games are assumed to
have no dead ends, and in which transitions of alternating automata
range over non-empty sets of non-empty subsets of $\Dir \times Q_{\At A}$.
However, the empty state would have appeared in the transitions of
$\oc\At A$ only in case there is some tree direction $d \in \Dir$
which is not available to $\Opp$ at some stage.
Since the empty state of $\oc\At A$
would have been unconditionally winning for $\Prop$,
this would have lead to an additional case to handle in
the proof of completeness of $\oc \At A$
(Proposition~\ref{prop:sim:sim:prom} below).
\end{rem}

So far we have defined for $\oc\At A$ a state set (with an initial state)
and a transition function.
As explained above, we will not directly equip it with a parity condition.
Instead, we will define its acceptance condition via an $\FSO$-formula
$\Win_{\oc\At A}$,
in the sense of Definition~\ref{def:aut:alt}.
Consider first $\funto{V}{\G(\oc \At A)}{\two}$
and $\funto{T}{\G(\At A)}{\two}$.
We say that $T$ is a \emph{trace} in $V$ if the following
formula $\Trace(T,V)$ holds,
\[
\begin{array}{c @{~~} l}
& \Path(\G(\At A),(\Root,\init q_{\At A}),T)
\\[0.5em]

  \land
& \big(\forall (x,q) \in \PP T \big) \big( \exists S \in Q_{\oc \At A} \big)
  \Big[ (x,S) \in \PP V ~~\land~~ q \in \pi_2(S) \Big]
\\[0.5em]

  \land
& \big( \forall (x,q),(y,q') \in \PP T \big)
  \big( \forall S \in Q_{\oc \At A} \big)
  \left[
  (x,q) \glt^{\Prop;\Opp}_{\G(\At A)} (y,q')
  ~~\limp~~
  (y,S) \in \PP V
  ~~\limp~~
  (q,q') \in S
  \right]
\end{array}
\]

\noindent
where we use the following formula:
\[
\begin{array}{r !{\quad\deq\quad} l}
  u \glt^{\Prop;\Opp}_{\G(\At A)} v
& \big( \exists w \in \OP{\G(\At A)} \big)
  \left(
    \gsucc_{\G(\At A)}(u,w) ~\land~ \gsucc_{\G(\At A)}(w,v)
  \right)
\end{array}
\]

\noindent
The formula $\Win_{\oc\At A}(V)$ is defined to be:
\[
\Win_{\oc\At A}(V)
\quad\deq\quad
\big( \forall \funto{T}{\G(\At A)}{\two} \big)
\left[
\Trace(T,V)
~~\limp~~
\Par(\At A,\hat\col_{\At A},n_{\At A},T)
\right]
\]

\noindent
Recall our notation $\Par(\At A,\hat\col_{\At A},n_{\At A},-)$
from Remark~\ref{rem:sim:inputgames}.
Note that $\Win_{\oc \At A}$ 
requires no condition 
\wrt\@ the \emph{transitions} of $\G(\At A)$.
We are now going to show that $\oc\At A$ has the same language as~$\At A$.

%%%%%%%%%%%%%%%%%%%%%%%%%%%%%%%%%%%%%%%%%%%%%%%%%%%%%%%%%%%%%%%%%%%%%%%%%%%
\begin{thm}
\label{thm:sim:sim:cor}
%%%%%%%%%%%%%%%%%%%%%%%%%%%%%%%%%%%%%%%%%%%%%%%%%%%%%%%%%%%%%%%%%%%%%%%%%%%
Fix a parity automaton $\At A : \Sigma$ and consider the automaton
$\oc \At A : \Sigma$ as defined above.
Then $\FSOD$ proves that for all $F:\Sigma$,
$\oc \At A$ accepts $F$ if and only if $\At A$ accepts $F$.
\end{thm}

The proof of Theorem~\ref{thm:sim:sim:cor}
is split into Propositions~\ref{prop:sim:sim:prom} and~\ref{prop:sim:sim:der} below.

%%%%%%%%%%%%%%%%%%%%%%%%%%%%%%%%%%%%%%%%%%%%%%%%%%%%%%%%%%%%%%%%%%%%%%%%%%%
\begin{conv}
%%%%%%%%%%%%%%%%%%%%%%%%%%%%%%%%%%%%%%%%%%%%%%%%%%%%%%%%%%%%%%%%%%%%%%%%%%%
In Propositions~\ref{prop:sim:sim:prom} and~\ref{prop:sim:sim:der},
for fixed automata $\At A$ and $\oc \At A$,
we let
\[
\iota ~~\deq~~ (\Root, \init q_{\At A})
\qquad\text{and}\qquad
\iota_\oc ~~\deq~~ (\Root, \init q_{\oc \At A})
\]
\end{conv}

%%%%%%%%%%%%%%%%%%%%%%%%%%%%%%%%%%%%%%%%%%%%%%%%%%%%%%%%%%%%%%%%%%%%%%%%%%%
\begin{prop}
\label{prop:sim:sim:prom}
%%%%%%%%%%%%%%%%%%%%%%%%%%%%%%%%%%%%%%%%%%%%%%%%%%%%%%%%%%%%%%%%%%%%%%%%%%%
Fix a parity automaton $\At A : \Sigma$ and consider the automaton
$\oc \At A : \Sigma$ as defined above.
Then $\FSOD$ proves that for all $F:\Sigma$,
if $\At A$ accepts $F$ then $\oc \At A$ accepts $F$.
\end{prop}

\begin{proof}
Let $\strat$ be a winning $\Prop$-strategy in $\G(\At A,F)$.
We define a winning $\Prop$-strategy $\stratbis$ in $\G(\oc \At A,F)$.
Note that
\[
\funto{\strat}{\PP{\G(\At A)}}{\OL{\G(\At A)}}
\qquad\text{and}\qquad
\funto{\stratbis}{\PP{\G(\oc \At A)}}{\OL{\G(\oc \At A)}}
\]

\noindent
First, given a $\Prop$-position $(x,S)$ in $\G(\oc \At A,F)$,
we define a conjunction
\[
\Conj_{(x,S)}
~\in~ \trans_{\oc \At A}(S,F(x))
~\sle~ \Pne(\Dir \times Q_{\oc \At A})
\]
as follows.

%\begin{description}
\begin{itemize}
%\item[Definition of $\Conj_{(x,S)}$]
\item {\it Definition of $\Conj_{(x,S)}$.}
For each $q \in \pi_2(S)$, $\strat(x,q)$
gives some $\Conj_q \in \trans_{\At A}(q,F(x))$.
$\HF$-Bounded $\HF$-Choice (\S\ref{sec:ax:choice})
then gives 
\[
\begin{array}{c c c c l}
  f
& :
& \pi_2(S)
& \longto
& \Pne(\Dir \times Q_{\At A})
\\
&
& q
& \longmapsto
& \Conj_q \in \trans_{\At A}(q,F(x))
\end{array}
\]
By $\HF$-Comprehension (Remark~\ref{rem:hfchoice}), we then let $\Conj_{(x,S)}$
be $\{(d,S'_d) \st d \in \Dir \land S'_d \neq \emptyset\}$
where each $S'_d$ is defined as in~\eqref{eq:sim:trans}.
%\end{description}
\end{itemize}

\noindent
We now define the $\Prop$-strategy $\stratbis$ on $\G(\oc \At A,F)$.
By $\HF$-Bounded Choice for Functions (\S\ref{sec:ax:choice}), 
we let
\[
\stratbis(x,S) \quad\deq\quad (S,\Conj_{(x,S)})
\qquad\text{for each $(x,S) \in \PP{\G(\oc \At A)}$}
\]

\noindent
We have $\Strat_\Prop(\G(\oc \At A,F),\stratbis)$ directly by definition of
$\trans_{\oc \At A}$.
It remains to check that $\stratbis$ is winning in $\G(\oc \At A,F)$.
Consider an infinite play of $\stratbis$, that is some
$\funto{V}{\G(\oc \At A)}{\two}$
such that
$\Play(\stratbis, \iota_\oc , V)$.
Since $\strat$ is winning in $\G(\At A,F)$,
by definition of $\Win_{\oc\At A}$
and by Remarks~\ref{rem:aut:propplays} and~\ref{rem:sim:inputgames},
we are done if we show that:
\[
\big( \forall \funto{T}{\G(\At A)}{\two} \big)
\Big(
\Trace(T,V)
~~\limp~~
\big( \exists \funto{U}{\G(\At A)}{\two} \big)
\big[
\PP U = \PP T ~~\land~~
\Play(\strat, \iota , U) 
\big]
\Big)
\]
Assume $\Trace(T,V)$.
By $\HF$-Comprehension for Product Types,
we let $\funto{U}{\G(\At A)}{\two}$
be such that $\PP U = \PP T$
and such that
$\OP U$ consists of the $\{(x,\strat(x,q))\}$ for $(x,q) \in \PP U$.
Note that we actually have $\funto{U}{\strat}{\two}$.
It remains to check that
\[
\Play(\strat,\, \iota,\, U)
\]
We apply Lemma~\ref{lem:games:pathplays},
whence it remains to show:

\begin{gather}
\label{eq:sim:sim:prom:path}
\Path(\G(\At A), \iota ,U)
\\
\label{eq:sim:sim:prom:succ}
\big( \forall u,u' \in U \big)
\left[
\gsucc_{\G(\At A)}(u,u')
~~\limp~~
u \edge{}{\strat} u'
\right]
\end{gather}

\noindent
Note that $\Path(\G(\At A),\iota ,T)$ since $\Trace(T,V)$.
\smallskip
\begin{itemize}
\item \begin{subproof}[Proof of~\eqref{eq:sim:sim:prom:path}]
We obviously have
$\iota \in \PP U = \PP T$.
Also, given $u \in U$,
if $u \in \PP U = \PP T$ then $\iota \gle_{\G(\At A)} u$,
and if $u \in \OP U$, then $v \edge{}{\strat} u$ for some
$v \in \PP U = \PP T$, so that
$\iota \gle_{\G(\At A)} v \glt_{\G(\At A)} u$.
Moreover, for each $u \in \PP U$, we have
$u \edge{}{\strat} v$ for some $v \in \OP U$,
and we get $\gsucc_{\G(\At A)}(u,v)$ by Proposition~\ref{prop:games:edges}.

It remains to show that $U$ is linearly ordered \wrt\@ $\gle_{\G(\At A)}$.
For $\PP U$ this follows from the same property for $\PP T$.
Now let $u \in \PP U$ and $v' \in \OP U$.
Hence $v'$ is of the form $(x,\strat(x,q))$ with $v \deq (x,q) \in \PP U = \PP T$.
If $u \gle_{\G(\At A)} v$ then 
$u \glt_{\G(\At A)} v'$ and we are done.
Otherwise, $v \glt_{\G(\At A)} u$.
But by definition of $\gle_{\G(\At A)}$, this implies $u = (y,q')$
with $x \Lt y$, so that $v' \glt_{\G(\At A)} u$.
Consider now $u',v' \in \OP U$
and let $u,v \in \PP U$ be their immediate predecessors.
If $u \glt_{\G(\At A)} v$ then $u' \glt_{\G(\At A)} v'$ and we are done.
Otherwise, without loss of generality we have that $u = v$.
But then $u' = v'$ by definition of $U$.
\end{subproof}

\item \begin{subproof}[Proof of~\eqref{eq:sim:sim:prom:succ}]
Assume first $u \in \PP U$.
In this case,
$u$ is of the form $(x,q)$ with $(x,q) \in \PP T$,
and $u'$ is of the form $(x,\strat(x,q'))$
with $(x,q') \in \PP T$.
But $T$ is linearly ordered \wrt\@ $\gle_{\G(\At A)}$,
so that $q = q'$.
It follows that $u \edge{}{\strat} u'$.

Assume now that $u \in \OP U$.
In this case, 
$u$ is of the form $(x,(q,\Conj_{\At A}))$ with $(x,q) \in \PP U = \PP T$
and $(q,\Conj_{\At A}) = \strat(x,q)$.
Moreover, $u' \in \PP U = \PP T$ is of the form
$(\Succ_d(x),q')$.
We thus get $u \edge{}{\strat} u'$ as soon as
%Hence we are done if
\[
(d,q') \in \Conj_{\At A}
\]

\noindent
Since $\Trace(T,V)$
and since $V$ is a play of $\stratbis$,
there are unique $S,S'$
with $(x,S),(\Succ_d(x),S') \in \PP V$
and such that
$q \in \pi_2(S)$ and $q' \in \pi_2(S')$.
Moreover, we necessarily have
$(d,S') \in \Conj_{(x,S)}$
for $(S,\Conj_{(x,S)}) = \stratbis(x,S)$.
But $\Trace(T,V)$ implies $(q,q') \in S'$,
and it follows
that $(d,q') \in \Conj_{\At A}$
by definition of $\Conj_{(x,S)}$.
\end{subproof}
\end{itemize}

\noindent
This concludes the proof of Proposition~\ref{prop:sim:sim:prom}.
%\qed
\end{proof}

%%%%%%%%%%%%%%%%%%%%%%%%%%%%%%%%%%%%%%%%%%%%%%%%%%%%%%%%%%%%%%%%%%%%%%%%%%%
\begin{prop}
\label{prop:sim:sim:der}
%%%%%%%%%%%%%%%%%%%%%%%%%%%%%%%%%%%%%%%%%%%%%%%%%%%%%%%%%%%%%%%%%%%%%%%%%%%
Fix a parity automaton $\At A : \Sigma$ and consider the automaton
$\oc \At A : \Sigma$ as defined above.
Then $\FSOD$ proves that for all $F:\Sigma$,
if $\oc \At A$ accepts $F$ then $\At A$ accepts $F$.
\end{prop}

\begin{proof}
Let $\stratbis$ be a winning $\Prop$-strategy in $\G(\oc \At A,F)$.
We will define a winning $\Prop$-strategy $\strat$ in $\G(\At A,F)$.
To this end, we invoke Corollary~\ref{cor:aut:nd:unique},
which tells us that since $\oc \At A$ is non-deterministic,
for each $x \in \univ$ there is at most one $S \in Q_{\oc \At A}$
such that $(x,S)$ belongs to an infinite play of $\stratbis$.
Moreover, 
using Remark~\ref{rem:ax:hf:well-order-hf},
for each $S \in Q_{\oc \At A}$
we fix a well-order $\preceq$ on $\Pne(\Dir \times Q_{\At A})^{\pi_2(S)}$.

We now define the strategy $\strat$.
\begin{itemize}
%\begin{description}
%\item[Definition of $\strat$]
\item {\it Definition of $\strat$.}
We apply $\HF$-Bounded Choice for Product Types (Theorem~\ref{thm:funto:choice}).
Consider $(x,q) \in \G(\At A,F)_\Prop$.
We first assign to $(x,q)$ an $S \in Q_{\oc \At A}$ such that $q \in \pi_2(S)$.
If there exists such an $S$ where furthermore $(x,S)$
belongs to an infinite play of $\stratbis$, then this $S$ is unique
and we choose that one.
Otherwise, 
by Comprehension for $\HF$-Sets (Remark~\ref{rem:hfchoice}),
we define an ad hoc $S \in Q_{\oc \At A}$ with $q \in \pi_2(S)$.

Let now $(S,\Conj_{(x,S)}) \deq \stratbis(x,S)$.
By definition of $\Conj_{(x,S)}$
there is some
\[
\begin{array}{c c c c l}
  f
& :
& \pi_2(S)
& \longto
& \Pne(\Dir \times Q_{\At A})
\\
&
& q
& \longmapsto
& \Conj_q \in \trans_{\At A}(q,F(x))
\end{array}
\]
such that 
$\Conj_{(x,S)} = \{(d,S'_d) \st d \in \Dir \land S'_d \neq \emptyset\}$
where each $S'_d$ is as in~\eqref{eq:sim:trans}.
Consider the $\prec$-least such $f$.
We let
\[
\strat(x,q) \quad\deq\quad (q,f(q))
\]
%\end{description}
\end{itemize}

\noindent
It remains to show that $\strat$ is winning.
To this end, given an infinite play $T$ of $\strat$,
we will define an infinite play $V$ of $\stratbis$ such that:
\[
\Trace(T,V)
\]

\noindent
Since $\stratbis$ is assumed to be winning, 
thanks to 
Remarks~\ref{rem:aut:propplays} and~\ref{rem:sim:inputgames},
%Remark~\ref{rem:games:parity:sub}
this will imply that $\strat$
is also winning.
Assume $\Play(\strat, \iota, T)$.
We define $V$ using the Recursion Theorem (Proposition~\ref{prop:games:rec}).
Let $\varphi(V,v)$ be a $\FSO$-formula stating that:
\begin{itemize}
\item either $v = \iota_\oc$,
\item or $v = (x,\stratbis(x,S))$ with $(x,S) \in V$,
\item or $v= (\Succ_d(x),S'_d)$ and 
\begin{itemize}
\item for some $q' \in Q_{\At A}$ we have $(\Succ_d(x),q') \in T$,
\item and for some $S \in Q_{\oc \At A}$, we have
$(x,S) \in V$ and $\stratbis(x,S) = (S,\Conj_{(x,S)})$ with $(d,S'_d) \in \Conj_{(x,S)}$.
\end{itemize}
\end{itemize}
Note that $\varphi(V,v)$ indeed satisfies the assumptions of the
Recursion Theorem (Proposition~\ref{prop:games:rec}), since
\begin{itemize}
\item
in the second clause we always have $(x,S) \glt_{\G(\oc \At A)} (x,\stratbis(x,S))$,
and $\stratbis(x,S)$ is uniquely determined from $(x,S)$;

\item
in the last clause,
we always have $(x,S) \glt_{\G(\oc \At A)} (\Succ_d(x),S'_d)$.
\end{itemize}
Note also that since $T$ is a play of $\strat$, there is at most one $d \in \Dir$
such that $(\Succ_d(x),q) \in T$ for some $q \in Q_{\At A}$,
and $S'_d$ is uniquely determined from $d$ and $\stratbis(x,S)$
by construction of $\oc \At A$.
So by the Recursion Theorem (Proposition~\ref{prop:games:rec}) we indeed let
$\funto{V}{\G(\oc \At A)}{\two}$ be unique such that
\[
\big( \forall v \in \G(\oc \At A) \big)
\Big[
v \in V ~~\liff~~ \varphi(V,v)
\Big]
\]

\noindent
We begin with a series of easy claims on $V$.

%%%%%%%%%%%%%%%%%%%%%%%%%%%%%%%%%%%%%%%%%%%%%%%%%%%%%%%%%%%%%%%%%%%%%%%%%%%
\begin{subclm}
\label{clm:sim:sim:der:Sunique}
%%%%%%%%%%%%%%%%%%%%%%%%%%%%%%%%%%%%%%%%%%%%%%%%%%%%%%%%%%%%%%%%%%%%%%%%%%%
For every $x \in \univ$, there is at most one $S \in Q_{\oc \At A}$
such that $(x,S) \in V$.
\end{subclm}

\begin{subproof}[Proof of Claim~\ref{clm:sim:sim:der:Sunique}]
We apply the Induction Axiom of $\FSO$ (\S\ref{sec:ax:ind}).
The property holds for $\Root$, since $(\Root,S) \in V$
implies $S = \init q_{\oc \At A}$ by definition of $V$.
Now assume the property for $x$ and let us show it for $\Succ_d(x)$.
So assume $(\Succ_d(x),S'_d),(\Succ_d(x),\tilde S'_d) \in V$.
By definition of $V$, there are $(x,S),(x,\tilde S) \in V$
such that $(d,S'_d) \in \Conj_{(x,S)}$
and $(d,\tilde S'_d) \in \Conj_{(x,\tilde S)}$
where
$(S,\Conj_{(x,S)}) = \stratbis(x,S)$
and
$(\tilde S,\Conj_{(x,\tilde S)}) = \stratbis(x,\tilde S)$.
But by induction hypothesis we get $S = \tilde S$,
which implies $\Conj_{(x,S)} = \Conj_{(x,\tilde S)}$.
This in turn implies $S'_d = \tilde S'_d$ by construction of $\oc \At A$.
\end{subproof}

%%%%%%%%%%%%%%%%%%%%%%%%%%%%%%%%%%%%%%%%%%%%%%%%%%%%%%%%%%%%%%%%%%%%%%%%%%%
\begin{subclm}
\label{clm:sim:sim:der:Vpreflin}
%%%%%%%%%%%%%%%%%%%%%%%%%%%%%%%%%%%%%%%%%%%%%%%%%%%%%%%%%%%%%%%%%%%%%%%%%%%
For every $u \in V$, the set $\{v \in V \st v \gle_{\G(\oc\At A)} u\}$
is linearly ordered \wrt\@ $\edge{*}{\stratbis}$.
\end{subclm}

\begin{subproof}[Proof of Claim~\ref{clm:sim:sim:der:Vpreflin}]
We reason by $\glt$-Induction (Theorem~\ref{thm:games:ind:pos}).
So let $u \in V$ be such that the property holds for all $w \glt_{\G(\oc \At A)} u$.

Assume first that $u \in \OP V$.
In this case, we must have $u = (x,\stratbis(x,S))$ with $(x,S) \in V$.
By induction hypothesis, the set 
$\{v \in V \st v \gle_{\G(\oc\At A)} (x,S)\}$
is linearly ordered \wrt\@ $\edge{*}{\stratbis}$.
On the other hand, it follows from Claim~\ref{clm:sim:sim:der:Sunique}
that $(x,S)$ is the only immediate $\edge{}\stratbis$-predecessor of $u$ in $V$.
Since $(x,S) \edge{}{\stratbis} u$,
we get the result by Proposition~\ref{prop:games:edges}.

Assume now that $u \in \PP V$.
If $u = \iota_\oc$ then the result is trivial.
Otherwise, $u$ is of the form $(\Succ_d(x),S'_d)$
and its membership to $V$ is given by the last clause defining $V$.
Let $S$ be such that $(x,S) \in V$ and such that
$\stratbis(x,S) = (S,\Conj_{(x,S)})$ with $(d,S'_d) \in \Conj_{(x,S)}$.
Since $(x,\stratbis(x,S)) \edge{}{\stratbis} u$ with $(x,\stratbis(x,S)) \in V$,
by induction hypothesis
the set
$\{v \st v \gle_{\G(\oc\At A)} (x,\stratbis(x,S))\}$
is linearly ordered \wrt\@ $\edge{*}{\stratbis}$.
In order to obtain the result for
$\{v \st v \gle_{\G(\oc\At A)} (\Succ_d(x),S'_d)\}$
we need to show that $(x,\stratbis(x,S))$ is the unique
immediate $\edge{}\stratbis$-predecessor of $(\Succ_d(x),S'_d)$ in $V$.
But if $(x,\stratbis(x,\tilde S)) \in V$
then we should have $(x,\tilde S) \in V$,
so that $\tilde S = S$ by Claim~\ref{clm:sim:sim:der:Sunique}.
\end{subproof}

%%%%%%%%%%%%%%%%%%%%%%%%%%%%%%%%%%%%%%%%%%%%%%%%%%%%%%%%%%%%%%%%%%%%%%%%%%%
\begin{subclm}
\label{clm:sim:sim:der:Vinclplay}
%%%%%%%%%%%%%%%%%%%%%%%%%%%%%%%%%%%%%%%%%%%%%%%%%%%%%%%%%%%%%%%%%%%%%%%%%%%
For every $u \in V$, there is an infinite play $U$ of $\stratbis$
such that:
\[
\big( \forall v \gle_{\G(\oc\At A)} u \big)
\Big(
v \in V ~~\liff~~ u \in U
\Big)
\]
\end{subclm}

\begin{subproof}[Proof of Claim~\ref{clm:sim:sim:der:Vinclplay}]
Let $u \in V$.
First, by Lemma~\ref{lem:games:infplay:pos}
there is an infinite play $U_0$ in the game  $\G(\oc\At A)\restr\{\stratbis\}$
such that $u \in U_0$
and $u \edge{*}{\stratbis} v$ for all $v \in U_0$.
By Comprehension for Product Types (Theorem~\ref{thm:funto:ca})
we let 
\[
U \quad\deq\quad U_0 \cup \{v \in V \st v \gle_{\G(\oc\At A)} u\}
\]
We then get $\Play(\stratbis,\iota_\oc,U)$
from Claim~\ref{clm:sim:sim:der:Vpreflin}
and $\Play(\stratbis,u,U_0)$.
\end{subproof}

%%%%%%%%%%%%%%%%%%%%%%%%%%%%%%%%%%%%%%%%%%%%%%%%%%%%%%%%%%%%%%%%%%%%%%%%%%%
\begin{subclm}
\label{clm:sim:sim:der:Vstep}
%%%%%%%%%%%%%%%%%%%%%%%%%%%%%%%%%%%%%%%%%%%%%%%%%%%%%%%%%%%%%%%%%%%%%%%%%%%
Let $(x,S) \in V$, and assume $(x,q),(\Succ_d(x),q') \in T$
with $q \in \pi_2(S)$.
Then there is some $S'_d \in Q_{\oc \At A}$ such that $(\Succ_d(x),S'_d) \in V$
and $(q,q') \in S'_d$.
Moreover, we have $(d,S'_d) \in \Conj_{(x,S)}$ for $(S,\Conj_{(x,S)}) = \stratbis(x,S)$.
\end{subclm}

\begin{subproof}[Proof of Claim \thesubclm]
Since $T$ is a play of $\strat$, we have $(d,q') \in \Conj$
for $(q,\Conj) = \strat(x,q)$.
Moreover, by Claim~\ref{clm:sim:sim:der:Vinclplay},
$(x,S)$ belongs to an infinite play of $\stratbis$.
Since $q \in \pi_2(S)$, by definition of $\strat$
this implies that there is some $S'_d$ such that
$(d,S'_d) \in \Conj_{(x,S)}$ for $(S,\Conj_{(x,S)}) = \stratbis(x,S)$
and $(q,q') \in S'_d$.
We then obtain $(\Succ_d(x),S'_d) \in V$ by definition of $V$.
\end{subproof}

We now proceed to show:
\[
\Play(\stratbis,\, \iota_\oc,\, V)
\quad\land\quad
\Trace(T,V)
\]

\noindent
We begin with $\Trace(T,V)$.
First, we have
$\Path(\G(\At A),\iota,T)$
since $T$ is a play of $\strat$.
Moreover

%%%%%%%%%%%%%%%%%%%%%%%%%%%%%%%%%%%%%%%%%%%%%%%%%%%%%%%%%%%%%%%%%%%%%%%%%%%
\begin{subclm}
\label{clm:sim:sim:der:Tlift}
%%%%%%%%%%%%%%%%%%%%%%%%%%%%%%%%%%%%%%%%%%%%%%%%%%%%%%%%%%%%%%%%%%%%%%%%%%%
\[
\big( \forall (x,q) \in \PP T \big) \big( \exists S \in Q_{\oc \At A} \big)
  \Big( (x,S) \in \PP V ~~\land~~ q \in \pi_2(S) \Big)
\]
\end{subclm}

\begin{subproof}[Proof of Claim \thesubclm]
Using the Induction Axiom of $\FSO$
(\S\ref{sec:ax:ind}), we show
\[
(\forall x) \big( \forall q \in Q_{\At A} \big)
\Big( (x,q) \in \PP T
~~\limp~~
\big( \exists S \in Q_{\oc \At A} \big)
  \Big[ (x,S) \in \PP V ~\land~ q \in \pi_2(S) \Big]
\Big)
\]

\noindent
For the base case $x = \Root$, if $(x,q) \in T$ then we must have
$q = \init q_{\At A}$, so $q \in \pi_2(\init q_{\oc \At A})$.
Assume now the property for $x$, and consider $\Succ_d(x)$ and $q,q' \in Q_{\At A}$
such that $(x,q) \in T$ and $(\Succ_d(x),q') \in T$.
Furthermore, by induction hypothesis, let $S \in Q_{\oc \At A}$
such that $(x,S) \in V$ and $q \in \pi_2(S)$.
By Claim~\ref{clm:sim:sim:der:Vstep},
we then get $(\Succ_d(x),S'_d) \in V$ for some $S'_d \in Q_{\oc \At A}$
with $q' \in \pi_2(S'_d)$.
%%
%Since $T$ is a play of $\strat$, we must have $(d,q') \in \Conj$
%with $(q,\Conj) = \strat(x,q)$.
%But since $(x,S) \in V$,
%it follows from Claim~\ref{clm:sim:sim:der:Vinclplay}
%that $(x,S)$ belongs to an infinite play of $\stratbis$.
%So by defintion of $\strat$, 
%we have $(q,q') \in S'_d$
%where $(d,S'_d) \in \Conj_{(x,S)}$ for $(S,\Conj_{(x,S)}) = \stratbis(x,S)$.
%Since $(\Succ_d(x),q')$ belongs to $T$, by defintion of $V$
%we have $(\Succ_d(x),S'_d) \in V$ and we are done.
\end{subproof}

We can now show the last required property for $\Trace(T,V)$,
namely:

%%%%%%%%%%%%%%%%%%%%%%%%%%%%%%%%%%%%%%%%%%%%%%%%%%%%%%%%%%%%%%%%%%%%%%%%%%%
\begin{subclm}
%%%%%%%%%%%%%%%%%%%%%%%%%%%%%%%%%%%%%%%%%%%%%%%%%%%%%%%%%%%%%%%%%%%%%%%%%%%
\[
  \big( \forall (x,q),(y,q') \in \PP T \big)
  \big( \forall S \in Q_{\oc \At A} \big)
  \left[
  (x,q) \glt^{\Prop;\Opp}_{\G(\At A)} (y,q')
  ~~\limp~~
  (y,S) \in \PP V
  ~~\limp~~
  (q,q') \in S
  \right]
\]
\end{subclm}

\begin{subproof}[Proof of Claim \thesubclm]
Let $(x,q),(y,q') \in T$
and $S' \in Q_{\oc \At A}$
such that
$(x,q) \glt^{\Prop;\Opp}_{\G(\At A)} (y,q')$
and
$(y,S') \in V$.
Then by definition of $\glt_{\G(\At A)}$ we must have $y = \Succ_d(x)$
for some $d \in \Dir$.
Moreover, by Claim~\ref{clm:sim:sim:der:Tlift}
there is some $S \in Q_{\oc \At A}$ such that $(x,S) \in V$
and $q \in \pi_2(S)$.
By Claim~\ref{clm:sim:sim:der:Vstep},
we then have $(\Succ_d(x),S'_d) \in V$ for some $S'_d \in Q_{\oc \At A}$
with $(q,q') \in S'_d$.
%%
%Since $T$ is a play of $\strat$, 
%we have $(d,q') \in \Conj$
%for $(q,\Conj) = \strat(x,q)$.
%But since $(x,S) \in V$,
%it follows from Claim~\ref{clm:sim:sim:der:Vinclplay}
%that $(x,S)$ belongs to an infinite play of $\stratbis$.
%So by defintion of $\strat$, 
%we have $(q,q') \in S'_d$
%where $(d,S'_d) \in \Conj_{(x,S)}$ for $(S,\Conj_{(x,S)}) = \stratbis(x,S)$.
%%
%On the other hand, $(\Succ_d(x),q')$ belongs to $T$,
%and by defintion of $V$
%we have $(\Succ_d(x),S'_d) \in V$.
It follows from Claim~\ref{clm:sim:sim:der:Sunique}
that $S' = S'_d$
so that $(q,q') \in S'$
and we are done.
%
%But by definition of $\strat$, this implies $(q,q') \in S'_d$
%where $(d,S'_d) \in \Conj_{(x,S)}$ for $(S,\Conj_{(x,S)}) = \stratbis(x,S)$.
%Since $S'_d$ is unique such that $(d,S'_d) \in \Conj_{(x,S)}$,
%by definition of $V$ we must have $S'_d = S'$, so that $(q,q') \in S'$.
\end{subproof}

We now turn to showing
$\Play(\stratbis, \iota_\oc, V)$.
Since $\iota_\oc \in V$,
thanks to Proposition~\ref{prop:games:edges:lin}
%Lemma~\ref{lem:games:infplay}
it remains to show:
\[
\left\{
\begin{array}{c l}
& (\forall u \in V) \big(
  \iota_\oc \edge{*}{\stratbis} u
  \big)
\\
  \land
& (\forall u \in V) (\exists ! v \in V) \big(
  u \edge{}{\stratbis} v
  \big)
\\
  \land
& (\forall v \in V) \Big(
     v\neq \iota_\oc ~~\limp~~
     (\exists u \in V) \big(
     u \edge{}{\stratbis} v
     \big)
  \Big)
\end{array}
\right.
\]

\noindent
First, we easily have:

%%%%%%%%%%%%%%%%%%%%%%%%%%%%%%%%%%%%%%%%%%%%%%%%%%%%%%%%%%%%%%%%%%%%%%%%%%%
\begin{subclm}
%%%%%%%%%%%%%%%%%%%%%%%%%%%%%%%%%%%%%%%%%%%%%%%%%%%%%%%%%%%%%%%%%%%%%%%%%%%
\[
 (\forall v \in V) \Big(
     v\neq \iota_\oc ~~\limp~~
     (\exists u \in V) \big(
     u \edge{}{\stratbis} v
     \big)
  \Big)
\]
\end{subclm}

\begin{subproof}[Proof of Claim \thesubclm]
The result follows from Claim~\ref{clm:sim:sim:der:Vinclplay},
but it can be proved directly,
without the inductions underlying Claim~\ref{clm:sim:sim:der:Vinclplay}.
Indeed, if $v = (x,\stratbis(x,S))$, with $(x,S) \in V$, then the result
directly follows from the definitions of $V$ and of
the game $\G(\oc \At A)\restr\{\stratbis\}$.
Otherwise, we have $v = (\Succ_d(x),S'_d)$,
and there is $(x,S) \in V$ such that
$\stratbis(x,S) = (S,\Conj_{(x,S)})$ with $(d,S'_d) \in \Conj_{(x,S)}$.
But $(x,S) \in V$ implies
$(x,\stratbis(x,S)) \in V$, and again the result 
directly follows from the definition of $\G(\oc \At A)\restr\{\stratbis\}$.
\end{subproof}

It then easily follows that:
%%%%%%%%%%%%%%%%%%%%%%%%%%%%%%%%%%%%%%%%%%%%%%%%%%%%%%%%%%%%%%%%%%%%%%%%%%%
\begin{subclm}
%%%%%%%%%%%%%%%%%%%%%%%%%%%%%%%%%%%%%%%%%%%%%%%%%%%%%%%%%%%%%%%%%%%%%%%%%%%
\[
 (\forall u \in V) \left(
  \iota_\oc \edge{*}{\stratbis} u
  \right)
\]
\end{subclm}

\begin{subproof}[Proof of Claim \thesubclm]
First, we have $\iota_\oc \in V$ by definition
of $V$.
Moreover, given $u \in V$
we have either $\iota_\oc \edge{*}{\stratbis} u$
or $u \edge{*}{\stratbis} \iota_\oc$
by Claim~\ref{clm:sim:sim:der:Vpreflin}.
The result then follows from
Proposition~\ref{prop:games:edges}.
\end{subproof}

It remains to show:
\begin{equation}
\label{eq:sim:sim:der:succ}
 (\forall u \in V) (\exists ! v \in V) \left(
  u \edge{}{\stratbis} v
  \right)
\end{equation}

\noindent
To this end, we first show:

%%%%%%%%%%%%%%%%%%%%%%%%%%%%%%%%%%%%%%%%%%%%%%%%%%%%%%%%%%%%%%%%%%%%%%%%%%%
\begin{subclm}
\label{clm:sim:sim:der:Vlift}
%%%%%%%%%%%%%%%%%%%%%%%%%%%%%%%%%%%%%%%%%%%%%%%%%%%%%%%%%%%%%%%%%%%%%%%%%%%
\[
\big( \forall (x,S) \in \PP V \big)
\big( \exists q \in Q_{\At A} \big)
  \Big( (x,q) \in \PP T ~~\land~~ q \in \pi_2(S) \Big)
\]
\end{subclm}

\begin{subproof}[Proof of Claim \thesubclm]
Using the Induction Axiom of $\FSO$
(\S\ref{sec:ax:ind}), we show
\[
(\forall x) (\forall S \in Q_{\oc \At A})
\Big((x,S) \in \PP V
~~\limp~~
(\exists q \in Q_{\At A})
  \big( (x,q) \in \PP T ~~\land~~ q \in \pi_2(S) \big)
\Big)
\]

\noindent
For the base case $x = \Root$, if $(x,S) \in V$ then we must have
$S = \init q_{\oc \At A}$.
Then we are done since $\iota \in T$ 
and $\init q_{\At A} \in \pi_2(S)$.
Assume the property for $x$, and consider $\Succ_d(x)$ and $S,S' \in Q_{\oc \At A}$
such that $(x,S),(\Succ_d(x),S') \in V$.
Furthermore, by induction hypothesis, let $q \in Q_{\At A}$
such that $(x,q) \in T$ and $q \in \pi_2(S)$.
By definition of $V$, we have $(\Succ_d(x),q') \in T$ for some $q' \in Q_{\At A}$.
It then follows from Claim~\ref{clm:sim:sim:der:Vstep}
that $q' \in \pi_2(S'_d)$ for some $S'_d \in Q_{\oc \At A}$
such that $(\Succ_d(x),S'_d) \in V$.
But Claim~\ref{clm:sim:sim:der:Sunique}
implies $S' = S'_d$
so that $q' \in \pi_2(S')$
and we are done.
\end{subproof}

We can now prove~\eqref{eq:sim:sim:der:succ}.

\begin{subproof}[Proof of~\eqref{eq:sim:sim:der:succ}]
If $u = (x,S) \in V$, then $v = (x,\stratbis(x,S)) \in V$ and is unique such that
$u \edge{}{\stratbis} v$.
Otherwise, $u = (x,\stratbis(x,S))$ for some $(x,S) \in V$,
and we have to show that there are some unique 
$d \in \Dir$ and $S'_d \in Q_{\oc \At A}$ such that
$(\Succ_d(x),S'_d) \in V$.
First, by Clim~\ref{clm:sim:sim:der:Vlift}
there is some $q \in Q_{\At A}$ such that $(x,q) \in T$
and $q \in \pi_2(S)$.
Moreover, since $T$ is a play of $\strat$, with have $(\Succ_d(x),q') \in T$
for some unique $d \in \Dir$ and $q' \in Q_{\At A}$.
It then follows from Claim~\ref{clm:sim:sim:der:Vstep}
that there is some $S'_d \in Q_{\oc \At A}$
such that $(\Succ_d(x),S'_d) \in V$
and $u \edge{}{\stratbis} (\Succ_d(x),S'_d)$.
The uniqueness of $S'_d$ follows from 
Claim~\ref{clm:sim:sim:der:Sunique}.
\end{subproof}

This concludes the proof of Proposition~\ref{prop:sim:sim:der}.
\end{proof}

In the proof of
Proposition~\ref{prop:sim:sim:der} above,
we have used Claim~\ref{clm:sim:sim:der:Vlift}
in order to show that $V$ is a play of $\strat$.
Let us state here for the record that this has a more general converse:
Claim~\ref{clm:sim:sim:der:Vlift}
holds whenever $T$ is a trace in $V$ for $V$ a play
in $\G(\oc \At A)$:

%%%%%%%%%%%%%%%%%%%%%%%%%%%%%%%%%%%%%%%%%%%%%%%%%%%%%%%%%%%%%%%%%%%%%%%%%%%
\begin{lem}
\label{lem:sim:sim:trace}
%%%%%%%%%%%%%%%%%%%%%%%%%%%%%%%%%%%%%%%%%%%%%%%%%%%%%%%%%%%%%%%%%%%%%%%%%%%
Given $\funto{V}{\G(\oc \At A)}{\two}$ and $\funto{T}{\G(\At A)}{\two}$,
in $\FSOD$ we have
\[
\big( \forall (x,S) \in \PP V \big)
\big( \exists q \in Q_{\At A} \big)
  \Big[ (x,q) \in \PP T ~~\land~~ q \in \pi_2(S) \Big]
\]
whenever
\[
\Play\big( \G(\oc \At A),\, (\Root,\init q_{\oc \At A}),\, V \big)
~~\land~~
\Trace(T,V)
\]
\end{lem}

\begin{proof}
Using the Induction Axiom of $\FSO$
(\S\ref{sec:ax:ind}), we show
\[
(\forall x) \big( \forall S \in Q_{\oc \At A} \big)
\Big((x,S) \in \PP V
~~\limp~~
(\exists q \in Q_{\At A})
  \Big[ (x,q) \in \PP T ~~\land~~ q \in \pi_2(S) \Big]
\Big)
\]

\noindent
For the base case $x = \Root$, if $(x,S) \in V$,
since
$\Play(\G(\oc \At A),(\Root,\init q_{\oc \At A}),V)$
we must have
$S = \init q_{\oc \At A}$, so $\init q_{\At A} \in \pi_2(S)$.
Assume now the property for $x$, and consider 
$d \in \Dir$ and $S,S' \in Q_{\oc \At A}$
such that $(x,S) \in V$ and $(\Succ_d(x),S') \in V$.
Furthermore, by induction hypothesis, let $q \in Q_{\At A}$
such that $(x,q) \in T$ and $q \in \pi_2(S)$.
It follows from $\Path(\G(\At A),(\Root,\init q_{\At A}),T)$
that $(\Succ_{d'}(x),q') \in T$ for some $d' \in \Dir$
and some $q' \in Q_{\At A}$.
Moreover, $\Trace(T,V)$ implies
$q' \in \pi_2(S'')$ for some $S''$ such that
$(\Succ_{d'}(x),S'') \in V$.
But 
$\Play(\G(\oc \At A),(\Root,\init q_{\oc \At A}),V)$
implies $d' = d$ and $S'' = S'$ and we are done.
\end{proof}

%%%%%%%%%%%%%%%%%%%%%%%%%%%%%%%%%%%%%%%%%%%%%%%%%%%%%%%%%%%%%%%%%%%%%%%%%%%
\subsection{Reformulating the Acceptance Condition of $\oc\At A$.}
\label{sec:sim:omega}
%%%%%%%%%%%%%%%%%%%%%%%%%%%%%%%%%%%%%%%%%%%%%%%%%%%%%%%%%%%%%%%%%%%%%%%%%%%

For an automaton $\At A$ which we now assume to be $\HF$-closed
(in the sense of Definition~\ref{def:aut:hfclosed}),
we are going to formulate the $\FSO$-formula $\Win_{\oc \At A}$
as a parity condition, which will allow us to obtain a parity automaton $\ND(\At A)$
in~\S\ref{sec:sim:par}.
In order to obtain a parity condition from $\Win_{\oc \At A}$ we note
(following~\cite{walukiewicz02tcs}) that (when read in the standard model) it defines
an $\omega$-regular condition, which can thus 
by McNaughton's Theorem~\cite{mcnaughton66ic}
(see also \eg~\cite{thomas97handbook,pp04book})
be formulated with a deterministic parity automaton on $\omega$-words.
We are actually not going to formalize McNaughton's Theorem in our setting.
Rather, we will apply Proposition~\ref{prop:bfsos:bfso},
which allows us to import in $\FSO$ any true $\FSO$-formula on the infinite
paths of $\univ$.
%For $\HF$-closed $\At A$, in the sense of Definition~\ref{def:aut:hfclosed},
%we now give a transformation of the formula $\Win_{\oc \At A}$
%to a formula allowing us to 
%apply Proposition~\ref{prop:bfsos:bfso},
%and thus by McNaughton's Theorem~\cite{mcnaughton66ic}
%(see also \eg~\cite{thomas97handbook,pp04book}),
%to formulate it as a parity $\omega$-word automaton.
Our way to the application of 
Proposition~\ref{prop:bfsos:bfso}
proceeds with constructions similar to some of those in the proof of
Theorem~\ref{thm:compl:red}.

Consider some $\funto{V}{\G(\oc \At A)}{\two}$
such that:
\[
\Play(\G(\oc \At A),\, (\Root,\init q_{\oc \At A}),\, V)
\]

\noindent
By Comprehension for Product Types (Theorem~\ref{thm:funto:ca}), 
let $\funto{|V|}{\univ}{\two}$
be the set of all $x \in \univ$ such that
$(x,S) \in V$ for some $S \in Q_{\oc \At A}$.
Note that $\TPath(|V|)$
(recall that $\TPath$ is defined in~\eqref{eq:bfsos:rpath}).
Furthermore, by $\HF$-Bounded Choice for Functions (\S\ref{sec:ax:choice}),
let $\funto{\tilde V}{|V|}{Q_{\oc \At A}}$
take $x \in |V|$ to the unique $S \in Q_{\oc \At A}$
such that $(x,S) \in V$.

In $\FSO$ we have that $\Win_{\oc \At A}$
is equivalent to the following formula
$\Win[\Lt]^{|V|}_{\oc \At A}(\tilde V)$:
\[
\left(\forall \funto{\tilde T}{|V|}{Q_{\At A}}\right)
\left[
\Trace[\Lt]^{|V|}(\tilde T,\tilde V)
~~\limp~~
\Par[\Lt]^{|V|}(\col_{\At A},n_{\At A},\tilde T)
\right]
\]

\noindent
where
\begin{itemize}
\item the formula
$\Trace[\Lt]^{|V|}(\tilde T,\tilde V)$
is
\[
\left\{
\begin{array}{r l}
& (\forall x \in |V|)
  \left[\tilde T(x) \in \pi_2(\tilde V(x)) \right]
\\
  \land
& (\forall x,y \in |V|)
  \left[
    \Succ_{\Lt}(x,y)
    ~~\limp~~
    (\tilde T(x),\tilde T(y)) \in \tilde V(y)
  \right]
\end{array}
\right.
\]
%where we use the formula
with
\[
\Succ_{\Lt}(x,y)
\quad\deq\quad
x \Lt y
~\land~ \lnot \exists z\left(x \Lt z \Lt y\right)
\]

\item 
and, for $\col : Q_{\At A} \fsoto [0,n]$,
the formula 
$\Par[\Lt]^{|V|}(\col,n,\tilde T)$
is (using Convention~\ref{conv:games:colors}):
\[
\big( \exists m \in \even(n) \big)
\left[
\begin{array}{r l}
& (\forall x \in |V|) (\exists y \in |V|)
  \big( x \Lt y ~~\land~~ \col(\tilde T(y)) \Eq m \big)
\\
  \land
& (\exists x \in |V|) (\forall y \in |V|)
  \big( x \Lt y ~~\limp~~ \col(\tilde T(y)) \geq m \big)
\end{array}
\right]
\]
\end{itemize}

Let us first note the following simple property.
Recall from Definition~\ref{def:aut:parity}
that $\col_{\At A} : Q_{\At A} \fsoto [0,n_{\At A}]$
is a coloring of the states of $\At A$,
while $\hat\col_{\At A}$ colors the positions of $\G(\At A)$,
by taking for $\Prop$-positions $(x,q)$
the color given by $\col_{\At A}$ to $q$ and for $\Opp$-positions the maximal color
$n_{\At A}$.

%%%%%%%%%%%%%%%%%%%%%%%%%%%%%%%%%%%%%%%%%%%%%%%%%%%%%%%%%%%%%%%%%%%%%%%%%%%
\begin{lem}
\label{lem:sim:omega:parity}
%%%%%%%%%%%%%%%%%%%%%%%%%%%%%%%%%%%%%%%%%%%%%%%%%%%%%%%%%%%%%%%%%%%%%%%%%%%
Assume given $V$ and $|V|$ as above.
Let $\funto{T}{\G(\At A)}{\two}$
and $\funto{\tilde T}{|V|}{Q_{\At A}}$
such that
\[
\Path\big( \G(\At A), (\Root,\init q_{\At A}),T \big)
\quad\land\quad
\big( \forall (x,q) \in \PP{\G(\At A)} \big)
\left[
  (x,q) \in T
  ~~\liff~~
  \left(
    x \in |V|
    ~\land~
    \tilde T(x) = q
  \right)
\right]
\]

\noindent
Then:
\[
\Par \big(\At A,\hat\col_{\At A},n_{\At A},T \big)
~~\liff~~
\Par[\Lt]^{|V|} \big(\col_{\At A},n_{\At A},\tilde T \big)
\]
\end{lem}

%%%%%%%%%%%%%%%%%%%%%%%%%%%%%%%%%%%%%%%%%%%%%%%%%%%%%%%%%%%%%%%%%%%%%%%%%%%
\begin{lem}
\label{lem:sim:omega:trace}
%%%%%%%%%%%%%%%%%%%%%%%%%%%%%%%%%%%%%%%%%%%%%%%%%%%%%%%%%%%%%%%%%%%%%%%%%%%
Given 
$V$, $|V|$ and $\tilde V$ as above,
$\FSOD$ proves that
\[
\Win_{\oc \At A}(V) ~~\liff~~
\Win[\Lt]^{|V|}_{\oc \At A}(\tilde V)
\]
\end{lem}

\begin{proof}
Recall that the formula $\Win_{\oc \At A}$
requires no condition 
\wrt\@ the \emph{transitions} of $\G(\At A)$.
We proceed as follows:
\begin{itemize}
\item
Assume first $\Win_{\oc \At A}(V)$, and let
$\funto{\tilde T}{|V|}{Q_{\At A}}$ such that
$\Trace[\Lt]^{|V|}(\tilde T,\tilde V)$.
Using Remark~\ref{rem:ax:hf:well-order-hf},
let $\prec$ be a well-order on $\Pne(\Dir \times Q_{\At A})$.
By Comprehension for Product Types
(Theorem~\ref{thm:funto:ca}),
let $\funto{T}{\G(\At A)}{\two}$
such that
for all $(x,q) \in \PP{\G(\At A)}$, we have
\[
  (x,q) \in \PP T ~~\liff~~
  \left(
    x \in |V|
    ~\land~
    \tilde T(x) = q
  \right)
  %\right]
\]
and such that $(x,\Conj) \in T_\Opp$ iff
$\Conj \in \Pne(\Dir \times Q_{\At A})$ is $\preceq$-minimal
such that $\tilde T(\Succ_d(x)) = q$ for some $(d,q) \in \Conj$.
Since
$V$ is an infinite play of $\G(\oc \At A)$ from
$(\Root,\init q_{\oc \At A})$,
we may conclude by Lemma~\ref{lem:sim:omega:parity}
as soon as we show:
%We first show
\[
\Trace(T,V)
\]

\noindent
We obviously have $\Path(\G(\At A),T)$
as well as $(\Root,\init q_{\At A}) \gle_{\G(\At A)} u$
for all $u \in T$.
Moreover we have:
\[
\big( \forall (x,q) \in \PP T \big) \big( \exists S \in Q_{\oc \At A} \big)
\Big[
(x,S) \in \PP V ~~\land~~ q \in \pi_2(S)
\Big]
\]
To see this, let $(x,q) \in T$, so that $\tilde T(x) = q$.
So by assumption we have $q \in \pi_2(\tilde V(x))$,
and we are done since $(x,\tilde V(x)) \in V(x)$.
%\begin{itemize}
%\item \begin{proof}
%Let $(x,q) \in T$, so that $\tilde T(x) = q$.
%So by assumption we have $q \in \pi_2(\tilde V(x))$,
%and we are done since $(x,\tilde V(x)) \in V(x)$.
%\end{proof}
%\end{itemize}

\noindent
Finally we have
\[
  \big( \forall (x,q),(y,q') \in \PP T \big)
  \big( \forall S \in Q_{\oc \At A} \big)
  \left[
  (x,q) \glt^{\Prop;\Opp}_{\G(\At A)} (y,q')
  ~~\limp~~
  (y,S) \in \PP V
  ~~\limp~~
  (q,q') \in S
  \right]
\]
To see this,
given $(x,q),(y,q') \in T$
such that
$(x,q) \glt^{\Prop;\Opp}_{\G(\At A)} (y,q')$
we necessarily have
$\Succ_{\Lt}(x,y)$,
so that $(q,q') \in \tilde V(y)$
since $\tilde T(x) = q$ and $\tilde T(y)= q'$.
But then we are done since $\tilde V(y)$ is the unique
$S \in Q_{\oc \At A}$ such that $(y,S) \in V$.
%\begin{itemize}
%\item \begin{proof}
%Given $(x,q),(y,q') \in T$
%such that
%$(x,q) \glt^{\Prop;\Opp}_{\G(\At A)} (y,q')$
%we necessarily have
%$\Succ_{\Lt}(x,y)$,
%so that $(q,q') \in \tilde V(y)$
%since $\tilde T(x) = q$ and $\tilde T(y)= q'$.
%But then we are done since $\tilde V(y)$ is the unique
%$S \in Q_{\oc \At A}$ such that $(y,S) \in V$.
%\end{proof}
%\end{itemize}

%It remains to show
%\[
%\Par[\Lt]^{|V|,c,n}(\tilde T)
%\]
%But this follows directly from 

\item
Conversely, assume
$\Win[\Lt]^{|V|}_{\oc \At A}(\tilde V)$
and let $\funto{T}{\G(\At A)}{\two}$
such that $\Trace(T,V)$.
Since
$V$ is an infinite play of $\G(\oc \At A)$ from
$(\Root,\init q_{\oc \At A})$,
Lemma~\ref{lem:sim:sim:trace}
implies:
\[
\big( \forall (x,S) \in \PP V \big)
\big(\exists q \in Q_{\At A} \big)
  \Big[ (x,q) \in \PP T ~~\land~~ q \in \pi_2(S) \Big]
\]
It follows that for all $x \in |V|$ there is $q \in Q_{\At A}$
such that $(x,q) \in T$, and this defines 
$\funto{\tilde T}{|V|}{Q_{\At A}}$
by $\HF$-Bounded Choice for Functions (\S\ref{sec:ax:choice}).
Note that we have:
\[
  (\forall x) (\forall q \in Q_{\At A}) \left[
  (x,q) \in \PP T ~~\liff~~
  \left(
    x \in |V|
    ~\land~
    \tilde T(x) = q
  \right)
  \right]
\]

\noindent
We can then conclude by Lemma~\ref{lem:sim:omega:parity}
as soon as we show:
\[
\Trace[\Lt]^{|V|}(\tilde T,\tilde V)
\]
To see this,
first, for all $x \in |V|$, we have $(x,\tilde T(x)) \in T$,
so that $T(x) \in \pi_2(\tilde V(x))$ by definition of $\tilde V$.
Moreover, given $x,y \in |V|$
with $\Succ_{\Lt}(x,y)$,
we have
\[
(x,\tilde T(x)) \glt^{\Prop;\Opp}_{\G(\At A)} (y,\tilde T(y))
\]
so that $(\tilde T(x),\tilde T(y)) \in \tilde V(y)$
since $\Trace(T,V)$.
%\begin{itemize}
%\item \begin{proof}
%First, for all $x \in |V|$, we have $(x,\tilde T(x)) \in T$,
%so that $T(x) \in \pi_2(\tilde V(x))$ by definition of $\tilde V$.
%%
%Moreover, given $x,y \in |V|$
%with $\Succ_{\Lt}(x,y)$,
%we have
%\[
%  (x,\tilde T(x)) \glt^{\Prop;\Opp}_{\G(\At A)} (y,\tilde T(y))
%\]
%so that $(\tilde T(x),\tilde T(y)) \in \tilde V(y)$
%since $\Trace(T,V)$.
%\end{proof}
%\end{itemize}
\end{itemize}

\noindent
This concludes the proof of Lemma~\ref{lem:sim:omega:trace}.
\end{proof}

We are now going to show that 
$\Win[\Lt]^{|V|}_{\oc \At A}(\tilde V)$
is equivalent in $\FSO$ to a \emph{parity} automaton on $\omega$-words.
This relies on McNaughton's Theorem~\cite{mcnaughton66ic}
applied in the usual standard model $\StdN$ of $\omega$-words,
and, via Proposition~\ref{prop:bfsos:bfso}, on the completeness of $\FSOW$.
In order to apply Proposition~\ref{prop:bfsos:bfso},
we rewrite 
$\Win[\Lt]^{|V|}_{\oc \At A}(\tilde V)$
as the relativization to $|V|$ of the $\FSOW$-formula
\[
\Win[\Lt]_{\oc \At A}(\tilde V)
\quad\deq\quad
\big( \forall \tilde T : Q_{\At A} \big)
\left[
\Trace[\Lt](\tilde V,\tilde T)
~~\limp~~
\Par[\Lt](\col_{\At A},n_{\At A},\tilde T)
\right]
\]

\noindent
where
$\Par[\Lt](\col,n,\tilde T)$ is the formula of 
Definition~\ref{def:compl:red:omega:parity},
and where
\[
\Trace[\Lt](\tilde V,\tilde T)\quad\deq\quad
\left\{
\begin{array}{r l}
& (\forall x)
  \left[\tilde T(x) \in \pi_2(\tilde V(x)) \right]
\\
  \land
& (\forall x) (\forall y)
  \left[
    \Succ_{\Lt}(x,y)
    ~~\limp~~
    (\tilde T(x),\tilde T(y)) \in \tilde V(y)
  \right]
\end{array}
\right.
\]

\noindent
Note that
$\Win[\Lt]^{|V|}_{\oc \At A}(\tilde V)$ is the relativization
to $|V|$ of 
$\Win[\Lt]_{\oc \At A}(\tilde V)$:
\[
\Win[\Lt]^{|V|}_{\oc \At A}(\tilde V)
\quad=\quad
\left(\Win[\Lt]_{\oc \At A}\right)^{|V|}(\tilde V)
\]

\noindent
Since $\At A$ is $\HF$-closed, the formula
$\Win[\Lt]_{\oc \At A}(\tilde V)$ is also $\HF$-closed,
and we can look at it
in the standard model $\StdN$ of $\omega$-words (see~\S\ref{sec:msow}).
By McNaughton's Theorem~\cite{mcnaughton66ic}
(see also \eg~\cite{thomas97handbook,pp04book}),
there is a deterministic parity $\omega$-word automaton
$\At D = (Q_{\At D},\, \init q_{\At D},\, \trans_{\At D}, \,c_{\At D})$
over $Q_{\oc\At A}$,
which
accepts $\tilde V$ exactly when:
\[
\StdN \models
\Win[\Lt]_{\oc \At A}(\tilde V)
\]

\noindent
It then follows
that in $\StdN$,
for all $\tilde V : Q_{\oc \At A}$,
the formula
$\Win[\Lt]_{\oc \At A}(\tilde V)$
is equivalent to
\begin{multline*}
\big( \forall \tilde R : Q_{\At D} \big)
\bigg(
\tilde R(\Root) = \init q_{\At D}
~~\limp~~
\\
(\forall x)(\forall y)
\left[
\Succ_{\Lt}(x,y) ~~\limp~~ \tilde R(y) = \trans_{\At D}\big(\tilde R(x),\tilde V(x)\big)
\right]
~~\limp~~
\Par[\Lt]\big( \col_{\At D},n_{\At D},\tilde R \big)
\bigg)
\end{multline*}

\noindent
Proposition~\ref{prop:bfsos:bfso} then implies that $\FSO$ 
proves that for $\funto{\tilde V}{|V|}{Q_{\oc \At A}}$,
the formula
$\Win[\Lt]_{\oc \At A}^{|V|}(\tilde V)$
is equivalent to
\begin{multline*}
\big( \forall \funto{\tilde R}{|V|}{Q_{\At D}} \big)
\bigg(
\tilde R(\Root) = \init q_{\At D}
~~\limp~~
\\
(\forall x \in |V|) (\forall y \in |V|)
\Big[
\Succ_{\Lt}(x,y) ~~\limp~~ \tilde R(y) = \trans_{\At D}\big(\tilde R(x),\tilde V(x)\big)
\Big]
~~\limp~~
\\
\Par[\Lt]^{|V|}\big( \col_{\At D},n_{\At D},\tilde R \big)
\bigg)
\end{multline*}

%%%%%%%%%%%%%%%%%%%%%%%%%%%%%%%%%%%%%%%%%%%%%%%%%%%%%%%%%%%%%%%%%%%%%%%%%%%
\subsection{Definition of the Parity Automaton $\ND(\At A)$.}
\label{sec:sim:par}
%%%%%%%%%%%%%%%%%%%%%%%%%%%%%%%%%%%%%%%%%%%%%%%%%%%%%%%%%%%%%%%%%%%%%%%%%%%
Consider an alternating parity tree automaton $\At A : \Sigma$
as in the beginning of~\S\ref{sec:sim},
and assume it to be $\HF$-closed.
Let $\oc \At A :\Sigma$ be defined as in~\S\ref{sec:sim:sim}.
Moreover, let $\At D : Q_{\oc \At A}$ be the parity
deterministic $\omega$-word automaton
of~\S\ref{sec:sim:omega}.
We then let
\[
\ND(\At A) \quad\deq\quad
(Q_{\oc \At A} \times Q_{\At D}
~,~ (\init q_{\oc \At A},\init q_{\At D})
~,~ \trans_{\ND (\At A)}
~,~ \col_{\ND(\At A)}
~,~ n_{\At D})
\]

\noindent
where:
\begin{itemize}
\item
the transition function
\[
\trans_{\ND(\At A)}
\quad:\quad
(Q_{\oc \At A} \times Q_{\At D}) \times \Sigma
\quad\longto\quad
\Pne(\Pne(\Dir \times (Q_{\oc \At A} \times Q_{\At D})))
\]

\noindent
takes $((S,q),\al a)$
to the set of all $\Conj \in \Pne(\Pne(\Dir \times (Q_{\oc \At A} \times Q_{\At D})))$
such that for some $\Conj_{\oc \At A} \in \trans_{\oc \At A}(S,\al a)$,
\[
\Conj
\quad=\quad
\Big\{
\big(d ~,~ \big(S'_d ~,~ \trans_{\At D}(q,S) \big) \big)
~~\st~~
(d,S'_d) \in \Conj_{\oc \At A}
\Big\}
\]

\item
the coloring
$\funto{\col_{\ND(\At A)}}{Q_{\oc \At A} \times Q_{\At D}}{[0,n_{\At D}]}$
takes $(S,q)$ to $\col_{\At D}(q)$.
\end{itemize}

\noindent
Note that $\ND(\At A): \Sigma$ is $\HF$-closed by
Remark~\ref{rem:aut:hfclosed}.
We shall now show that $\ND(\At A)$ is equivalent to $\At A$,
thus completing the proof of the Simulation Theorem~\ref{thm:sim}.
The proof is split into
Propositions~\ref{prop:sim:par:prom}
and~\ref{prop:sim:par:der}.
As expected, we invoke Theorem~\ref{thm:sim:sim:cor}, that $\FSOD$ proves the equivalence of $!\At A$ and $\At A$.

%%%%%%%%%%%%%%%%%%%%%%%%%%%%%%%%%%%%%%%%%%%%%%%%%%%%%%%%%%%%%%%%%%%%%%%%%%%
\begin{conv}
%%%%%%%%%%%%%%%%%%%%%%%%%%%%%%%%%%%%%%%%%%%%%%%%%%%%%%%%%%%%%%%%%%%%%%%%%%%
In Propositions~\ref{prop:sim:par:prom} and~\ref{prop:sim:par:der},
for fixed automata $\oc \At A$ and $\ND(\At A)$,
we let
\[
\iota_\oc ~~\deq~~ (\Root, \init q_{\oc \At A})
\qquad\text{and}\qquad
\iota_\ND ~~\deq~~ (\Root, \init q_{\ND(\At A)})
\]
\end{conv}

%%%%%%%%%%%%%%%%%%%%%%%%%%%%%%%%%%%%%%%%%%%%%%%%%%%%%%%%%%%%%%%%%%%%%%%%%%%
\begin{prop}
\label{prop:sim:par:prom}
%%%%%%%%%%%%%%%%%%%%%%%%%%%%%%%%%%%%%%%%%%%%%%%%%%%%%%%%%%%%%%%%%%%%%%%%%%%
Fix an $\HF$-closed automaton $\At A : \Sigma$ and consider 
$\ND(\At A) : \Sigma$ as defined above.
Then $\FSOD$ proves that for all $F:\Sigma$,
if $\At A$ accepts $F$ then $\ND(\At A)$ accepts $F$.
\end{prop}

\begin{proof}
Thanks to Theorem~\ref{thm:sim:sim:cor},
%Proposition~\ref{prop:sim:sim:prom},
we are done if we show
that $\ND(\At A)$ accepts $F$ whenever $\oc\At A$ accepts~$F$.
%\begin{itemize}
%\item
%If $\oc\At A$ accepts $F$ then $\ND(\At A)$ accepts $F$.
%\end{itemize}
%
Let $\funto{\strat}{\G(\oc \At A,F)_\Prop}{\Opp_{\G(\oc \At A,F)}}$ 
be the winning $\Prop$-strategy in $\G(\oc \At A,F)$.
We define a strategy 
$\funto{\stratbis}{\G(\ND(\At A),F)_\Prop}{\Opp_{\G(\ND(\At A),F)}}$
as follows.
\begin{itemize}
\item \emph{Definition of $\stratbis$.}
By $\HF$-Bounded Choice for Functions (\S\ref{sec:ax:choice}),
we let $\stratbis(x,(S,q_{\At D}))$
be $((S,q_{\At D}),\Conj)$,
where
$\Conj \in \Pne(\Pne(\Dir \times (Q_{\oc \At A} \times Q_{\At D})))$
is defined by Comprehension for $\HF$-Sets
(Remark~\ref{rem:hfchoice})
as the set of all $(d,(S'_d,\trans_{\At D}(q_{\At D},S)))$
such that $(d,S'_d) \in \Conj_{\oc \At A}$,
where $\strat(x,S) = (S,\Conj_{\oc \At A})$.
\end{itemize}

\noindent
It remains to show that $\stratbis$ is winning.
So let $T$ such that
%be an infinite play of $\stratbis$, \ie\@ such that
\[
\Play(\stratbis,\, \iota_\ND ,\, T)
\]

\noindent
By Comprehension for Product Types
(Theorem~\ref{thm:funto:ca}),
let $\funto{||U||}{\univ \times Q_{\At A}}{\two}$
consist of the $(x,q_{\oc \At A})$
for which there is $q_{\At D} \in Q_{\At D}$
such that $(x,(q_{\oc \At A},q_{\At D})) \in T$.
By $\HF$-Bounded Choice for Functions (\S\ref{sec:ax:choice}),
now let $\funto{\tilde U}{||U||}{Q_{\At D}}$
take $(x,q_{\oc \At A}) \in ||U||$ to 
(the necessarily unique) $q_{\At D}$ such that
$(x,(q_{\oc \At A},q_{\At D})) \in T$.
We have:
\[
\Par\big( \ND(\At A), \hat \col_{\ND(\At A)},n_{\ND(\At A)},T \big)
\quad\liff\quad
\Par[\Lt]^{||U||}\big(\col_{\At D},n_{\At D},\tilde U \big)
\]

\noindent
It then follows from Lemma~\ref{lem:sim:omega:trace}
that we are done if we show that
$||U||$ is the set of all $(x,q_{\At A}) \in \PP V$
for some $\funto{V}{\G(\oc \At A)}{\two}$
such that:
\[
\Play(\strat,\iota_\oc,V)
\]
But this is immediate from Comprehension for Product Types (Theorem~\ref{thm:funto:ca})
by letting $V$
be the union of $||U||$ with the set of
all $(x,\strat(x,q_{\oc \At A}))$ for $(x,q_{\oc \At A}) \in ||U||$.
\end{proof}

%However,
When proving that
$\Lang(\ND(\At A)) \sle \Lang(\At A)$
in Proposition~\ref{prop:sim:par:der} below,
in order to apply
Proposition~\ref{prop:sim:sim:der},
we have to extract a $\Prop$-strategy on $\G(\oc \At A,F)$
from a $\Prop$-strategy on $\G(\ND(\At A),F)$.
But $\ND(\At A)$ has more states than $\oc\At A$,
so we have to resort to Corollary~\ref{cor:aut:nd:unique},
stating that in plays of strategies on \emph{non-deterministic automata},
states are uniquely determined from tree positions.

%%%%%%%%%%%%%%%%%%%%%%%%%%%%%%%%%%%%%%%%%%%%%%%%%%%%%%%%%%%%%%%%%%%%%%%%%%%
\begin{prop}
\label{prop:sim:par:der}
%%%%%%%%%%%%%%%%%%%%%%%%%%%%%%%%%%%%%%%%%%%%%%%%%%%%%%%%%%%%%%%%%%%%%%%%%%%
Fix an $\HF$-closed automaton $\At A : \Sigma$ and consider %the automaton
$\ND(\At A) : \Sigma$ as defined above.
Then $\FSOD$ proves that for all $F:\Sigma$,
if $\ND(\At A)$ accepts $F$ then $\At A$ accepts $F$.
\end{prop}

%%%%%%%%%%%%%%%%%%%%%%%%%%%%%%%%%%%%%%%%%%%%%%%%%%%%%%%%%%%%%%%%%%%%%%%%%%%
\begin{proof}
%%%%%%%%%%%%%%%%%%%%%%%%%%%%%%%%%%%%%%%%%%%%%%%%%%%%%%%%%%%%%%%%%%%%%%%%%%%
Thanks to Theorem~\ref{thm:sim:sim:cor},
%Proposition~\ref{prop:sim:sim:der},
%we just have to show:
we are done if we show
$\oc\At A$ accepts $F$
whenever
$\ND(\At A)$ accepts $F$.
%\begin{itemize}
%\item
%If $\ND(\At A)$ accepts $F$ then $\oc\At A$ accepts $F$.
%\end{itemize}

\noindent
Let
$\funto{\stratbis}{\G(\ND(\At A),F)_\Prop}{\Opp_\G(\ND(\At A),F)}$
be the winning $\Prop$-strategy in $\G(\ND(\At A),F)$.
We are going to define a winning strategy
$\funto{\strat}{\G(\oc \At A,F)_\Prop}{\Opp_\G(\oc \At A,F)}$ 
in $\G(\oc \At A,F)$.
Note that $\ND(\At A)$ has more states than $\oc \At A$
and that,
\[
\funto{\stratbis}{\univ \times (Q_{\oc \At A} \times Q_{\At D})}{\two}
\]
whereas we need to define:
\[
\funto{\strat}{\univ \times Q_{\oc \At A}}{\two}
\]

\noindent
As mentioned, we resort to 
Corollary~\ref{cor:aut:nd:unique}.
The strategy $\strat$ is defined by $\HF$-Bounded Choice for
Functions (\S\ref{sec:ax:choice}) as follows.
Let $(x,S) \in \univ \times Q_{\oc \At A}$.
\begin{itemize}
\item
Assume that there is a play $U$ of $\stratbis$
such that
\[
(\exists q_{\At D} \in Q_{\At D})
\Big( (x,(S,q_{\At D})) \in U \Big)
\]

\noindent
Then it follows from 
Corollary~\ref{cor:aut:nd:unique}
there is a unique $q_{\At D}$ such that
\[
(\exists U)
\Big(
\Play(\stratbis ,\, \iota_\ND ,\, U)
~~\land~~
(x,(S,q_{\At D})) \in U
\Big)
\]

\noindent
In this case, we let
$\strat(x,S)$ be $(S,\Conj_{\oc \At A})$ where,
by Comprehension for $\HF$-Sets (Remark~\ref{rem:hfchoice}),
$\Conj_{\oc \At A}$ is the set of all $(d,S'_d)$
such that there is some $q'_{\At D} \in Q_{\At D}$
with $(d,(S'_d,q'_{\At D})) \in \Conj_{\ND(\At A)}$
for $((S,q_{\At D}),\Conj_{\ND(\At A)}) = \stratbis(x,(S,q_{\At D}))$.

\item Otherwise, we let $\strat(x,S)$
be $(S,\Conj_{\oc \At A})$ where,
by Comprehension for $\HF$-Sets (Remark~\ref{rem:hfchoice}),
$\Conj_{\oc \At A}$ is the set of all $(d,S'_d)$
such that there is some $q_{\At D} \in Q_{\At D}$
with $(d,(S'_d,q_{\At D})) \in \Conj_{\ND(\At A)}$
for $((S,\init q_{\At D}),\Conj_{\ND(\At A)}) = \stratbis(x,(S,\init q_{\At D}))$.
\end{itemize}

\noindent
We are now going to show that $\strat$ is winning.
To this end, consider an infinite play of $\strat$, that is some
$\funto{V}{\G(\oc \At A)}{\two}$
such that
\[
\Play(\strat,\, \iota_\oc,\, V)
\]

\noindent
We are going to define an infinite play of $\stratbis$, that is some
$\funto{U}{\G(\ND(\At A))}{\two}$
with
\[
\Play(\stratbis,\, \iota_\ND,\, U)
\]

\noindent
First, note that we are done if $U$ satisfies the following property:
\begin{equation}
\label{eq:sim:par:cor:der:playND}
(\forall (x,S) \in V) (\exists q_{\At D} \in Q_{\At D})
\big(
(x,(S,q_{\At D})) \in U
\big)
\end{equation}

\noindent
Indeed, by Comprehension (Theorem~\ref{thm:funto:ca}),
let $\funto{|V|}{\univ}{\two}$
be the set of all $x \in \univ$ such that
$(x,S) \in V$ for some (necessarily unique) $S \in Q_{\oc \At A}$.
Moreover, 
by $\HF$-Bounded Choice for Functions (\S\ref{sec:ax:choice}),
let $\funto{\tilde V}{|V|}{Q_{\oc \At A}}$
take $x \in |V|$ to the unique $S \in Q_{\oc \At A}$
such that $(x,S) \in V$.
By $\HF$-Bounded Choice for Functions (\S\ref{sec:ax:choice}),
let now $\funto{\tilde U}{|V|}{Q_{\At D}}$
take $x \in |V|$ to the unique $q_{\At D} \in Q_{\At D}$
such that $(x,(\tilde V(x),q_{\At D})) \in U$.
Since 
$\Par(\G(\ND(\At A)), \hat \col_{\ND(\At A)},n_{\ND(\At A)},U)$,
we have
$\Par[\Lt]^{|V|}(\col_{\At D},n_{\At D},\tilde U)$,
so that
$\Win[\Lt]^{|V|}_{\oc \At A}(\tilde V)$
and we conclude by Lemma~\ref{lem:sim:omega:trace}.

We now define an infinite play $U$ of $\stratbis$
satisfying~\eqref{eq:sim:par:cor:der:playND},
for which we appeal to the 
Recursion Theorem (Proposition~\ref{prop:games:rec}).
Let $\varphi(U,u)$ be a $\FSO$ formula
stating the disjunction of the following:
\begin{itemize}
\item  $u = \iota_\ND$; or
\item  $u = (x,\stratbis(x,(S,q_{\At D})))$ with
$(x,(S,q_{\At D})) \in U$; or
\item  $u = (\Succ_d(x),(S'_d,q'_{\At D}))$,
where
\begin{itemize}
\item $(\Succ_d(x),S'_d) \in V$, and
\item for some $q_{\At D} \in Q_{\At D}$
and some $S \in Q_{\oc \At A}$ such that 
$(x,S) \in V$ and $(x,(S,q_{\At D})) \in U$,
we have
$q'_{\At D} =\trans_{\At D}(q_{\At D},S)$.
\end{itemize}
\end{itemize}

\noindent
By the Recursion Theorem (Proposition~\ref{prop:games:rec}) we let
$\funto{U}{\G(\ND(\At A))}{\two}$ be unique such that:
\[
\big( \forall u \in \G(\ND(\At A)) \big)
\Big[
u \in U ~~\liff~~ \varphi(U,u)
\Big]
\]

\noindent
We need to show~\eqref{eq:sim:par:cor:der:playND}
and:
\[
\Play(\stratbis,\, \iota_\ND,\, U)
\]

\noindent
We first show
that $U$ is a play of $\stratbis$.
Since $\iota_\ND \in U$,
by Proposition~\ref{prop:games:edges:lin}
%by Lemma~\ref{lem:games:infplay}
it suffices to show:
\[
\left\{
\begin{array}{c l}
& (\forall u \in U) \big(
  \iota_\ND \edge{*}{\stratbis} u
  \big)
\\
  \land
& (\forall u \in U) (\exists ! v \in U)
  \big(
  u \edge{}{\stratbis} v
  \big)
\\
  \land
& (\forall v \in U)
  \bigg(
     v\neq \iota_\ND ~~\limp~~
     (\exists u \in U)
     \big[
     u \edge{}{\stratbis} v
     \big]
  \bigg)
\end{array}
\right.
\]

\noindent
We proceed similarly to the proof of 
Proposition~\ref{prop:sim:sim:der}.
First, we prove:

%%%%%%%%%%%%%%%%%%%%%%%%%%%%%%%%%%%%%%%%%%%%%%%%%%%%%%%%%%%%%%%%%%%%%%%%%%%
\begin{subclm}
\label{clm:sim:nd:cor:der:pred}
%%%%%%%%%%%%%%%%%%%%%%%%%%%%%%%%%%%%%%%%%%%%%%%%%%%%%%%%%%%%%%%%%%%%%%%%%%%
\[
  (\forall v \in U) \left(
     v\neq \iota_\ND ~\limp~
     (\exists u \in U) \left[
     u \edge{}{\stratbis} v
     \right]
  \right)
\]
\end{subclm}

\begin{subproof}[Proof of Claim \thesubclm]
The result directly follows from the definition of $\varphi$
and the definition of 
$\G(\ND(\At A))\restr\{\stratbis\}$ (Definition~\ref{def:games:strat}).
If $v = (x,\stratbis(x,(S,q_{\At D})))$, with $(x,(S,q_{\At D})) \in U$, then the result
is trivial.
Otherwise, we have $v = (\Succ_d(x),(S'_d,q'_{\At D}))$,
and there is $(x,(S,q_{\At D})) \in U$ such that
$(x,S) \in V$ and $q'_{\At D} = \trans_{\At D}(q_{\At D},S)$.
Note that $(x,(S,q_{\At D})) \in U$ implies
$(x,\stratbis(x,(S,q_{\At D}))) \in U$,
and similarly, that $(x,S) \in V$ implies $(s,\strat(x,S)) \in V$.
By definition of $\strat$, we have $\strat(x,S) = (S,\Conj_{\oc \At A})$
where $\Conj_{\oc \At A}$ is the set of all $(\tilde d, S'_{\tilde d})$
such that $(\tilde d,(S'_{\tilde d},q'_{\At D})) \in \Conj_{\ND(\At A)}$,
where 
$((S,q_{\At D}),\Conj_{\ND(\At A)}) = \stratbis(x,(S,q_{\At D}))$.
But then we are done since we indeed have:
\[
\big( x ~,~ \stratbis(x,(S,q_{\At D})) \big)
~~\edge{}{\stratbis}~~
\big( \Succ_d(x) ~, (S'_d,q'_{\At D}) \big)
\tag*{\qedhere}
\]
\end{subproof}

Now we prove:

%%%%%%%%%%%%%%%%%%%%%%%%%%%%%%%%%%%%%%%%%%%%%%%%%%%%%%%%%%%%%%%%%%%%%%%%%%%
\begin{subclm}
%%%%%%%%%%%%%%%%%%%%%%%%%%%%%%%%%%%%%%%%%%%%%%%%%%%%%%%%%%%%%%%%%%%%%%%%%%%
\[
  (\forall u \in U) \left(
  \iota_\ND \edge{*}{\stratbis} u
  \right)
\]
\end{subclm}

\begin{subproof}[Proof of Claim \thesubclm]
We proceed by $\glt$-Induction
(Theorem~\ref{thm:games:ind:pos}).
So let $u \in U$ s.t.\@
$\iota_\ND \edge{*}{\stratbis} v$
for all $v \glt u$ with $v \in U$.
The result is trivial if $u = \iota_\ND$.
Otherwise, by Claim~\ref{clm:sim:nd:cor:der:pred},
there is $v \in U$ such that $v \edge{}{\stratbis} u$.
But $v \glt u$ by Proposition~\ref{prop:games:edges},
so we have 
$\iota_\ND \edge{*}{\stratbis} v$
by induction hypothesis and we conclude by Proposition~\ref{prop:games:edges},
again.
\end{subproof}

It remains to show
\begin{equation}
\label{eq:sim:nd:cor:der:succ}
  (\forall u \in U) (\exists ! v \in U) \left(
  u \edge{}{\stratbis} v
  \right)
\end{equation}

\noindent
We first prove:

%%%%%%%%%%%%%%%%%%%%%%%%%%%%%%%%%%%%%%%%%%%%%%%%%%%%%%%%%%%%%%%%%%%%%%%%%%%
\begin{subclm}
\label{clm:sim:nd:cor:der:Vproj}
%%%%%%%%%%%%%%%%%%%%%%%%%%%%%%%%%%%%%%%%%%%%%%%%%%%%%%%%%%%%%%%%%%%%%%%%%%%
\[
\big( \forall (x,(S,q_{\At D})) \in U \big)
\Big[
(x,S) \in V
\Big]
\]
\end{subclm}

\begin{subproof}[Proof of Claim \thesubclm]
The property follows from a case analysis according
to the following usual consequence of Induction
(see Proposition~\ref{prop:ax:tree}, \S\ref{sec:ax:tree}):
\[
(\forall x)
\left(x \Eq \Root ~~\lor~~ (\exists y) \bigdisj_{d \in \Dir} x \Eq \Succ_d(y)
\right)
\]

\noindent
In the case of $x \Eq \Root$, if $(x,(S,q_{\At D})) \in U$ then we must have
$S = \init q_{\oc \At A}$, so that $(x,S) \in V$.
Consider now the case of $x \Eq \Succ_d(y)$.
If $u= (x,(S,q_{\At D})) \in U$, then it follows from 
$\varphi(U,u)$ that we have $(x,S) \in V$ and we are done.
\end{subproof}

We can now prove~\eqref{eq:sim:nd:cor:der:succ}.

\begin{subproof}[Proof of~\eqref{eq:sim:nd:cor:der:succ}]
If $u = (x,(S,q_{\At D}))$, then $(x,\stratbis(x,(S,q_{\At D})))$
is the unique successor of $u$ in $U$.
Assume $u = (x,\stratbis(x,(S,q_{\At D})))$.
It then follows from Claim~\ref{clm:sim:nd:cor:der:pred}
that $(x,(S,q_{\At D})) \in U$,
and by Claim~\ref{clm:sim:nd:cor:der:Vproj}
we also get $(x,S) \in V$.
Since $(x,S) \in V$, we have $(\Succ_d(x),S'_d) \in V$
for some unique $d \in \Dir$ and $S'_d \in Q_{\oc \At A}$.
It follows from $\varphi(U,u)$
that $v= (\Succ_d(x),(S'_d,q'_{\At D})) \in U$,
where $q'_{\At D} = \trans_{\At D}(q_{\At D},S)$.
It remains to show that $v$ is unique such :
\[
u \edge{}{\stratbis} v
\]
Uniqueness follows from $\varphi(U,u)$ and the fact that $V$ is a play,
so it remains to show
$u \edge{}{\stratbis} v$.
Note that $(x,(S,q_{\At D})) \in U$ implies
$(x,\stratbis(x,(S,q_{\At D}))) \in U$,
and similarly, that $(x,S) \in V$ implies $(s,\strat(x,S)) \in V$.
By definition of $\strat$, we have $\strat(x,S) = (S,\Conj_{\oc \At A})$
where $\Conj_{\oc \At A}$ is the set of all $(\tilde d, S'_{\tilde d})$
such that $(\tilde d,(S'_{\tilde d},q'_{\At D})) \in \Conj_{\ND(\At A)}$,
where 
$((S,q_{\At D}),\Conj_{\ND(\At A)}) = \stratbis(x,(S,q_{\At D}))$.
This finishes the proof since we indeed have:
\[
\big(x ~,~ \stratbis(x,(S,q_{\At D})) \big)
~~\edge{}{\stratbis}~~
\big( \Succ_d(x) ~,~ (S'_d,q'_{\At D}) \big)
\qedhere
\]
\end{subproof}

Finally, we prove~\eqref{eq:sim:par:cor:der:playND},
that is:
\[
\big(\forall (x,S) \in V \big) \big(\exists q_{\At D} \in Q_{\At D} \big)
\Big[
(x,(S,q_{\At D})) \in U
\Big]
\]

\begin{subproof}[Proof of~\eqref{eq:sim:par:cor:der:playND}]
Using the Induction Axiom of $\FSO$
(\S\ref{sec:ax:ind}), we show
\[
(\forall x) (\forall S \in Q_{\oc \At A})
\Big((x,S) \in V
~~\limp~~
(\exists q_{\At D} \in Q_{\At D})
  \big( (x,(S,q_{\At D})) \in U \big)
\Big)
\]

\noindent
For the base case $x \Eq \Root$,
if $(x,S) \in V$ then we must have $S = \init q_{\oc \At A}$,
and we indeed obtain $(x,(S,\init q_{\At D})) \in U$.
Assume now the property for $x$, and consider some $d \in \Dir$
and $S'_d$ such that $(\Succ_d(x),S'_d) \in V$.
Since $V$ is a play, it follows from
the Predecessor Lemma~\ref{lem:games:predplays} for Infinite Plays
(applied twice)
that $(x,S) \in V$ for some $S \in Q_{\oc \At A}$.
It follows by induction hypothesis that $(x,(S,q_{\At D})) \in U$
for some $q \in Q_{\At D}$.
But now, taking $q'_{\At D} = \trans_{\At D}(q_{\At D},S)$,
we have
$(\Succ_d(x),(S'_d,q'_{\At D})) \in U$ and we are done.
\end{subproof}

This concludes the proof of Proposition~\ref{prop:sim:par:der}.
%%%%%%%%%%%%%%%%%%%%%%%%%%%%%%%%%%%%%%%%%%%%%%%%%%%%%%%%%%%%%%%%%%%%%%%%%%%
\end{proof}
%%%%%%%%%%%%%%%%%%%%%%%%%%%%%%%%%%%%%%%%%%%%%%%%%%%%%%%%%%%%%%%%%%%%%%%%%%%

%%% Local Variables:
%%% mode: latex
%%% TeX-master: "main.tex"
%%% End:

%%%%%%%%%%%%%%%%%%%%%%%%%%%%%%%%%%%%%%%%%%%%%%%%%%%%%%%%%%%%%%%%%%%%%%%%%%%
\section{Conclusion}
%%%%%%%%%%%%%%%%%%%%%%%%%%%%%%%%%%%%%%%%%%%%%%%%%%%%%%%%%%%%%%%%%%%%%%%%%%%

\noindent
In this paper, we proposed for each non-empty (hereditarily) finite
set $\Dir$ the theory $\FSOD$ of
\emph{Functional (Monadic) Second-Order Logic}
on the full (infinite) $\Dir$-ary tree.
The theory $\FSOD$ (henceforth $\FSO$) is 
a uniform extension of $\MSO$ on the full $\Dir$-ary tree
with hereditarily finite sets.
We formalized in $\FSO$ a basic theory of (alternating) tree automata
and (acceptance) games.
This allowed us, in the theory of $\FSO$
augmented with an axiom $(\PosDet)$ of positional determinacy of parity games,
to formalize a translation of $\MSO$-formulae to automata
adapted from~\cite{walukiewicz02tcs}.
We then deduced the completeness of $\FSO + (\PosDet)$
thanks to a variant of the Büchi-Landweber Theorem~\cite{bl69tams},
stating that $\MSO$ decides winning for (definable) games of finite graphs
(and obtained thanks to the completeness of $\MSO$ over
$\omega$-words~\cite{siefkes70lnm}).
By naive proof enumeration, this gives a proof
of Rabin's Tree Theorem~\cite{rabin69tams},
the decidability of $\MSO$ on infinite trees.
Moreover, since the formal theory $\FSO$ is conservative
(\wrt\@ the faithful translation $(-)^\circ : \MSO \to \FSO$)
over a natural set of axioms for $\MSO$,
we also get a complete axiomatization of $\MSO$ on infinite trees, namely $\MSOD + \MI{\PosDet}$ (cf.~Definition~\ref{def:posdet:mso}, \S\ref{sec:posdet:mso}).

\subsection{Proof theoretic strength of complementation}
The present paper does not discuss proof theoretic strength.
In the context of second-order arithmetic
(in the sense of~\cite{simpson10book}), it is known that 
complementation of tree automata is between $\Pi^1_3$ and
$\Delta^1_3$-comprehension~\cite{km16lics}.
As far as only games are concerned
(as opposed to proving the correctness of 
an internal function for complementation), only $\Pi^1_2$-comprehension
is required for the positional determinacy of (each level of)
parity games~\cite[Lemma 4.6]{km15arxiv}.

\subsection{Clarifying the status of $\MI{\PosDet}$}
%An outstanding
A problem arising from this work is whether the axiom schema $\MI{\PosDet}$ is indeed independent of $\MSO$. The latter may be seen as the monadic fragment of $\PAS$ (over the appropriate language) and, as we have mentioned, is complete when restricted to infinite words.
While it might therefore be natural to suspect that $\MSO$ is already complete without $\MI{\PosDet}$, we point out that the axiomatization of Weak $\MSO$
over infinite trees given in~\cite{siefkes78ijm} also augments the natural fragment of Peano arithmetic by an axiom of induction over finite trees.
As we mentioned in the Introduction, the completeness of $\MSOD$ was erroneously claimed in the preliminary version of this work \cite{dr15lics}.

%%%%%%%%%%%%%%%%%%%%%%%%%%%%%%%%%%%%%%%%%%%%%%%%%%%%%%%%%%%%%%%%%%%%%%%%%%%%
%\subsection{Further Work.}
%%%%%%%%%%%%%%%%%%%%%%%%%%%%%%%%%%%%%%%%%%%%%%%%%%%%%%%%%%%%%%%%%%%%%%%%%%%%

\subsection{On the notion of proof for $\MSO$}
One outcome of this work is that our complete deduction
system for $\MSO$
gives a new decision algorithm.
Of course, the naive decision algorithm by proof enumeration is
not very sophisticated,
and it is worth
restating that its correctness is itself driven by the usual automata-theoretic argument.
Such an algorithm, nonetheless, makes no mention of automata and so can be
adapted and improved purely in the setting of proof theory.
In this sense, the algorithm is the first of its kind:
a decision procedure for $\MSO$ on infinite trees that remains internal to the language, rather than requiring intermediate translations to automata.

%%%%%%%%%%%%%%%%%%%%%%%%%%%%%%%%%%%%%%%%%%%%%%%%%%%%%%%%%%%%%%%%%%%%%%%%%%%
\renewcommand\fntext{To our knowledge, the \textsc{Mona} tool
(\url{https://www.brics.dk/mona/}) only handles \emph{Weak} $\MSO$.}
%%%%%%%%%%%%%%%%%%%%%%%%%%%%%%%%%%%%%%%%%%%%%%%%%%%%%%%%%%%%%%%%%%%%%%%%%%%
A basic motivation for such algorithms
is that, even if Rabin's Tree Theorem proves the existence
of decision procedures for $\MSO$ on infinite trees, there is 
(as far as we know) no working implementation of such procedures.\fn\@
%The reason is that all known translations from $\MSO$ formulae to tree automata
%involve at some stage the determinization of automata on $\omega$-words
%(McNaughton's Theorem~\cite{mcnaughton66ic}),
%which is believed not to be amenable to tractable implementation
%%which has not been yet been amened to tractable implementations
%(see \eg~\cite{kv05focs}).
%
Our axiomatization instead allows the targeting of
(semi) automatic approaches, for instance based on proof assistants.
%, in which
%proof-search procedures can be combined with user proof-steps
%and with automata procedures for sufficiently simple subgoals.
As we mentioned in the Introduction, our axiomatization is polynomial-time recognizable and so indeed yields a meaningful notion of `proof certificate': a proof of a theorem may be easily checked, without having to reprove the theorem again.

\subsection{Constructive systems and proof interpretations.}
%%%%%%%%%%%%%%%%%%%%%%%%%%%%%%%%%%%%%%%%%%%%%%%%%%%%%%%%%%%%%%%%%%%%%%%%%%%%
A further direction of research is to look for constructive interpretations of $\MSO$.
In the case of $\omega$-words, preliminary steps were made
in~\cite{pr19lmcs}.
The general idea is to proceed along the following steps:
\begin{enumerate}
\item
Determine the relevant computational information one should be able to extract
from constructive proofs.
\full{In the case of $\MSO$ on $\omega$-words, the 
approach taken in~\cite{pr17fscd,pr18lics,pr19fossacs}
(and specializing~\cite{riba20mscs})
was to consider the provably total \emph{causal}
(or 1-Lipschitz) functions of $\MSO$.}

\item
Devise constructive variants of $\MSO$
(together with suitable proof interpretations),
which are correct and complete
for the chosen class of functions
\wrt\@ to their provable $\forall\exists$-sentences.
%To this end, a realizability model has already been proposed in~\cite{riba20mscs},
%in which the underlying logic is not only constructive but also
%\emph{linear} (in the sense of~\cite{girard87tcs}).
%\full{So far the model of~\cite{riba20mscs} has been applied
%to $\omega$-words~\cite{pr18lics,pr19fossacs},
%leading to a linear version of $\MSO$
%refining the intuitionistic variant of~\cite{pr17fscd}.}
\end{enumerate}

\noindent
A realizability model for $\MSO$ has been proposed in~\cite{riba20mscs},
in which the underlying logic is not only constructive but also
\emph{linear} (in the sense of~\cite{girard87tcs}).
\short{Of course, similar approaches may also be considered in more
traditional settings for constructive interpretations of
proofs~\cite{tvd88cim,kohlenbach08book}.}

\full{Of course, similar approaches may also be considered in more
traditional settings for constructive interpretations of
proofs~\cite{tvd88cim,kohlenbach08book}.
In particular, it is not clear (at least to us)
what usual computational interpretations of Comprehension,
following either Girard's \emph{System F}~\cite{girard72phd} or
Spector's \emph{bar-recursion} (see \eg~\cite{tvd88cim,kohlenbach08book}),
could say in the context of $\MSO$.
It is not yet clear
what in this context should be the correct analogue in $\MSO$
of the quantifier free formulae of arithmetic.
Regarding the quantifier-free formulae
as those formulae with trivial realizers
\wrt\@ usual proof interpretations, 
the model of~\cite{riba20mscs} suggests
in the case of $\omega$-word
that for languages based on linear logic,
correct analogues of quantifier formulae are formulae
which are both negative and positive.
These formulae, called \emph{deterministic} in~\cite{pr18lics,pr19fossacs}
may contain unbounded quantifiers, but these must be guarded
by exponential modalities of suitable polarity
($\oc$ for $\forall$ and $\wn$ for $\exists$).}

%\note{In the same vein, a complete axiomatization provides a means to
%\emph{communicate} proofs 
%that can be checked efficiently, rather than having to redecide formulae.
%Such a tool might be useful for \emph{cyclic} theorem provers,
%where correctness criteria boil down to checking the inclusion between B\"uchi automata,
%a priori a very complex procedure.}

%%% Local Variables:
%%% mode: latex
%%% TeX-master: "main.tex"
%%% End:

\bibliographystyle{alpha}
\bibliography{bibliographie}

%\newpage
%\tableofcontents

\end{document}